\documentclass[12pt,reqno]{amsart}

\usepackage{mathabx}
\usepackage{bbold}
\usepackage{mathrsfs}         
\usepackage{amssymb}
\usepackage{graphicx, amsmath}
\usepackage{amsfonts}
\usepackage{dsfont}
\usepackage{amssymb}
\usepackage{fancyhdr}
\usepackage{a4wide}
\usepackage{xcolor}
\usepackage{enumerate} 
 \usepackage{enumitem}
\usepackage[colorlinks]{hyperref}
\usepackage{mathtools}\usepackage{mathtools}

\numberwithin{equation}{section}
\usepackage[margin=3cm]{geometry}

\newcommand{\vp}{\varphi}
\newcommand{\ve}{\varepsilon}

\renewcommand{\O}{\mathcal{O}}

\newcommand{\N}{\mathcal{N}}  
\newcommand{\C}{\mathbb{C}} 
\newcommand{\Z}{\mathbb{Z}} 
\newcommand{\R}{\mathbb{R}} 
 \newcommand{\h}{\mathfrak{h}}
\renewcommand{\S}{\mathcal{S}}
\newcommand{\po}{  \N_\S  }

\newcommand{\B}{\mathcal B}

\newcommand{\F}{\mathscr{F}}
 \newcommand{\calR}{ \mathcal{R} }

\newcommand{\Ha}{ \mathcal{H}  }  

\newcommand{\supp}{\textnormal{supp}}
\newcommand{\1}{\mathds{1}}
\renewcommand{\i}{ \mathrm{ i }  }
 \renewcommand{\d}{\mathrm{d}}

\newcommand{\blu}{\color{black}}

\renewcommand{\t}[1]{\textnormal{#1}}

\newcommand{\<}{\left\langle}
\renewcommand{\>}{\right\rangle}

\renewcommand{\leq}{\leqslant}

\renewcommand{\geq}{\geqslant}

\newtheorem{condition}{Condition}[section]

\newtheorem{theorem}[condition]{Theorem} 
\newtheorem{proposition}[condition]{Proposition} 
\newtheorem{lemma}[condition]{Lemma}

\newtheorem{definition}[condition]{Definition} 

\theoremstyle{remark}

\theoremstyle{definition}

\newtheorem{remark}[condition]{Remark}

 \title[
 Quantum Boltzmann dynamics 
around the Fermi ball
 ]
 {Quantum Boltzmann dynamics 
 	and 
 	bosonized particle-hole interactions 
 	in fermion gases}

\author[Esteban C\'ardenas]{Esteban C\'ardenas}

\address[Esteban C\'ardenas]{Department of Mathematics,
	University of Texas at Austin,
	2515 Speedway,
	Austin TX, 78712, USA}
\email{eacardenas@utexas.edu}

\author[Thomas Chen]{Thomas Chen}

\address[Thomas Chen]{Department of Mathematics,
	University of Texas at Austin,
	2515 Speedway,
	Austin TX, 78712, USA}
\email{tc@math.utexas.edu}

\makeatletter

\begin{document}
 
 \begin{abstract}
In this paper, we study a   gas of $N \gg 1 $ weakly interacting 
fermions. 
We describe the  time evolution of 
states that are perturbations of 
      the Fermi ball, 
      and analyze   the dynamics
      in particle-hole variables. 
Our main result states that,
 for small values of   the coupling constant 
 and for appropriate initial data, 
  the  effective dynamics of the 
  momentum distribution  
 is   determined 
 by a discrete collision  operator of quantum  Boltzmann form.  
 \end{abstract}

\maketitle

{\hypersetup{linkcolor=black}
\tableofcontents}

\section{Introduction }
\label{section introduction}

The quantum Boltzmann equation  was first proposed 
 on phenomenological grounds 
by Nordheim \cite{Nordheim} and Uehling and  Uhlenbeck \cite{UehlingUhlenbeck}
as  a quantum-mechanical correction to the    classical Boltzmann equation, taking into account the statistics of the particles. 
For  a spatially homogeneous gas of fermions in three dimensions, 
it reads 
\begin{align} 
	\label{boltzmann}
& \partial_T F 
	  =   
	\int_{\R^9}
b (  p  \, p_*  p' p_* '   )
\Big(
F' F'_* (1 - F) (1 - F_*) - 
F F_* (1 - F' ) (1 - F_*')
\Big) 
	\d p_* \d p ' \d p_* '    
\end{align} 
Here,   the unknown $F = F_T(p)$ is a  probability density function  in momentum space,
and we  use the short-hand notations   $F = F(p)$, $F'  = F (p')$, $F_* = F(p_*)$
and $F_* ' = F(p_* ' )$. 
 The function $b (p p_* p' p_*')$
is the  scattering amplitude and its precise form
depends on the scaling limit under which \eqref{boltzmann} is obtained.
In this article, we are mostly interested  in a 
regime of weak interactions.
Here,  the amplitude  is calculated quantum-mechanically in the Born approximation  to second order
and takes the form  (up to geometric constants)
\begin{equation}
	\label{def:b}
	b  (  p \,  p_* p' , p_* '  )
	=   
	\delta (p + p_* - p' - p'_*) \, 
	\delta \big( p^2 + p_*^2 - (p')^2  - (p_* ')^2	\big)  \, 
	\big|  \hat v( p  - p') - \hat v (p - p_*')\big|^2   
\end{equation}
where  $\hat v(k)$ is the Fourier transform of the microscopic 
interaction potential.  In other regimes -- for instance,  in the low-density limit described below -- two-body interactions
are not weak and  the scattering amplitude
  needs  to be computed taking into account   all orders in the potential. 
 

 \vspace{2mm}
  
Despite important mathematical efforts,   
the rigorous derivation of \eqref{boltzmann}  
from first principles 
remains an  important   open problem  in mathematical physics. 
This amounts to analyzing  the 
microscopic dynamics of
an interacting particle system governed by the Schr\"odinger equation,
and the standing conjecture 
states
that the equation
 \eqref{boltzmann}
emerges 
 in a suitable scaling limit.

\vspace{2mm}

{\blu  In this regard,  one of 
the main lines of research consists 
of studying  the \textit{kinetic scaling    limit}
of an interacting quantum gas. 
Here, 
one considers a 
particle system   with an interaction of strength 
$\lambda>0$
and rescales time by 
\begin{equation}
T  = \lambda^{ 2     }  t 
\end{equation}
where $t >0$ is the microscopic time scale. 
The  time  scale 
$T $ is   known as the \textit{kinetic time scale}, 
and it is  conjectured to be the smallest scale 
that produces $O(1)$ corrections to the dynamics. 
The rigorous analysis  of this phenomenon was initiated by Hugenholtz \cite{Hugenholtz} and Ho and Landau \cite{holandau}, for spatially homogeneous systems. 
See also \cite{ESY2004,Lukkarinen2009} for more results in this direction. 
In the space inhomogeneous setting, it was observed by Spohn in \cite{Spohn1994}
that   if one implements the following scaling for a system of $N$ quantum particles
\begin{equation}
	\label{weak}
	X = \ve x  \ , \qquad T = \ve t \ , \qquad \lambda = \ve^{1/2} 
	\qquad
	\t{and} 
	\qquad 
	N\ve^3=1 
\end{equation}
one should conjecturally obtain    the space inhomogeneous quantum Boltzmann equation.
The scaling \eqref{weak} is often called the weak coupling limit. 
Using   BBGKY hierarchy methods, the weak coupling limit    was further studied for quantum systems 
 in \cite{Bendetto2005} 
and more recently by 
\cite{ChenGuo2015,ChenHolmer2023}. 
 In addition to the weak coupling limit, there is the so-called   \textit{low density} limit.
It  differs from the weak coupling regime
  in that one sets $\lambda=1$ and $ N = \ve^{-2}$ and takes      
  $\ve \rightarrow  0$; in particular, the scattering amplitude entering the quantum Boltzmann equation \eqref{boltzmann} needs to be computed to all order in the potential. 
  See for instance \cite{Bendetto2006} for partial results. 
  One should note that this limit (and \textit{not} the weak coupling limit) is formally the same as   taking the  \textit{Boltzmann-Grad} limit. 
  The latter is considered in the   derivation of the classical Boltzmann equation \cite{Gallagher2012} and is now well understood even at arbitrarily long times \cite{Deng}. 
  Finally, for further reading,  we refer the reader to \cite{Bendetto2004,Bendetto2007,Bendetto2008,Pulvirenti2006}.

   \vspace{2mm}

To this day, the derivation of the
quantum Boltzmann equation 
in an interacting quantum system
remains an open problem. 
Indeed, previous results at kinetic time scales may be classified into the following two types. They are   either conditional (the solution is assumed to satisfy key properties, yet these remain unproven) or they correspond to truncation results (the partial series of a Duhamel expansion is fully analyzed, by the tail is not controlled).

\vspace{2mm}

%

In this article, we turn to the study of the following  question: } 

\vspace{1.5mm}
\noindent \textit{Is there a  scaling window 
for which 
the many-body fermionic dynamics
is, to leading order, governed by an operator 
of the form \eqref{boltzmann}, with complete error analysis? }

\vspace{1.5mm}

We refer to this question as the 
\textit{emergence} of quantum Boltzmann dynamics.  {\blu 
Let us note that questions in the same spirit have been investigated in \cite{Hott,Hott2}
for Bose gases, and also in \cite{Deng2021} for systems of waves.
}

\subsection{Description of main results}
Let us     informally 
describe   the main results of this paper, 
which  {\blu give an answer to the proposed question.} 

\vspace{2mm}

We consider the microscopic dynamics
of $N$ identical  {\blu spinless} fermions
on the torus of side length $ L>0 $
in dimension $d \geq 1$, 
which we denote by
$	 		 \Lambda = (\R/ L \Z)^d  \ . $
The states are $N$-body wave functions  belonging to
the Hilbert space 
\begin{equation}
	L^2_a ( \Lambda^N ) 
	 = 
	 \big\{ \Psi \in L^2 (\Lambda^N) : 
	 \Psi (x_{\sigma (1)}	 , \ldots, x_{ \sigma (N)	})
	  = 
	  (-1)^\sigma 
	  	 \Psi (x_{   1 }	 , \ldots, x_{ N 	} )  \   \forall \sigma \in S_N 
	 		 \big\}   . 
\end{equation} 
Here, $S_N$ is the group of permutations of $N$ elements; the antisymmetry of the functions $\Psi \in L_a^2 (\Lambda^N)$  characterizes  
the fermionic  statistics of the particles.
We then   study   the solutions of the  Schr\"odinger equation, 
written   in microscopic units as 
\begin{equation}
\label{schrodinger}
	i \partial_t \Psi 
	 = 
	 \bigg(	 \frac{1}{2}	  \,   \sum_{  i  } ( - \Delta_{x_i})	 + 
\frac{	 \lambda }{2 }
	 \sum_{  i\neq j } V(x_i - x_j )			 \, 	\bigg)
	  \Psi  \  , \qquad \Psi(0) = \Psi_0  \in 	L^2_a ( \Lambda^N )   \ . 
\end{equation}
Here,  $\lambda >0$ is the coupling constant of
the interaction, mediated by a    pair potential $V( x - y)$, 
and the initial datum is denoted by $\Psi_0$.
We   denote by
$	\Lambda^* \equiv  ( 	 \frac{2\pi}{L }\Z			 )^d $
the dual momentum lattice, 
and write for convenience
$\int_{\Lambda_*} \d p \equiv \frac{1}{|\Lambda|} \sum_{ p \in \Lambda^*}$.

\vspace{2mm}
 
We consider quasi-free, translation-invariant initial data $\Psi_0$
corresponding to states that are perturbations
 of  the \textit{Fermi ball}  
\begin{equation}
	\Psi_F \equiv \bigwedge_{   p  \in \B } e_p 
\qquad
\t{where}
\qquad 
		\B   \equiv  \{   p\in   ( 2\pi \Z / L  )^d : |p| \leq p_F 		\}     \ , 
\end{equation}
in the sense of Condition \ref{condition initial data}.  
Here,  we denote by 
\begin{equation}
\blu 
	e_p (x) =  
	|\Lambda|^{	 - \frac{1}{2}}
	e^{ i  p \cdot x }  \ , \qquad p \in \Lambda^* 
\end{equation}
the plane-wave basis,  and 
	  $p_F>0$ stands for   the \textit{Fermi momentum}.
We
	     tune $N>0$
 so that $ | \B|  =  N $
 and thus, consequently,  for fixed $L>0$
	  \begin{equation}
	  	p_F =  C_d \bigg( \frac{N}{|\Lambda|}	\bigg)^{ 1/d }  (1 + o(1)) \qquad N  \rightarrow  \infty  \  , 
	  \end{equation}
for some constant $C_d>0$.  
As is well-known, the Fermi ball 
$\Psi_F$ corresponds to the non-interacting ground state of the system. It is, therefore, 
an approximately stationary solution of the Schr\"odinger equation for small values of $\lambda>0$.  
As we explain below,
 we   study the dynamics of $\Psi (t)$ relative to the Fermi ball $\Psi_F$ in the particle-hole formalism.

\vspace{2mm}

The    main object  
of interest  in this article 
{\blu{is}}   the 
momentum distribution of the wave-function $\Psi (t)$.
We shall denote {\blu it} by $F_t(p)$, where $p $ belongs to the dual lattice 
$	\Lambda^* $, 
and {\blu satisfies} 
\begin{equation}
 \int_{\Lambda_*} F_t(p) \d p = N \qquad \t{and} \qquad 0 \leq F_t(p) \leq 1   \ . 
\end{equation}
For a precise definition of $F_t(p)$, see \eqref{momentum F} in the next section, where we introduce the Second Quantization formalism. 
In particular, we are interested in understanding the emergence of
quantum Boltzmann dynamics
for the time evolution of $F_t(p)$ 
for small values of the coupling constant $\lambda>0$, 
large values of the  time variable $t \in \R $, and at fixed length $L>0$.

\vspace{2mm}

In this regard,  our main result is contained in Theorem \ref{thm3} 
and establishes the emergence  
of the quantum Boltzmann operator \eqref{boltzmann}
for  the dynamics of $F_t(p)$, 
in a particular scaling window.
We   consider the following scaling 
for $d = 3 $
\begin{equation}
	\label{scaling0}
	\lambda 
	=  
	 N^{ - \frac{3}{2} - \delta }
	 \qquad 
	t = N^{ \frac{1}{6} +   \frac{\delta}{2 } } \,  T    \ , 
	\qquad
\t{and } 
\ L>0 \ 
 \t{ fixed} 
\end{equation}
where $\delta\in(0,1)$ is small enough, and 
$F_0$ is a small enough 
perturbation of the Fermi ball.  
The coupling in the 
scaling \eqref{scaling0} is chosen
so that higher-order
 many-body effects can be controlled.
 At the same time, the time-scale in \eqref{scaling0}
 is chosen  to be large enough 
 so that completed collisions can be detected. 

\vspace{1mm}

In the scaling \eqref{scaling0}, 
we prove the validity   of the following  asymptotic expansion 
\begin{equation}
	\label{F}
	F_t = F_0 +   (\lambda t )^2 \, 
\mathscr C  [F_0 ] 	  + \cdots   
\end{equation}
 with respect to an appropriate metric, 
 with rigorous   control of the error terms. 
The 
 operator  in    \eqref{F} 
is a 
	discrete version of the   right-hand side of 
	\eqref{boltzmann}, given by 
	\begin{align}
\mathscr C  [F] 
	\equiv  
\int_{ ( \Lambda_*)^3 }
     b  (  p  \, p_*  p' p_* '   )
     \Big(
     F' F'_* (1 - F) (1 - F_*) - 
     F F_* (1 - F' ) (1 - F_*')
     \Big) 
     \d p_* \d p ' \d p_* '    
      \end{align} 
where now  
the scattering
{\blu amplitude} 
 contains a \textit{discrete}
 energy-conservation delta function
      \begin{align} 
\label{kernel}
           b(p p_* & p' p'_* )  \\
            &= 
            \frac{\pi}{ ( \pi/2)  }
             \,  \delta (p + p_* - p' - p'_*)
             \delta_{\Z } 
 \big( p^2 + p_*^2 - (p')^2  - (p_* ')^2	\big) 
      | \widehat V ( p - p ' )			-		\widehat V ( p - p_* ' )				|^2 
      \nonumber 
	\end{align}
where    we denote
$\delta_{\Z} (x) \equiv  \delta_{x,0} $
and 
$
  	\delta(p -q )  =   |\Lambda|^{-1 }   \delta_{p,q}$.
  	 In particular,  
  	we note that 
  	the  leading order term grows \textit{quadratically} with time, 
  	and includes an additional   factor $   \pi/2   $
  	relative to \eqref{def:b}. 
See  	 Remark \ref{remark:t} below for a discussion. 
Here we have a  factor $\pi$ rather than the  conventional  $4\pi$ thanks
to the $\frac{\lambda}{2}$ in \eqref{schrodinger}.

\subsection{Discussion}
 
\begin{remark}
Throughout   this paper, 
	the side length of the box $L$ is kept fixed.
	 Thus,    we obtain a 
	discrete operator   $\mathscr{C}$
rather  than  its  continuous counterpart   \eqref{boltzmann}. 
Notice that this is a  {high-density regime}, and the Fermi momentum $p_F>0$ becomes large. 
In contrast,  the weak-coupling limit 
considers systems with constant density. 
\end{remark}

\begin{remark}\label{remark:t}
Let us comment on the quadratic time dependence
of the
leading order term of \eqref{F}.
The
relevant time-energy delta function  
that appears  naturally 
in collision operators is    
\begin{eqnarray} 
	\delta_t (E) \equiv   t \delta_{1}(tE)  
	\qquad
	\t{ where }
	\qquad 
	\delta_1(E) 
	\equiv 
	\frac{2}{\pi }\frac{\sin^2\big( 	\frac{E}{2}\big)}{E^2} 
\end{eqnarray}
for $t, E \in \R$.  
It is well-known
that after the limit 
$L \rightarrow \infty$,  the function $\lim_{t \rightarrow \infty}	\delta _t(E)$
becomes a   Dirac delta $\delta(E)$. 
Indeed, for illustration purposes,   consider a smooth 
 test function $\vp : \R \rightarrow \R $
and note  that 
 \begin{equation}
 	\label{macro}
\lim_{t \rightarrow \infty }
\lim_{L \rightarrow \infty }
\frac{1}{L} \sum_{ n \in 		 \frac{1}{L }	\Z} 
\delta_t (n )	 \vp (n) 
 = \vp (0) \ . 
 \end{equation}
In this paper, we consider fixed $L>0$
and, consequently, 
	the same quantity behaves like $O(t)$ for large times.
	Taking $L=1$ for illustration yields  
	 \begin{equation}
	 	\label{moll}
 \sum_{ n \in 	 \Z} 
		\delta_t (n )	 \vp (n) 
		= 		
\frac{t}{\pi /2 }
		 \  \vp(0) + 
		O(1/t )
		\qquad
		\t{as}
		\qquad
		t\rightarrow \infty 
	\end{equation}
which we interpret as a lattice effect.
	See for reference  Lemma \ref{lemma delta}. 
 	As a consequence of \eqref{moll}, 
	the leading order term
	in \eqref{F}  is $O(\lambda^2 t^2)$;
	much larger than the expected $O (\lambda^2 t )$ in the macroscopic limit \eqref{macro}. 
	\end{remark}

\begin{remark}
	We consider initial data $F_0$
	for which $\mathscr C[F_0 ]$ is of order $N^{1/3}$
	in an appropriate norm. 
	Thus, the first-order correction
	to the momentum distribution 
	in \eqref{F}
	in the scaling \eqref{scaling0}
	is of order 
	\begin{equation}
		\label{correction}
	O ( \lambda^2 t^2 N^{1/3} )  = O ( N^{ -7/3 - \delta})		
	\end{equation}
	While this correction is certainly small, 
	the leading order term in the expansion \eqref{F}
 is determined by the quantum Boltzmann operator 
 $\mathscr C [F_0]$, 
  with full error control over the subleading   terms.
	
\end{remark}

\subsection{Strategy}
The heart of our approach is 
the  study  of the time evolution
of       the momentum distribution $F_t(p)$ 
 relative to    the stationary Fermi ball. 
 It is then very convenient 
 to introduce the momentum distributions
 of $\Psi_F\in L_a^2(\Lambda^N)$ as well as its complement: 
 \begin{equation}
 	\label{chi definition}
 	\chi
 	\equiv
 	\1 ( |p |\leq p_F )
 	\qquad 
 	\t{ and }
 	\qquad 
 	\chi^\perp  \equiv 1 - \chi    \  .
 \end{equation}
We   
change variables
and work in   
particle-hole space.
Namely, we study the 
 the  momentum distribution $f_t(p)$  of excited particles and  empty holes (see Definition \ref{definition momentum distribution}) 
 which satisfies 
\begin{equation}
\label{F and f}
	F_t ( p) 
	= 
	\begin{cases}
		f _t(p) 	  \quad   &  |p| > p_F \\ 
		1 - f_t (p)     \quad   & | p | \leq p_F 
	\end{cases} \ . 
\end{equation} 
{\blu  Our attention is on     states  that  consist of  suitable  perturbations of  the  Fermi ball.
See Condition \ref{condition initial data} for the precise assumptions. }

 \vspace{2mm}
 
 The proof of the expansion \eqref{F} can then   be broken down into three steps.
 \begin{enumerate}[leftmargin=*]
 	\item  In Theorem \ref{theorem 1}, 
 	we study the dynamics of $f_t(p)$
 	in terms of the \textit{unconstrained parameters} of the theory
 	and find a general estimate in \eqref{thm 1 estimate} in dimension $d \geq 1 $. 
{\blu  	This means that we do not fix any functional relationship  (i.e., scaling) between the parameters $\lambda$, $N$, $\Lambda$, and $t$}. 
 	
 	\item In Theorem \ref{corollary},   we specialize this estimate to  the scaling \eqref{scaling0} in $d =3 $
 	and show that the remainder terms are smaller in an appropriate sense, and for well-chosen initial data.
 	 Here, extracting the  dependence of the collision operators with respect to  $t \gg 1 $ is essential.

 	\item In Theorem \ref{thm3}, we go back to the momentum distribution $F_t(p)$
 	and conclude the validity of the expansion \eqref{F}.  
 \end{enumerate}

 Let us now describe these steps in more detail.

 \vspace{2mm}

\textit{Step 1.} 
Here, our main goal is to  analyze the  effective dynamics of $f_t(p)$, 
which is driven by a particle-hole Hamiltonian $\h$ described 
in   Section \ref{section preliminaries}.  

\vspace{1mm}
In Theorem \ref{theorem 1}, 
we consider 
the following   asymptotic expansion  in terms of  the coupling $\lambda>0 $
and microscopic time $t \in \R $
\begin{equation}
\label{f eq}
	f_t   = f_0  + \lambda^2 t  \, B_t[f_0] + \lambda^2  t \, Q_t[f_0] + \cdots  \  
\end{equation}
and give  an explicit estimate  on  the remainder,  
in terms of the unconstrained parameters of the theory.    
In particular, our estimates  focus on the description of the dynamics of fermions  whose momenta are  \textit{away} from the Fermi surface.
This is encoded in the norm $\|	 \cdot \|_{\ell_m^{1*}}$ that we use to measure distances. 
Note that this is consistent with the fact that
particles and holes near the Fermi surface couple together to form \textit{bosonized} particle-hole pairs, 
and experience different dynamical behaviour.
This bosonization phenomenon was first explained by 
Sawada \cite{Sawada1957} and Sawada, Brueckner,  Fukuda, and Brout \cite{Sawadaetal1957}
in the 1950s. 
{\blu   
More recently,    rigorous formulations  have  been constructed to study  
 large  Fermi gases in the random phase approximation.
In \cite{Benedikter2020},  the authors introduced a rigorous   decomposition of  the Fermi surface
into  several small patches, and on each one a localized quasi-bosonic operator is defined.
 These techniques are then used to prove  an optimal upper bound for the correlation energy of a Fermi gas in the mean-field regime;
these methods were later refined  and improved in \cite{Benedikter2021,Benedikter2022,Benedikter 2023} to find the matching lower bound, study dynamical properties of the system, 
and include large interactions. 
A different bosonization approach  (not based on patch localization)
was developed  in  \cite{Christiansen2022} that  approximately diagonalizes the fermionic Hamiltonian and derives an effective quasi-bosonic Hamiltonian describing the correlation energy and the elementary excitations of the system. These methods were subsequently employed in 
     \cite{Christiansen2022-2} to study plasmon modes and in \cite{Christiansen2023}
 to derive a mean-field version of the Gell-mann--Brueckner formula. }

\vspace{2mm}

{\blu 
 The operator  $B_t$ is bilinear and is  given in Definition \ref{definition B}.
It describes     physical processes between fermions that are mediated 
by \textit{virtual}   {bosonized} particle-hole pairs near the Fermi surface. 
 Here, virtual refers to the fact that in these second-order processes,  both the initial and final number of particle-hole pairs is zero, yet they still participate      in intermediate interactions.
 These processes are allowed in view of the fact that the number of particle-hole pairs is not conserved.  
Let us note that under  our assumptions on the initial data and on the dynamics, 
we can guarantee that the Fermi surface is \textit{depleted} 
throughout  the time scale  under consideration; 
see Proposition \ref{prop N estimates 2} for a precise statement.  
This means that  all physical processes 
in which  the final number of  particle-hole pairs is non-zero  can be neglected. 
} 
On the other hand, the  operator $Q_t$ is given in Definition \ref{definition Q1}, and corresponds to  an energy-mollified
collisional operator  of quantum Boltzmann form, 
describing 
collisions of the form:  particle-particle, particle-hole, and hole-hole.

\vspace{2mm}

Both operators satisfy $ B_t  \big|_{t=0} = Q_t  \big|_{t=0}=0$.
Therefore, they do \textit{not} dominate over the remainders
for small times. 
In the next step, we consider longer time scales
and extract the leading order time dependence.

\vspace{2mm}

\textit{Step 2}.
In  Theorem \ref{corollary}, 
we consider     the dynamics    for longer time scales 
and for small values of the coupling. 
More precisely, we specialize to the
scaling \eqref{scaling0}
in three dimensions.
The scaling is chosen so that $\lambda $ is small enough to control error terms arising from Theorem \ref{theorem 1}, 
but $t$ is large enough to extract the leading order terms of $B_t$ and $Q_t$.
 In particular,  it is long enough  to  observe completed collisions. 

\vspace{1mm}

We    identify the leading order terms     of the operators
$B_t$  and $Q_t$  as $t   \rightarrow \infty  $, with rigorous error control. 
These are given by     limits ({\blu to be understood in a certain topology})
\begin{equation}
	\mathscr B  \equiv  	\lim_{t \rightarrow \infty } \lim_{\lambda \downarrow 0} \frac{1}{t}  B_t   
\qquad 
\t{and}
\qquad 
\mathscr Q  \equiv  	\lim_{t \rightarrow \infty } \lim_{\lambda \downarrow 0} \frac{1}{t}  Q _t   \ . 
\end{equation}
Observe that we take the factor $\frac{1}{t}$
in order to extract the leading order term,  in agreement with 
  \eqref{moll}.   In particular, after we take the limits $t \rightarrow \infty$, 
  it is possible to establish \textit{lower bounds}
  for the operators $\mathscr{ B}$
  and $\mathscr{Q}$.
  This then allows for comparison with the remainder terms obtained in Theorem \ref{theorem 1}.
For small enough perturbations, 
we  construct   initial data
such that in an $\ell^\infty$-norm: 
 \begin{equation}
 	 \mathscr{B}[f_0]  \, \sim  \, N^{1/3}
 	 \qquad
 	 \t{and}
 	 \qquad 
 	 \mathscr{Q} [f_0]	\,  \lesssim 	\, N^{1/6} \ . 
 \end{equation}

\vspace{2mm}

\textit{Step 3.} Finally, in Theorem  \ref{thm3}, 
 we use the relation  \eqref{F and f}
 to go back to the original momentum distribution function $F_t(p)$.
In particular, one must here compare $\mathscr{B}[f]$ and $\mathscr{C}[F]$.
We find that for appropriate perturbations,
 there holds
\begin{equation}
\label{C and B}  
	\mathscr{C}[F]  \ =  \ ( 1 - \chi)  \mathscr{B}[f]  \ - \      \chi   \mathscr{ B}[f]   
\end{equation}   
modulo a small error term, see e.g Proposition \ref{prop F}. 
 A combination of  \eqref{C and B} and Theorem \ref{corollary} allows us to finish the proof of \eqref{F}. 
 
 {\blu 
 From a conceptual level, \eqref{C and B} is a crucial observation.  
It states  that the  emergence of the     operator that drives the   quantum Boltzmann equation \eqref{boltzmann}
 for states $F$ near the Fermi ball, 
 is associated  with   an       operator $\mathscr{B}$  in   particle-hole space.
In particular, the operator $\mathscr B$ is 
 {\blu \textit{bilinear} in the variables $f$ and $1-f$}, 
 and may be regarded as a ``quadratic approximation" of the full operator $\mathscr C$. 
The physical meaning of the operator  $\mathscr B$   is associated  with the    interactions between particles   
with a distinguished boson field --  the field of bosonized particle-hole pairs near the Fermi surface. }

 \bigskip 
{\color{black} 
 Finally, we would like to recall that our results are valid only for 
 \textit{spatially homogeneous systems}.  
 This assumption is  physically relevant  and simplifies  the analysis considerably. 
Indeed, this leads to the consideration  of  the expectation  of  field operators  $a_p ^*a_p$ instead of 
 $a_{p}^* a_q$ for $q \neq p$. 
 This has the advantage of having a simpler representation in the interaction picture,
and has been significantly exploited in the past
 \cite{Hott,Hott2,ChenSasaki,ESY2004,Hugenholtz}.
 For spatially inhomogeneous systems,  partial results are available in
 \cite{Bendetto2005,Bendetto2006,Bendetto2007,Bendetto2008,ChenGuo2015,ChenHolmer2023}.
 Note that the methods differ vastly from those in this article, 
  and 
  the BBGKY  approach  seems   to be preferred over  field-theoretic methods.
  At the moment, it is unclear if the results in this article may be generalized  to spatially inhomogeneous systems.

   }

\vspace{2mm }

\subsection{Organization of this article }
In Section  \ref{section main results},  we state the main result of this article, 
and in Section \ref{section preliminaries},  we introduce the preliminaries that are needed to set up the proof. 
In Sections  \ref{toolbox1} and \ref{section number estimates} 
we introduce and develop the machinery 
that we use  in  our analysis. 
In Sections \ref{section TFF} and \ref{section TFB},   we show how the  operators $Q $ and $B $, 
respectively, emerge from the many-body dynamics, giving rise to the \textit{leading order terms}. 
In Section \ref{section subleading}, we estimate \textit{subleading order terms} and  in Section \ref{section proof of theorem},  we prove our main result, Theorem \ref{theorem 1}. 
Finally, in Section \ref{section fixed volume} we analyze the fixed volume case.

\subsection{Notation}
The following notation is going to be used   throughout this article. 
\begin{itemize}[leftmargin=*]
	\item[-] $\Lambda^* \equiv (2 \pi \Z / L)^d$ denotes the dual lattice of $\Lambda$.
	\item[-] We write  $\int_{\Lambda^*} F(p) \d p \equiv |\Lambda|^{-1} \sum_{p\in \Lambda^*} F(p)$
	for any function $F : \Lambda^* \rightarrow \C$. 
	\item[-] $\delta_{p,q}$ denotes the standard Kronecker delta.
	\item[-] We denote
	 $\delta_\Z (x) = \delta_{x,0}$
	 and 
	 $\delta(p -q  ) = |\Lambda|\delta_{p,q}$. 
	\item[-] $\ell^p(\Lambda^*)$  denotes   the space of functions  
	with finite norm $\| f\|_{\ell^p} \equiv (	\int_{  \Lambda^{*} } |f(p)|^p \d p 	)^{1/p}$. 
	\item[-] $B(X)$ denotes the space of bounded linear operators acting on $X$. 
	\item[-] We denote $\widetilde f \equiv 1 -f $  for any function  $f: \Lambda^* \rightarrow \C$  . 
	\item[-] $\widehat G(k) \equiv (2\pi)^{-d/2} \int_\Lambda e^{- i k\cdot x } G(x) \d x $ 
	denotes the Fourier transform of $G : \Lambda \rightarrow \C.$
	\item[-] We say that a positive real number  $C>0$ is a \textit{constant}, if it is    independent of  the physical parameters 
	$N$, $|\Lambda|$, $\lambda$, $n$ and $t$. 
	\item[-]	Given two real-valued  quantities $A 	$
	and 
	$B
	$, 
	we say that $A \lesssim B $ if there exists a constant $C>0$ 
	such that 
	\begin{eqnarray}
		A \leq C  \,  B   \ . 
	\end{eqnarray}
	Additionally, we say that  $A \simeq B $ if  both $A \lesssim  B $ and $B \lesssim  A$ hold true. 
	
	\item[-] We shall frequently omit subscripts from Hilbert spaces 
	norms throughout proofs. 
	
	\item[-] {\blu We define the bracket
	   $\<t\>   \equiv  (1 + t^2)^{1/2}$.}   
\end{itemize}

\section{Main results}
\label{section main results}
The main result of this article is a rigorous interpretation of the  expansion 
\eqref{F}  that arises from the   many-body fermionic dynamics. 
First, we  present   the model  that we study and introduce particle-hole variables. 
Secondly, we present our first 
 main result in Theorem \ref{theorem 1}.
 It contains an estimate 
in a weighted $\ell^\infty$ norm for the error  of  the momentum distribution of     particles and holes. 
Thirdly,  in Theorem \ref{corollary},  we discuss the consequences of this estimate in a precise scaling. 
Here, we extract the leading order terms
and prove appropriate lower bounds.
Finally, in Theorem \ref{thm3},  we go back to the momentum distribution $F_t (p)$. 
  
 \subsection{The model}
 We consider the torus 
 $\Lambda  \equiv   (\R / L \Z)^d $ where $L>0$ is its   length.
 We study the dynamics on the 
 Fock space   
 \begin{equation}
 	\textstyle 
 	\F  \equiv  \C \oplus \bigoplus_{n \geq 1  } \F_n  \qquad {\t{where}} \qquad 
 	\F_n \equiv   \bigwedge_{i=1}^n  L^2 (\Lambda  )   , \ \forall n \geq 1     \     .
 \end{equation} 
 As usual, $\F$ is equipped  with  creation and annihilation operators  $a_p$ and $a_q^*$.
 In  momentum space 
 {\blu they} satisfy the Canonical Anticommutation Relations (CAR)
 \begin{equation}
 	\{   a_p , a_q^*  \}   
 	=
 	\delta (p-q)
 	\equiv
 	|\Lambda | \delta_{p,q}
 	\qquad
 	\t{and}
 	\qquad 
 	\{   a_p^* , a_q^*  \}    = \{	 a_p , a_q	\}  = 0 
 \end{equation}
 for all $p,q \in  \Lambda^*  \equiv     (2 \pi \Z / L)^d   $.
 Here,  $\delta_{p,q}$ stands for the Kronecker delta and 
 and $ \{ \cdot, \cdot  \}$ denotes   the anticommutator. 
 The Fock vacuum vector will be denoted by $\Omega  \in \F $, and
 for notational convenience we denote sums by 
 $\int_{\Lambda^*} \d p \equiv |\Lambda|^{-1} \sum_{p\in \Lambda^*}$. 
 
 \vspace{1mm }
 The Hamiltonian on Fock space  
 that we study     corresponds to the 
 operator 
 \begin{equation}
 	\label{hamiltonian}
 	\Ha
 	\equiv 
 	\frac{1}{2}
 	\int_{\Lambda^*}
 	|p|^2 a_p^* a_p  \d p 
 	+
 	\frac{\lambda}{2}
 	\int_{ ( \Lambda^*)^3 }
 	\hat{V}(k)
 	\ 
 	a^*_{p+k} a^*_{q-k} a_q a_p 
 	\  \d p  \,  \d q \,   \d k   \  , 
 \end{equation}
 where $\hat V(k) \equiv (2\pi)^{-3/2} \int_\Lambda V(x) \d x    $ is the Fourier transform of a two-body  potential
 $ V (x)  $. 
 We are interested in the 
 dynamics generated by $\Ha$ on  initial states
 that are appropriate
 perturbations of the Fermi ball. 
 These are described 
 in detail in Condition \ref{condition initial data}.

 \vspace{2mm}

   Let us describe   the time evolution      generated by    the Hamiltonian $\Ha$, defined in \eqref{hamiltonian}. 
  Denote by $B(\F)$   the $C^*$-algebra of bounded operators  on Fock space, 
  and consider  an initial state $\rho$.
  Then, the dynamics  of the system is given by 
  \begin{equation}
  	\label{dynamics}
  	\rho_t  (    \O) = \rho ( e^{  i t \Ha} \O e^{ -  i t \Ha }   )  \ , \qquad \O \in B(\F) \ . 
  \end{equation}
  Here and in the sequel, the time variable $t\in \R$  should be understood as being measured in microscopic units. 
  In particular, the momentum distribution per unit volume of the system is defined as
  \begin{equation}
\label{momentum F}
  	F_{t}(p) \equiv |\Lambda|^{-1} \rho_t(     a_p^*a_p) \ , 
  \end{equation}
  for all $t \in \R$ and $ p \in \Lambda^*$.

  \vspace{2mm}

  The initial states that we study are  a  {\blu special class of    perturbations} of the 
Fermi ball. 
  That is,  $\rho$  corresponds to   a translation-invariant state, 
which is a 
  suitable perturbation 
of  the  pure state 
  \begin{equation}
  	\Psi_F =  
  	\prod_{ p \in \mathcal B }
  	a^*_{p} \Omega   
  	\qquad \t{where} \qquad 
  	\label{fermi sphere}
  	\mathcal{B}
  	= 
  	\{    p  \in   \Lambda^*    :    | p | \leq    p_F 		 \}   
  \end{equation} 
{\blu in the sense of Condition \ref{condition initial data}}. 
The  state $\Psi_F$ 
corresponds to the    Slater determinant
of plane waves 
$e_p (x) = |\Lambda|^{-1/2 } e^{ i  p \cdot x }$ 
with momenta in $\mathcal B$, 
minimizing the kinetic energy of the system in compliance with 
Pauli's exclusion  principle.  

\vspace{2mm}

  We refer  interchangeably to $\Psi_F$ and $\mathcal{B}$    as the \textit{Fermi ball}  
  defined in terms of the Fermi momentum $p_F$.  
  For simplicity, we assume here that $p_F$ is given, and define the number of particles  to be  $N \equiv | \B|$.  
  The relationship between $p_F$ and $N$ is then given by the formula 
  \begin{equation}
  	\label{fermi momentum}
  	p_F 
  	\, = \,  C 
  	\big( 
  	N / |\Lambda|
  	\big)^{1/d} \ , 
  \end{equation}
{\blu 
  where $C = C_d  + o( 1 )$ as $p_F \rightarrow \infty$, 
  and $C_d>0$ is a constant depending only  on  the  dimension.}  In particular, it is important to note that  the high-density regime that we study
  corresponds to large values of the Fermi momentum $p_F \gg1$.

\subsection{Assumptions and definitions}
Let us  state  the precise mathematical conditions under which our main theorems are  formulated.
We first discuss the conditions on the interaction potentials, and then the conditions on the initial data.

\vspace{2mm}

 \textit{Potentials.}
Throughout this work, 
 we shall  consider      real-valued functions  
 $V : \Lambda \rightarrow \R$
 that satisfy     Condition \ref{cond potentials} below. 
 In particular, under these conditions, the many-body Hamiltonian 
 $\Ha$ defined in   \eqref{hamiltonian} is self-adjoint, and
 its  dynamics  is well-defined.  
 
  \begin{condition} 
 	\label{cond potentials}
 	$V : \Lambda \rightarrow \R$
 	is a real-valued function
 	whose 
 	Fourier transform  $\hat V(k)$
 	satisfies  the following   conditions. 
 	
 	\begin{enumerate}
 		\item 	It has compact support in a ball of radius $r>0$ . 
 		
 		\item 
 		$\hat V  (  - k ) = \hat V  (k)  $ for all $ k \in  \Lambda^*$. Thus, $\hat V$ is real-valued . 
 		
 		\item 
 		$\hat V (0)    = 0 $ .   
 		
 		\item 
 		$V$  is chosen relative to the box $\Lambda$
 		so    that $     \sup_{	 |\Lambda | >0	} \|	 \hat V	\|_{\ell^1 (\Lambda^*)} < \infty $.
 		
 	\end{enumerate}
 \end{condition}

  \begin{remark}
The radius $r>0$ will be fundamental in our analysis. 
From a physical point of view, it determines a momentum scale 
that allows for particle interactions. 
In particular, it will determine an $\mathcal{ O }(1)$ neighborhood around the Fermi surface that separates  excited particles and holes 
into either bosonizable or non-bosonizable. 
  \end{remark}

 \vspace{2mm}

 \textit{States}. 
The initial states 
that we consider are regarded as   perturbations of the Fermi ball
 $\Psi_F$. 
It will   be   convenient to use   
the following notations 
for    the momentum distribution 
of $\Psi_F$, as well as its complement 
\begin{equation}
	\chi(p)
	 = 
	\1 ( |p |\leq p_F )
	\qquad 
\t{ and }
	\qquad 
	\chi^\perp   =  1 - \chi    \  , 
\end{equation}
where $\1$ stands for a  characteristic function.  In particular,  it is now a  standard calculation  using the 
CAR     to show that   
\begin{equation}
	\<  \Psi_{ F	}, a_p^*a_q \Psi_{F	}  \>  =
	\delta(p -q ) \chi(p)   
\end{equation}
for all $p,q \in \Lambda^*$. 

\vspace{2mm}

We analyze the dynamics of fermions
relative to the Fermi ball. 
In this regard, we think of 
excited fermions with momentum $|p| >  p_F$
as \textit{particles}, 
and the anti-particles
they leave behind inside of the Fermi ball  as
\textit{holes}, with momentum $|h| \leq p_F$
This change of variables is implemented by a 
\textit{particle-hole transformation}. 
It  corresponds to the unitary transformation on Fock space 
\begin{equation}
	\calR : \F \longrightarrow \F 
\end{equation}
that 
can be explicitly defined 
through its action on creation  and annihilation operators  as  follows 
\begin{align}
\label{particlehole transformation}
\calR^* a_p^* \calR 
 \ = \ 
 \begin{cases}
 \ 	a_p^* 		\qquad  |p| > p_F  \\
 \  a_p  	\qquad  |p| \leq p_F 
 \end{cases} 
\qquad
\t{and}
\qquad 
\calR \Omega \equiv  \Psi_F \ . 
\end{align}
 %
In particular, we  are     interested in describing 
the time evolution of the momentum distribution of states relative to the Fermi ball. 
In particle-hole space, the dynamics of     these states is   described  by the  
\textit{particle-hole Hamiltonian }
\begin{equation}
\label{conjugated hamiltonian}
 \h  \equiv  	\calR^* \Ha \calR    \    . 
\end{equation}	
A more explicit representation of the Hamiltonian $\h$ will be given in the next section.

\vspace{1mm}

Thus, we   study the evolution in time of the 
corresponding momentum distribution, defined as follows. 
 Recall   that a state is a positive linear functional on $B(\F)$, with
$\nu(\1)=1$.

\begin{definition}
\label{definition momentum distribution}
	Given an initial state $\nu$, we 
	define  
	the momentum distribution per unit volume of particles and holes 
	as 
	\begin{align}
\label{momentum distribution definition}
		f_t ( p  ) 
		 \equiv 
		 |\Lambda|^{ -1 }
		 	\nu(  e^{\i t \h }   a_p^*a_p e^{ - \i t \h } ) 	  \ , 
	\end{align}
 	for  $ (t,p) \in \R \times \Lambda^*$. 
\end{definition}

\begin{remark}
	Let $F_t(p) $ be the momentum distribution  \eqref{momentum F}
of a state $\rho$, 
	evolving {\blu according} to the Hamiltonian $\Ha$, in the original many-body problem. 
	Then, a straightforward calculation using the CAR shows that 
\begin{equation}
	F_t ( p) 
	= 
	\begin{cases}
		f_t (p) 	  \quad   &  |p| > p_F \\ 
		1 - f_t (p)     \quad   & | p | \leq p_F \  , 
	\end{cases}
\end{equation}
where $f_t(p)$ is given as in Definition \ref{definition momentum distribution}, 
with respect to the unitarily transformed initial state
\begin{equation}
\label{nu and rho}
	\nu (\mathcal O ) \equiv  
	\rho ( \calR \mathcal O \calR^*  ) \ , \qquad \O  \in B (\F) \ . 
\end{equation}
Note that, if $\rho$ and $\nu$ are determined by pure states
$\Psi$ and $ \psi$, respectively, then
\eqref{nu and rho} 
is equivalent to 
$\Psi = \calR \psi$. 
In particular, $\Psi = \Psi_F$ if and only if $\psi = \Omega$. 
\end{remark}
 
%

We find it convenient to introduce conditions on the initial data with respect to the particle-hole variables. 
In order to state them, 
recall  that  the  interaction potential $\hat V$
has support of size $r>0$.
We introduce the following neighborhood around the surface of the Fermi ball  
\begin{eqnarray}
	\label{S definition}
	\S   \equiv 
	\{   p\in \Lambda^* : p_F - 3 r \leq | p | \leq p_F +3 r  \}    
\end{eqnarray}
which (under a slight abuse of notation)    we shall refer to  as    the \textit{Fermi surface}. 
The pre-factor $3$ is included for technical reasons.  

\vspace{1mm}

The conditions for the initial data in the particle-hole representation are given as follows.

\vspace{1mm }

\begin{condition}  
	\label{condition initial data}
The initial state   $\nu$   {\blu satisfies} the following conditions. 
	\begin{enumerate}[label = (C\arabic*) , leftmargin=*]
\item 
There {\blu exist} sequences $ 0 \leq  \nu_j  \leq 1 $ and $\Psi_j \in \F$
such that $\nu(\O) = \sum_{j=1}^\infty \nu_j \< \Psi_j , \O  \Psi_j \>_\F $ 
with $\sum \nu_j  =1 $ and $\| 	 \Psi_j\|_\F =1 $ . 

		\item 
		$\nu$ is number-conserving and quasi-free:
		for all $k,k' \in \mathbb N $, $p_1, \ldots, p_k \in \Lambda^*$ and $q_1, \ldots, q_{k'} \in \Lambda^*$ there holds
		\begin{equation}
\nu 
\Big(
a_{p_1}^* \cdots a_{p_k}^* a_{q_{k ' } }   \cdots a_{q_1}  
\Big)
			\, 		= \, 
			\delta_{ k , k' }
			\det 
			\big[
			\nu(a_{p_i}^* a_{q_j})
			\big]_{1 \leq i,j \leq k } \ . 
			\end{equation}
		
		\item 
		$\nu$ is translation-invariant: 
		for all $p,q \in \Lambda^*$,  there holds 
		$
		\nu (a_p^* a_q) = \delta(p-q) \nu (a_p^* a_p ) \ . 
		$
		
		\item 
		
		$\nu$ has zero charge:
		$
		\int_{ \mathcal{B} } \nu (a_p^*a_p) \d p 
		= 
		\int_{ \mathcal{B}^c } \nu (a_p^*a_p) \d p  \ .
		$
		
		\item 
		
			There exists a constant $C \geq 0 $ such that
		$\int_\S \nu (a_p^*a_p)\d p \leq C (\lambda |\Lambda| p_F^{d-1} )^{2}$. 
	\end{enumerate}
\end{condition}

\vspace{1mm}

\begin{remark}   
Analogously  as in   \cite{ESY2004}, 
we could have considered states that are only \textit{restricted quasi-free} up to a certain   degree, 
and our main results  would remain unchanged.  
{\blu
	Since all the examples that we consider are indeed quasi-free, we choose to require condition (C2) instead of the more general restricted quasi-free assumption. 
}

We note that translation invariance is  a natural assumption that greatly simplifies the analysis, while at the same time being physically relevant. 
The same comments apply to the requirement  of  states  having  zero charge.

{ \color{black} Finally, let us comment on  Condition (C5). We refer to it as the \textit{depletion}  of the Fermi surface, in the sense that it contains only a small number  of  particle-hole pairs. 
	Let us observe that if the initial datum was equal to the Fermi ball, then $F_0  = \chi$
or, equivalently,  $f_0 =0$. 
Thus, $f_0$ quantifies the deviations of the initial datum from the Fermi ball.
In particular, (C5) states that  the   perturbations  of $\chi$    around its    surface are negligible, and so the distribution of   excited particles and holes described by $f_0$ is mostly supported away from the surface.  
	Physically,  such states represent a special class of excited states,  and our main interest lies in the dynamics that is observed at later times. 
	Let us note that
these excited states do not occur spontaneously in the system, and can only be achieved by \textit{external} means;
 the reader may, for instance, think of a cold gas of electrons in a semiconductor that has been carefully excited by an external source of light (photoexcitation). 
In this regard, we consider them to be external perturbations of a system originally at (or close to) equilibrium. 
 }
 \end{remark}

\vspace{2mm}

 \textbf{Example.}
We may  construct a pure  state $\nu$ that {\blu satisfies} Condition \ref{condition initial data}
as a Slater determinant. 
Namely, given $ n \in \mathbb N$, 
let $h_1, \ldots, h_n \in \B \backslash \S $
and
$p_1 , \ldots, p_n \in \B^c \backslash \S $. 
Then, we set   
\begin{equation}
	\nu(\mathcal O)
	\equiv 
	\< \Psi_0  , \mathcal O  \Psi_0 \>_\F
	\qquad
	\t{where}
	\qquad
	\Psi_0 \equiv
	a^*_{h_1} \cdots a^*_{h_n} a^*_{p_1} \cdots a^*_{p_n }
 \Omega  \ . 
\end{equation} 
Since Slater determinants are always number-conserving and  quasi-free, this {\blu satisfies} (C2). 
One may verify that translation invariance in  (C3)
is satisfied by direct computation of the two-point function
\begin{equation}
	\nu (a_p^*a_q) 
	=
	\delta(p-q) 
	\Big(	\delta(p - h_1) + \ldots + \delta(p-h_n) + \delta(p - p_1) +
	\ldots   +  \delta(p-p_n) 	\Big)  \ . 
\end{equation}
The state $\nu$ has zero charge in (C4) 
because 
we have chosen an equal number of $h_i's$ and $p_i's$ in 
$\B$ and $\B^c$, respectively. 
Finally,  (C5) follows from  $\nu(a_p^*a_p)=0$
for all $  p  \in \S. $

 \vspace{2mm}


\subsection{Statement of the main theorem}

Our main result identifies the time evolution of the  momentum distribution $f_t(p)$
in terms of two non-linear operators that act on functions on $\Lambda^*$.  
In order to define these,   we introduce 
the following three  notations. 

\vspace{1mm}
\noindent  \textit{Notations}
 
 \begin{enumerate}[leftmargin=*  , label = (\roman*)]
 	\item 
 	We denote 
 	by $E_p$ the \textit{dispersion relation} of particles and holes.
 	It is defined    for $ p \in \Lambda^*$
 	\begin{equation}
 		\label{dispersion relation}
 		E_p   \,  \equiv  \, 
 		-  
 		\chi(p)
 		\Big( \ 
 		\frac{p^2}{2} 
 		\, +  \, 
 		\frac{\lambda}{2}
 		\, 
 		(\hat V * \chi^\perp)(p) \ 
 		\Big)
 		+ 
 		\chi^\perp (p)
 		\Big( \ 
 		\frac{p^2 }{2}
 		\, 	- \, 
 		\frac{		\lambda }{2}  
 		\, 
 		( \hat V * \chi)(p) 
 		\ 		\Big) \ . 
 	\end{equation}
 	
 	\item 
 	For  $(t,E)\in \R^2$ 
 	we denote by $\delta_t(E) $
 	the following 
 	\textit{mollified Delta function}
 	\begin{eqnarray}
 		\label{delta function}
 		\delta_t (E) =   t \delta_{1}(tE)  
 		\qquad
 		\t{ where }
 		\qquad 
 		\delta_1(E) 
 		= 
 		\frac{2}{\pi }\frac{\sin^2\big( 	\frac{E}{2}\big)}{E^2} \ . 
 	\end{eqnarray}

 	\item 
 	Third,  we introduce the following convenient notation for products 
 	of $\chi$. Namely, for any $ k \in \mathbb N $ 
 	we  will  write 
 	\begin{equation}
 		\chi(p_1 , \ldots,  p_k)  \equiv  \chi(p_1)   \cdots   \chi (p_k)
 	\end{equation}
 	for all  $p_1, \cdots, p_k \in \Lambda^*$ 
 	and similarly for $\chi^\perp. $

 \end{enumerate}

 \vspace{2mm}

The following operator describes describes Boltzmann-type interactions
between particles/particles, particles/holes and holes/holes. 
Here and in the sequel  we denote 
$$\widetilde f \equiv 1 -f $$ 
for any function $f : \Lambda^* \rightarrow \R $.

\vspace{3mm}


%

\begin{definition}
\label{definition B}
For all  $t \in \R $
we define    in terms  of particle and hole interactions 
\begin{align}
	\label{mollified B}
	B_t   & \equiv  B_t^{(H)}  +  B_t^{(P)}
	 : \ell^1(\Lambda^*) \rightarrow \ell^1 (\Lambda^* )
\end{align}
where 
  $B_t^{(H)} : \ell^1 (\Lambda^* ) \rightarrow \ell^1  (\Lambda^* ) $ 
and 
  $B_t^{(P)} : \ell^1 (\Lambda^* ) \rightarrow \ell^1  (\Lambda^* ) $ 
  are defined as follows 
 \begin{align}
 	B_t^{(H)}  [f](h)   
\nonumber 
&   	 \equiv 
2\pi 
 	 \int_{  \Lambda^{*} }
 	 |\hat V(k)|^2 
 	 \Big(
 	 \alpha^{H}_t (h - k ,k)    f(h - k) \widetilde f(h) 
 	  - 
 	   	 \alpha^{H}_t (h,k)    f(h) \widetilde f(h+k ) 
 	 \Big) \d  h   \ ,  \\ 
\nonumber 
	B_t^{(P)}  [f](p) 
 & 	\equiv 
2\pi 
	\int_{  \Lambda^{*} } 
	|\hat V(k)|^2 
	\Big(
	\alpha^{P}_t ( p + k ,k)    f( p + k) \widetilde f(p) 
	- 
	\alpha^{P}_t (p,k)    f(p) \widetilde f( p  - k ) 
	\Big) \d  p    \   , 
\end{align} 
for $f \in \ell^1$ and $p,h\in \Lambda^*$. 
Here,  the coefficients  $\alpha^H_t$   and  $\alpha^P_t$ are defined as  
 \begin{align}
 	\label{alpha H}
 	\alpha^H_t 
 	( h, k)
 	 &   \equiv 
 	 \chi(h)\chi(h+k)
 	\int_{  \Lambda^* }
 	\chi( r )
 	\chi^\perp ( r + k )
 	\delta_t 
 	\big[   
 	E_h - E_{ h + k  }  -  E_r     -  E_{r+k }    
 	\big]   
 	\d r   \ ,   \\ 
 		\label{alpha P}
 	\alpha^P_t ( p, k)
 & 	\equiv 
 	\chi^\perp (p)\chi^\perp ( p -k)
 	\int_{  \Lambda^* }
 	\chi( r )
 	\chi^\perp ( r + k )
 	\delta_t 
 	\big[   
 	E_p - E_{ p - k  }  -  E_r     -  E_{r+k }    
 	\big]   
 	\d r \      ,
 \end{align}
for all $p,h, k  \in \Lambda^*.$
\end{definition}

\begin{definition}
	\label{definition Q1}
	For $f   \in \ell^1   (\Lambda^* )$
	and
	$t \in \R $
	we 
	define  
	\begin{align}
		\label{mollified Q}
		Q_t [f] (p)
		& \equiv 
		\pi 
		\int_{  \Lambda^{*4 } }
		\d \vec p 
		\,
		\sigma (\vec p  )
		\,
		\Big[ 
		\delta(p - p_1)
		+
		\delta(p - p_2)
		-
		\delta(p - p_3)
		-											  
		\delta(p - p_4)
		\Big]    \\
		& \times \delta_t [     E_{ p_1 }   +   E_{ p_2 } -   E_{ p_3 } -   E_{ p_4  }] 
		\Big(
		\nonumber 
		f (p_3 ) f (p_4)  \widetilde f (p_1 ) \widetilde f (p_2 )
		-
		f (p_1 ) f (p_2)  \widetilde f (p_3 ) \widetilde f (p_4 )
		\Big)  
		\ . 
	\end{align}
	The   coefficient function  $\sigma   : ( \Lambda^*)^4 \rightarrow \R$  is {\blu defined} as  
	\begin{equation}
		\sigma = \sigma_{HH} + \sigma_{P P} + \sigma_{ H P } + \sigma_{P H } 
	\end{equation}
	where the coefficient functions are defined for $\vec p = (p_1 , p_2, p_3, p_4) \in (\Lambda^*)^4 $
	as follows 
	\begin{align}
		\sigma_{HH} (\vec p )
		& 	= 
		\chi(p_1, p_2 , p_3 , p_4)
		\delta(p_1 + p_2 - p_3 - p_4 )
		|   \hat V (p_1 - p_4 )  - \hat V (p_1 - p_3 ) |^2   \\ 
		\sigma_{PP} (\vec p )
		& 	= 
		\chi^\perp    (p_1, p_2 , p_3 , p_4)
		\delta(p_1 + p_2 - p_3 - p_4 )
		|   \hat V (p_1 - p_4 )  - \hat V (p_1 - p_3 ) |^2 	\\ 
		\sigma_{H P} (\vec p )
		& 	= 
		2
		\chi(p_1, p_3 )		\chi^\perp (p_2 , p_4 ) 
		\delta(p_1  -  p_2 - p_3 +  p_4 )
		|   \hat V (p_1 - p_3 ) |^2 	  \\  
		\sigma_{P H } (\vec p ) 				 & 	= 
		2
		\chi^\perp (p_1, p_3 )		\chi (p_2 , p_4 ) 
		\delta(p_1  -  p_2 - p_3 +  p_4 )
		|   \hat V (p_1 - p_3 ) |^2 	 \ . 
	\end{align}
	%
\end{definition}

   \vspace{3mm}

\textit{Definition of the norm}. 
We will analyze   error terms with respect to a weighted norm which we now introduce. 
Indeed, for  $m> 0 $ we introduce the following weight 
\begin{equation}
	w_m(p)   \equiv 
	\begin{cases}
\< p\>^m \ ,  &  p \in \S  \\ 
1  \ ,   & p \in  \Lambda^*  \backslash \S   
		\end{cases} \    . 
\end{equation}
where   $\<  p\> \equiv  (1+ p^2)^{1/2}$ denotes  the standard Japanese bracket. 
   We define  the Banach space $\ell_m^1 \equiv \ell^1_m (\Lambda^*) $ 
of functions $\vp : \Lambda^* \rightarrow \C$
for which the norm 
\begin{eqnarray}
	\label{def:metric}
	\|	 \vp 	\|_{\ell^1_m}  \equiv  \int_{\Lambda^*}	  |\vp (p) | w_m(p)  \d p 
\end{eqnarray}
is finite. 
We will measure distances in the norm associated with the dual space of $\ell^1_m(\Lambda^*)$. 
Namely,  we regard  $ \ell_m^{1*} \equiv  [ \ell^{1}_m (\Lambda^*) ]^*$ as the Banach 
space of functions $f : \Lambda^* \rightarrow \C$
endowed with the norm 
\begin{equation}
\label{f sup}
	\|	f 	\|_{ \ell^{1*}_m }  \equiv 
	\sup_{ p \in \Lambda^*}
	 w_m(p)^{ -1 }
	 |f(p)|  
	  = 
	  \sup_{ \vp  \in \ell^1_m}
 \frac{	 |    \< \vp , f  \> | 	}{\|	\vp\|_{\ell_m^1}	}
\end{equation}	
where we denote by $\<\vp ,f \> \equiv  \int_{  \Lambda^{*} }  \overline{\vp(p) }  f(p) \d p $ 
the coupling between $\ell_m^1$
and $\ell^{1*}_m$.

\begin{remark}
\label{remark norms}
	As vector spaces, 
	$\ell_m^1 (\Lambda^*) = \ell^1(\Lambda^*)$
	and 
	$\ell_m^{1* } (\Lambda^* ) = \ell^\infty(\Lambda^*) $ for all $m> 0$. 
	However, we   {\blu we equip}  these spaces with the norms 
	$\|		\cdot \|_{\ell_m^1 }$
	and 
	$	\|	 \cdot 	\|_{\ell_m^{1* }}	$
since the weight $w_m(p) $	   records the decay near the Fermi surface $\S$. 
For completeness, we record here the following inequality 
\begin{equation}
\label{f norms}
	\|	f	\|_{\ell^\infty (\Lambda^*\backslash \S )} 
	\leq 
	\|	 f 	 \|_{	 \ell_m^{1*}		}
	\leq 
	\|	f	\|_{\ell^\infty(\Lambda^*)} \ , 
	\qquad \forall f \in \ell^\infty(\Lambda^*)  
\end{equation}
which we shall make use of   when studying the fixed volume case in the next subsection. 
 \end{remark}

\begin{remark}
	If $f \in \ell_m^{1* } $ is real-valued,  one may  restrict the supremum over  $\vp \in \ell_m^1$ 
	on the right-hand side of \eqref{f sup} to be real-valued as well. 
{	\blu This observation will be useful when computing double commutators,   as it makes 
	various observables of interest self-adjoint. 
	In particular, it will simplify certain expressions that would then need to be decomposed into real and imaginary parts otherwise. 
}
\end{remark}

\begin{remark}
	{\blu 
		Let us further comment on this norm. 
		First, we note that it is an $\ell^\infty$ based norm. 
		Such  a choice of base space is, in particular, quite useful when working with \textit{fermionic} field variables. 
		Indeed, one immediately obtains uniform $\ell^\infty$ bounds for  observables of interest thanks to the boundedness of the operators $(a_p)_{ p \in \Lambda_*}$; this is a consequence of the Pauli Exclusion Principle. 
		Furthermore, it is natural to decompose this norm into two parts that are weighted differently: either \textit{inside} or \textit{outside} the Fermi surface. 
		In this article, we are mostly interested in the dynamical properties of electrons and holes outside the Fermi surface, which do not exhibit bosonization phenomena. 
		For particles with such momenta, the weight $w_m(p) =1$ indicated that the norm in consideration is only the $\ell^\infty$ norm. 
		On the other hand, the weight $ w_m(p)^{ -1 } = \< p\>^{-m}  \sim p_F^{-m}$ on the surface indicates that the norm is unable to resolve   further  errors in this region. 
		In other words, the Fermi surface becomes  small in this norm. 
				This is compatible with the choice of initial data under consideration.
Indeed, it will be a consequence of our analysis 
 that the Fermi surface remains depleted of particle-hole pairs  throughout the analysis.  
 In particular, we will use the norm $\| f \|_{\ell_m^{1*}}$ to quantify the size of various error terms
involving physical processes with a non-zero number of initial particle-hole pairs. 
} 
\end{remark}

 \vspace{2mm}

Finally, let us now introduce 
two important new parameters. 
Namely,  
these are the numbers 
$n$ and $R$  defined as 
\begin{equation}
	\label{R definition} 
	n 
	\equiv 
	|\Lambda | \int_{\Lambda^* } f_0 (p) \d p 
	\qquad
	\t{and} 
	\qquad 
	R \equiv|\Lambda| p_F^{d-1 }
	\simeq 
	 |\S| \ . 
\end{equation}
Here, $n$ corresponds  to   
the initial number of particles and holes in the system, 
and it  measures the size of the perturbation of the Fermi ball. 
On the other hand,  $R$ corresponds
to the maximal number  of bosonized particle-hole pairs
that can populate  the Fermi surface $\S$, defined in \eqref{S definition}.  
Our main result is now  stated as follows.

\begin{theorem}\label{theorem 1}
Let  $f_t(p)$ be the momentum distribution of particles and holes, 
as given in Definition  \ref{definition momentum distribution}. 
We assume that  
Conditions \ref{cond potentials} 
and 
    \ref{condition initial data} are satisfied, 
    as well as the   bounds   $1  \leq   n \leq C   R^{1/2}$. 
Then, for all $m >  0 $ there exists $C = C(m,d)>0$ such that 
for all  $t  \geq 0 $
there holds 
\begin{align}
\label{f t expansion}
	f_t 
	 \,   =  \, 
	f_0 
 \, 	+  \, 
	\lambda^2 t  \, B_t    	[  f_0 ]   
\, 	  +  \, 
	\lambda^2 t  \, Q_t[f_0  ]
\,	 +	\, 
	 \lambda^2 t \, \t{Rem}_1 (t)  \ , 
\end{align}
where $\t{Rem}_1 (t) $ is a remainder term that satisfies 
\begin{equation}
\label{thm 1 estimate}
	\|	 \t{Rem}_1 (t) 	\|_{\ell_m^{1*}} 
\ 	\leq  \ C  
\, 	t  \,  e^{  C \lambda R \<t \>} \, 
	\bigg(
	\lambda R^2 
	(R^{\frac{1}{2}}  + n^2 )   \< t \>
	+ 
	\frac{R^3}{p_F^m} 
	\bigg) \ . 
\end{equation}
\end{theorem}

 
 \begin{remark}
 	At time zero, there holds $Q_0= B_0=0$. 
 	Hence, in order to prove that  the collision operators   dominate the remainder terms, we consider longer time scales in the next subsection, 
 	for three-dimensional boxes.
We show that for appropriately chosen initial data, 
and for $t \gg 1: $
\begin{equation}
c_\Lambda 
\, t \,  N^{1/3}
 \ \leq  \ 
	\|	 B_t [f_0]	\|_{\ell_m^{1*}}
 \  \leq  \ 
	  C_\Lambda \, t  \, N^{1/3}
	\qquad 
	\t{and}
	\qquad 
	 \|	 Q_t[f_0]	\|_{\ell_m^{1*}} \leq t n  \ . 
\end{equation} 
We will compare these sizes with the right-hand side of 
\eqref{thm 1 estimate}, for  $n \ll N^{1/6}$. 
 \end{remark}
 
 \begin{remark}
 	In view of Remark \ref{remark norms}, the above result  describes the dynamics  
 	of particles and holes away from the Fermi surface, i.e. on $\Lambda^*/ \S$. 
This is consistent with the fact \cite{Benedikter2022} that bosonization occurs inside    the Fermi surface. 
 \end{remark}

 \subsection{A particular scaling}
 Let us discuss in this section 
 a scaling regime for which   Theorem \ref{theorem 1} turns
 into an effective approximation. 
 
 \vspace{1mm }
 Namely,  we specialize to  the three-dimensional case
  amd consider the following scaling regime 
\begin{equation}
	\label{scaling2}
	\lambda = 1 / N^{ 3/2  + \delta_1} , 
	\qquad 
	t = N^{1/6 + \delta_2} T   \ , 
	\qquad
	L \  \t{ fixed }
\qquad 
	n    \leq N^{\frac{1}{6}}
\end{equation}
for positive parameters 
$ 0<  \delta_1$, $\delta_2 <1 $  specified below.

\vspace{2mm}

The time scales that we consider are short enough to maintain control over the remainder terms
of Theorem \ref{theorem 1}, 
but long enough to observe  {\blu exact energy conservation}.
More precisely, the mollified delta function $\delta_t( \Delta E)$ 
incorporated in the definition of $Q_t$ and $B_t$, 
now becomes a Kronecker delta function with respect to the free energy.
In other words, we prove  in Lemma \ref{lemma delta} in a suitable sense that 
\begin{equation}
\label{delta 1}
	\delta_t  ( \Delta E) = \frac{2t }{\pi} 
	\delta_\Z(\Delta e)
	\Big(
1   	+  \mathcal{ O }(1/t^2) + \mathcal{ O }(\lambda^2 t^2)
	\Big)   
\end{equation}
{\blu where $\delta_\Z (x)  = \delta_{x,0}$, and }
  \begin{equation}
	\Delta e  \equiv  e(p_1) + e(p_2) - e(p_3) - e(p_4) \ , 
	\qquad 
	\t{where}
	\qquad e(p) \equiv   
	\big[\chi^\perp(p) - \chi(p) \big] \frac{p^2}{2} \ . 
\end{equation}
As a consequence, 
in Lemma \ref{lemma Q difference} and \ref{lemma B difference} 
we are able to identify the leading order terms of the operators $Q_t$ and $B_t$
as follows 
\begin{align}
		\label{large time eq 1}
	 	B_t[f]  
	 	 & = 
	 	 t    \mathscr{B} [f]  
	 +  \mathcal{ O }_{\ell^\infty } 					
	 (1/t  )
+  \mathcal{ O }_{\ell^\infty   } 					 (  t^3  \lambda^2  \|	 \hat V	\|_{\ell^1}^2 )    \\
	\nonumber 
Q_t[f]  
& = 
t \mathscr{Q} [f]    
+ 			 \mathcal{ O }_{\ell^\infty   } 					(1/t  )
+  \mathcal{ O }_{\ell^\infty  } 					
(  t^3 \lambda^2   \|	 \hat V	\|_{\ell^1}^2 )   \ . 
\end{align}
where $f \in \ell^1(\Lambda^* ).$
  Here, the   operator $\mathscr B[f]$ 
  is defined as in Definition
  \ref{definition B}
  but with $\delta_t ( \Delta E )$  
    replaced by   
   $ (2/\pi) \delta_\Z   ( \Delta e)$. 
 The definition of $\mathscr{Q}$ is analogous. 
 
 \vspace{1mm}
%

The  following   result now follows as a corollary of  Theorem \ref{theorem 1}, 
Lemma \ref{lemma Q difference} and \ref{lemma B difference}, 
and the inequalities found in Eq. \eqref{f norms}.

\begin{theorem}[First collision time]
	\label{corollary}
Let  $f_t(p)$ be the momentum distribution of particles and holes, 
as given in Definition  \ref{definition momentum distribution}. 
We assume that  
Conditions \ref{cond potentials} 
and 
\ref{condition initial data} are satisfied, 
and consider  the  scaling    \eqref{scaling2} in three dimensions.
	  Assume that $\delta_2 \leq \delta_1/2 \equiv \delta$. 
	Then, for all $m>5$ 
there exists $C>0$
such that for all $ T \in [ N^{-\delta/2}  ,1]$
there holds 
\begin{equation}
 f_{t  } 
 = 
  f_0 
  + 
 \lambda^2 t^2   
    \Big(
\mathscr{B}[f_0]
+
 \mathscr{Q}[f_0] 
+  
\t{Rem}_2 (  T ) 
    \Big) 
\end{equation}
where $\t{Rem}_2$ is a remainder term that satisfies 
\begin{equation}
 \|	  \t{Rem}_2 ( T) \|_{    \ell^\infty(  \Lambda^*  \backslash \S )  }
		\leq  
C  N^{\frac{1}{3 }}
\bigg(
\frac{1}{N^\delta}
+ 
\frac{1}{N^{(m-5)/3    }}
\bigg)   
 = o (N^{\frac{1}{3}}) \ . 
\end{equation}
\end{theorem}

 \begin{remark}
 	\label{remark:L1}
Theorem \ref{corollary} 
describes  $f_t(p)$
 	for $p\in \Lambda^* \backslash\S$, i.e.
 	 away from the Fermi surface.
 	For $ p \in \S$ and $T \in [0,1]$,  
 	one actually has an $\ell^1$-bound 
\begin{equation}
	 \|	 f_{\epsilon^{-1} T  }	\|_{\ell^1( \S )}
	 \,   \leq  \, 
	  C /  N^{5/3 + 2 \delta} 
\end{equation}
which    follows as a propagation-in-time of  the depletion of the Fermi surface, as stated in Condition \ref{condition initial data} for the initial data
$f_0$. 
 See Proposition \ref{prop N estimates 2}. 
 In words, the scaling is chosen so that the Fermi surface remains almost entirely depleted over the   scale $T.$
  \end{remark}

\begin{remark}[Sizes of $\mathscr Q$ and $\mathscr B$]
	\label{remark lower bound}
	    The inequality contained in   Theorem \ref{corollary}
shows that       $\mathscr{B}$ and $\mathscr Q$ dominate    the remainder terms if 
 \begin{equation}
\label{initial data lower bound}
 	\|	\mathscr Q[f_0]+ 	\mathscr B[f_0]	\|_{    \ell^\infty( \Lambda^*   \backslash \S )  }
 \,  	 	\gg  \, 
 	 	 \|	  \t{Rem}_2 (T) 	\|_{    \ell^\infty( \Lambda^*   \backslash \S )  } \ . 
 \end{equation}
In Section \ref{section fixed volume} 
we prove that this holds 
for initial data satisfying additionally Condition \ref{condition 3}, 
and   $\hat V(k)$  satisfies additionally 
Condition \ref{condition 4}. 
Let us further explain. 
\begin{enumerate}[leftmargin=*]
		\item
We take   {\blu $f_0 (p)$  }as a linear combination of Kronecker deltas in $ \Lambda^* $ (see Def. \ref{definition initial data}). 
Further, we assume that they are supported
away from each other by a distance $r>0$, 
and that at least one of their cartesian components satisfies $|p_i| \sim p_F$.
We show that 
\begin{equation}
c_\Lambda N^{1/3}
 \ \leq  \ 
	\|	 \mathscr {B}[f_0]	\|_{\ell^\infty (   \Lambda^*   / \S 		)} 
	  \ \leq \ 
	  C_\Lambda  N^{1/3}
\end{equation}
where $c_\Lambda, C_\Lambda>0$. 
The heuristics for the $p_F \sim N^{1/3}$ dependence are as follows: given $ k \in  \supp \hat V / \{ 0\}$ 
a fermion can interact with any of the particle-hole pairs in the \textit{lune set}
\begin{equation}
	L(k) \equiv 	\{  q \in  \Lambda^*   :  |q | \leq p_F  , \  |q +k| \geq p_F	 	\} \ , 
\end{equation}   
which is of order $|L (k) | \sim N^{2/3}$.
On the other hand, energy conservation $\delta_\Z ( \Delta e) $
introduces a geometric constraint that reduces this number by an additional factor $ p_F \sim   N^{1/3}$. {\blu This  reduction arises from the  intersection of the  two-dimensional  lune set $L(k)$ with a straight line. Hence, reducing the number of lattice points to be counted within a two-dimensional figure of area $ p_F^2 \sim N^{2/3} $, to a one-dimensional set of length $p_F \sim N^{1/3}$}.

	\item  {\blu  Using elementary estimates for  $f_0 \in \ell^1(\Lambda^*)$ with $ 0 \leq f_0 \leq 1 $ one may show the following bound for the collision operator
	\begin{equation}
		\|	 \mathscr{Q} [f_0 ]	\|_{\ell^\infty( \Lambda^*  )} 
		\leq C   \|	f _0	\|_{\ell^1  (\Lambda_* )} 
		 = C n \  , 
	\end{equation}
 which is sufficient  
 to conclude the validity of \eqref{initial data lower bound}.
 We refer the reader to Section \ref{section fixed volume} for more details. 
}

\end{enumerate}

\end{remark}

\subsection{The original distribution function}
 Let us now state our final result, regarding the expansion 
 \eqref{F} introduced in the first section. 
 
 \begin{theorem}
 	\label{thm3}
 	 Let $F_t(p)$ be  the momentum distribution of the system, 
 	 defined   in 
 	 \eqref{momentum F}.
Assume that  
 	 Conditions  \ref{cond potentials} 
 	 and 
 	 \ref{condition initial data} are satisfied, 
 	 and consider the scaling \eqref{scaling2}
 	 in three dimensions, 
 	 with
 $\delta_2 \leq \delta_1/2  =  \delta$. 
 	Then, for all $m>5$  there exists
 	$C>0$
 	such that
 	 \begin{equation}
 	 	F_t = F_0 +  \lambda^2 t^2 
 	 	\Big(	 \mathscr C [F_0] +  \t{Rem}_3 (T)  				\Big)
 	 	\qquad 
 	 	\forall T \in [N^{ - \delta/2} ,1  ]
 	 \end{equation}
where $ \t{Rem}_3 (T)  			$ 
is a remainder term 
that satisfies
	\begin{equation}
		\|	  \t{Rem}_3 ( T) \|_{    \ell^\infty(  \Lambda^*  )  }
		\leq  
		C  N^{\frac{1}{3 }}
		\bigg(
		\frac{1}{N^\delta}
		+ 
		\frac{1}{N^{(m-5)/3    }
	}
		+ 
\frac{1}{N^{ 1/6 }} 
		\bigg)   
		= o (N^{\frac{1}{3}}) \ . 
\end{equation}
Additionally, if $F_0$ and $\hat V$ satisfy  Condition \ref{condition 4} and \ref{condition 3}, respectively, 
there exist positive constants $c_\Lambda, C_\Lambda>0$
such that 
\begin{equation}
 c_\Lambda N^{1/3}
 \ \leq  \ 
\|	 \mathscr C   [F_0 ]		\|_{\ell^\infty(	\Lambda^*
	\backslash \S )} 
 \ \leq \ 
 C_\Lambda  N^{1/3}  \ . 
\end{equation}

\end{theorem}

 \begin{remark}
 	Remark \ref{remark:L1}
 	 implies that $\| F_t - \chi	\|_{\ell^1(\S)} \leq C N^{ - 5/3 - 2 \delta}$. 
 	 That is, the Fermi surface remains almost stationary. 
 \end{remark}

 \subsection{Comparison}
 Let us compare our work 
 with      previous results that are available in the literature. 
 \begin{enumerate}[leftmargin=*]
 	\item 
 	 {\blu  
 	Benedikter, Nam, Porta, Schlein, and Seiringer \cite{Benedikter2022}
 	studied the dynamics of a three-dimensional Fermi gas around the Fermi ball, 
 	in the   semi-classical $\hbar = N^{-1/3}$, mean-field regime
 	$\lambda = 1/N$. 
 	The main focus is on the \textit{bosonization} of particle-hole pairs inside   a suitable neighborhood of the Fermi surface. 
 	The   initial states   $\psi $   considered  by the authors consist     of   approximations of the true  (interacting) ground state, constructed via   explicit  Bogoliubov transformations;    excitations of particle-hole pairs  within the surface are allowed in such data.  	 
 These  bosonization methods are based on  a  patch decomposition  of the Fermi surface, 
and  	have     also been employed    	 in	\cite{Benedikter2020, Benedikter2021,Benedikter 2023}; 
see   the discussion in Section \ref{section introduction}.  
 	In their terms, here     we consider  the contribution to the dynamics associated  with  the ``non-bosonizable terms''. These interaction  terms describe the dynamics of particles and holes away from the Fermi surface. 
} 
 	
 	\vspace{1mm}
 	
 	\item Hott and the second author  \cite{Hott} 
 	have    studied the emergence of quantum Boltzmann dynamics for the fluctuations around a  Bose-Einstein condensate.
 	In particular, our scaling regimes  are similar to  one      another  in the sense that both of them contain a large number of particles per unit volume.
 	In other words, the density of particles acts as an expansion parameter. 
 	Similarly, they both employ  quasi-free initial states. 
 	While of course the difference in statistics plays  a crucial role in the analysis, 
 	the approach presented here is largely    inspired by    that   of   \cite{Hott}. 
 	See also \cite{Hott2}
 	for a more recent refined analysis 
 	which incorporates additional  renormalized terms
 	into the Hartree-Fock-Bogoliubov dynamics, 
 	  yielding control over longer time-scales.

 	\vspace{1mm}
 	
 	\item Erd\H{o}s  \cite{Erdos2002} studied the weak-coupling limit 
 	of an electron interacting with a thermal bath of phonons. In particular, 
 	a linear Boltzmann equation is shown to emerge from the long-time dynamics. 
 	It turns out that this is not very different than the situation under consideration. 
 	Namely, the dynamics of particles and holes 
 	{outside} of the Fermi surface   can  also be described as particles 
 	that  interact with a  boson field, i.e., the  bosonized electron-hole pairs around the Fermi surface, as described by \cite{Benedikter2022}. 
 	The situation here is more complicated, however.  On the one hand, bosonization is only approximate. On the other hand,   several other interactions influence the dynamics of particles and holes, 
 	and rigorous error  control  over these interactions is already   demanding. 
 	Finally, let us recall that the weak coupling limit  of electrons interacting with a random medium is intimately related to the model studied in \cite{Erdos2002}; for results in this direction,  see e.g  
 	\cite{Chen2005,ChenRodnianski,ChenSasaki,ErdosYau2000,Spohn2006}. 
 	
 	\vspace{1mm}
 	
 	\item 
 	{\blu  
 	The scaling regime considered in this article 
 	\eqref{scaling2} contains      an  interaction strength  that is \textit{much} weaker than  the  microscopic mean-field scaling regime, which sets  $ \lambda = N^{-1/3}$ in three dimensions for fermions. 
In the dynamical description of  mean-field theory,  
 	  an approximation can be found in terms of transport equations.
 	In microscopic regimes, one finds
 	the Hartree-Fock (HF) equation
 	and  for macroscopic (semi-classical) regimes 
 	one can also derive the Vlasov equation. 
 	 	The literature of mean-field theory  for fermions is vast  
and will not be reviewed in detail here. We refer the interested reader  to the following non exhaustive list of references regarding the derivation of HF dynamics
	\cite{BardosEtal2003,Elgart2004,FrohlichKnowles2011,PetratPickl2016,BachEtal2016,Benedikter2014,Porta2017,Benedikter2016}
 and Vlasov dynamics
\cite{Narnhofer1981,Benedikter2016-2,Lafleche2021}.

 } 
 \end{enumerate}

 \section{Preliminaries}
\label{section preliminaries}

In this section, we introduce   preliminaries that are needed to prove our main result. 
First,  we give an explicit representation of the particle-hole Hamiltonian $\h $, introduced in \eqref{conjugated hamiltonian}. 
Second, based on this representation, we introduce the interaction picture framework that we shall use to study the dynamics of the momentum distribution $f_t(p)$, defined in \eqref{momentum distribution definition}. 
Third, we perform a double commutator expansion and identify nine terms, 
from which we shall extract   leading order and subleading order terms. 
Finally, we introduce number estimates that we use to analyze the nine terms found in the double commutator expansion.

\subsection{Calculation of $\h$}
Let us  introduce  two   fundamental {\blu collections} of operators. 
 We shall   refer to them    as $D$- and $b$-operators, respectively. 
 
 \begin{definition}
 	Let $k \in \Lambda^*$. 
 	
 	\vspace{1mm}
 	
 	\noindent (1)  
 		  We define the $D$-operators as 
 		\begin{equation}
 	  D_k 
 		\equiv 
 		\int_{  \Lambda^* }
 		\chi^\perp (p)  \chi^\perp(p  - k )
 		a_{p -k }^* a_{p  } \,  \d p 
 		-
 		\int_{  \Lambda^* }
 		\chi (h )  \chi ( h+ k  )
 		a_{h + k } ^* a_{ h  } \,  \d h   \  .
 	\end{equation}

 \noindent (2) 
 			We define the $b$-operators as 
\begin{equation}
 		 			 b_k 
 		 			\equiv  
 		 		  \int_{  \Lambda^* }   \chi^\perp (p)  \chi (p-  k )  		 a_{p - k} a_{p } \, \d p   \ . 		
\end{equation}
 
 \end{definition}
 
 \begin{remark}
For the rest of the article, we denote the corresponding adjoint operators by 
{\blu 
 $D_k^* \equiv   (  D_k )^*$ and $ b^*_k \equiv ( b_k)^*$, respectively. 
} 
Additionally, we shall extensively use the basic relation 
\begin{equation}
	D^*_k = D_{-k} \qquad \forall k \in \Lambda^*  \ . 
\end{equation}
 \end{remark}

 \begin{remark}[Heuristics]
 	\label{remark D and b}
{\blu  	  $D$ is   a combination of  \textit{fermionic} operators.  
 	  They contain interactions  between   holes and holes,  together with particles and particles,
 	that are away from the Fermi surface.  } 
 	On the other hand, the operators $b$  should be understood as   approximate \textit{bosonic} operators; they   annihilate  bosonized particle-hole pairs near the Fermi surface. 
 	 	In fact, the following commutation relation holds 
 	\begin{equation}
 		[  b_k  , D^*_k ] = 0 \qquad \forall k \in \Lambda^* \   .
 	\end{equation}
However, we shall not need any estimates on the
commutation relations satisfied by $b$'s, and this interpretation will remain at a heuristic level. 
 \end{remark}

 The following lemma contains the explicit representation for the   particle-hole Hamiltonian, in terms of 
 a ``solvable Hamiltonian", plus interaction terms depending on $D$ and $b$ operators.

 \begin{lemma}
Let $\h$ be the operator defined in \eqref{conjugated hamiltonian}. 
Then,  	the following identity holds 
 	\begin{equation}
\label{lemma h eq 1}
 		\h  -   \mu_1 \1  - \mu_2 \mathcal Q  =    \h_0 + \lambda \mathcal{V}  
 	\end{equation}
 	for some real-valued constants $\mu_1 , \mu_2 \in \R$. 
 	Here 
 	$\mathcal{Q}$
 	corresponds to the 
 	\textit{charge} operator
 	\begin{equation}
 		\mathcal{Q}
 		 \equiv 
 		\int_{  \Lambda^{*} } \chi^\perp( p)a_p^* a_p \d p 
 		- 
 		\int_{  \Lambda^{*} } \chi( p)a_p^* a_p \d p    ; 
 	\end{equation}
 	$\h_0$ corresponds to the quadratic, diagonal operator 
 \begin{equation}
\label{h0 definition}
 		\h_0  = \int_{  \Lambda^{*} } E_p a_p^*a_p \d p 
 \end{equation}
 with $E_p$ the dispersion relation defined in \eqref{dispersion relation}; 
and $\mathcal V = V_F + V_{FB} + V_B$ contains the following three interaction terms 
  \begin{align}
	\label{VF}
	V_F  & 
	 \equiv 
	\frac{1}{2}
	\int_{  \Lambda^* }
	\hat V (k) D^*_k  D_k   \  \d  k \\
	\label{VFB}
	V_{FB}  & 
	\equiv 
	\int_{  \Lambda^* } 
	\hat V (k)
	D^*_k
	\big[ 
	b_k
	+ b^*_{	-k }
	\big]  \d  k \\
	\label{VB}
	V_{B} &  
	\equiv 
	\int_{  \Lambda^* } 
	\hat V (k) 
	\big[ 
	b^*_k b_k 
	+ 
	\frac{1}{2} \, 
	b^*_k  b^*_{- k }
	+ 
	\frac{1}{2} \, 
	b_{ -k }b_k  
	\big]  \d  k   \ . 
\end{align} 
 \end{lemma}
 
\begin{remark}
	The labeling of       $V_F$, $V_{FB}$, and $V_B$  is of course related to Remark \ref{remark D and b}. Namely, $V_F$ contains    fermion/fermion interactions, $V_{FB}$ contains fermion/boson interactions and $V_B$ contains boson/boson interactions. 	 
	{\blu Let us note that the interactions containing $b$ operators are not \textit{exactly} bosonic,
		and the present terminology may be somewhat misleading.  
		 In particular,  the term quasi-bosonic is    more precise and is often used in the literature. However, in order to ease the overall terminology and notation,  we will choose the $B$ labels. We hope this will not cause much confusion.}
\end{remark}

\begin{remark}
	The charge operator $\mathcal Q$ is irrelevant for  the dynamics in the system.
	Indeed, one may easily check that $[\h_0 , \mathcal Q]  = [D, \mathcal Q] = [b, \mathcal Q] = 0 $ and, therefore, 
	$[ \h , \mathcal{Q}  ] = 0 $. In other words, the charge is a constant of motion
	and only the right hand side of \eqref{lemma h eq 1} is relevant      regarding the   time evolution of the momentum distribution of the system. 
	We make this argument precise in the next subsubsection. 
\end{remark}

  The proof of the above Lemma will not be given here, for it has already been considered in the literature 
 in a very similar form. The reader is referred, for instance, to \cite[pps 897-899]{Benedikter2021}.

 	\subsection{The interaction picture}
  Let us now exploit the identity found in \eqref{lemma h eq 1}.  
First, recalling that  the Hamiltonian $\h_0$ is quadratic and diagonal with respect to creation and annihilation operators, 
  we  may  easily calculate the associated  Heisenberg evolution to be given by 
  \begin{align}
  	a_p(t)  & \equiv e^{ it \h_0   } a_p e^{ - i t \h_0 }  = e^{ - i t E_p} a_p   \ , \\
  	a_p^* (t)  &  \equiv e^{ it \h_0   } a_p^* e^{ - i t \h_0 }  = e^{  +  i t E_p} a_p^*  	  \ , 
  \end{align}
for all $ p \in \Lambda^*$ and $t \in \R $ ; the dispersion relation $E_p$ was defined in \eqref{dispersion relation}. 
 Secondly, we introduce the 
 \textit{interaction Hamiltonian} 
 \begin{equation}
  \h_I (t)  \equiv   \lambda \, e^{ it \h_0 }\mathcal{V} e^{ - i t \h_0 } \qquad \forall t\in \R \  , 
 \end{equation}
where $\h_0$ and $\mathcal{V}$ are defined in Lemma \ref{lemma h eq 1}.

\vspace{2mm}

We now introduce the dynamics associated to the interaction picture. 

\begin{definition}
\label{definition interaction dynamics}
	Given an initial state $ \nu  :  B (\F) \rightarrow \C$, we denote by $(  \nu_t  )_{t \in \R}$ 
	the solution of the {\blu initial} value problem 
  \begin{align}
 	\label{interaction dynamics}
 	\begin{cases}
 		&  	 i 	\partial_t \nu_t (  \mathcal{O} )  
 		= 
 		\nu_t \big(   [  \h_I(t)  , \mathcal{ O } ]		\big)		
 		\quad \forall \mathcal{ O } \in B (\F ) 
 		\\
 		&  \nu_ 0 = \nu   
 	\end{cases}  
 \end{align}
which we shall refer to as the interaction dynamics. 
\end{definition}

The momentum distribution of the system $f_t(p)$, 
introduced in Def. \ref{definition momentum distribution}, 
is now linked to the interaction dynamics.
{\blu Indeed, a standard calculation shows that   the Schr\"odinger and the interaction picture agree.} That is, 
for all $t \in \R$ and
$ p \in \Lambda^*$   
	\begin{equation}
\label{interaction picture}
		f_t(p) =  |\Lambda|^{-1 } {\nu_t (  a_p^* a_p )}   \ . 
	\end{equation}
 In the next subsection, we shall use Eq. \eqref{interaction picture} 
 to expand $f_t(p).$

 \subsection{Second order perturbative expansion}  
Let $f_t(p)$ be  as in Eq. \eqref{interaction picture}, 
and let us recall that $\nu$ is  an initial state satisfying Condition \ref{condition initial data}.  
In particular, quasi-freeness and translation invariance imply that 
  \begin{equation}
  	\label{nu vanishing}
	\nu (  [a_p^* a_p  ,  a_{k_1}^\hash   a_{k_2}^\hash    a_{k_3}^\hash   a_{k_4}^\hash       ]  ) = 0 \ , \qquad \forall \, k_1, k_2, k_3, k_4 \in \Lambda^* \ . 
\end{equation}
Thus, upon expressing the Hamiltonian $\h_I(t)$
in terms of creation- and annihilation operators, 
one finds that 
 $\partial_t |_{t=0}   f_{t }  (p) 
 =
   i |\Lambda|^{-1} \nu (  [  a_p^*a_p  ,  \h_I (0)  ]  )=0$.
Hence,  
the following second   order expansion holds true
 \begin{equation}
 	\label{double commutator}
 	f_t(p )
 	=
 	f_0(p) 
 	- 
| \Lambda|^{-1 } 
 	\int_0^t 
 	\int_0^{t_1}
 	\nu_{ t_2  }
 	\Big( 
 	[[ a_p^* a_p , \h_I (t_1)    ] , \h_I( t_2 ) ]   
 	\Big)  
 	\d t_1 \d t_2
 \end{equation}
 for any $t \in \R $ and $ p \in \Lambda^*. $
 We dedicate the rest of this article to the study of the right-hand side of the above equation.
 
 \vspace{1mm}
 
 Let us identify   the terms in the {\blu second order perturbative expansion} found above. 
A straightforward expansion of the interaction Hamiltonian yields the decomposition 
\begin{equation}
\label{interaction hamiltonian 2}
	\h_I(t) = \lambda \big(
	V_F(t) + V_{FB} (t) + V_B(t) 
 	\big) \qquad \forall t \in \R 
\end{equation}
where the interaction terms  {\blu evolve}  according to the Heisenberg picture. Namely, we set 
\begin{equation}
\label{Heisenberg evolution V}
	V_\alpha (t)   \equiv 
 \, e^{ it \h_0 }
 V_\alpha  
 \, e^{ -  it \h_0 }
 \qquad \forall t\in \R \  , \alpha \in  \{  F , FB  , B  \} \ . 
\end{equation}
Upon expanding  the right hand side of
\eqref{double commutator}, 
  one finds the following nine terms 
 \begin{align}
	\nonumber 
	f_t  - f_0  =  & 
	- 
	\lambda^2 |\Lambda|^{ -1 }
	\Big(
	T_{F,F} (t) 
	+ T_{F, FB}  (t)
	+  T_{F , B } (t)
	\Big)
	\\
	\nonumber 
	&  	- 
		\lambda^2 |\Lambda|^{ -1 }
	\Big(
	T_{FB,F} (t) 
	+ T_{FB, FB}  (t)
	+  T_{FB , B } (t)
	\Big)
	\\ 
	\label{expansion f}
	&
	-  
		\lambda^2 |\Lambda|^{ -1 }
	\Big(
	T_{B,F} (t) 
	+ T_{B, FB}  (t)
	+  T_{B , B } (t)
	\Big)
\end{align}
  where we set, for $t \in \R$  and $ p \in \Lambda^*$
 \begin{align}
\label{T alpha beta}
 	T_{ \alpha , \beta }( t,p )
  \equiv 
 	\int_0^t 
 	\int_0^{t_1}
 	\nu_{ t_2  }
 	\Big( 
 	[[ a_p^* a_p , V_\alpha   (t_1 )    ] , V_\beta (t_2 ) ]   
 	\Big)  
 	\d t_1 \d t_2 
 	\qquad 
 	\alpha , \beta  \in \{ F, FB , B \}   \ . 
 \end{align} 
 We shall analyze in detail      the quantities $T_{\alpha , \beta}: \R \times \Lambda^* \rightarrow \R$ 
 when tested against a smooth function. 
 To this end, let  us introduce some notation we shall be using for the rest of this work.
 For $\vp : \Lambda^* \rightarrow \C $ we let 
 \begin{equation}
 	N (\vp )  
 	\equiv 
 	\int_{  \Lambda^* }
 \overline{ 	\vp (p) } 
 	a_p^* a_p  
 	\d p  
 \end{equation}
 together with 
 \begin{align}
\label{T alpha beta phi}
 	T_{\alpha , \beta } 
 	(t , \vp )
    \equiv 
\< \vp ,   T_{\alpha,\beta} (t) \>
 	= 
 	\int_0^t 
 	\int_0^{t_1}
 	\nu_{t_2 }
 	\Big( 
 	[[  N(\vp ), V_\alpha   (t_1)    ] , V_\beta (t_2 ) ]   
 	\Big)  
 	\, \d t_1 \d t_2 \ . 
 \end{align}

\subsection{Excitation operators}
The following two operators will play a major role in our analysis. 
{\blu They are}  the number operator (per unit volume) that counts the  total number of particles and holes in the system, 
together with the number operator that  only counts the number of particles and holes in the Fermi surface $\S$. 
More precisely, we consider

\begin{definition}
We define the two following operators in $\F$.

\noindent (1) The number operator as 
\begin{equation}
	\N   \equiv 
	 \int_{ \Lambda^*} a_p^*a_p \d p  \ . 
\end{equation} 

\noindent (2) The 
surface-localized number operator as 
\begin{equation}
\label{number boson}
	\po   
	\equiv 
	 \int_{  \S } a_p^*a_p  \d p 
\end{equation}
where $\S$ is the Fermi surface,    defined in \eqref{S definition} .
\end{definition}

\begin{remark}
	Let us recall that in Section \ref{section main results} 
	we have introduced the parameter $R = |\Lambda| \int_\S \d p $. 
	In particular, it follows from the boundedness of creation- and annihilation- 
	operators
	that $\N_\S$ is a {\blu bounded} operator
	and 
	$\|  \N_\S	\|_{B(\F)} \leq R $. 
\end{remark}

\begin{remark}[Domains]
$\N$ is an unbounded 
self-adjoint operator in $\F$
with domain
$\mathcal D (\N) 
=
\{	 \Psi = (\psi_n)_{n\geq0 } \in \F :  \sum_{n\geq0} n^2 
\| \psi_n 		\|^2_{L^2(\Lambda^n)} < \infty
		\} \  .  $
As initial data, the mixed states that we work with 
satisfy 
\begin{equation}
	\nu (\N) 
	 \equiv 
	 \int_{  \Lambda^{*} } \nu(a_p^*a_p) \d p 
	  = 
	  \int_{  \Lambda^{*} } f_ 0 (p) \d p < \infty  \  , 
\end{equation}
and similarly for higher powers $\N^k$.
It is standard to show that 
the time evolution generated by  the particle-hole Hamiltonian $\h$, 
as defined in \eqref{conjugated hamiltonian}, 
preserves $\mathcal D (\N)$, 
in the sense that $\nu_t( \N^k) < \infty$
for $t \in \R$ and $ k \in \mathbb N$.
In order to simplify the exposition, 
we shall  purposefully not refer to the unbounded nature
of the operator $\N$ in the rest of the article.
\end{remark}

%
%

The proof of Theorem \ref{theorem 1} relies on the fact that 
the subleading order terms that arise from the double commutator expansion --written in terms of $ b$- and $D$-operators--
can be bounded above by expectations of the operators $\N$ and $\N_\S $,
with respect to the evolution of the state $\nu$ driven by the interaction Hamiltonian $\h_I(t)$. 
This analysis is carried out in Section \ref{toolbox1}. 
Further,   in Section \ref{section number estimates} 
we prove bounds for the {\blu growth in  time}  of 
  the expectations $\nu_t(\N)$ and $\nu_t (\N_\S )$. 
  This two-step analysis is combined in Section \ref{section proof of theorem} to prove Theorem \ref{theorem 1}.

 \section{Tool Box I: Analysis of $b$- and $D$-operators}
 \label{toolbox1}
 
 In the last section, we introduced 
 the time evolution of certain observables in the Heisenberg picture, 
 with respect to the solvable Hamiltonian $\h_0$, introduced in \eqref{h0 definition}. 
 In particular, the evolution of the creation- and annihilation- operators $a$ and $a^*$
 takes the simple form 
 \begin{align}
 	a_p(t)    = e^{ - i t E_p} a_p   
 	\qquad \t{and}
 	\qquad 
 	a_p^* (t)   = e^{  +  i t E_p} a_p^*  	  \ , 
 \end{align}
 for all $ p \in \Lambda^*$ and $t \in \R $ ; the dispersion relation $E_p$ was defined in \eqref{dispersion relation}.  
 Let us now introduce the Heisenberg evolution of the $b$- and $D$-operators 
 as follows.
 
 \begin{definition}
 	\label{definition Heisenberg b D}
 	Let  $k \in \Lambda^*$ and $t \in \R$.\\
 	\noindent (1) 
 	The Heisenberg evolution of the $D$-operators is given by 
 	$$ 
 	D_k (t)
 	\equiv 
 	e^{ i t \h_0}
 	D_k 
 	e^{ -  i t \h_0}
 	=
 	\int_{  \Lambda^* }
 	\chi^\perp (p, p  - k )
 	a_{p -k }^* (t)  a_{p  } (t)  \,  \d p 
 	-
 	\int_{  \Lambda^* }
 	\chi (h , h+ k  )
 	a_{h + k }^* (t)  a_{ h  } (t)  \,  \d h    
 	$$
 	and $D_k^*(t) \equiv [ D_k(t)]^* = D_{-k} (t)$.
 	
 	\noindent (2) 
 	The Heisenberg evolution of the $b$-operators is given by 
 	$$
 	b_k(t) 
 	\equiv 
 	e^{ i t \h_0}
 	b_k  
 	e^{ -  i t \h_0}
 	 = \int_{  \Lambda^* }   \chi^\perp (p)  \chi (p-  k )  		 a_{p - k}(t) a_{p } (t)  \, \d p     
 	$$ 		
 	and $b_k^*(t) \equiv [ b_k(t)]^*$.
 \end{definition}

The main goal of this section is to introduce a systematic calculus that lets us deal with a combination of the operators $b_k(t)$ and $D_k(t)$ -- together with multiple combinations of their commutators -- 
as they show up in the analysis of the double commutator expansion found in 
\eqref{expansion f}. 
First, we introduce many useful identities required for the upcoming analysis. 
Secondly, we state estimates for several combinations of $b$-  and $D$-operators.

\subsection{Identities}
In this subsection, we record useful identities between operators in $\F$ that we shall use extensively in the rest of this article. 
Most importantly, in the next subsection, we shall use these identities  to obtain estimates of 
{\blu important} commutator observables. 
 \\

\noindent \textit{Preliminary  identities }.  
First, we write general time-independent relations. 

\vspace{2mm}

\noindent 1) For all $ p,q ,r \in \Lambda^*$ the CAR   imply that  
 \begin{align}
 [ a_r^*a_r, a_p^*  a_q ] 
 =
  \big(
  \delta ( r - q ) - \delta(r - p )
  \big) 
  a_p^* a_q
 \end{align}

\vspace{1mm }

\noindent  2) For all $p,q \in \Lambda^*$ and $ \vp \in \ell^1(\Lambda^*)$ there holds 
 \begin{equation}
 [N(\vp) , a_p^* a_q] =  \Big(		\overline{ \vp(p)}  -  \overline{\vp(q ) } \Big) a_p^* a_q
 \end{equation}
 where we recall $N(\vp)  =  \int_{  \Lambda^{*} }  \overline{\vp(p)}  a_p^* a_p  \d p .$   

 \vspace{3mm}

\noindent \textit{Commutator  identities}.   
 The following lemma contains useful operator  identities, to be used in the next section. 
 Since they only rely on the CAR and straightforward commutator calculations,
  we leave their proof to the reader.

\begin{lemma}
Let $k, \ell \in \Lambda^*$ and  $t,s  \in \R$ . \\
\noindent (1) For $ p \in \mathcal B^c$ and $h \in \mathcal B$ there holds 
\begin{align}
	[b_k(s), a_p^* (t) ]
 & 	\  = \ 
	\chi  (p - k )
	\, 
	e^{ i(t-s) E_p}
	a_{p -  k } (s)    \ ,   \\ 
	[ b_k (s) , a_h^* (t) ]
 & 	\  = \ 
	- \chi^\perp 
	( h +  k  )
	\, 
	e^{ i(t-s) E_h }
	a_{h +  k  } (s)    \ . 
\end{align}

\vspace{2mm}

\noindent (2) There holds 
 \begin{align}
	[b_\ell (s)	  , 		D_k^* (t) 	]
	&  = 
	\int_{  \Lambda^* }
	\chi^\perp (p)
	\chi^\perp(p - k )
	\chi  (p-\ell )
	e^{ i (t-s)E_p}
	a_{p- \ell} (s)
	a_{p - k }(t)			
	\d  p 	\nonumber \\
	& 
	+ 
	\int_{  \Lambda^* }
	\chi ( h ) 
	\chi  ( h + k )
	\chi^\perp  ( h + \ell  )
	e^{ i (t-s)E_h}
	a_{  h + \ell } (s)
	a_{  h + k  }(t)
	\d  h \  . 
	\label{commutator b D*}
\end{align}
In particular, $[b_k(t) , D_k^* (s)] = 0 $. 

\vspace{2mm }

\noindent (3) There holds 
\begin{align}
	\label{boson commutator}
	[b_k (t), b_\ell^* (s)]
	& 
	=   \ 
	\delta ( k - \ell )
	\int_{  \Lambda^* }
	\chi^\perp(p)
	\chi(p-k)
	\nonumber 
	e^{  - i (t-s ) (E_p   +  E_{p - k})  }					
	\d p  \\
	& 	 - 
	\int_{  \Lambda^* }
	\chi^\perp (p) 	 \chi^\perp (p+ \ell - k )	\chi(p-k) 			\nonumber  
	e^{  - i (t-s ) E_{p-k}    }	
	a_p^* (t)  a_{p + \ell - k } (s) \, \d p 	\\ 
	&  - 	  \int_{  \Lambda^* }
	\chi  ( h )
	\chi (  h + \ell - k )
	\chi^\perp (h + \ell)
	e^{  -i (t-s )  E_{h+k}    } 
	a_h^* (t)  a_{h + \ell - k } (s) \, \d h \ . 
	%
	%
	%
	%
	%
	%
	%
	%
	%
\end{align}
\end{lemma}

Besides the $b$- and $D$- operators, we shall work extensively with
their \textit{contracted} versions.  These are defined as follows
in terms of $N(\vp) = \int_{\Lambda^*}  \overline{\vp(p)} a_p^* a_p \d p . 	$

\begin{definition}[Contractions]
	Let $\vp : \Lambda^* \rightarrow \C$. 
Then, we define the $D(\vp)$-operators as the collection of operators for $t\in \R$ and $ k \in \Lambda^*$
	\begin{align}
 D_k(t,\vp) \equiv  [	N(\vp) , D_k(t)	]	 \qquad   \t{and}\qquad 	D_k^*(t, \vp) \equiv [N(\vp) , D_k^*(t)]  \ . 
	\end{align}
	Similarly, we define the $b(\vp)$-operators 
	as the collection of operators for $t\in \R$ and $ k \in \Lambda^*$
	\begin{align}
		b_k(t,\vp) \equiv [N(\vp ) , b_k(t)] \qquad  \t{and} \qquad  b_k^*(t,\vp)	\equiv [N(\vp) , b_k^*(t)] \ . 
	\end{align}
	We call them the contractions of $b$ and $D$ with the 
	  function $\vp$. 
	\end{definition}

\begin{remark}
	We immediately note that
\begin{align}
 	 [ D_k (t,\vp) ]^* =   - D^*_k(t , \bar \vp) \qquad [b_k (t,\vp)]^* =  - b^*_k(t,\bar \vp)
\end{align}
for all $t \in \R$ and $ k\in \Lambda^*$. 
Thus, the contractions are not adjoints of each other.
 However,  the following relations hold true for all $\Psi \in \F$
 \begin{equation}
 	\|	 	 [ D_k (t,\vp) ]^* \Psi	\|_\F  =  	\|	 	 D_k^* (t, \bar \vp)  \Psi	\|_\F 
 	\qquad
 	\t{and}
 	\qquad 
 	 	\|	 	 [ b_k (t,\vp) ]^* \Psi	\|_\F  =  	\|	 	 b_k^* (t, \bar \vp)  \Psi	\|_\F  \ .
 \end{equation}
Since final estimates are given in terms of $\ell^p$ norms of $\vp$, the complex conjugation 
does not affect the end result.
 Thus, when proving estimates, one may regard them as adjoints of each other. 
\end{remark}
 
\begin{remark}
	The contractions can of course be calculated explicitly using the CAR. 
	Let us record here the two following calculations 
	\begin{align}
		\label{operator D phi}
		D_k^* (t, \vp) 
		&    =
		\int_{  \Lambda^* }
		\chi^\perp(p)
		\chi^\perp(p-k)
		\big[											\nonumber 
		\vp(p) - \vp( p - k)
		\big] 
		a_p^*(t)  a_{p-k}(t)
		\d p 		\\
		&  \quad -   
		\int_{  \Lambda^* }
		\chi ( h )
		\chi( h + k )
		\big[
		\vp( h) - \vp( h+ k )
		\big] 
		a_h^*(t)  a_{ h + k }(t)  
		\d h \ , 
		\\ 
		b_k
		(t, \vp) 
		\label{b phi}
		& =
		\int_{  \Lambda^* }
		\chi^\perp  (q)
		\chi (q-k)
		\Big[
		\vp(q -k )
		+ 
		\vp(q  )
		\Big] 
		a_{q-k}(t) a_q (t ) 
		\ \d q \ . 
	\end{align} 
\end{remark}
 
Let us now state in the following Lemmas 
some useful commutation relations.
Since they all follow from straightforward manipulation of the CAR, 
we leave them as an exercise for the reader.


\begin{lemma}
Let  $k, \ell \in \Lambda^*$, $t,s \in \R$ and $\vp \in \ell^1$. \\
\noindent (1) There holds 
	 \begin{align}
\label{commutator b D phi}
		[ b_\ell  (s)  , 
		& 	D_k^*(t,\vp)	
		] \\
		&   =
		\int_{  \Lambda^* }
		\chi^\perp(p)
		\chi^\perp(p-k)
		\chi(p - \ell )
		\big[
		\vp(p)
		-
		\vp(p - k) 
		\big]
		e^{i(t-  s) E_p }
		a_{p - \ell } (s) 			\nonumber 
		a_{ p - k } (t) 		
		\d p 
		\\
		& 
		+
		\int_{  \Lambda^* }
		\chi(h)
		\chi(h+k )
		\chi^\perp ( h + \ell )
		\big[
		\vp(h)
		-
		\vp(h  + k ) 
		\big]
		e^{i(t-  s) E_h }
		a_{  h + \ell } (s) 			\nonumber 
		a_{  h +  k  } (t) 	
		\d h  \ . 	
	\end{align}
\end{lemma}   
 
%
%

 \begin{lemma}[$\N$ commutators]
\label{lemma N commutators}
For all $k \in \Lambda^*$ and  $t \in \R$   
the following holds true.

\vspace{1.5mm}

\noindent (1) For the $D$-operators
\begin{equation}
	[  D_k(t) , \N 			] = [ D_k^*(t) , \N ]  = 0  
\end{equation}
 	and similarly for the contracted operators 
 	$D_k(t,\vp)$ and $D_k^*(t,\vp)$. 
 	
 	\vspace{1.5mm}
 	
 	\noindent (2) For the $b$-operators, 
 	for any measurable function $f : \R \rightarrow \C $ the
 	pull-through formulae holds true 
 	  \begin{align} 
\label{pull thru formula b}
 	f( 	\N  ) 
 		b_k (t) 
 		= 
 		b_k (t)   f  (\N -2 ) 
 		\quad
 		\t{and}
 		\quad 
 		 	f(	\N  ) 
 		b_k^*(t) 
 		= 
 		b_k^* (t ) f   (\N + 2 ) 
 	\end{align}
 	and similarly for the contracted operators
 	$b_k(t,\vp)$
 	and
 	$b_k^* (t,\vp)$. 
 	
 \end{lemma}

  \begin{lemma}[$\N_\S$ commutators]
  	\label{lemma NS commutators}
  	For all $k \in 3 \supp \hat V$ and  $t \in \R$   
  	the following commutation relations hold true 
  	\begin{eqnarray}
  		[\N_\S , b_k(t) ] = - 2 b_k(t) 
  		\quad
  		\t{and}
  		\quad 
  		  		[\N_\S , b_k^*(t) ] = + 2 b_k^* (t)  \ . 
  	\end{eqnarray}
  \end{lemma}

\subsection{Estimates}
In this subsection we state estimates that shall be used extensively for the rest of this article. 
Most of these are operator estimates for observables in $\F$ containing the fermionic creation
and annihilation  operators $a_p$ and $a_p^*$.
We remind the reader that these are bounded operators with norm
$ 
		\|	a_p	 \|_{B (\F)}    \,   =  \,      \|	 a_p^*  	\|_{B(\F) }  \  \leq  \  | \Lambda|^{1/2}   
	 $ 
	 for all $ p \in \Lambda^* \ . $ \\

 \noindent \textit{Preliminary  estimates}.   
  Let us state    some  elementary estimates that we shall make use of. 

\vspace{2mm }

\noindent 1) For any   function $f : \Lambda^* \rightarrow \C$, 
  $ k \in \Lambda^* $ and   $\Psi \in \F$ there holds 
 \begin{equation}
 \label{cross estimate}
\Big\|		
\int_{  \Lambda^* } f(p) a_{p+k}^* a_{p} \d p 
\Psi 
 \, 
 \Big\|_\F 
\leq 
  \ 	\|f 	\|_{ \ell^\infty   }
	\|	 \N \Psi 	\|_\F  \ . 
 \end{equation}

\vspace{1mm}

\noindent 2)   The Heisenberg evolution  of the creation- and annihilation- operators $a_p (t) $ and $a_p^*(t)$ are bounded operators in $\F$,  with norms 
  \begin{equation}
\label{bounds on a}
  \|	a_p	(t) \|_{B (\F)}   \,   = \,    \|	 a_p^* (t) 	\|_{B(\F) }  \  \leq  \  | \Lambda|^{1/2} 
   ,
   \qquad \forall t  \in \R , \   p \in \Lambda^* \ . 
  \end{equation}

\vspace{1mm}

\noindent 3) 
The Heisenberg evolution of the 
 $b$-operators  are bounded operators in $\F$ with norms   
\begin{equation}
\label{bounds on b}
\|	 b_k(t) 	\|_{B(\F) }
\, = \, 
 \|   b^*_k(t) 	\|_{B(\F) 	}
 \ \leq  \ 
 | \Lambda|
  \int_{  \Lambda^* } \chi^\perp (p) \chi(p-k) \d p 
 \  \lesssim \ 
R 
\end{equation}  
for all $ k \in \supp \hat V $ and $ t \in \R$. Let us recall that $R =   |\Lambda| p_F^{d-1 }$ . \\ 

{\blu 
 \begin{proof}
 (1) 	 Let $\Phi \in \F$ and note that a  two-fold application of  the Cauchy-Schwarz inequality  gives
 \begin{align*}
 \Big|    	 \< \Phi,   \int_{  \Lambda^* } f(p) a_{p+k}^* a_{p} \d p  \Psi \>_\F  \Big| 
 & \leq \| f \|_{\ell^\infty}  \int_{  \Lambda^* }   dp \|  a_{p+k } \Phi  \|_\F \|    a_p \Psi \|_\F  \\ 
 & \leq \| f \|_{\ell^\infty}  \|  \N^{1/2} \Phi \|_\F \| \N^{1/2} \Psi	 	\|_\F \ . 
 \end{align*}
The identity  $\int_{  \Lambda^* } f(p) a_{p+k}^* a_{p} \d p  = \N^{-1/2 } \int_{  \Lambda^* } f(p) a_{p+k}^* a_{p} \d p  \N^{1/2}$ combined with the previous estimate is sufficient to finish the proof after taking the supremum over $\Phi$. \\

(2) This is an easy consequence of the unitarity of  the evolution group 
$e^{itH}$ and $\| a_k \|_{ B (\F)} \leq  |\Lambda|^{1/2}$, which follows from the CAR.  \\

(3) 
From the unitarity of the evolution group,  we obtain 
\begin{equation}
	 \|  b_k(t)	\|_{ B(\F)}	  = 
	  \|  b_k  \|_{ B(\F)}	 \leq 
	  \int_{\Lambda^*} \chi^\perp (p) \chi(p -k) \|   a_{p-k} a_p \|_{ B(\F)} dp 
	  \leq 
 |\Lambda|
	  \int_{\Lambda^*} \chi^\perp (p) \chi(p -k)   dp  \ . 
\end{equation}
For the final estimate, we observe that $ p\mapsto \chi^\perp (p) \chi (p-k)$ is supported 
on a neighborhood of order $|k|$ around the Fermi surface, with radius $p_F \gg1 $.
Thus, a   point counting estimate shows that 
$	  \int_{\Lambda^*} \chi^\perp (p) \chi(p -k)   dp   \lesssim  |\Lambda|  | k| p_F^{d-1}$. 
Since $ k \in \supp \hat V$, the value $|k|$ may be absorbed in the constant. This finishes the proof. 
 \end{proof}
}

\noindent \textit{Commutator  estimates}.  
Let us now describe  the most important estimates concerning $b$- and $D$-operators. 
Essentially, commutators between $b$- and $D$-operators--together with their contracted versions $b(\vp)$ and $D(\vp)$--can be classified into four types, depending on the estimate they verify. 
It turns out that these four types of estimate exhaust \textit{all} possibilities that show up 
in the double commutator expansion for $f_t(p)$.
In other words, these estimates are enough to analyze the nine terms 
$\{T_{\alpha ,  \beta } (t,p )\}_{	\alpha, \beta	 \in \{ F, FB,B\} } $.

\vspace{2mm}

We remind the reader of the relation $D_k^*(t)   =  D_{-k} (t)$, valid for all $ k \in \Lambda^*$ and $t \in \R$.
In particular, \textit{all} of the upcoming inequalities are valid if we replace $D$ by  $D^*$.
On the other hand, 
we warn the reader that 
this property  \textit{does not} hold for $b$-operators in general.


  \vspace{2mm}

  The first type of estimate concerns the combination of operators 
  that are relatively bounded with respect to the number operator $\N = \int_{  \Lambda^{*} }a_p^*a_p \d p $, 
  or any of its powers. 
  We call these \textit{Type-I} estimates. 
  They are contained in the following lemma.

\begin{lemma}[Type-I estimates]
	\label{lemma type I}
There exists a constant $C>0$
such that  for any 
$\Psi \in \F $,  $k,\ell \in \Lambda^*$,  
		  and $t ,s,r   \in \R$
		  the following  inequalities hold true  
  \begin{align}
\label{lemma type I.1}
	\|		D_k(t) \Psi	\|_\F  
	&  \ \leq  \  
	C 
  	\|		\N \Psi	\|_{\F } \\ 
  	\label{lemma type I.1 2}
	\| [   D_k (t)  , D_\ell (s)    ]		\Psi \|_\F 
	&   \   \leq  \ 
	C
	\|	 \N \Psi	\|_\F 		\\
	  	\label{lemma type I.1 3}
	\|	[D_k  (t)  , D_\ell  (s)  D_\ell (r) ]  
	\Psi	\|_\F 
	&    \  \leq  \ 
	C 
	\|	 \N^2 \Psi 	\|_\F  	 \  .  
\end{align}
	\end{lemma}
  
  \begin{remark}
  	{\blu  
  	 Previously in the literature,  the first estimate \eqref{lemma type I.1} was considered
  	 in a similar fashion in \cite{Benedikter2020}. 
  	 To the best of our knowledge, higher-order commutator estimates  like
  	 \eqref{lemma type I.1 2}  and 
  	   	 \eqref{lemma type I.1 3}   are new. 
  	   	} 
  \end{remark}

  The second type of estimates concerns combination of operators that can be bounded  above 
  by 
  the surface-localized number operator
  $\po  =  \int_\S a_p^* a_p \d p$,
  up to pre-factors that can grow with the recurring parameter 
  $R   =  |\Lambda| p_F^{d-1}$. 
  We call these \textit{Type-II} estimates, and they are contained in the following lemma

  \begin{lemma}[Type-II estimates]
\label{lemma type II}
There exists a constant $C>0$
such that for any $\Psi \in \F $,  $ k, \ell ,q   \in \supp \hat V $,    
and $t ,s,r   \in \R$
the following inequalities  hold true

  	\begin{align}
\label{lemma type II.1}
  		\|	b_k(t) \Psi 	\|_\F 
  		& 
   \ 	\leq  \  C 
  	R^{\frac{1}{2}}
  	\, 
  		\|	 \po^{1/2} \Psi 	\|_\F 
  		\\
  		\|	[b_\ell  (t) , D_k(s)] \Psi 	\|_\F 
  		&    
 \   		\leq   \ C 
  	R^{\frac{1}{2}}
  	 \, 
  		\|		 \po^{1/2} \Psi   	\|_\F  
  		\\
  		\|	 [ [b_\ell  (t) , D_k (s) ] , D_q (r)   ]\Psi 	\|_\F 
  		&  \    \leq  \  C 
  	R^{\frac{1}{2}} 
  	\, 
  		\|		 \po^{1/2} \Psi   	\|_\F  \ . 
  	\end{align}
  \end{lemma}
  
  \begin{remark}
  	In certain proofs, it will be convenient to use the upper bound 
\begin{equation*}
  	\N_S \leq \N . 
  	\end{equation*}
  	The reader should  then have in mind   that the (weaker) version 
  	of the      estimates contained in Lemma \ref{lemma type II}, 
     	in which     $\N_\S$ is replaced by $\N $, 
     	 also holds true. 
  \end{remark}

  The third     type of estimate corresponds to a combination of 
  operators that have been contracted with a test function $\vp \in \ell_m^1$, 
  and  their operator norm can be bounded above in terms of the integral 
  \begin{equation}
  	\int_\S |\vp(p) | \d p   \  \lesssim \    p_F^{- m }	\|	 \vp \|_{\ell_m^1 }  \  . 
  \end{equation}
We call these \textit{Type-III} estimates, and they are contained in the following lemma.

  \begin{lemma}[Type-III estimates]
\label{lemma type III}
Let $m>0$. 
There exists a constant $C>0$ such that   	
for all $ k, \ell , q  \in  \supp \hat V  , \, $
 $t , s, r \in \R$ and $\vp  \in \ell^1_m(\Lambda^*) $
  	the following inequalities {\blu hold}  true 
  	\begin{align}
\label{eq:1}
  		  \|	b_k(t, \vp)	\|_{	B(\F) 	} 
   &     	 \ 	\leq    	 \  
   C  \, 
|	\Lambda	| 
\, p_F^{-m}   \, 
\|	\vp 	\|_{\ell_m^1 }  \\ 
\label{eq:2}
  		   	\|		 [b_\ell (t)  , D_k(s, \vp)]	 	\|_{B(\F) }
  		 &     	 \ 	\leq   	 \ 
  		 C\, 
  		|	\Lambda	| 
  		\, p_F^{-m}   \, 
  		\|	\vp 	\|_{\ell_m^1 }  \\ 
  		\label{eq:3}
  		 \|	 [	[ 	 b_k(t )  , D_\ell(s)  	]	 , D_q(r, \vp)]	\|_{B(\F) 	}
  		 &     	 \ 	\leq   	 \ 
  		 C\, 
  		|	\Lambda	| 
  		\, p_F^{-m}   \, 
  		\|	\vp 	\|_{\ell_m^1 }  
  	\end{align}
  	\end{lemma} 
\begin{remark}
	Type-III estimates are symmetric with respect to the exchange of $b$ and $b^*$.
	This property follows from the relation $\| 	 \O	\|_{B(\F)} = \|	 \O^*	\|_{B(\F)}$
		and the symmetry  $D_k^* (t) = D_{-k}(t)$. 
\end{remark}

    The fourth and final     type of estimate corresponds to combination of 
  operators that have been contracted with a test function $\vp \in \ell_m^1$, 
  and  their operator norm can be bounded above in terms of the integral 
  \begin{equation}
  	\int_{\Lambda^*} |\vp(p) | \d p \   = \  \|	 \vp 	\|_{\ell^1 }  \  \lesssim \   	\|	 \vp \|_{\ell_m^1 }  \ , 
  \end{equation}
and a pre-factor, depending on the volume of the box $|\Lambda|. $
  We call these \textit{Type-IV} estimates, and they are contained in the following lemma.

  \begin{lemma}[Type-IV estimates]
\label{lemma type IV}
There exists a constant $C>0$ such that   	
for all $ k, \ell , q  \in  \Lambda^* ,  \, $
$t , s, r \in \R$ and $\vp  \in \ell^1(\Lambda^*) $
the following inequalities  hold true
  	  	\begin{align}
  		\|	D_k (t, \vp)	\|_{	B(\F) 	} 
  	& 	\ 	\leq \  
  	C   
  		\, 
  		|\Lambda |
  		\|	\vp\|_{\ell^1 }  \\
  		  	\|    [  D_k (t, \vp ) , D_\ell (s) ]		\|_{	{B} (\F) 	} 	 
  	& 	\ 	\leq    \ 
  	C \, 
|\Lambda |
\|	\vp\|_{\ell^1 }   \  . 
  	\end{align}

  	  \end{lemma}

Let us now turn to the   proofs of  Lemmas \ref{lemma type I}, \ref{lemma type II}, \ref{lemma type III} and \ref{lemma type IV}. 
  
\begin{proof}[Proof of Lemma \ref{lemma type I}]
Let us fix $\Psi \in \F $,  $k,\ell \in \Lambda^*$,  
and $t ,s,r   \in \R$. 

\vspace{1mm}

\noindent  \textit{\blu Proof of} \eqref{eq:1}.
We shall make use of the elementary estimate 
found in \eqref{cross estimate}.
To this end, starting from \eqref{definition Heisenberg b D} 
we decompose
\begin{equation}
	D_k(t) 
	 = 
	 \int_{  \Lambda^{*} }  
	 f^{(1)} (t,k,p) 
	 a_{p-k}^* a_p \d p 
	  + 
	  	 \int_{  \Lambda^{*} }  
	  f^{(2)} (t,k,h) a_{h + k }^* a_h \d h 
\end{equation}
where 
$	 f^{(1)} (t,k,p) = \chi^\perp(p,p-k) e^{it(E_{p-k}-E_p)}$
and 
$	 f^{(2)} (t,k,h) = \chi(h,h+k) e^{it(E_{h+k}-E_h)}$. 
Clearly, $\|	 f^{(1)}	 (t,k )\|_{\ell^\infty}  =  \|	 f^{(2)}	 (t,k )\|_{\ell^\infty}  =  1$.
Hence, it follows that 
$\|	 D_k(t) \Psi	\|_\F \leq 2 \|   \N\Psi	\|_\F$.

\vspace{1mm}

\noindent  {\textit{\blu Proof of} \eqref{eq:2}.} 
The proof is extremely similar--it suffices to note that the commutator can be calculated explicitly to be 
\begin{align}
\nonumber 
	[D_k(t), D_\ell(s)]
	= 
& 	\int_{  \Lambda^{*} }    \chi^\perp(	p , p -\ell , p - k - \ell 	)	e^{  i  (s-t) E_{p-\ell} }  a^*_{p-k-\ell} (t) a_p(s)		 \d p 	\\
\nonumber 
&  \quad -  
	\int_{  \Lambda^{*} }    \chi^\perp(	p , p - k , p - k - \ell 	)	e^{  i  (t-s) E_{p- k } }  a^*_{p-k-\ell} (s) a_p(t)		 \d p \\
\nonumber 
	&  \quad + 
		\int_{  \Lambda^{*} }    \chi (	h  , h + \ell ,  h+ k +  \ell 	)	e^{  i  (s-t) E_{h+ \ell} }  a^*_{ h + k + \ell} (t) a_h (s)		 \d h   \\
		&  \quad - 
			\int_{  \Lambda^{*} }    \chi  (	h  , h  + k  , h + k+  \ell 	)	e^{  i  (t-s ) E_{ h+ k } }  a^*_{h + k \ell} (s) a_h(t)		 \d h   \ . 
\end{align}
Hence,  the same argument shows that 
$\| [D_k(t) ,D_\ell(s)]  \Psi \|_\F  
 \leq 4 \|   \N  \Psi		\|_\F$.

\vspace{1mm}

\noindent  {\textit{\blu Proof of} \eqref{eq:3}.} 
For simplicity, let us suppress the time labels, and the momentum variables. 
In what follows $C>0$ is a constant whose value may change from line to line. 
We calculate using the previous results, and the commutation relations $[\N ,D]= 0$
\begin{align}
\nonumber
	\|  [ D ,D D ] \Psi		\|_\F 
& 	\leq 
	\| 		D [D,D]	\Psi		\|_\F
	+ 
	\| 	 [D,D]	 D	\Psi		\|_\F			\\ 
\nonumber
& 	\leq 
	C 	\| 		\N [D,D]	\Psi		\|_\F
	+ 
	C 	\| 		\N D	\Psi		\|_\F		\\ 
\nonumber
 & 	= 
	C 	\| 		 [D,D]	 \N \Psi		\|_\F
	+ 
	\| 	 C 	D\N 	\Psi		\|_\F		\\ 
& 	\leq C \|	 \N^2 \Psi	\|_\F  \ . 
\end{align}
This finishes the proof.
\end{proof}

  \begin{proof}[Proof of Lemma \ref{lemma type II}]
Let us fix  $\Psi \in \F $,  $ k, \ell ,q   \in \supp \hat V $,    
  and $t ,s,r   \in \R$. 
  \vspace{1mm}
  
  Let us give the main ideas behind the proof. 
  Let us recall that $\supp\hat V$ is contained in a ball of radius $r>0$. 
For $ n \in \mathbb N $,   define the Fermi surfaces 
  \begin{equation}
  	\S ( n) \equiv \{ 	 p \in \Lambda^* \ : \  p_F - n r \leq |p| \leq p_F + nr 		\}  , 
  \end{equation}
and the
number operators $\N_{\S(n)	} \equiv \int_{\S(n)} a_p^* a_p \d p$. 
  In particular, we are denoting $\S = \S(3)$ in \eqref{fermi sphere}. 
Given $k,\ell \in \supp \hat V$, 
consider operators of the form
\begin{align}
\beta_k   
\equiv
  \int_{  \Lambda^{*} }  \,  \1_{ \S(1)} (p)		\,  a_{p+k } a_p  \, \d p  
\  ,  \qquad \t{and} \qquad 
\mathcal D_\ell   \equiv  \int_{  \Lambda^{*} }  a_{p+\ell}^* a_p \d p  \ . 
\end{align}
  One should think generically  of $\beta_k$  as $b_k(t)$
  and $\mathcal D_\ell$ as $D_\ell(s)$. 
We make the following two observations.
First, $\beta_k$ can be controlled   by $\N_{\S(1) }$ in the following sense 
\begin{align}
		\nonumber 
		\| \beta_k  \Psi	\|_\F
		&  \leq 
		|\Lambda|^{\frac{1}{2}} 
		\int_{  \Lambda^{*} } 
		\1_{	\S(1)	}(p)
		\|	  a_p \Psi	\|_\F 			 \\
		\nonumber 
		& \leq
		|\Lambda|^{\frac{1}{2}} 
		\Big(
		\int_{  \Lambda^{*} }
		\1_{\S(1) 	}(p) \d p 
		\Big)^{\frac{1}{2}}
		\Big(
		\int_{  \Lambda^{*} }
		\1_{\S(1) }(p) 
		\|	   a_p \Psi	\|_\F^2
		\Big)^{\frac{1}{2}}					\\
		&  \lesssim 
				|\Lambda|^{\frac{1}{2}}  
p_F^{	\frac{d-1}{2}	}
		\|	 \N_{\S(1) }^{\frac{1}{2}} \Psi	\|_\F 	  
		 =
		R^{\frac{1}{2}}
				\|	 \N_{\S(1) }^{\frac{1}{2}} \Psi	\|_\F 		\  , 
\label{beta bound}
	\end{align}
where we used a basic geometric estimate 
to find that 
$
	\int_{  \Lambda^{*} }
\1_{\S(1) 	}(p) \d p  
\lesssim p_F^{d-1}.
$
Secondly, the commutator between $\beta_k$ and $\mathcal D _\ell$ can be calculated to be 
\begin{align}
	[\beta_k , \mathcal D _\ell]
	= 
	\int_{  \Lambda^{*} }
		\1_{	\S(1)	} ( p -  \ell) 
	a_{ p + k - \ell}	 a_p \d p 
	+ 
	\int_{  \Lambda^{*} }
	\1_{	\S(1)	} ( p ) 
	a_{ p + k - \ell}	 a_p \d p   \ . 
\end{align}
Since both $k,\ell \in \supp \hat V$, 
it holds that $\1_{\S(1)}  (p - \ell ) \leq \1_{\S(2)} (p)$, 
and of course 
$\1_{\S(1)} (p) \leq \1_{\S(2)}(p)$.
Consequently, the same argument that we used to obtain \eqref{beta bound} can now be repeated on each term of the above equation 
to obtain
\begin{equation}
	\|	 [\beta_k , \mathcal D _\ell] 	\Psi 	\|_\F 
	\lesssim 
			R^{\frac{1}{2}}
	\|	 \N_{\S(2) }^{\frac{1}{2}} \Psi	\|_\F 		\ .   
\end{equation}
  The same argument  can be repeated for the next commutator with $\mathcal D _q$,
  provided one enlarges the Fermi surface
  from $\S(2)$ to $\S(3)$.
  In other words, it holds that 
  \begin{equation}
  		\|	 [  [\beta_k , \mathcal D _\ell]  , \mathcal D _q]	\Psi 	\|_\F 
  	\lesssim 
  	R^{\frac{1}{2}}
  	\|	 \N_{\S(3) }^{\frac{1}{2}} \Psi	\|_\F 		\ .   
  \end{equation}

 \vspace{1mm}
 The above motivation contains the main ideas for the proof of the lemma. 
One merely has to include additional  bounded coefficients 
in the definition of $\beta_k$ and $\mathcal D _\ell $
to account for the dependence on $t \in \R$ and $ k \in \Lambda^*$,
that comes from $b_k(t)$ and $D_\ell(s)$. 
We leave the details to the reader.
   	  \end{proof}

  \begin{proof}[Proof of Lemma \ref{lemma type III}]
Let us fix $m>0$,  $ k, \ell , q  \in  \supp \hat V  , \, $
$t , s, r \in \R$ and $\vp  \in \ell^1_m(\Lambda^*) $.
  	Starting from Eq. \eqref{b phi} we easily estimate that 
  	 \begin{align}
  	 	\|	 b_k(t,\vp)	\|_{B(\F)}
  	 	\leq 
2   	 	|\Lambda|  \int_{\Lambda^*} \1_\S(p) |\vp(p)|  \d p  \ . 
  	 \end{align}
It suffices then to note that 
$\int_{\S} |\vp(p)| \d p \lesssim  p_F^{-m}  \|	 \vp	\|_{\ell_m^1 }$. 
For the next estimate,  the same analysis can be carried out, 
starting from the commutator identity found in Eq. \eqref{commutator b D phi}. 
For the last estimate, one has to calculate the upcoming 
commutators and bound each term in the same way.
  \end{proof}

  \begin{proof}[Sketch of Proof of Lemma \ref{lemma type IV}]
Let us fix $k \in \Lambda^*$ and $\vp \in \ell^1$.
  	Starting from Eq. \eqref{operator D phi} 
  	we use $ 0 \leq \chi , \chi^\perp \leq1$
  	and
  	$\|	 a_p (t)	\|_{	B(\F)	} = \|	 a_p^* (t)	\|_{B(\F)} \leq |\Lambda|^{\frac{1}{2}}$
to find 
  	\begin{align}
  		\|	 D_k(t,\vp)				\|_{	B(\F)	} 
  		\leq 4 |\Lambda| 		 \int_{\Lambda^*}		|	\vp(p)	| \d p \ . 
  	\end{align}
A similar inequality can be found upon calculation
of the commutator $ [   D_k(t)  ,  D_\ell ( s  , \vp )   ].$  	
This finishes the proof. 
  \end{proof}

 \section{Tool Box II: Excitation Estimates}
\label{section number estimates}
In Section \ref{section preliminaries}
we introduced the two following observables: 
\begin{equation}
	\N = \int_{  \Lambda^{*} } a_p^* a_p \d p 
	\qquad 
\t{ 	and } 
	\qquad 
	\N_\S  = \int_{ \S  } a_p^* a_p \d p 
\end{equation}
The  main purpose of this section is to prove
estimates
that control 
 the {\blu growth in time} of  the expectation of $\N$ and $\N_\S$ 
  with respect to the interaction dynamics   $(\nu_t)_{t\in\R }$, defined in \eqref{interaction dynamics}.
These estimates are precisely stated in the following two propositions, which we prove in the {\blu remainder} of this section.

\begin{proposition}
	\label{prop N estimates 1}
	Let   $( \nu_t)_{t \in \R}$
	solve the interaction dynamics defined in \eqref{interaction dynamics}, with initial data $\nu_0 \equiv \nu$
	satisfying Condition \ref{condition initial data}. 
	Assume that $n = \nu(\N) \geq 1 $. 
	Then, 
	for all $ \ell \in \mathbb N$ there exists a constant $C> 0 $ such that  
	\begin{equation}
		\nu_t (\N^\ell   ) \leq C n^\ell 
		\exp(  C \lambda R t )   
		\  , \qquad \forall t  \geq 0 \ . 
	\end{equation}
\end{proposition}

\begin{proposition}
	\label{prop N estimates 2}
	Let   $( \nu_t)_{t \in \R}$
	solve the interaction dynamics defined in \eqref{interaction dynamics}, with initial data $\nu_0 \equiv \nu$
	satisfying Condition \ref{condition initial data}. 
	Further, assume that $ n = \nu(\N)  \lesssim R^{1/2}$. 
	Then, there exists a constant  $C>0$ such that  
	\begin{equation}
		\nu_t (\po)
		\leq
		C(  \lambda R \<t\> )^2   
		\exp(  C   \lambda R  t )    \  , 
		\qquad \forall t  \geq 0 \ , 
	\end{equation}
where $ \< t \>  = (1 + t^2 )^{\frac{1}{2}}. $
\end{proposition}

 {\blu  
The idea behind the proof of our estimates
relies on a 
  Gr\"onwall argument, in which we
bound expectations of  
commutators  $[\N, \h_I(t)]$  and $[\N_\S, \h_I(t) ]$
in terms of combinations of  expectations of 
 $\N$ and $\N_\S$. 
 This then allows us to close the estimates after paying with a constant that grows exponentially fast with time. 
  The proof of these number estimates is heavily inspired by previous work on  the derivation of mean-field dynamics for Bose and Fermi  gases, and they are nowadays considered a standard tool in the derivation of nonlinear equations from quantum many-body systems.  
 	See for instance, \cite{Rodnianski2009} for Bose gases, and \cite{Benedikter2014} for Fermi gases, respectively. 
  
In the situation considered in this article, the proof  of the commutator estimates
 relies heavily on  the fact that 
  the interaction Hamiltonian decomposes into three parts, corresponding to fermion-fermion, fermion-boson and boson-boson interactions.
  Indeed, each term gives rise to different commutators with $\N$ and $\N_\S$, respectively, which require different estimates; see e.g. Lemma \ref{lemma commutator estimates 1} and \ref{lemma commutator estimates 2}.  }

  Let us recall that this decomposition reads
\begin{eqnarray}
	\label{h_I decomposition }
	\h_I(t)
 \ 	=  \ 
 	\lambda  \, 
	\big(	V_F(t)  + V_{F, B}(t)   + V_B(t)  \big) \ , \qquad \forall t \geq 0 \ . 
\end{eqnarray}
Here, time-dependence corresponds to the Heisenberg evolution associated to the solvable Hamiltonian $\h_0$--see Eq. \eqref{Heisenberg evolution V}. 
In particular, using the formulae \eqref{VF}, \eqref{VFB} and \eqref{VB} for $V_F$, $V_{F,B}$ and $V_B$, respectively, we may write that  
for all $t \in \R$   
\begin{align}
\label{VF(t)}
	V_F   (t) 
	&   = 
	\frac{1}{2}
	\int_{  \Lambda^* }
	\hat V (k) D_k^*(t )  D_k(t)   \  \d  k 
	\\ 
\label{VFB(t)}
	V_{FB}  (t) 
	& 	 = 
	\int_{  \Lambda^* } 
	\hat V (k)
	D_k^*(t) 
	\big[ 
	b_k(t) 
	+
	b_{ -k }^*(t) 
	\big]  \d  k \\
\label{VB(t)}
	V_{B}  (t) 
	& = 
	\int_{  \Lambda^* } 
	\hat V (k) 
	\big[ 
	b_k^*(t) b_k(t ) 
	+ 
	\frac{1}{2} \, 
	b^*_k (t)   b_{-k}(t ) 
	+ 
	\frac{1}{2} \, 
	b_{- k } (t) b_k(t) 
	\big]  \d  k   
\end{align}  
where $b_k(t)$ and $D_k(t)$ correspond to the Heisenberg evolution of the $b$- and $D$-operators, respectively, as given in    Definition \ref{definition Heisenberg b D}.

\subsection{Number Operator Estimates}
The main purpose of this section is to prove     Proposition \ref{prop N estimates 1}.
The first step  in this direction 
is to prove  appropriate commutator estimates between $\N$ and the generator of the interaction dynamics, $\h_I(t)$.
 The commutator estimates  that we prove are contained  in the upcoming Lemma. 
We recall that $R = |\Lambda| p_F^{d-1}$. 

\begin{lemma}
	[Commutator Estimates for $\N$]
\label{lemma commutator estimates 1}
For all $\ell \geq 1 $ there exists a constant
 $C = C(\ell) >0$ 
such that: 

\begin{enumerate}
 
	\item 
For all    $\Psi \in \F$ and $t \geq 0$ there holds  
\begin{align}
	 \<   
	\Psi 
	,   
	[\N^\ell 
	, 
	V_{F } (t) ]						\nonumber 
	\Psi   \>_\F     =0
	\end{align} 
	\item 
	For all    $\Psi \in \F$ and $t \geq 0$ there holds  
	\begin{align}
		|  \<   
		\Psi 
		,   
		[\N^\ell 
		, 
		V_{F,B} (t) ]						\nonumber 
		\Psi   \>_\F     |
\leq C 
		R 
		 \< 
		 \Psi ,  (\N^\ell  + \1 )
		 \Psi 
		 \>_\F   
	\end{align}
	\item 
For all    $\Psi \in \F$ and $t \geq 0 $  there holds  
	\begin{equation}
		|  \<   
	\Psi 
	,   
	[\N^\ell 
	, 
	V_{B} (t) ]						\nonumber 
	\Psi   \>_\F      |
\leq C 
	R 
	\< 
	\Psi ,  (\N^\ell  + \1 )
	\Psi 
	\>_\F  
	\end{equation}
\end{enumerate}
\end{lemma}

\begin{remark}
	\label{remark states}
{\blu Recall that we assume that the initial data  $ \nu $
satisfies (C1) from Condition \ref{condition initial data}. 
 }
	Namely,
	there exists 
	sequences 
	$(\lambda_n)_{n=0}^\infty \subset (0 , \infty) $ 
	and $(\Psi_n)_{n=0}^\infty  \subset \F $
	satisfying the normalization condition $\sum_{n=0}^\infty \lambda_n =1 $
	and $ \|	 \Psi_n 	\|_\F =1 $, respectively, such that 
	 the following decomposition holds true  
	\begin{eqnarray}
		\nu (\O)  = \sum_{n=0}^\infty \lambda_n  \<  \Psi_n		 ,		 \O	\Psi_n \>_\F  \ , 
		\qquad \forall \O \in B(\F)  \ . 
	\end{eqnarray}
	In particular, the   estimates contained in Lemma \ref{lemma commutator estimates 1}
{\blu 	can be easily converted into estimates for $\nu$. }
	For instance, if $ \O_1 , \O_2 , \O_3 $  are operators such that 
	\begin{equation}
	 | 	\< \Psi , \O_1   \Psi  \>_\F  |   \leq C    \|  	 \O_2   \Psi 		\|_\F  \|  \O_3 \Psi \|_\F  \ , \qquad \forall \Psi \in \F 
	\end{equation}   
for a constant $C>0$, 	then  it follows from the above decomposition of $\nu$ and the Cauchy-Schwarz inequality that 
	\begin{equation}
	 | 		\nu (\O_1)	 |      \leq C     \nu (\O_2^* \O_2 )^{\frac{1}{2}} \,  \nu (\O_3^* \O_3 )^{\frac{1}{2}} \ . 
	\end{equation}
	In most applications, $\O_2$ and $\O_3$ shall {\blu correspond} to either $\N$ or $\N_\S$.

\end{remark}

Let us briefly postpone the proof of the above Lemma to the next subsubsection. 
First, we turn to the proof of the important Proposition \ref{prop N estimates 1}.

\begin{proof}[Proof of Proposition \ref{prop N estimates 1}]
 The decomposition for $\h_I(t)$ from  \eqref{h_I decomposition } 
combined with 		the commutator estimates
	from     Lemma   \ref{lemma commutator estimates 1}
	implies that for all $\ell \geq 1$ there exists $C  =  C(\ell) > 0 $
	such that 
	\begin{equation}
		\partial_t 
		\nu_t (  \N^\ell + \1    ) 
		 = 
		  \nu_t \big(   i 	[\h_I  (t) , \N^\ell ]	\big)
		\leq 
		C 
		\lambda R 
		\nu_t (  \N^\ell  + \1  )  ,  \qquad \forall  t \geq 0  \ . 
	\end{equation} 
 Gronwall's inequality  now  easily implies that 
 there exists  a constant $C>0 $ such that 
 \begin{eqnarray}
 	\nu_t 
 	(\N^\ell) \,  \leq \,   C \lambda R  \,   \nu_t (\N^\ell + \1 )e^{C\lambda R t } \ , \qquad \forall t \geq 0 \ . 
 \end{eqnarray}
 To finalize the proof, we use the fact that for quasi-free states it holds true that $\nu(\N^\ell) \lesssim \nu(\N)^\ell $, 
 together with the assumption 
  $\nu(\N) = n \geq 1 $ . 
\end{proof}

\subsubsection{Commutator Estimates for $\N$}

\begin{proof}[Proof of Lemma \ref{lemma commutator estimates 1}]

Throughout this proof, $\Psi \in \F$ denotes an element in $ \cap_{k  = 1 }^\infty D(\N^k ) $, 
which will justify all of the upcoming calculations. 
Let us now fix $\ell \in \mathbb N . $

\vspace{2mm}

\noindent \textit{Proof of (1).}  This is an immediate consequence of the fact that $[D_k(t), \N]=0$ for all $ k\in \Lambda^*$ and $t \in \R$. See Lemma \ref{lemma N commutators}. 

\vspace{2mm}

\noindent \textit{Proof of (2).}
Using the fact that $ D_k^* (t) = D_{-k} (t)$ and $ [D_k^*(t), b_k(t)] = 0 $  
we may rewrite 
the fermion-boson interaction term as 
\begin{eqnarray}
	V_{FB}  (t) = 
	\int_{  \Lambda^{*} } \hat V(k) D^*_k(t) b_k (t) \d k   
 \ + \ 	\mathrm{h.c} \ . 
\end{eqnarray}
Thus, we find that 
for all $t \in \R$
\begin{equation}
	\< 	 \Psi ,  [	\N^\ell , V_{FB} (t)	] \Psi 	\>
	 = 
	 2 \mathrm{Im}
 \int_{\Lambda^* }	  \hat V(k) 
   \langle    \Psi , [\N^\ell , D^*_k	 (t)  b_k(t) ] \Psi   \rangle  \ . 
\end{equation}
In view of Lemma \ref{lemma N commutators}, 
we see that $[ D_k^* (t) , \N^\ell  ]  = 0$.
{\blu Further,  using the pull-through formulae  for $b$-operators 
  in \eqref{pull thru formula b}}
with $f (x) = x^\ell $
we find the following useful identity 
\begin{align}
	[\N^\ell ,   b_{k} (t) ]		 
	= 
	\sum_{ n = 0 }^{\ell-1} 
	\binom{\ell }{ n } 
	(-2)^{\ell  - n } 
	\N^n  b_ k (t)  \  , \qquad \forall   k \in \Lambda^*, \ t \in \R \ . 
\end{align}
Consequently, 
we can estimate that 
\begin{align}
		| 	\<
	\Psi , 
	[\N^\ell  , V_{FB} ]
	\Psi 
	\> | 
	& 	 \leq  
		\sum_{ n = 0 }^{\ell-1} 
	\binom{\ell }{ n } 
	(-2)^{\ell  - n } 
	\int_{  \Lambda^{* } }
	|   \hat V (k)  |
	\ 
	|
	\<  
	\Psi , 
	D_k^*(t) 
	\N^n  b_ k (t) 
	\Psi 
	\>
	|	
	\d k 
	\\ 
	& 	 \leq  
	\sum_{ n = 0 }^{\ell-1} 
	\binom{\ell }{ n } 
	(-2)^{\ell  - n } 
	\int_{  \Lambda^{* } }
	|   \hat V (k)  |
	\ 
\|		
 \N^{\frac{n-1}{2}}	 D_k(t) \Psi  
 \|	  
\, 
\|	\N^{\frac{n+1}{2}	}  b_ k (t) 
	\Psi 
 \|	
	\d k  	\ . 	
	\nonumber 
\end{align}
We can now combine Lemma \ref{lemma N commutators},
the Type-I estimate \eqref{lemma type I.1} 
and the norm bound \eqref{bounds on b}
to find that there exists a constant $C>0$
such that 
\begin{equation}
	\|		
	\N^{\frac{n-1}{2}}	 D_k(t) \Psi  
	\|	  
	\, 
	\|	\N^{\frac{n+1}{2}	}  b_ k (t) 
	\Psi 
	\|	 
 \leq C  R 
 \|		
\N^{\frac{n+1}{2}} \Psi  
 \|^2 	  \ , \qquad \forall n\geq 0 \ . 
\end{equation}
Finally, we put
the two above  estimates together and use the elementary fact
$\N^\frac{n+1}{2} \lesssim \N^\ell + 1 $  (valid for $ n \leq \ell -1$)
to find that for some $C >0$ there holds 
\begin{equation}
			| 	\<
	\Psi , 
	[\N^\ell  , V_{FB} ]
	\Psi 
	\> |  
	\leq 
	C  R \|	 \hat V	\|_{\ell^1 } 	\|	 (  \N^\ell +1 )  \Psi	\|^2  \ , \qquad \forall t \geq 0 \ 
\end{equation}
which gives the desired estimate.

\vspace{2mm}
 
\noindent \textit{Proof of (3)}.
First, we note that $[	 \N, b_k^* (t) b_k(t)	] =0$ for all $t \in \R$ and $ k \in \Lambda^*$.
Hence, we can readily check that 
\begin{equation}
	 \< 			
	 \Psi ,   [\N^\ell , V_B(t)]	\Psi 
	 \>
	 = 
	   \mathrm{ Im } 
\int \hat V(k) 
	 \<
	 \Psi   ,
	 [\N^\ell  , b_k (t) b_{-k} (t) ]
	 \Psi 
	  \> \d k 
	  \qquad 
	  \forall t \in \R \ . 
\end{equation}
In view of the commutation relation 
 $\N b_k (t)  b_{-k} (t)  = b_{k}  (t) b_{-k} ( t)  (\N  - 4 ) $ we can calculate  
 using the pull-through formula for $f(x) = x^\ell $ that 
\begin{align}
	[   \N^\ell  , b_{k} (t)  b_{-k}   (t) ] 
	=
	\sum_{n  =  0  }^{ \ell  -1 }
	\binom{ \ell  }{ n  }
	4^{ \ell  - n  }
	(\N +4 )^{ \frac{n  -1 }{2} }
	b_k  (t) b_{- k } (t) 
	\N^{ \frac{n  + 1 }{2}  }  \ . 
\end{align}
Consequently, putting the last two displayed equations together one finds that 
for all $ t \in \R$
\begin{align}
	|
		 \< 			
	\Psi ,   [\N^\ell , V_B(t)]	\Psi 
	\>
	|
	& \leq 
		 	\sum_{n  =  0  }^{ \ell  -1 }
		 \binom{ \ell  }{ n  }
		 4^{ \ell  - n  }
	\int_{\Lambda^* 	}	
	| \hat  V(k) |
	\|	(\N +4 )^{ \frac{n + 1 }{2} }	 \Psi  	\|
\
\|	 b_k (t)  b_{- k } (t) 
\N^{ \frac{n - 1 }{2}  } 	 \Psi  	\| 
\d k 
\ . 
\end{align}
We estimate the right hand side as follows. 
First, 
we note that 
$
	\|	(\N +4 )^{ \frac{n + 1 }{2} }	 \Psi  	\| \leq  C(\ell)  \|	 (\N +1 )^{\ell/2 } \Psi \|
$
for all $0 \leq n \leq \ell -1 $. 
Secondly, we use the Type-II estimate \eqref{lemma type II.1} 
and the commutation relation  \eqref{pull thru formula b} for $f \equiv 1$
to find that 
\begin{align}
\nonumber 
\|			b_k (t) 		b_{-k }  (t) \N^{	 \frac{n -1 }{2 }	} \Psi  \| 
&   
 \ \lesssim  \ 
R^{\frac{1}{2}} 
 \, 
\|	 (\N+2 )^{\frac{1}{2}} b_{-k} (t) \N^{\frac{n-1}{2}  	}	 \Psi 	\|    \\ 
\nonumber 
& 
\  =  \ 
R^{\frac{1}{2}} \, 
 \|	   b_{-k} (t) \N^{\frac{1}{2 }}   \N^{\frac{n-1}{2}  	}	 \Psi 	\|   \\ 
\nonumber 
&  
\   \lesssim  \ 
  R  \, 
   \|	 \N^{\frac{1}{2 } } \N^{\frac{1}{2}}	\N^{\frac{n-1}{2}  	}	 \Psi 	\|   \\ 
  & 
    \  \lesssim \ 
      R \, 
       \|	  ( \N +1)^{\frac{\ell}{2}}	 \Psi \|
\end{align}
where again we used the fact that $  n \leq \ell -1 $. 
The proof of the {\blu lemma} is easily 
finished after we put together the last two displayed estimates. 
\end{proof}

\subsection{Surface-localized Number Operator Estimates}

The main purpose of this section is to prove Proposition \ref{prop N estimates 2}. 
In order to control the time evolution of  $\po$ with respect to $\h_I (t) $, we establish the following commutator estimates. 
Recall that $R  =  |\Lambda| p_F^{d-1 }$. 

\begin{lemma}
	\label{lemma commutator estimates 2}
There exists a constant $C>0$ such that the following estimates hold true 
	\begin{enumerate}
		\item 
		For all   $\Psi \in \F$  
		\begin{equation}
			|  \<  \Psi ,   
			[\po, V_F  (t) ]
			\Psi   \>_\F    |
 \leq C 
			\|	 \po^{1/2} \Psi	\|_\F 
			\|	 \N^{3/2}	 \Psi\|_\F    \ . 
		\end{equation}
		\item 
		For all   $\Psi \in \F$  
		\begin{align}
			| \<
			\Psi  
			,
			[\po  , V_{FB} (t)  ] 			\nonumber 
			\Psi 
			\>_\F   | 
			& 
 \leq C 
			R^{1/2}
			\|	 \po^{1/2}	 \Psi\|_\F  	\|	 \N \Psi	\|_\F  	 \ . 
		\end{align}
		\item 
		For all   $\Psi \in \F$  	
		\begin{equation}
			|\<
			\Psi ,
			[\po , V_B (t)  ] 
			\Psi 											\nonumber 
			\>_\F    | 
 \leq C 
			R
			\|	 \po^{1/2} \Psi	\|^2_\F 
			+ 
			C R 
			\|	 \po^{1/2} \Psi	\|_\F  
			\|	 \Psi	\|_\F   		\ . 
   		\end{equation}
	\end{enumerate}
	
\end{lemma}

We shall defer the proof of Lemma  \ref{lemma commutator estimates 2} to  the  next subsubsection. Now
we turn our attention to the proof of Proposition \ref{prop N estimates 2}. 

\begin{proof}[Proof of Proposition \ref{prop N estimates 2}]
Throughout the proof, $C>0$ is a constant whose value may change from line to line.
First, 	in view of the decomposition of $\h_I (t) $ given in \eqref{h_I decomposition }, 
	 Lemma \ref{lemma commutator estimates 2}
	 and Remark \ref{remark states}, 
there holds 
 for all $t \in \R$  
	 \begin{align}
\nonumber 
\frac{d}{dt}
	 	\nu_t  (\N_\S)
	 	= \nu_t    (   i  [ \h_I (t) , \N_\S] )
	 	& \leq 
	 	C  \lambda 
	 	[	 \nu_t (\N_S)	]^{\frac{1}{2}}
	 	[	 \nu_t (\N^3)	]^{\frac{1}{2}} \\
\nonumber 
	 	&  \ + 
	 		 	C  \lambda  R^{\frac{1}{2}}
	 		 	[	 \nu_t (\N_S)	]^{\frac{1}{2}}
	 	[	 \nu_t (\N^2)	]^{\frac{1}{2}} \\
\nonumber 
	 		&  \ + 
	 	C  \lambda  R 
	 	[	 \nu_t (\N_S)	] \\ 
	 		&  \ + 
	 	C  \lambda  R 
	 	[	 \nu_t (\N_S)	]^{\frac{1}{2}}
	 	[	 \nu_t (\1 )	]^{\frac{1}{2}}  \ . 
	 \end{align}
Thus, we divide\footnote{Technically, one should introduce a regularization $u_\delta(t) \equiv (\delta + \nu_t(\po))^{1/2}$ in order to avoid possible singularities whenever $\nu_t(\po)=0$. 
{\blu 	Namely, in order to avoid division by zero in the present argument. }
	One should then close the estimates  after taking the limit $\delta \downarrow 0$. We leave the details to the reader} 
 by $\nu_t(\po)^{1/2}$
 to find that thanks to Proposition \ref{prop N estimates 1} 
\begin{align}
\nonumber 
	\frac{d}{dt} 
	\nu_t (\po)^{\frac{1}{2}}
& 	\leq 
 C 	\lambda R 	\nu_t (\po)^{\frac{1}{2}}
 + 
 C 
 \lambda R \Big(
 \nu_t(\N^3)^\frac{1}{2}/R
 + 
  \nu_t(\N^2)^\frac{1}{2}/R^{\frac{1}{2}}
  +1 
 \Big)		\\
 & \leq 
 C 	\lambda R 	\nu_t (\po)^{\frac{1}{2}}
 + 
 C
 \lambda R  \exp(\lambda R t )
 \Big(
 n^{\frac{3}{2}}/ R + n / R^{\frac{1}{2}} +1 
 \Big)	 \ . 
\end{align}
The  Gr\"onwall inequality 
now 	implies  that  for all $ t \geq 0 $
	\begin{align}
	\nu_t (\po)^{\frac{1}{2}} 
		\leq 
		C
\exp( C \lambda R  t)
\Big(
	\nu_0 (\po)^{\frac{1}{2}} 
+ 
\lambda R t  
\, 
(
n^{\frac{3}{2}}/ R + n / R^{\frac{1}{2}} +1 
)
\Big)   \ . 
	\end{align}
Finally, we notice that in view of Condition \ref{condition initial data}  we have $
\nu_0 (\N_\S)  \lesssim (\lambda R)^2$.
The proof is then finished once we simplify the right hand side using the bound
$ n \lesssim R^{1/2}$, 
and take squares on both sides of the inequality. 
\end{proof}

%
%
%
%

\subsubsection{Commutator Estimates for  $\N_\S$}
\label{section proof commutator estimates}
In order to prove Lemma \ref{lemma commutator estimates 2}, we shall first establish the following useful  lemma.
Here and in the sequel, $\1_\S$ denotes the characteristic function of the Fermi surface $\S $. 

\begin{lemma}
	\label{lemma O}
	For all $ k  \in \Lambda^*$ and $g \in \ell^\infty  $ the operator
	\begin{equation}
		\mathcal{O} (k) \equiv 
		\int_{  \Lambda^{* } }  \1_\S (p) g(p) a_{p+k}^* a_p \d p 
	\end{equation}
	satisfies the following estimate
	\begin{align}
		\big| 
		\<   \Phi , \mathcal{O} (k) \Psi   \>_\F 
		\big| 
		\leq
		\|	 g	\|_{\ell^\infty }
		\| \N^{1/2} \Phi		\|
		\|	 \po^{1/2} \Psi 	\| 
		 \ , \qquad \forall  \Phi , \,  \Psi \in \F.
	\end{align}
\end{lemma}

\begin{proof}
Let $\Phi,\Psi\in \F$, $ k \in \Lambda^*$	 and $g \in \ell^\infty $. 
Then, we calculate
\begin{align}
|	
 \<  \Phi , \O (k) \Psi \>_\F 
 	|		  
\nonumber 
&  = 
 \bigg|
 		\int_{  \Lambda^{* } }  \1_\S (p) g(p) 
\<  a_{p+k } \Phi  , a_p   \Psi \>_\F 
 		\d p 
 \bigg|  		\\ 
\nonumber 
 &  \leq 
 \int_{  \Lambda^{* } }  \1_\S (p)  | g(p)  | 
 \| a_{p+k } \Phi \|_\F   \|	 a_p   \Psi \|_\F 
 \d p 			\\ 
\nonumber 
  &  \leq 
  \|	 g \|_{\ell^\infty }
\biggl(
\int_{  \Lambda^{*} } 
 \|		 a_{p+k } \Phi 	\|_\F^2 \d p 
\biggr)^{\frac{1}{2}	}
\biggl(
 \int_{  \Lambda^{*} } 
\1_\S(p) 
 \|		 a_{p  } \Psi  	\|_\F^2 \d p 
\biggr)^{\frac{1}{2}	}  \\
& = 
  \|	 g \|_{\ell^\infty }
  \|	 \N^{\frac{1}{2}} \Phi 	\|_\F
    \|	 \N_\S^{\frac{1}{2}} \Psi	\|_\F  \ . 
\end{align}
In the last line we used the fact that $  \|	 a_p \Phi	\|^2_\F  = \< \Phi , a_p^*a_p\Phi \>_\F $ for all $ p \in \Lambda^*$, plus a change of variables $p \mapsto p-k$. 
A similar argument holds  for the term containing $\Psi $. This finishes the proof. 
\end{proof}

\begin{proof}
	[Proof of Lemma \ref{lemma commutator estimates 2}]  
Throughout this proof, $\Psi \in \F$ is fixed.
In addition, in order to ease the notation, we shall drop the explicit  time dependence in our estimates -- since the estimates are  uniform in $t \in \R $, there is no risk of confusion. 
Let us now fix $\ell \in \mathbb N . $

\vspace{2mm}

\noindent \textit{Proof of (1).}  
Starting from \eqref{VF} we can first calculate that 
	\begin{align}
		\< \Psi , [\po  , V_F] \Psi  \> 
		& =
		2 \i  
		\ 
		\int_{  \Lambda^{* } }
		\hat V (k)
		\mathrm{ Im } 
		\<  \Psi ,  
		[\po, D_k^*] D_k 
		\Psi \> 
		\d  k \ . 
	\end{align}
We now put the above commutator in an appropriate form. 
Using the explicit 
expression of $D_k^* $ in terms of creation- and annihilation- operators  (see Def. \ref{definition Heisenberg b D}) together with the CAR, 
we find that
for all $ k \in \Lambda^*$  there holds 
	\begin{align}
		\nonumber 
		[ \po , D_k^*]
		& 	= 
		\int_{  \Lambda^{* } }
		\Big(
	\1_\S 
	(p)  - 	\1_\S  ( p - k)
		\Big)
		\chi^\perp (p)  \chi^\perp(p  - k )
		a_p^* a_{p - k} \,  \d p  \\ 
\nonumber 
		&  - 
		\int_{  \Lambda^{* } }
		\Big(
		\1_\S  ( h )  - 	\1_\S  ( h + k )
		\Big)
		\chi (h )  \chi ( h+ k  )
		a_h ^* a_{ h + k} \,  \d h    \\
\label{eq comm 1}
		& \equiv \mathcal{O}_1 (k) + \mathcal{O}_2(k) 
	\end{align}
	where we introduce the two following auxiliary operators  (notice the change of variables $p \mapsto p+k$ and $h\mapsto h-k$ in the second operator)
	\begin{align}
		\mathcal{O}_1(k)
		&  \equiv 
		\int_{  \Lambda^{* } }
	\1_\S  (p)  
		\chi^\perp (p, p-k )  
		a_p^* a_{p - k} \,  \d p   
		- 
		\int_{  \Lambda^{* } }
	\1_\S    (   h ) 
		\chi (h  , h +  k )  
		a_h ^* a_{ h + k} \,  \d h    		\\
		\mathcal{O}_2(k)
		&  \equiv 
		- 	\int_{  \Lambda^{* } }
	\1_\S  (p   )  
		\chi^\perp (p , p + k )  
		a_{p+k}^* a_{p } \,  \d p   
		+  
		\int_{  \Lambda^{* } }
	\1_\S  ( h  ) 
		\chi (h, h - k  )  
		a_{ h - k } ^* a_{ h } \,  \d h    	
	\end{align}
	where for simplicity we denote $\chi^\perp( p , p - k) \equiv \chi^\perp (p) \chi^\perp (p -k)$  and similarly for $\chi(h,h+k)$. 
	We are now able to write 
	\begin{align}
		\< \Psi , [\po  , V_F] \Psi  \> 
		& =
		2 \i  
		\ 
		\int_{  \Lambda^{* } }
		\hat V (k)
		\mathrm{ Im } 
		\<  \Psi ,  
		\mathcal{O}_1 (k) D_k 
		\Psi \> 
		\d  k   \\ 
		& 	+ 
		2 \i  
		\ 
		\int_{  \Lambda^{* } }
		\hat V (k)
		\mathrm{ Im } 
		\<  \Psi ,  
		D_k 
		\mathcal{O}_2 (k) 
		\Psi \> 
		\d  k  \\ 
		& 	+
		2 \i  
		\ 
		\int_{  \Lambda^{* } }
		\hat V (k)
		\mathrm{ Im } 
		\<  \Psi ,  
		[ 	\mathcal{O}_2 (k)  , D_k  ] 
		\Psi \> 
		\d  k \ . 
	\end{align}
	The first term in the above equation can be estimated using Lemma \ref{lemma O} for $\mathcal{O}(k) = \mathcal{O}^*_1(k)$. Namely, 
	\begin{align}
		\big|
		2 \i  
		\ 
		\int_{  \Lambda^{* } }
		\hat V (k)
		\mathrm{ Im } 
		\<  \Psi ,  
		\mathcal{O}_1 (k) D_k 
		\Psi \> 
		\d  k  
		\big|
		\leq 
		2 \| \hat V		\|_{\ell^1}
		\|	 \po^{1/2} \Psi	\| 
		\|	 \N^{3/2} \Psi 	\|
	\end{align}
	The second term in the above equation is estimated using Lemma \ref{lemma O} for $\mathcal{O}(k) = \mathcal{O}_2(k)$. We get 
	\begin{align}
		\big|
		2 \i  
		\ 
		\int_{  \Lambda^{* } }
		\hat V (k)
		\mathrm{ Im } 
		\<  \Psi ,  
		D_k  \mathcal{O}_2 
		\Psi \> 
		\d  k  
		\big|
		\leq 
		2 \| \hat V		\|_{\ell^1}
		\|	 \N^{3/2}  \Psi 	\|
		\|	 \po^{1/2} \Psi	\| 
	\end{align}
	The third term in the above equation is actually zero. This comes from the fact that the commutator between $\mathcal O _2(k)$ and $D_k$ is self-adjoint. More precisely, we can calculate
	using the CAR
	\begin{align}
		[
		\mathcal{O} _2(k) , D_k]	
		& = 
		\int_{  \Lambda^{* } }
		\Big(
	\1_\S  (p + k )
		- 
	\1_\S   (p  )
		\Big)
		\chi^\perp (p,p +  k )
		a_p^*a_p \d p \\
		& 
		- 
		\int_{  \Lambda^{* } }
		\Big(
	\1_\S   ( h - k  )
		- 
	\1_\S  ( h   )
		\Big)
		\chi(h,h - k)
		a_h^*a_h  \d h  \ . 
	\end{align}
	We put our results together to find that 
	\begin{align}
		\big | 
		\< \Psi , [\po  , V_F] \Psi  \> 
		\big | 
		& 
		\leq 
		4 \|	 \hat V	\|_{\ell^1 }
		\|	 \po^{1/2} \Psi	\|
		\|	 \N^{3/2}	 \Psi\| \ . 
	\end{align}

	\vspace{2mm}
	\noindent{\textit{Proof of (2).}}
Starting from  \eqref{VFB} 	we  can calculate that 
	\begin{align}
		| 
		\<
		\Psi  
		,
		[\po , V_{FB}]			\nonumber 
		\Psi 
		\>
		| 
		& 
		\leq  
		2 
		\int_{  \Lambda^{* } }
		| \hat  V (k) |
		\, 
		|
		\<
		\Psi, 
		[						\nonumber 
		\po, D_k^* b_k 
		]
		\Psi 
		\>
		|
		\d k 
		\ , 	\\ 
		& \leq 
		2 
		\int_{  \Lambda^{* } }
		| \hat  V (k) |
		\, 
		\|
		[\po , D_k ] \Psi 
		\|								\nonumber 
		\,
		\| 
		b_k \Psi 
		\|
		\d k 
		\\	
		& 
		\quad 
		+ 
		2
		\int_{  \Lambda^{* } }
		| \hat  V (k) |
		\, 
		\|
		[\po ,b_k ] \Psi 
		\|
		\,
		\| 
		D_k  \Psi 
		\|  
		\d k  	 \ . 
		\label{eq comm 2}
	\end{align}

	Let us estimate the first term contained in  the right hand side of \eqref{eq comm 2}.
In view of $D_k^* = D(-k)$ and \eqref{eq comm 1} we have
that  $[\N_\S , D_k ] = \O_1(- k)  + \O_2(-k)$.
Each $\O_i(k)$ can be estimated using \eqref{cross estimate} --we conclude that 
$
		\|  
		[\po ,D_k ] \Psi 
		\|
		\lesssim 
		\|   
		\N \Psi 
		\| \ . 
$
On the other hand, we use the Type-II estimate \eqref{lemma type II.1} on  $b_k$. We conclude that 
\begin{eqnarray}
		\int_{  \Lambda^{* } }
	| \hat  V (k) |
	\, 
	\|
	[\po , D_k ] \Psi 
	\|							 
	\,
	\| 
	b_k \Psi 
	\|
	\d k 
 \	\lesssim  \ 
	R^{\frac{1}{2}}
	\|	 \N \Psi 	\| 	\, \|	 \N_\S^{1/2} \Psi 	\| \  . 
\end{eqnarray}

\vspace{1mm}

	Let us now look at the second term contained in \eqref{eq comm 2}. 
	First, we  
	recall that for $ k \in \supp \hat V$
	there holds $[\N_\S  ,  b_k]  =  -2 b_k$, see Lemma \ref{lemma NS commutators}. 
	Consequently, using 
 the Type-II estimate \eqref{lemma type II.1}  we
 see that $	 		\|
 [\po ,b_k ] \Psi 
 \|	\lesssim R^{1/2 } \|	\N_\S^{1/2}	 \Psi\|$. 
 On the other hand, we can use the Type-I estimate  \eqref{lemma type I.1}		to find $\|	 D_k \Psi 	\|\lesssim \|	 \N \Psi 	\|$. 
These upper bounds  can be put    together to find that 
	\begin{align}
	\int_{  \Lambda^{* } }
	| \hat  V (k) |
	\, 
	\|
	[\po ,b_k ] \Psi 
	\|
	\,
	\| 
	D_k  \Psi 
	\|  
	\d k  	 	& 
		\  \lesssim  \    
		R^{1/2}
		\|	 \po^{1/2}	 \Psi\| 	\|	 \N \Psi	\| \ . 
	\end{align}
	A direct combination of the last three displayed estimates finishes the proof of (2).

	\vspace{2mm }
	
	\noindent
	{\textit{Proof of (3)}}. 
Starting from  \eqref{VB} we  decompose the boson-boson interaction into a diagonal, and off-diagonal part. 
	Namely, we  write $V_B  = V_1 + V_2$, where we 
	set
	\begin{align}
		V_1 
		 \equiv 
		\int_{  \Lambda^{* } }
		\hat V  (k)
		b^*_k b_k   \d  k 	
		\quad
		\t{and}
		\quad 
		V_2  
		 \equiv 
		\frac{1}{2}
		\int_{  \Lambda^{* } }
		\hat V  (k)
		\Big(
 b_k  b_{-k}   + \mathrm{h.c} 
		\Big)  
		\d k 
		\ . 
	\end{align}

\vspace{1.5mm }

For $V_1$ we can quickly verify that  its   commutator with $\po$   vanishes. 
Indeed, thanks to Lemma \ref{lemma NS commutators} 
we find that 
$[\N_S, b^*(k) b_k] = +2 b_k^* b_k - 2 b_k^*b_k =0$
for all $ k \in \supp \hat V$. 
Hence, $[\N_\S , V_1]=0$ upon summing over $ k \in \Lambda^*$. 

\vspace{2mm}

	For $V_2$, we   have the preliminary upper bound as our starting point  
	\begin{align}
\label{eq comm 3}
		|\<
		\Psi ,
		[\po , V_2]
		\Psi 										 
		\>  | 
		& 
\leq  2 
		\int_{  \Lambda^{* } }
		| \hat{V} (k) | 
		|	 \<  \Psi , b_k b_{-k} \Psi \>	|
		\d k 
		\ . 
	\end{align} 
We estimate the integrand of the  right hand side as follows --let us fix $ k \in \supp \hat V$. 
First, recalling that   $[\N_\S, b_k]=0 $    (see Lemma \ref{lemma NS commutators}) 
we find that for any measurable function $ \vp  : \R \rightarrow \C $ 
the following \textit{pull-through formula} holds true
	\begin{eqnarray}
\label{pull through formula Ns}
		\vp ( \po) b_k = b_k \vp (\po -2  )    \ . 
	\end{eqnarray} 
Thus, using  $\vp (x)  = (x+5)^{1/2}$	we find 
	\begin{align}
\nonumber 
		|	 \<  \Psi , b_k b_{-k} \Psi \>	|
		&    = 
		|	   \<   (\po  + 5  )^{1/2} \Psi , b_k  b_{-k}   ( \po +1 )^{-1/2} \Psi    \> |  \\
		& \leq 
		\|	  (\po  + 5  )^{1/2} \Psi	\| 
		\, 
		\|	 b_k  b_{-k}  ( \po +1 )^{-1/2} \Psi   	\|    \ . 
\label{eq comm 4}
	\end{align}
{\blu We use   
 the Type-II estimate \eqref{lemma type II.1} for $b$-operators 
 and the commutation relation $ (\N_S + 2 )^{1/2} b_k = b_k \N_S^{1/2 }$ 
 to find that 
} 
	\begin{align}
		\nonumber 
		\|	 b_k b_{-k}   ( \po +1 )^{-1/2} \Psi   	\|    
		& 
		\,   \lesssim  \, 
		R^{1/2}
		\|	 \po^{1/2}	  b_{-k} ( \po +1 )^{-1/2} \Psi    \| 		\\ 
		\nonumber 
		& 
	 \, 	\leq  \, 
		R^{1/2}
		\|	  (\po+2)^{1/2}	  b_{-k}  ( \po +1 )^{-1/2} \Psi    \| 		\\ 
		\nonumber 
		& 
	 \, 	=  \, 
		R^{1/2}
		\|	  	  b_{-k}   \po^{1/2 }   ( \po +1 )^{-1/2} \Psi    \| 		\\ 
		\nonumber 
		 &  \,  \lesssim  \, 
		R \|	\po^{1/2} \po^{1/2} 	 ( \po +1 )^{-1/2} \Psi    \| 	  \\ 
		&  \, \leq  \, 
		R 	\|	 \po^{1/2} \Psi	\|  \  . 
		\label{eq comm 5}
	\end{align}
On the other hand, the other term multiplying in \eqref{eq comm 4}
can be bounded as follows 
$
\|	  (\po  + 5  )^{1/2} \Psi	\|  
\lesssim \|	 \po^{1/2} \Psi	\|	+ \|	 \Psi 	\|  \ . 
$
A straightforward combination of the estimates contained in \eqref{eq comm 3}, \eqref{eq comm 4} and \eqref{eq comm 5} now finish the proof. 
\end{proof} 

 \section{Leading Order Terms I: Emergence of $Q$}
\label{section TFF}
 In Section \ref{section preliminaries} we considered a double commutator expansion \eqref{expansion f} for the momentum distribution  of particles and holes,
 $f_t(p)$. 
 This expansion     is expressed in terms of the nine quantities $\{T_{\alpha,\beta}(t) \}$ 
 that 
 arise from the three different {\blu interaction} potentials $V_F$, $V_{FB}$ and $V_{B}$, respectively. 
 The main goal of this section is to analyze the single term $T_{F,F}$. 
In particular, we prove that one may extract the mollified collision operator $Q_t$--originally introduced in Def. \ref{definition Q1}--
up to {\blu remainder} terms that we have control of. 
A precise statement is given in the following proposition. 
 We remind the reader that $R  =  |\Lambda| p_F^{d-1}$

 \begin{proposition}
 	[Analysis of $T_{F,F}$]
 	\label{prop TFF}
Let $T_{F,F}(t,p)$ be the quantity defined in Eq.  \eqref{T alpha beta} for $\alpha  =  \beta = F$,
and let $m>0$.
Then, there exists a constant  $C>0$
such that for all  $\vp \in \ell_m^1 $ and $ t \geq 0 $ 
the following inequality holds true 
 	\begin{align}
 		\big|		
 		T_{F,F}
 		(t,\vp)
 		+
 		\,    |\Lambda| \, 
 		t 
 		\< \vp , Q_t [f_0]\>
 		\big|	
 \leq C 
 		 |\Lambda| 
 		\lambda 
 		t^3 
 		\|	\hat V	\|_{\ell^1}^3 
 		\|	 \vp	\|_{\ell^1 }
 		\sup_{\tau \leq t }
 		\Big(
 		R^2 
 		\nu_\tau (\N^4)^{\frac{1}{2} } 
 		+
 		\nu_\tau (\N^4)
 		\Big)
 	\end{align} 
 	where  $T_{F,F} (t,\vp) \equiv \<\vp, T_{F,F}(t)\>$ and $Q_t $   is   given in Def. \ref{definition Q1}. 
 \end{proposition}

In order to  prove Proposition \ref{prop TFF}  we  shall perform an additional expansion of $\nu_t$ with respect to the interaction Hamiltonian $\h_I(t)$.
Namely, we consider 
 \begin{align}
 T_{F,F} (t,\vp) 
  & =
 \int_0^t 
 \int_0^{t_1 }				\nonumber 
 \nu 
 \big(
 [ [ N(\vp) , V_F(t_1) ] , V_F (t_2)]
 \big) 
 \d t_1 \d t_2   \\
 & \quad  
   - i 
 \int_0^t 
\int_0^{t_1 }
\int_0^{t_2}
    \nu_{t_2}
   \big(
  [  [ [ N(\vp) , V_F(t_1) ] , V_F (t_2)] ,  \h_I( t_3)  ] 
   \big) 
   \d t_1 \d t_2 \d t_3 \ , 
   \label{TFF eq 1}
 \end{align}
where we recall $N(\vp) \equiv \int_{  \Lambda^{*} } \overline{\vp(p)}a_p^*a_p \, \d p $. 
We then analyze  the two terms of the right hand side of \eqref{TFF eq 1} separately. 
Thus, we split the proof into two parts, which are contained in the following two lemmas. 

\begin{lemma}
	\label{lemma TFF 1}
Let $\nu: B(\F) \rightarrow \C$ be an initial state
satisfying Condition \ref{condition initial data},
and let 
   $f_0(p)  = |\Lambda|^{-1} \nu(a_p^*a_p)$ for all $p \in \Lambda^*$. 
Let $V_F(t)$ be the Heisenberg evolution of the fermion-fermion interaction, 
defined in \eqref{Heisenberg evolution V} for $\alpha  = F$.
Then, for all $\vp \in \ell^1$  and $ t \geq 0 $
	\begin{equation}
		 \int_0^t 
		\int_0^{t_1 }			 
		\nu 
		\big(
		[ [ N(\vp) , V_F(t_1) ] , V_F (t_2)]
		\big) 
		\d t_1 \d t_2   = - t 
		|\Lambda| \< \vp , Q_t[f_0] \>  \ .
	\end{equation}
\end{lemma}

The proof of the identity contained in   Lemma \ref{lemma TFF 1}
will be heavily inspired by the work of Erd\H{o}s, Salmhofer and Yau \cite{ESY2004},
on a heuristic derivation of the quantum Boltzmann equation. 
In fact, we shall make use of some of their algebraic relations.

\begin{lemma}
	\label{lemma TFF 2}
	Let $(\nu_t)_{t\in\R}$ be the interaction dynamics as given in Def. \ref{definition interaction dynamics}, with initial data  $\nu  =  \nu_0 $ satisfying Condition \ref{condition initial data}.
Let $V_F(t)$ be the Heisenberg evolution of the fermion-fermion interaction, 
defined in \eqref{Heisenberg evolution V} for $\alpha  = F$.
Then, there exists a constant $C>0$ such that 
for all $\vp \in \ell^1$  and $ t \geq 0 $
	  \begin{align}
	    \bigg{|}
	  \int_0^t 
	  \int_0^{t_1 }
	  \int_0^{t_2}
	  \nu_{t_2}
	  \big(
	  [  [ [ N(\vp) , &  V_F(t_1) ] ,  V_F (t_2)] ,  \h_I( t_3)  ] 
	  \big) 
	  \d t_1 \d t_2 \d t_3 
	    \bigg{|}
\nonumber 
	  \\
	  	& \leq C  
\lambda 	  	t^3 
	  	\|	\hat V	\|_{\ell^1}^3 
	  	|\Lambda|
	  	\|	 \vp	\|_{\ell^1 }
	  	\sup_{\tau \leq t }
	  	\Big(
	  	\nu_\tau (\N^4)
	  	+ 
	  	R^2 
	  	\nu_\tau (\N^4)^{\frac{1}{2}	}
	  	\Big) \ . 
	  \end{align}
\end{lemma}

We remind the reader that   the interaction Hamiltonian $\h_I(t)$
admits the decomposition given in \eqref{h_I decomposition } 
in terms of  the Heisenberg evolution of $b$ and $D$-operators--see \eqref{VF(t)}, \eqref{VFB(t)} and \eqref{VB(t)}.

\begin{proof}[Proof of Proposition \ref{prop TFF}]
It suffices to put together Eq. \eqref{TFF eq 1} and  Lemmas \ref{lemma TFF 1} and \ref{lemma TFF 2}. 
\end{proof}

We dedicate the rest of this section to the proof of Lemmas \ref{lemma TFF 1} and \ref{lemma TFF 2}. 
Before we jump into the proof of Lemma \ref{lemma TFF 1}, 
we shall rewrite the fermion-fermion interaction term $V_F(t)$
in a form that will be suitable for our analysis. 
This representation is recorded in Lemma \ref{lemma normal order VF}, 

\medskip 
 \textit{{Normal ordering of   $V_F(t)$}.}
 Let us fix the time label $ t \in \R$. 
First, we see from \eqref{VF(t)} that $V_F(t) = \int_{  \Lambda^{*} }\hat V(k) D_k^*(t) D_k(t) \d k $ can be written in terms of
the Heisenberg evolution of the $D$-operators, as given in Def. \ref{definition Heisenberg b D}. 
These can be written explicitly in terms of creation- and annihilation- {\blu operators}
in the following way 
\begin{align}
	D_k(t) =
	\int_{(\Lambda^*)^2 	}
	d_t (k,p , q ) 
	a_{p}^* 
	a_{q} \d p \d q
\end{align}
where the  coefficients in the above expression are given as follows 
\begin{equation}
	d_k(t, p , q)
	  \,  \equiv  \, 
	 e^{it(E_{p  }    - E_{q}  ) }
	 \big[
	 \chi^\perp (p)
	 	 \chi^\perp (q)
\delta(p - q + k )
	 	 -
	 	 	 \chi(p)
	 	 \chi(q)
	 	 \delta(p - q - k )
	 \big]
\end{equation}
for all   $k, p, q \in \Lambda^*$. 
Since $D_k^*(t) = D_{-k}(t)$  it readily follows that we can   write the fermion-fermion interaction in the following form 
\begin{equation}
\label{TFF VF(t)}
	V_F(t) 
	=  \int_{  \Lambda^{*4 } } 
\bigg[
\int_{  \Lambda^{*} }
\hat V(k) 
d_t(-k , p_1, q_1 ) \,  d_t(k, p_2, q_2) \d k 
\bigg]
a_{p_1}^*  a_{q_1} 
a_{p_2}^*  a_{q_2} 
\d p_1  \d p_2  
\d q_1 \d q_2   \  . 
\end{equation}

\vspace{2mm}

Clearly, the expression in \eqref{TFF VF(t)} is \textit{not} normally ordered. 
Our next goal is then to  put $V_F(t)$ in normal order, with explicit coefficients.
To this end, we introduce the following coefficient  function 
\begin{equation}
\label{coefficient phi}
	\phi_t(  \vec p )
	\equiv 
	\int_{  \Lambda^{*} }
	\hat V(k) 
 \, 	d_t(-k , p_1, p_4 ) \,  d_t(k, p_2, p_3) \d k    
\end{equation}
where $\vec p = (p_1, p_2, p_3,p_4) \in (\Lambda^*)^4 $.  
A straightforward calculation using the CAR  in  Eq. \eqref{TFF VF(t)}
now yields 
\begin{align}
\nonumber 
	V_F(t)    
		& = 
		 \int_{  \Lambda^{*4 } } 
\phi_t(p_1,p_2,q_2,q_1)
	a_{p_1}^*   a_{p_2}^*
	a_{q_2}   a_{q_1} 
	\d p_1  \d p_2  
	\d q_1 \d q_2   \\
\label{TFF VF(t) 2}
	& + 
	 \int_{  \Lambda^{*2 } } 
 \bigg[\
  \int_{  \Lambda^{*2 } } 
 \phi_t(p_1,p_2,q_2,q_1)
 \delta(q_1 - p_2)
  \d p_2 \d q_1
 \bigg]
	a_{p_1}^*    a_{q_2} 
	\d p_1 \d q_2  \ . 
\end{align}
We shall denote by $ :	 V_F(t):$ the normal  ordering of $V_F(t) $, that is, the first term in Eq. \eqref{TFF VF(t) 2}.

\vspace{2mm}

Next, we shall put the above normal order form in a more explicit representation by calculating
 explicitly the coefficient function $\phi_t$, together with its contraction for $q_1=p_2$.
Before we do so, let us introduce some convenient notation:  

 \begin{itemize}[leftmargin=*]
 	\item  
 	When $\vec p  = (p_1 , p_2 , p_3 , p_4 )\in (\Lambda^*)^4 $ is known from context, we let 
 	$$ \chi_{1234}  \equiv  
 	\chi (p_1) 
 	\chi (p_2) 
 	\chi (p_3) 
 	\chi (p_4) 
 	\qquad
 	\t{and}
 	\qquad 
 	\chi_{1234}^\perp \equiv  1 - \chi_{1234}  
 	$$
 	and similarly for $\chi_{ ij } $ and $\chi_{ ij }^\perp$ for any combination of 
 	$i,j \in \{ 1,2,3,4\}$.
 	
 	\item 
 	For any $ \vec p   =  (p_1, p_2,p_3,p_4)\in (\Lambda^*)^4 $ we let 
 	\begin{eqnarray}
 		\Delta E (\vec p  ) \equiv  E_{p_1} + E_{p_2} - E_{p_3} - E_{p_4} 
 	\end{eqnarray}
 	where $E_p$ is the dispersion relation of the system--see \eqref{dispersion relation}. 
 \end{itemize}

Starting from \eqref{coefficient phi} and using the definition of $d_t(k,p,q)$  
we may explicitly calculate that
for all $\vec p  \in (\Lambda^*)^4 $ there holds 
 \begin{align}
\nonumber 
\label{coefficient phi 2}
 	\phi_t( \vec p )
&  	 = 
 	 e^{it \Delta E  (\vec p )}
 	 \delta (p_1 + p_2 - p_3 - p_4 )
 	 \hat V (p_1  - p _4 ) 
 	 \big(
 	 \chi_{1234} + 
 	 \chi_{1234}^\perp 
 	 \big)  \\ 
 	 & -
 	 e^{it \Delta E  (\vec p )}
 	 \delta (p_1 -    p_2  +  p_3  - p_4 )
 	  	 \hat V (p_1  - p _4 ) 
 	 \big(
 \chi_{13} \chi^\perp_{24} + 
 \chi_{13}^\perp  \chi_{24}
 	 \big)   \ . 
 \end{align}
In particular, a straightforward calculation using \eqref{coefficient phi 2}
shows that the integrand of   the quadratic term in \eqref{TFF VF(t) 2} can be written as 
\begin{equation}
\label{TFF eq 3}
	  \int_{  \Lambda^{*2 } } 
	\phi_t(p_1,p_2,q_2,q_1)
	\delta(q_1 - p_2)
	\d p_2 \d q_1 
	= 
	\delta(p_1 - p_3) g(p_1) 
\end{equation}
where 
$	g(p)  
\equiv 
\chi(p) ( \hat V * \chi)(p) 
+ 
\chi^\perp (p) ( \hat V * \chi^\perp )(p) $--the explicit form of $g(p)$ is not important, 
but the   $\delta(p_1 - p_3)$ dependence in the last equation implies that the second term in \eqref{TFF VF(t) 2} 
\textit{commutes} with $a_p^*a_p$. 
This fact we shall use in the proof of Lemma \ref{lemma TFF 1}.

\vspace{1mm}

Finally, thanks to the CAR,   the coefficients $\phi_t(p_1,p_2,p_3,p_4)$ 
inside of $:V_F(t):$ can be antisymmetrized with respect to the permutation of the variables $ (p_1,p_2) \mapsto (p_2,p_1)$
and $(p_3,p_4)\mapsto (p_4,p_3)$, respectively. 
Namely, the coefficients $\phi_t$ in $: V_F(t):$ may be replaced by 
\begin{equation}
\Phi_t( \vec p )
\equiv 
 \frac{1}{4}
\Big(
\phi_t(p_1 , p_2 , p_3 , p_4 )
-
\phi_t(p_2 , p_1 , p_3 , p_4 )
+
\phi_t(p_2 , p_1 , p_4 , p_3 )
-
\phi_t(p_1 , p_2 , p_3 , p_4 )
\Big) \ ,
\end{equation}
which can be put in an explicit form, using \eqref{coefficient phi 2}. 
 We record all these results in the following lemma.

\begin{lemma}[Normal ordering]
\label{lemma normal order VF}
Let $t\in \R$ and 
$V_F(t)$ the Heisenberg evolution of the fermion-fermion interaction. 
Then, the following identity holds 
\begin{equation}
	V_F(t)  \ = \    : \! V_F(t) \! :   \ + \  N(g)  \ . 
\end{equation}
Here, $ : \! V_F(t) \! :  \ = \   \int_{  \Lambda^{*4 } } \Phi_t(p_1 \cdots p_4 ) a_{p_1}^* a_{p_2}^* a_{p_3} a_{p_4} \d p_1 \cdots \d p_4$
is the normal ordering of $V_F(t)$,
and 
$N(g) = \int_{  \Lambda^{*} } g(p) a_p^* a_p \d p $, 
where
$g(p)  
\equiv 
\chi(p) ( \hat V * \chi)(p) 
+ 
\chi^\perp (p) ( \hat V * \chi^\perp )(p) $. 
\vspace{1.5mm}

The coefficient function $\Phi_t :  (\Lambda^*)^4 \rightarrow \C $
is partially antisymmetric
\begin{equation}
	\Phi_t (p_1,p_2,p_3,p_4) 
	= - 
\Phi_t (p_2, p_1, p_3,p_4)   
=  + 
	\Phi_t (p_2,p_1,p_4,p_3) 
	=
	- 
		\Phi_t (p_1,p_2,p_3,p_4 ) 
\end{equation}
 and admits the following decomposition   
 $$\Phi_t = \Phi_t^{(1)} + \Phi_t^{(2)}$$    
 where $\Phi_t^{(1)}$ is given by 
 \begin{equation}
 		\Phi_t^{(1)} 
 	(\vec p ) 	 
 	= 	 
 	\,  \frac{1}{2}\, 
 	e^{  it  \Delta E  (  \vec p )  }   
 	\delta(p_1 + p_2 - p_3 - p_4 )
 	\big(
 	\hat V  (p_1 - p_4 )
 	-			  
 	\hat V  (p_1 - p_3 )
 	\big) 	
 	\big( 
 	\chi_{1234}
 	+
 	\chi^\perp_{1234}
 	\big) 		
 \end{equation}
  and $\Phi_t^{(2) }$ is given by 
			\begin{align}
				\Phi_t^{(2)} ( \vec p ) 
			&   = 
		\, 	\frac{1}{2}\, 
 	e^{  it  \Delta E  (  \vec p )  }   
			\delta(p_1 + p_3 - p_2 - p_4 )
			\hat V (p_1 - p_4 ) 
			\big(
			\chi^\perp_{14} \chi_{23} 
			+
			\chi^\perp_{23} \chi_{14}
			\big) 					\nonumber \\
			& \   \ \ -  
			\frac{1}{2}
 	e^{  it  \Delta E  (  \vec p )  }   
			\delta(p_1 + p_4 - p_2 - p_3)
			\hat V (p_1 - p_3 )
			\big(
			\chi^\perp_{13} \chi_{24} 
			+
			\chi^\perp_{24} \chi_{13}
			\big) 	 \ . 
		\end{align}		
		\end{lemma}

\begin{proof}[Proof of Lemma \ref{lemma TFF 1}]
We start with the   normal ordering of $V_F(t)$ found in Lemma \ref{lemma normal order VF}.

\vspace{1.5mm}

	First, we observe that we may disregard the quadratic term $ N(g) \equiv \int_{  \Lambda^* }g(t,p) a_p^* a_p \d p $. 
	Indeed, since $[ a_p^* a_p , N(g)]  = 0 $ we find that for any $ p \in \Lambda^* $
	\begin{align}
		\nu 
		\big(
		[   [ a_p^* a_p , V_F(t)] , V_F(s) ] 
		\big) 
		=
		\nu 
		\big(
		[   [ a_p^* a_p ,   : \! V_F(t) \! : ] , V_F(s) ] 
		\big) 
	\end{align}
	Furthermore, since $ \nu $ is 
	quasi-free and translation invariant, 
	it {\blu satisfies} the identities \eqref{nu vanishing}. 
	Thus, since $[a_p^* a_p ,:V_F(t):]$ is quartic in creation- and annihilation operators, we find that 
	\begin{align}
		\nu 
		\big(
		[   [ a_p^* a_p , V_F(t)] , V_F(s) ] 
		\big) 						\nonumber 
		& =
		\nu 
		\big(
		[   [ a_p^* a_p ,   : \! V_F(t) \! : ] , V_F(s) ] 
		\big) 	\\ 
		& =
		\nu 
		\big(
		[   [ a_p^* a_p ,   : \! V_F(t) \! : ]  \, , \,    :  \! \! V_F(s) \! \!  :  ] 
		\big) 									\nonumber  
		- 
		\nu 
		\big(
		[   N(g) ,  [ a_p^* a_p ,   : \! V_F(t) \! : ]   ] 	
		\big) 		\\
		& = 
		\nu 
		\big(
		[   [ a_p^* a_p ,   : \! V_F(t) \! : ]  \, , \,    :  \! \! V_F(s) \! \!  :  ] 
		\big)  \  , 
	\end{align}
for all $ p \in \Lambda^*$.

\vspace{1.5mm}
Secondly, we note that a standard calculation using the CAR implies that 
	\begin{align}
		\nu 
		\big(
		[   [ a_p^* a_p , V_F(t)] , V_F(s) ] 
		\big) 		
		& = 
		\int_{  \Lambda^{*4}  \times \Lambda^{*4}}
		\big(
		\delta(p - k_1 )
		+
		\delta (p-k_2 )
		-											\nonumber 
		\delta(p -  k_3)
		-
		\delta ( p - k_4 )
		\big) 		\\
		&  \qquad 
		\times \nonumber 
	\Phi_t(\vec k )
\Phi_s(\vec \ell )
		\nu 
		\big(
		[  \,   a_{k_1}^*a_{k_2}^*a_{k_3}a_{k_4}  \,   ,  \,   a^*_{\ell_1}a^*_{\ell_2}a_{\ell_3 }a_{\ell_4}       \,   ] 
		\big) 	 \d \vec k \d \vec \ell 	\\
		& = 
		\int_{  \Lambda^{*4}  \times \Lambda^{*4}}
		M_p(\vec k , \vec \ell)
		\, 
		\nu (   a_{k_1}^*a_{k_2}^*a_{k_3}a_{k_4} a^*_{\ell_1}a^*_{\ell_2}a_{\ell_3 }a_{\ell_4}     )
		\d \vec k \d \vec \ell  
	\end{align}
	where (suppressing the explicit $t,s\in \R $ dependence)
	\begin{align}
		M_p(\vec k , \vec \ell)
		& 
		\equiv
	\Phi_t(\vec k )
\Phi_s(\vec \ell )
		\big(
		\delta(p - k_1 )
		+
		\delta (p-k_2 )
		-											\nonumber 
		\delta(p -  k_3)
		-
		\delta ( p - k_4 )
		\big) 		\\
		&   \  \ 
		-
		\Phi_t(\vec \ell  )
	\Phi_s(\vec k  )
		\big(
		\delta(p - \ell_1 )
		+
		\delta (p-\ell_2 )
		-											\nonumber 
		\delta(p -  \ell_3)
		-
		\delta ( p - \ell_4 )
		\big) 	 \  . 
	\end{align}
	A change of variables $ ( k_3 , k_4 ,\ell_1 , \ell_2 , \ell_3 , \ell_4  ) \mapsto (\ell_3 , \ell_4, k_3 , k_4 , \ell_1 , \ell_2 )$ now yields 
	\begin{align}
		\nu 
		\big(
		[   [ a_p^* a_p , V_F(t)] & , V_F(s) ] 	 
		\big) 		 \\
		& = 
		\int_{  \Lambda^{*4}  \times \Lambda^{*4}}
		M_p( k_1 k_2 \ell_3 \ell_4 , k_3 k_4 \ell_1 \ell_2  )
		\, 
		\nu (   a_{k_1}^*a_{k_2}^* a_{\ell_4} a_{\ell_3} a^*_{k_3} a^*_{k_4 } a_{\ell_2} a_{\ell_1}     ) 
		\d \vec k 				\nonumber  
		\, 
		\d \vec \ell \ .  
	\end{align}
	It is important to note that the coefficient function $M_p(\vec k  , \vec \ell )$ is antisymmetric with respect to $k_1 \mapsto k_2$, $k_3 \mapsto k_4$,   $\ell_1 \mapsto \ell_2$ and $\ell_3\mapsto \ell_4 $ respectively. 
	In addition, $M_p(\vec k  , \vec \ell ) = -  M_p(   \vec \ell  , \vec k )$. 
Indeed, these symmetries allow us to   simplify the right hand side of the last equation as follows.  
	First,   quasi-freeness  of the state $\nu$   implies that 
	\begin{align}
		\nu (   a_{k_1}^*a_{k_2}^* a_{\ell_4} a_{\ell_3} a^*_{k_3} a^*_{k_4 } a_{\ell_2} a_{\ell_1}     ) 
		=
		\det 
		\begin{pmatrix}
			\nu_{11} & \nu_{12} & \nu_{13} & \nu_{14 }	\\
			\nu_{21} & \tilde \nu_{22} & \tilde  \nu_{23} & \tilde  \nu_{24 }	\\
			\nu_{31} & \tilde \nu_{32} &  \tilde \nu_{33} & \tilde  \nu_{34 }	\\
			\nu_{41} &  \tilde \nu_{42} &  \tilde \nu_{43} & \tilde \nu_{44 }	\\
		\end{pmatrix}
	\end{align}
	where  we denote $ \nu_{i j } \equiv  \nu (a_{k_i}^* a_{\ell_j})$ and 
	$\tilde \nu_{ij }\equiv \delta(k_i - \ell_j )  - \nu_{ij} $. 
	Secondly, based on the   symmetries of $M_p (\vec k , \vec{\ell} )$, 
	we may  follow  the algebraic analysis carried out in \cite[pps 374--375]{ESY2004}
to find that 
	\begin{align}
		\nu 
		\big(
		[   [ a_p^* a_p , &  V_F(t)] , V_F(s) ] 
		\big) 		\\
		& = 
		\int_{  \Lambda^{*4}  \times \Lambda^{*4}}
		M_p( k_1 k_2 \ell_3 \ell_4 , k_3 k_4 \ell_1 \ell_2  )
		\, 
		4 
		\big(
		\nu_{1 1}
		\nu_{2 2 }
		\tilde \nu_{3 3 }
		\tilde \nu_{4 4 }			\nonumber 
		+
		4 
		\nu_{ 1 1 }
		\nu_{2 3 }
		\nu_{ 4 2  }
		\tilde \nu_{3 4}
		\big)  
		\d \vec k 
		\d \vec \ell \ . 
	\end{align}
	Thirdly, translation invariance $\nu (a_p^* a_q) =  \delta(p-q) f_0 (p)$ 
	now yields two terms 
	\begin{align}
\label{TFF eq 2}
		\nu 
		\big(
		[   [ a_p^* a_p , V_F(t)] ,  & V_F(s) ] 		\nonumber 
		\big) 		 \\
		&  = 
		4 
		\int_{  \Lambda^* }
		M_p( k_1 k_2 k_3 k_4  , k_3 k_4 k_1 k_2  )					\nonumber  
		f_0 (k_1 ) f_0 (k_2)  \tilde f_0 (k_3 ) \tilde f_0 (k_4 )
		\d \vec k 	\\	
		& \quad+ 	
		16 
		\int_{  \Lambda^* }
		M_p( k_1 k_2 k_2 k_3   , k_3 k_4 k_1 k_4  )
		f_0 (k_1 ) f_0 (k_2)    f_0 (k_3 ) \tilde f_0 (k_4 )
		\d \vec   k . 
	\end{align}
Similarly as in \cite{ESY2004}, we look at the two terms of the right hand side of \eqref{TFF eq 2} by evaluating the function $M_p$ in the different cases. 
	
	\vspace{1.5mm}
	
	  \textit{The second term of \eqref{TFF eq 2}}. Let us show that the second term vanishes. 
Indeed,  we use the fact that  
	$\Phi_t (k_3 k_4 k_1 k_2) = \Phi_{ -t } ( k_1 k_2 k_3 k_4   )$
	together with antisymmetry with respect to $k_1 \mapsto k_2$ and $k_3 \mapsto k_4 $  to find  that 
	\begin{align}
		M_p( k_1 k_2 k_2  & k_3    
		,  
		k_3 k_4 k_1 k_4  )\\
		& =
		2 
		\cos \big[ (t-s)  (E_1 - E_3)   \big] 
		\big( \delta(p - k_3 ) - \delta (p - k_1)  \big)
		\Phi (  k_1 k_2 k_3 k_2  )		\nonumber
		\Phi (  k_1 k_4 k_3 k_4  ) 
	\end{align}
	where we denote $\Phi  (\vec k ) \equiv \Phi_0 (  \vec k ) $. 
	One may verify that 
	$\Phi (k_1 k_2 k_3 k_2)  $ 
	is proportional to $\delta (k_1 - k _3 )$ 
	and, 
	consequently, 
	it holds that 
	$
	\big( \delta(p - k_3 ) - \delta (p - k_1)  \big)
	\Phi (k_1 k_2 k_3 k_2) = 0$.

	\vspace{1.5mm}
	
	  \textit{The first term of \eqref{TFF eq 2}}. 
	Using the fact that  
	$\Phi_t ( k_3 k_4 k_1 k_2) = \Phi_{-t} (   k_1 k_2 k_3 k_4   )$
	one finds
	\begin{align}
		M_p( k_1 k_2  & k_3 k_4   	 , k_3 k_4 k_1 k_2  )					\\ 
		& =
		2 
		\cos \big[ (t-s) \Delta E (\vec k)   \big] 
		|   \Phi ( \vec k )  |^2 
		\big(
		\delta(p - k_1)
		+
		\delta(p - k_2)
		-
		\delta(p - k_3)
		-											\nonumber 
		\delta(p - k_4)
		\big)   \ . 
	\end{align}
	We plug this result back in \eqref{TFF eq 2} to find that 
	after a change of variables $(k_1 k_2 ) \mapsto (k_3 k_4)$, 
	\begin{align}
		\nu 
		\big(
		[   [ a_p^* a_p ,  & V_F(t)] , V_F(s) ] 
		\big) 		\\
		& = 
		4
		\int_{  \Lambda^* }
		\big(
		\delta(p - k_1)
		+
		\delta(p - k_2)
		-
		\delta(p - k_3)
		-											\nonumber 
		\delta(p - k_4)
		\big) 
		\  |   \Phi ( \vec k )  |^2  
		\\
		& \times \cos \big[ (t-s) \Delta E (\vec k)   \big] 
		\Big(
		f_0 (k_1 ) f_0 (k_2)  \tilde f_0 (k_3 ) \tilde f_0 (k_4 )	\nonumber 
		-
		f_0 (k_3 ) f_0 (k_4)  \tilde f_0 (k_1 ) \tilde f_0 (k_2 )
		\Big)  \ . 
	\end{align}
	
	\vspace{1mm}
	
	Finally, we integrate against time and a test function $\vp(p)$
	to find that
	\begin{align}
		\int_0^{t}
		\int_0^{t_1}
		\nu 
		\big(
		[ [ N(\vp) , V_F(t_1) ] , V_F (t_2)]
		\big) 
		\d t_1 \d t_2 
		=
		- 
		t  |\Lambda |
		\int_{  \Lambda^* } 
		\vp (p) 
		Q_t [f_0]  (p) 
		\d p \  
	\end{align}
	where $Q_t [f_0]$ 
is the expression given by 
	\begin{align}
		Q_t [f_0] (p)
		& =
		4
	\pi 
		\int_{  \Lambda^{*4 } }
		\,
\frac{		|   \Phi ( \vec k )  |^2}{|\Lambda|}
		\,
		\Big[ 
		\delta(p - k_1)
		+
		\delta(p - k_2)
		-
		\delta(p - k_3)
		-											  
		\delta(p - k_4)
		\Big]    \\
		& \times \delta_t [     \Delta E (\vec k )    ] 
		\Big(
		\nonumber 
		f (k_3 ) f (k_4)  \tilde f (k_1 ) \tilde f (k_2 )
		-
		f (k_1 ) f (k_2)  \tilde f (k_3 ) \tilde f (k_4 )
		\Big)  
		\d \vec k 
		\ . 
	\end{align}
	where 
	we recall 
	$\delta_1 (x)  =   \frac{2}{\pi}\frac{ \sin^2 (x/2)}{ x^2 } $
	and
	$\delta_t (x)  =  t \delta_1 (t x ). $
Upon expanding $\Phi = \Phi^{(1)} + \Phi^{(2)}$ in the above expression 
with respect to the decomposition found in Lemma \ref{lemma normal order VF}, 
one may check that the formula is in agreement with the operator $Q_t$, as given by Def. \ref{definition Q1}. This finishes the proof of the lemma. 
	\end{proof}

Finally, we prove the last lemma of this section. 
 
 \begin{proof}[Proof of Lemma \ref{lemma TFF 2}]

Let $\vp \in \ell^1$ and $t,s\in\R$,  
let us introduce the following notation for the fermion-fermion double commutator
 \begin{align}
\nonumber 
 	C_F (\vp, t, s)
 & 	\equiv 
 	[ [ N(\vp) , V_F (t)    ] , V_F(s) ] \\
 	&  = 
 	\int_{  \Lambda^{*2 } } 
 	\hat V (k) \hat V(\ell) 
 \Big[  	\Big[	 N(\vp) , D_k^*(t) D_k (t) 	\Big] , D_\ell^*(s) D_\ell(s)  \Big] 
 \ \d k \d \ell  
\label{definition CF}
 \end{align}
where  we have written $V_F(t)$ 
in terms of $D$-operators, see \eqref{VF(t)}. 
For simplicity, we shall assume that $\vp$ is real-valued so that $C_F(\vp,t,s)$ is self-adjoint (in the general case, one may decompose $\vp = \mathrm{Re}  \vp + \i \mathrm{Im}  \vp $ and apply linearity of the commutator). 
We claim that there exists a constant $C>0$ such that 
 \begin{equation}
 	\| C_F 	(\vp,t,s ) \Psi 	\| 
 \leq C 
 	\| \hat  V		\|_{\ell^1}^2
 	|\Lambda|
 	\|	\vp	\|_{\ell^1 }
 	\| \N^2  \Psi  		\|   \   ,
\label{CF bound}
 \end{equation}
for all $\Psi \in \F$. 
To see this, 
we shall expand the 
double commutator of the right hand side of \eqref{definition CF} into eight terms.
In order to ease the notation, we shall 
drop the time labels $t,s\in \R $. Since our estimates are uniform in time, there is no risk in doing so. 
In terms of the contraction operators
 $
D_k^* (\vp) \equiv [N(\vp) , D_k^*]
  $
 and 
 $D_k(\vp) \equiv    [ N(\vp) , D_k]$
 we find 
\begin{align}
\nonumber 
	 \Big[  	\Big[	 N(\vp) , D_k^* D_k 	\Big] , D_\ell^*D_\ell \Big]  
	 &  \ =  \ 
D_k^*(\vp)			 \Big[D_k , D_\ell^*\Big]    D_\ell 
 \ +  \ 
D_k^*(\vp) D_\ell^*
 \Big[D_k , D_\ell \Big]  \\ 
\nonumber 
 &  \ +  \ 
\Big[
D_k^*(\vp) , D_\ell^*
\Big]
D_\ell D_k
\  + \ 
D_\ell^*
 \Big[
D_k^*(\vp) , D_\ell
\Big] 
D_k  \\ 
\nonumber 
& \ + \ 
D_k^* D_\ell^* 
\Big[
D_k(\vp) , D_\ell
\Big]
\ +  \ 
D_k^* 
\Big[
D_k(\vp)  , D_\ell^*
\Big]
D_\ell \\
& \ + \ 
D_\ell^* 
\Big[
D_k^*, D_\ell
\Big]
D_k (\vp) 
\ + \ 
\Big[
D_k^* ,  D_\ell^*  
\Big]
D_\ell D_k(\vp) \ . 
\label{CF expansion}
\end{align}
All these operators 
can be controlled 
 using the  Type-I and Type-IV estimates, found in Lemma \ref{lemma type I}
	and Lemma \ref{lemma type IV},  respectively, 
	together with the commutator identities
	 $[D_k , \N] = [D_k(\vp) , \N]=0$, see Lemma \ref{lemma N commutators}.
	 For instance, given $\Psi \in \F$ 
	 the first term can be estimated as follows 
	 \begin{align}
\nonumber 
	 \|	 D_k^*(\vp)			 \big[D_k , D_\ell^*\big]    D_\ell  \Psi \|
	&  \,  \leq  \, 
	 \|	D_k^*(\vp)			\|		
	 \|	 \big[D_k , D_\ell^*\big]    D_\ell  \Psi 	\| \\
\nonumber 
	 & \,   \leq  \, 
C  \|	 \vp	\|_{\ell^1 } |\Lambda |
	 \|	 \N   D_\ell  \Psi 	\| \\
\nonumber 
	 &  \, =  \, 
	 C  \|	 \vp	\|_{\ell^1 } |\Lambda |
	 \|	   D_\ell \N  \Psi 	\| \\
	 & \, \leq  \, 
	 C  \|	 \vp	\|_{\ell^1 } |\Lambda |
	 \|	      \N^2   \Psi 	\|    
\label{CF estimate}
	 \end{align} 
for a constant  $C>0$. 
 Every other term in the expansion \eqref{CF expansion}
can be analyzed in the same fashion,  and satisfy the same bound --we leave the details to the reader.
Thus,  we plug   the estimate \eqref{CF estimate}
back in the expansion \eqref{CF expansion} 
and integrate over $k,\ell \in \Lambda^*$. 
One then obtains 
 \eqref{CF bound}.

\vspace{1mm }

Let us now estimate the integral remainder term, we fix  $0 \leq t_3 \leq t_2 \leq t_1 $. 
As a first step, since $C_F$ and $\h_I$ are self-adjoint, 
we use the following rough  upper bound 
 \begin{align}
	\nu_{ t_3 }
	\big(
	&  [   [ [ N(\vp)  , V_F(t_1) ] V_F(t_2) ]    ,  \h_I( t_3)  ] 
	\big)
	\leq 2  
	\nu_{ t_3 }
	\Big(
	C_F(\vp , t_1, t_2 )^2 
	 \Big)^{\frac{1}{2}}
\nu_{ t_3 }
\Big(  \h_I(t_3)^2  \Big)^{\frac{1}{2}}
\label{lemma 6.2 eq 1}
%
\end{align}
In view of Remark \ref{remark states}, we can turn the estimate \eqref{CF bound} into the upper bound 
\begin{equation}
	 \nu_{t_3} \Big(C_F(\vp, t_1 , t_2 )^2 \Big)^{\frac{1}{2}}
	 \leq 
	 C \|		\hat V  \|_{\ell^1}^2 
	  | \Lambda| 		\|	 \vp	\|_{\ell^1 }
	  \nu_{t_3}\big( \N^4 \big)^{\frac{1}{2}} \ .  
\label{lemma 6.2 eq 2}
\end{equation}
On the other hand,  using the operator norm estimates \eqref{bounds on b}, 
 a simple but rough estimate 
 for the interaction Hamiltonian is found to be 
 \begin{align}
 	\|	 \h_I (t) \Psi	\|				\nonumber
 	& \  \leq  \ 
\lambda  	\|	 V_F(t ) \Psi 	\|
 	+ 
\lambda  	\|	 V_{FB} (t ) \Psi 	\|
 	+
\lambda  	\|	 V_{BB} (t ) \Psi 	\|		\\
 	&  \ \lesssim  \ 
\lambda  	\|	 \hat V	\|_{\ell^1}		\|	 \N^2 \Psi	\|
 	+ 
\lambda  	\|	 \hat V	\|_{\ell^1}
 	R 
 	\|	 \N  \Psi \|
 	+ 
\lambda  	\|	 \hat V	\|_{\ell^1}
 	R^2	\nonumber 
 	\|	 \Psi	\|    \\
 	& \  \lesssim  \ 
\lambda  	\|	 \hat V	\|_{\ell^1}
 	\big( 
 	\|	 \N^2 \Psi	\|
 	+ 
 	R^2	
 	\|	 \Psi	\|    
 	\big) 
 	\nonumber 
 \end{align}
 where we recall that  $R =|\Lambda| p_F^{d-1}$. 
 Consequently, in view of Remark \ref{remark states} we find that 
 \begin{equation}
\label{lemma 6.2 eq 3}
 		 \nu_{t_3}
 		  \Big( 
 		  \h_I(t_3)^2 
 		  \Big)^{\frac{1}{2}} 
 	\leq 
 	C\lambda   \|		\hat V  \|_{\ell^1}^2 
\Big(
	\nu_{t_3}\big( \N^4 \big)^{\frac{1}{2}} + R^2  
\Big)    
 \end{equation}
 where we used the fact that $\nu_t (\1) = 1 $ for all $t \in \R . $ 
 The proof of the lemma is now finished once we combine 
 Eqs. \eqref{lemma 6.2 eq 1}, \eqref{lemma 6.2 eq 2} and \eqref{lemma 6.2 eq 3}, 
 and integrate over the   time variables $ 0 \leq t_3 \leq t_2 \leq t_1 \leq t $.
  \end{proof}

 \section{Leading Order Terms II: Emergence of $B$}
 \label{section TFB}
 The main purpose of this section 
 is to analyze the term $T_{FB,FB}(t)$
 found in the double commutator expansion \eqref{expansion f}, introduced in Section \ref{section preliminaries}. 
 In particular, we show that this term gives rise to the operator $B_t$, as given in Def. \ref{definition B},
 corresponding to the second leading order term describing the dynamics of $f_t(p).$ 
 It describes interactions between particles/holes 
 as mediated by 
  \textit{virtual bosons} around the Fermi surface. 
 This is manifest in the fact that, as we shall see, 
 it contains  the  \textit{propagator} of free bosons
 \begin{equation}
 	G_k (t-s) \equiv \< \Omega, [b_k  (t) ,  b_k^* (s)]  \Omega 	\>_\F
 \end{equation}
defined for  $ k \in \Lambda^*$, and $t,s \in \R$.

\vspace{2mm}
We state the main result of this section in the following proposition, 
which we prove in the remainder of the section.

 \begin{proposition}
[Analysis of $T_{FB,FB}$]
 	\label{prop TFB}
Let $T_{FB,FB}(t,p)$ be the quantity defined in Eq.  \eqref{T alpha beta} for $\alpha  =  \beta = FB$,
and let $m>0$.
Then, there exists a constant  $C>0$
such that for all  $\vp \in \ell_m^1 $ and $ t \geq 0 $ 
the following inequality holds true 
 	\begin{align}
\nonumber 
 		\big|		
 		T_{FB,FB}
 	 & 	(t,   \vp)
 		  + 
 		|\Lambda |
 		t \< \vp , B_t [f_0]\>
 		\big|	   \\ 
 	 &  \leq 
 C | \Lambda|  		t^2 
 		\|	 \vp	\|_{\ell^1 }  	\|	 \hat V	\|_{\ell_1}^2   
 		\nonumber 
 		\sup_{ \tau \leq t }
 		\Big(
 		R^{\frac{1}{2}	}
 		\nu_\tau (\po)^{\frac{1}{2}}
 		\nu_\tau  (\N)^{\frac{1}{2}}
 		+ 
 		C R^{\frac{3}{2}}
 		\nu_\tau (\po)^{\frac{1}{2}} 
 		+
 		R p_F^{-m}
 		\nu_\tau (\N_1^2)
 		\Big)	 \\
 		&  \,  +  \, 
| \Lambda|   	 t^3 	\lambda  R 
 		\|	 \vp	\|_{\ell^1 }  	\|	 \hat V	\|_{\ell_1}^3 
 		\sup_{ \tau \leq t }
 		\Big(
 		R^{  \frac{3}{2} } \nu_\tau (\po)^{ \frac{1}{2 }}  
 		+ 
 		R \nu_\tau  ( \po )
 		+ 
 		R p_F^{ -m}
 		\nu_\tau (\N )^{\frac{1}{2 } }
 		\Big)
 	\end{align}
 	where  $T_{FB,FB} (t,\vp) \equiv \<\vp, T_{FB,FB}(t)\>$ and $B_t $   is   given in Def. \ref{definition B}. 
 \end{proposition}
 
 \vspace{2mm}
 
\begin{remark}
	 In order to prove Proposition \ref{prop TFB}, we expand 
	$
	T_{FB,FB}
	$ 
	into several terms and analyze each one  separately. 
	This expansion is based {\blu on} the following two observations: 

\vspace{1mm}

\noindent \textit{(i)}  For any self-adjoint operators $N, T ,S   $ and  state $\mu $,
		there holds: 
		\begin{equation}
			\label{eq:N}
			\mu  \big( [[N,  T + T^*] , S]		\big)
			=
			2 \mathrm{Re} \, 	\mu  \big( [[N, T  ] , S]		\big) \ . 
		\end{equation}

\noindent\textit{ (ii)  }Thanks to the symmetries  
		$D_k = D^*_{-k }$,  
		$\hat V  ( - k) =  \hat V (k)$
		and the vanishing commutator $[D^*_k , b_k ]= 0$,
		starting from the representation \eqref{VFB(t)} 
		we may rewrite the fermion-boson interaction term 
		as  
		\begin{equation}
			V_{FB}   (t) 
			= 
			\int_{  \Lambda^* } \hat V  (k)
			B^*_k (t) D_k (t) 
			\d k 
			\qquad 
			\t{where}
			\qquad 
			B^*_k  (t) 
			\equiv 
			b^*_k (t)   + b_{-k }  (t) \ . 
		\end{equation} 
\end{remark}

  \vspace{2mm}

Starting from \eqref{T alpha beta}, based on these two observations  we are able to  rewrite
the term $T_{FB,FB}$ 
 for all $ t\in \R$
and $\vp \in \ell^1$  in the following form  
 \begin{align}
 	T_{FB,FB} &
 	(t,\vp) 		 
 	\nonumber 
 	\\ 				
 	&  = 
 	2 \mathrm{Re}  			\nonumber 
 	\int_0^t 
 	\int_0^{t_1}
 	\int_{  \Lambda^{*2 } }
 	\hat V (k) 
 	\hat V(\ell)
 	\nu_{t_2 }
 	\Big(
 	[ [	 N(\vp) , D_k^*(t_1) b_{k}(t_1)		] , 		B^*_{\ell}(t_2) D_{\ell}(t_2 ) ]  
 	\Big) 
 	\d t_1 \d t_2  
 	\d k \d \ell 
 	\\
 	& \equiv 
 	M (t,\vp) 
 	+
 	R^{(1)} (t,\vp)
 	+
 	R^{(2)} (t,\vp)
 	+
 	R^{(3)} (t,\vp)
 	+
 	R^{(4)} (t,\vp)
\label{TFB expansion}
 \end{align}
 where  in the second line we have expanded the commutator into five terms.
 The first one we shall refer to as the \textit{main term}, and is defined as follows
 \begin{align}
\label{TFB main term}
 	M (t,\vp)
 	&   = 
 	2 \mathrm{Re} \, 		 
 	\int_0^t 
 	\int_0^{t_1}
 	\int_{  \Lambda^{*2 } }
 	\hat V (k) 
 	\hat V(\ell)
 	\nu_{t_2 }
 	\Big(
 	D_{k}^*(t_1,\vp)
 	[b_{k}(t_1), b^*_\ell(t_2)]
 	D_{\ell}( t_2  )
 	\Big) 
 	\, 
 	\d t_1 \d t_2  
 	\d k \d \ell     \ . 
 \end{align} 
The last four, which we shall refer to as the \textit{remainder terms}, are defined as follows 
 \begin{align}
 	R^{(1)} (t,\vp)
 	&    = 
 	2 \mathrm{Re} \, 			\nonumber 
 	\int_0^t 
 	\int_0^{t_1}
 	\int_{  \Lambda^{*2 } }
 	\hat V (k) 
 	\hat V(\ell)
 	\nu_{t_2 }
 	\Big(
 	D_{ k }^*(t_1,\vp)
 	B^*_\ell (t_2)
 	[b_{ k}(t_1),  D_{\ell}(t_2)]
 	\Big) 
 	\, 
 	\d t_1 \d t_2  
 	\d k \d \ell 
 	\\
 	R^{(2)} (t,\vp)
 	&    = 
 	2 \mathrm{Re} \, 			\nonumber 
 	\int_0^t 
 	\int_0^{t_1}
 	\int_{  \Lambda^{*2 } }
 	\hat V (k) 
 	\hat V(\ell)
 	\nu_{t_2 }
 	\Big(
 	[ D_{k}^*(t_1 ,\vp) , B^*_\ell(t_2)	] 
 	D_{\ell}( t_2 )  b_{ k }(t_1 )  
 	\Big) 
 	\, 
 	\d t_1 \d t_2  
 	\d k \d \ell 
 	\\
 	R^{(3)} (t,\vp)
 	&    = 
 	2 \mathrm{Re} \, 			\nonumber 
 	\int_0^t 
 	\int_0^{t_1}
 	\int_{  \Lambda^{*2 } }
 	\hat V (k) 
 	\hat V(\ell)
 	\nu_{t_2 }
 	\Big(
 	B_{\ell}^* (t_2) 
 	[D_{k}(t_1,\vp)	 , D_{\ell}(t_2)	]
 	b_{k}(t_1)
 	\Big) 
 	\, 
 	\d t_1 \d t_2
 	\d k \d \ell 
 	\\
 	R^{(4)} (t,\vp)
 	&    = 
\label{TFB reminder terms}
 	2 \mathrm{Re} \, 		 
 	\int_0^t 
 	\int_0^{t_1}
 	\int_{  \Lambda^{*2 } }
 	\hat V (k) 
 	\hat V(\ell)
 	\nu_{t_2 }
 	\Big(
 	[
 	D_{k}(t_1)
 	b_{k }(t_1,\vp)
 	,
 	B_{\ell}^* (t_2) 
 	D_{\ell}(t_2)	
 	]
 	\Big) 
 	\, 
 	\d t_1 \d t_2  
 	\d k \d \ell  \   .
 \end{align}
\begin{remark}
We remind 
  the reader that we have previously  introduced the notation 
\begin{equation}
	D_k^* ( t ,\vp )
	= 
	[N(\vp)  , D_k^*(t)]	 
	\quad \t{and}
	\quad 
	b_k(t, \vp)
	= 
	[N(\vp)  ,  b_k (t)	]		  
\end{equation}
for any $ k \in \Lambda^*$ and $ t \in \R $. We have also used the fact that $[b_k(t), b_{\ell}(s)]=0$. 	
\end{remark}

\vspace{2mm}

 In the remainder of this section, we shall study these five terms separately.
The proof of Proposition \ref{prop TFB} follows directly from the following two lemmas. 
Here, we remind the reader that $R= |\Lambda|p_F^{d-1}$ is our recurring parameter. 

\begin{lemma}
	[The main term]
		\label{lemma TFB main term}
Let $M$ be the quantity defined in \eqref{TFB main term}, and let $ m> 0 $. 
Then, there exists a constant 
$C>0$ 
such that for all $\vp \in \ell_m^1  $
and $t \geq 0 $ the following estimate holds true 
	\begin{align}
		|
		M (t,\vp)
		+ 
|\Lambda |		t    & 
		\langle 
		\vp ,   B_t[f_0]
		\rangle 
		|		\\ 
		\nonumber 
		&  
		\leq C 
		t^2 
		\|		 \hat V	\|_{\ell^1 }^2 
		|\Lambda|
		\|	 \vp 	\|_{\ell^1_m} 
		\sup_{ \tau \leq t }
		\Big(
		R^{\frac{1}{2}}
		\nu_\tau (\po)^{ 	\frac{1}{2}	 }
		+ 
		p_F^{ -m  }
		\Big) 					
		\nu_\tau (\N^2)^{ \frac{1}{2} }    \\ 
		\nonumber 
		& + 
 C 		\lambda 
		t^3 
		R
		|  \Lambda |
		\|	 \vp	\|_{\ell^1_m }
		\|	 \hat V	\|_{\ell^1}^3 
		\sup_{0 \leq \tau \leq t 	}
		\Big(	
		R^{\frac{3}{2}	}
		\nu_\tau (\po)^{  \frac{1}{2}}
		+ 
		R \nu_\tau (\po )
		+ 
		Rp_F^{ -m  }
		\nu_\tau (\N )^{ \frac{1}{2} }
		\Big)   
	\end{align}
	where the operator $B_t$ was introduced in Def. \ref{definition B}. 
\end{lemma}

\begin{lemma}[The remainder terms]
	\label{lemma TFB reminder terms}
Let $R^{(1)}$, $R^{(2)}$, $R^{(3)}$	  and $R^{(4)}$
be the quantities defined in \eqref{TFB reminder terms}, and let $m>0. $ 
Then, there exists  a constant $C> 0 $ such that for all $ \vp \in \ell_m^1 $ and $t \geq 0 $

\vspace{1.5mm} 
\noindent (1) There holds 
\begin{equation}
		\label{lemma remainder 1}
	| R^{(1)} (t,\vp) | 
	\leq C  
	t^2 
	\|	 \hat V	\|_{\ell^1 }^2 
	\|	 \vp	\|_{\ell^1 }
	|	\Lambda	|
	R^{\frac{3}{2}}
	\sup_{ 0 \leq t \leq \tau}
	\nu_\tau (\po)^{\frac{1}{2}	}  \  . 
\end{equation}

\vspace{1.5mm} 
\noindent (2) There holds 
\begin{align}
\label{lemma reminder 2}
		| R^{(2)} (t,\vp) | 
\leq C 
	t^2 
	\|	 \hat V	\|_{\ell^1 }^2 
	\|	 \vp	\|_{\ell^1_m }
	|\Lambda |
	\frac{R}{p_F^m }
	\sup_{ 0 \leq t \leq \tau}
	\nu_\tau 
	(\N^2)^{	\frac{1}{2 }} \ . 
\end{align}

\vspace{1.5mm} 
\noindent (3) There holds 
\begin{equation}
\label{lemma reminder 3}
		| R^{(3 )} (t,\vp) | 
\leq C 
	t^2 
	\|	 \hat V	\|_{\ell^1 }^2 
	\|	 \vp	\|_{\ell^1 }
	|	\Lambda	|
	R^{\frac{3}{2}}
	\sup_{ 0 \leq t \leq \tau}
	\nu_\tau (\po)^{\frac{1}{2}	}  \ . 
\end{equation}
\vspace{1.5mm} 
\noindent (4) There holds 
\begin{equation}
	 		\label{lemma remainder 4}
	| R^{(4)} (t,\vp) | 
\leq C 
	t^2 
	\|	 \hat V	\|_{\ell^1 }^2 
	\|	 \vp	\|_{\ell^1_m }
	|\Lambda |
	\frac{R}{p_F^m }
	\sup_{ 0 \leq t \leq \tau}
	\nu_\tau 
	(\N^2)  \ . 
\end{equation}
\end{lemma}

\begin{proof}[Proof of Proposition \ref{prop TFB}]
Straightforward combination of the expansion given in Eq.  \eqref{TFB expansion}, 
 and the estimates contained in Lemmas \ref{lemma TFB main term} and \ref{lemma TFB reminder terms}. 
\end{proof}

We dedicate the rest of the section to the proof of   Lemma \ref{lemma TFB main term} and \ref{lemma TFB reminder terms}, respectively. 
This is done in the two following subsections. 
 
 \subsection{Analysis of the main  term}
The main goal of this  subsection is to prove Lemma \ref{lemma TFB main term}
by analyzing the main term $M $. 
Our first step in this direction is to    give an additional
 decomposition of $M $.
 Indeed, we start by noting that the commutator 
 of the bosonic operators may be written as  (see \eqref{boson commutator} in Section \ref{toolbox1})
 \begin{align}
 	\label{commutator of b}
 	[b_k (t), b_\ell^* (s)]
 	& 
 	=   \ 
 	\delta ( k - \ell )
 	G_k(t -s )   \1 
 	- 
 	\calR_{ k , \ell  }(t,s) \ , 
 \end{align} 
which corresponds to a decomposition into its  ``diagonal" and ``off-diagonal" parts,
with respect to the variables $k, \ell \in \Lambda^*$. 
Here,  $G_k(t-s)$ is a scalar that 
 corresponds to the   \textit{propagator}
 of the boson field--it can be explicitly calculated to be 
 \begin{equation} 
\label{Gk}
 	G_k(t -s  )   = 
 	\<	\Omega, 
 	[b_k (t), b_k^* (s)] , \Omega	\>_\F 
 	= 
 	\int_{  \Lambda^* }
 	\chi^\perp(p)
 	\chi(p-k)
 	e^{  -i  (t -s )  (E_p   +  E_{p - k})  }					
 	\d p		\ . 
 \end{equation}
for all $k\in \Lambda^*$ and $t,s \in \R$. 
On the other hand, the second term of \eqref{commutator of b} corresponds to an operator  remainder term 
 \begin{align} 
\label{cal R k ell}
 	\calR_{ k ,\ell  }(t,s)   & 
 	\equiv 
 	\int_{  \Lambda^* }
 	\chi^\perp (p) 	 \chi^\perp (p+ \ell - k )	\chi(p-k) 			\nonumber  
 	e^{  -i (t-s ) E_{p-k}    }	
 	a_p^* (t)  a_{p + \ell - k } (s) \, \d p 	\\ 
 	&   + 
 	\int_{  \Lambda^* }
 	\chi  ( h )
 	\chi (  h + \ell - k )
 	\chi^\perp (h + \ell)
 	e^{  - i (t-s ) E_{h+k}    } 
 	a_h^* (t)  a_{h + \ell - k } (s) \, \d h   \ . 
 \end{align}

\vspace{2mm}

The decomposition of the bosonic commutator given in \eqref{commutator of b}
now suggests that we split the main term into two parts. 
 The first one contains the  $\delta(k - \ell) $ function, 
 and the
 second one contains the operator  $\mathcal R _{k , \ell}$.
 In other words, we shall consider 
 \begin{align}
 	\label{T FB 1 }
 	M (t, \vp )
 	=
 	M^{\delta }(t, \vp )
 	+ 
 	M^{\calR }(t, \vp ) \ . 
 \end{align}
We   analyze $M^\delta$ and $M^\calR $ separately.
The proof of Lemma \ref{lemma TFB main term} is given at the end of the subsection.

Upon expanding the bosonic commutator \eqref{commutator of b} in \eqref{TFB main term}, we evaluate  the $\delta(k-\ell)$ function to find that 
\begin{align}
	M^{\delta }(t, \vp )
	&  
	= 
	2 \mathrm{Re} \, 		 
	\int_{  \Lambda^{* } }
	| \hat V (k)  |^2 
	\int_0^t 
	\int_0^{t_1}
	G_k(t_1 - t_2 )
	\nu_{t_2 }
	\Big(
	D_{ k}^*( t_1 ,\vp)
	D_{k }( t_2 )
	\Big) 
	\, 
	\d t_1 \d t_2  
	\d k    \ . 
\end{align}
In order to analyze the above expectation value, 
we shall expand $\nu_{t_2}$
with respect to the interaction dynamics \eqref{interaction dynamics}. 
Namely, we consider 
\begin{align}
\label{M delta expansion}
	M^\delta = M^\delta_0 + M^\delta_1 
\end{align}
where for all $t\in \R$ and $\vp \in \ell^1$ 
we define 
\begin{align} 
M^\delta_0     (t,\vp )
	& \equiv 
	2 \mathrm{Re} \, 		 
	\int_{  \Lambda^{*2 } }
	| \hat V (k)  |^2 
	\int_0^t 
	\int_0^{t_1}
	G_k(t_1 - t_2 )
	\nu
\big( 
	D_{ k}^*( t_1 ,\vp)
	D_{k }( t_2 )
\big) 
	\, 
	\d t_1 \d t_2     
	\d k  
	\end{align}
together with 
\begin{align}	
\label{M 1 delta}
	M^\delta_1  &  (t,\vp ) \\
\nonumber 
& 	  \equiv 
2  
\mathrm{ Im } 	 
 \int_{  \Lambda^{*  } }
| \hat V (k)  |^2 
\int_0^t 			 
\int_0^{t_1}
\int_0^{t_2  }
G_k(t_1 - t_2 )	 
\nu_{  t_3  }
\big(  
[   D_k^* (t_1 ,\vp)
D_k ( t_2 ) ,
\h_I( t_3 )		 
] 
\big) 
\d t_1 \d t_2\d t_3     
\d k   \ . 
\end{align}

\vspace{3mm} 

 First, we identify that from the  first term in the above expansion will the $B_t$ operator emerge. Namely, we claim that

 \begin{lemma}
\label{claim 1}
 	For all $t\in\R$ and real-valued $\vp \in \ell^1$, the following identity holds true 
 	\begin{align}
 		M_0^\delta(t,\vp)
 		=
 		-t 
 		\< \vp ,B_t[f_0] \>  
 	\end{align}
 	where $B_t$ is the operator given in Def. \ref{definition B}. 
 \end{lemma}

Once this is established, it suffices to control the second term in the expansion of $M^\delta$, that is, 
the extra integral remainder term in \eqref{M delta expansion}, $M_1^\delta$.

\begin{lemma}  
	\label{claim 2}
For all $m>0$ there exists a constant $C>0$
such that for  all $t \geq 0 $ and $\vp \in \ell^1$ the following estimate holds true 
	\begin{align}
		|  
		M^\delta_1(t,\vp)  
		|    
 		\leq C  
		\lambda 
		t^3 
		R
		| \Lambda |
		\|	 \vp	\|_{\ell^1_m }
		\|	 \hat V	\|_{\ell^1}
				\|	 \hat V	\|_{\ell^2}^2
		\sup_{0 \leq \tau \leq t 	}
		\Big[
		R^{\frac{3}{2}	}
		\nu_\tau (\po)^{  \frac{1}{2}}
		+ 
		R \nu_\tau (\po )
		+ 
	 \frac{R}{p_F}	 
		\nu_\tau (\N )^{ \frac{1}{2} }
		\Big] 
		  \ . 
	\end{align}
 \end{lemma}

\begin{proof}[Proof of Lemma \ref{claim 1}]
	Let us fix       $ k \in \Lambda^*$, $t,s\in\R$ and $\vp \in \ell^1 $, which we   assume is real-valued in the remainder of the proof.
	In order to prove our claim,  we write 
\begin{align}
\nonumber 
		D_k^*(t,\vp)
		& = 
		\int_{  \Lambda^{*} }		
		\chi^\perp(p_1,p_1-k) 
		[\vp(p_1) - \vp (p_1 - k )]
		a^*_{p_1}(t) a_{p_1 - k } (t)
		 \d p_1 
		 \\ 
		& - 
				\int_{  \Lambda^{*} }		
		\chi(h_1,h_1+  k) 
		[\vp(h_1) - \vp (h_1  + k )]
		a^*_{h_1}(t) a_{h_1 + k } (t)
		 \d h_1 \ ,  \\ 
\nonumber 
		D_k(s)
& = 
\int_{  \Lambda^{*} }		
\chi^\perp(p_2,p_2+ k) 
 a^*_{p_2  }(s) a_{p_2 +k} (s)
  \d p_2  \\ 
& - 
\int_{  \Lambda^{*} }		
\chi(h_2,h_2 -  k)  
a^*_{h_2 }(s) a_{h_2  -  k } (s)
 \d h_2  \ . 
\end{align}
Thus  we are able to calculate that  the following four terms arise 
	\begin{align}	
		\nu  \big( 					
		\nonumber 
		D_k^* (t,\vp)
		D_k (s)
		\big)
		&  =  
		\int_{  \Lambda^{*2 } }
	\chi^\perp (p_1, p_2, p_1 - k, p_2+k) 
		[  \vp( p_1 ) - \vp (p_1 - k ) ]	\\
		\nonumber
		& \qquad \times 
		\ \nu
		\Big( 
		a_{p_1}^* (t) a_{p_1-k}(t) a_{p_2}^*(s) a_{p_2+k} (s)   
		\Big)
		\d p_1  \d p_2     \\
		& + 
				\int_{  \Lambda^{*2 } }
		\chi(h_1,  h_2, h_1+k , 					  h_2- k )		
			\nonumber
		[  \vp(  h_1 ) - \vp ( h_1+ k  ) ]
		\\
		&	\qquad  \times   					
			\nonumber 
		\ \nu
		\Big( 
		a_{h_1}^* (t) a_{ h_1+ k }(t) a_{ h_2 }^*(s) a_{h_2 - k } (s)   
		\Big) 	\d h_1  \d h_2 	 \\
\nonumber
			&  - 	\int_{  \Lambda^{*2 } }
		\chi^\perp (p_1,   p_1 -k )    \chi(h_2 ,h_2 -  k )
		[  \vp( p_1 ) - \vp (p_1 - k ) ]	\\
		\nonumber
		& \qquad \times 
		\ \nu
		\Big( 
		a_{p_1}^* (t) a_{p_1-k}(t) a_{h_2}^*(s) a_{h_2 -k} (s)   
		\Big)
		\d p_1  \d h_2   \\
\nonumber
 & - 
 				\int_{  \Lambda^{*2 } }
      \chi(h_1, h_1 + k) \chi^\perp (p_2, p_2+k)
		[  \vp( h_1) - \vp ( h_1 + k  ) ]	\\ 
		& \qquad \times 
		\ \nu
		\Big( 
		a_{h_1}^* (t) a_{h_1+ k}(t) a_{p_2}^*(s) a_{p_2+k} (s)   
		\Big)
		\d h_1  \d p_2     \ . 
		\label{TFB eq 1}
	\end{align}
In order to calculate the four terms displayed {\blu on the right hand side} of \eqref{TFB eq 1} we use the   fact
that   $\nu$ is translation invariant and quasi-free.
In particular,  
it is possible to calculate that 
for any 
$p_1 ,p_2,q_1,q_2 \in \Lambda^*$
the following relation {\blu holds} true 
\begin{align}
	\nonumber 
	\nu\Big(
	a^*_{p_1}(t) a_{q_1 } (t)  a^*_{p_2  }(s) a_{q_2 } (s)
	\Big) 
	&   = 
	\delta(q_1 - p_1)
	\delta(q_2  - p _2)
	f_0 (p_1) 
	f_0(p_2) \\
	& + 
	\delta(q_1 - p_2 )
	\delta(q_2  - p_1)
	e^{i(t-s) (E_{p_1}- E_{p_2})}
	f_0(p_1) 
	\widetilde f _0 (p_2 )   \  . 
\label{nu evaluation}
\end{align}

\vspace{1mm}

This implies that the third and fourth term in \eqref{TFB eq 1} are zero. 
Indeed, for the third term we choose in \eqref{nu evaluation}
$p_1  = p_1$, $q_1 = p_1 - k$, $p_2 = h_2$ and $q_2 = h_2 - k $ to find  that 
\begin{align}
	\nonumber 
	\nu\Big(
	a^*_{p_1}(t) a_{ p_1 - k  } (t)  a^*_{ h_2  }(s) a_{ h_2 - k } (s)
	\Big) 
	&   = 
|\Lambda |
	\delta( k )
	f_0 (p_1) 
	f_0(h_2 ) \\
	& + 
|\Lambda |
	\delta(k)
	\delta(p_1 - h_2)
	e^{i(t-s) (E_{p_1}- E_{h_2})}
	f_0(p_1) 
	\widetilde f _0 (h_2 )     . 
\end{align}
It suffices to note that the right hand side is proportional to $\delta(k)$, and that  $		[  \vp( p_1 ) - \vp (p_1 - k ) ] \delta(k) = 0$.
This shows that the third term has a null contribution. The same analysis holds for the fourth term in \eqref{TFB eq 1}.

 \vspace{1mm}

In a similar fashion, the first and second term in \eqref{TFB eq 1} can be collected and rewritten thanks to \eqref{nu evaluation} to find that 
	\begin{align}
		\nu  \big( 			 
		D_k^*  & (t,\vp)D_k (s)
		\big)\\
	\nonumber
		&  = 
|\Lambda |
\int_{  \Lambda^* }
\chi^\perp (p , p-k)
[  \vp(p) - \vp (p - k ) ]
e^{
	i (t-s)
	(E_p - E_{ p -k  })
}
\ 
f_0(p )
\widetilde{ f _0} 
( p - k )
\d p 
\\
	\nonumber
& + 
		|\Lambda |
		\int_{  \Lambda^* }
		\chi(h,  h+k)									 
		[  \vp(  h ) - \vp ( h+ k  ) ]
		e^{
			i (t-s)
			(E_h - E_{ h + k   })
		}
		\ 
		f_0( h  ) 
		\widetilde{ f _0} 
		(h + k )
		\d h  
	\end{align}
where we have dropped all terms in \eqref{nu evaluation} containing $\delta(k)$. 
Now, we integrate  in time the above equation to find that 
	\begin{align}
\nonumber     
		\int_0^t 
		\int_0^{t_1}
		 \nu_{t_2}
		\Big(
		G_k(t_1 - t_2  ) & 
		D_k^*(t_1,\vp)
		D_k (t_2)
		\Big)
		\d t_2
		\d t_1 
		\\
		& = 
				|\Lambda | 
		\int_{  \Lambda^* }
\chi^\perp (p , p-k)
		\Bigg( 
		\int_0^t 
		\int_0^{t_1}
		G_k( t_2  )
		e^{		\nonumber  
			i t_2
			(E_p - E_{ p  -   k   })
		}
		\d t_2 \d t_1 
		\Bigg)  \\
\nonumber 
 & \qquad  \qquad \qquad \qquad 	 \times 	[  \vp(  p  ) - \vp (  p -k   ) ]
		f_0( p ) \widetilde f _0 (p - k )
		\d   p 
		\\
& + 
		|\Lambda | 
		\int_{  \Lambda^* }
		\chi (h , h + k)
		\Bigg( 
		\int_0^t 
		\int_0^{t_1}
		G_k( t_2  )
		e^{		\nonumber  
			i t_2
			(E_h - E_{ h + k   })
		}
		\d t_2 \d t_1 
		\Bigg)  \\
&  \qquad  \qquad \qquad \qquad  	 \times 	[  \vp(  h ) - \vp ( h+ k  ) ]
		f_0( h  ) \widetilde f _0 (h+k )
		\d   h 
		\label{TFB eq 2} \  . 
	\end{align}

To finalize the proof, let us identify the right hand side of the last displayed equation, with the operator $B_t$ as given by Def. \ref{definition B}. 	
Indeed, consider the second term of Eq. \eqref{TFB eq 2}. 
We may calculate explicitly the integrals with respect to time as follows. 
First, we  rewrite $G_k(t)$  in terms of the variables $r= p-k $ 
\begin{equation}
		G_k(t -s  )    
	= 
	\int_{  \Lambda^* }
	\chi(r)
		\chi^\perp(r+k)
	e^{  -i  (t -s )  (E_r   +  E_{r + k})  }					
	\d r		\ . 
\end{equation}
Let $h \in \B\cap \B-k$. After integration in time and taking the real part 
we find 
	\begin{align}
		2 \mathrm{Re} \, 		 
		\int_0^t 
		\int_0^{t_1}					
		\nonumber
		G_k( t_2  )
		& e^{	 
			i t_2
			(E_h - E_{ h + k   })
		}	 					 
		\d t_2 \d t_1 		\\ 
		& =
		\int_{  \Lambda^* }
		\chi(r)
	\chi^\perp(r+k)
														\nonumber
		\  2 \mathrm{Re}  
		\int_0^t 
		\int_0^{t_1}
		e^{   i  t_2  (E_h - E_{ h + k  } - E_r  -  E_{r + k})  }		
		\d t_2 \d t_1 			
		\d r 	\\ 
		& =  						 
		\int_{  \Lambda^* }
		\chi(r)
		\chi^\perp ( r+ k )
 \, 		2 \pi 
  t    		\delta_t  \, 
		\big(  
		E_h - E_{ h + k  }  -  E_r    -  E_{r+k }
		\big)  
		\d r 			\nonumber \\ 
		&  =  2 \pi 
		t \,  \alpha^H_t( h,k) \ . 
	\end{align}
Here, $\delta_t (x) $
corresponds to the mollified Delta function defined as 
$		\delta_t(x) = t \delta_1 (tx)$ 
where
	$	\delta_{1}(x) =  \frac{2}{\pi}  \sin^2(x/2)/x^2 .$ 
	On the other hand, $\alpha^H_t $ 
	corresponds to  the object defining $B_t$, see \eqref{alpha H} in Def. \ref{definition B}. 
A   similar calculation shows that the first term of the right hand side of Eq. \eqref{TFB eq 2}
can be put in the following form 
	\begin{align}
\chi^\perp(p,p-k)
	2 \mathrm{Re} \, 		 
	\int_0^t 
	\int_0^{t_1}					
	\nonumber
	G_k( t_2  )
	 e^{	 
		i t_2
		(E_p - E_{ p -  k   })
	}	 					 
	\d t_2 \d t_1 		 
	 =  2\pi 
	t \,  \alpha^P_t(p ,  k)  
\end{align}
where $\alpha_P$ is the quantity given in \eqref{alpha P}, see Def. \ref{definition B}. 
We integrate against $ | \hat V (k)  |^2 $  and change variables $h \mapsto h-k, \ p \mapsto p + k $   
	in the ``gain term"
of \eqref{TFB eq 2}
	to find that 
	\begin{align}
		2 \mathrm{Re} \, 		 
		\int_{  \Lambda^{*2 } }
		| \hat V (k)  |^2 
		\int_0^t 
		\int_0^{t_1}
		\nu
		\Big(
		G_k(t_1 - t_2 )
		D_{ k}^*( t_1 ,\vp)
		D_{k }( t_2 )
		\Big) 
		\, 
		\d t_1 \d t_2   \nonumber 
		\d k      
		= -t 
		\< \vp ,B_t[f_0] \>  
	\end{align}
	where $B_t$ is the operator  given  in Eq. \eqref{mollified B}. This finishes the proof. 
\end{proof}

\begin{proof}
	[Proof of Lemma \ref{claim 2}]
Let us fix throughout the proof the time label $t \in \R$, the parameter $m>0$
and the test function $\vp \in \ell^1$. 
Based on the fact that $ \|	 G_k(\tau)	\|_{B(\F)}\lesssim R$ for all $ k\in\Lambda^*$ and $ \tau \in \R$, 
our starting point is the  following elementary inequality
\begin{align}	
 | 	M^\delta_1   (t,\vp ) | 
 \, 	  \lesssim   \, 
R 	  	\,	\|	 \hat V	\|_{\ell_2}^2 		\, 
	 	t^3  
	 \sup_{  k\in \supp \hat V, t_ i \in [0,t]}
	 \Big| 	\nu_{  t_3  }
	 \Big(  
	 \Big[   
	 \, D_k^* (t_1 ,\vp)
	 D_k ( t_2 ) \,  , \, 
	 \h_I( t_3 )		 \, 
	 \Big] 
	 \Big)  
	 \Big| \    . 
	 \label{starting point M delta 1}
\end{align}
Thus, it suffices to estimate the $\sup$
quantity in Eq.  \eqref{starting point M delta 1}.
For notational convenience 
	we do not write explicitly the time variables $t_i \in [0,t]$ for $i=1,2,3$--since our estimates are uniform in these variables, there is no risk in doing so. 
	In addition, we shall  only 
	give estimates for pure states $  \< \Psi , \  \cdot \  \Psi \> $
	and then apply Remark \ref{remark states} 
	to conclude estimates for the mixed state $\nu$. 
	Finally, we shall extensively use the results contained in Section \ref{toolbox1} --that is,  the estimates of Type-I, Type-II, Type-III and Type-IV, 
	contained in Lemma \ref{lemma type I}, \ref{lemma type II}, \ref{lemma type III} and \ref{lemma type IV}, respectively,
	together with the  several commutation relations. 
	
	\vspace{3mm}
	
Let us fix $k\in \supp \hat V$. 	We begin by expanding the commutator in \eqref{starting point M delta 1} as follows 
	\begin{align}
\label{TFB eq 3}
		\nu
		\big(
		[
		D_k^*(\vp)
		& D_k , \h_I 
		]			\nonumber 
		\big) \\
		& =
\lambda 
		\nu
		\big(
		[
		D_k^*(\vp)
		D_k , V_F 
		]
		\big) 
		+
\lambda 
		\nu
		\big(
		[
		D_k^*(\vp)
		D_k , V_{FB }
		]
		\big) +
\lambda 
		\nu
		\big(
		[
		D_k^*(\vp)
		D_k , V_{B}
		]
		\big)  \ . 
	\end{align}
	Let us estimate the three terms on the right hand side of Eq. \eqref{TFB eq 3}, separately. 
	We do this in the following items (I), (II) and (III). 
	
	\vspace{3mm}
	
	\noindent \textit{(I) the F term of \eqref{TFB eq 3}.}
	A straightforward expansion of $V_F$ based on the representation \eqref{VF(t)}
	 yields
	\begin{align}
\nonumber 
		[
		D_k^*(\vp)
		D_k , V_F 
		]
		& = 
		\int_{  \Lambda^* } 	\hat V  (\ell) \
		D_k^* (\vp )	 [D_k , D^*_\ell ]	D_\ell 
		\ \d \ell 	
		+ 
		\int_{  \Lambda^* } 	\hat V  (\ell) \
		D_k^* (\vp )	 D_\ell^* [D_k , D_\ell ] 
		\   \d \ell 	  \\
& 		+ 
		\int_{  \Lambda^* } 	\hat V  (\ell)  \
		D_\ell^* [D_k^* (\vp ) ,  D_\ell ]  D_k
		\ \d \ell 		
 + 
 \int_{  \Lambda^* } 	\hat V  (\ell)  \ 
		[D_k^* (\vp ) ,  D_\ell^*]  D_\ell  D_k 
		\  \d \ell 	 \ . 
	\end{align} 
	Each of the   four  terms {\blu on the right hand side} above is estimated
	in the same way. 
	Let us look in detail at the first one. 
	For   
	 $\Psi \in \F$ and $\ell \in \Lambda^*$, 
	  we  find using
	 the   Type-I estimate for $D_\ell$ and $[D_\ell ,  D_k]$,  
	 the  Type-IV estimate  for $D_k(\vp)$
	  	 and the commutation relation $   [\N , D(\vp)] =0 $
	\begin{align}
		\label{lemma reminder estimate 1}
		| \<	\Psi , D_k^* (\vp )	 [D_k , D^*_\ell ]	D_\ell  \Psi \> | 								
		\nonumber
		& =
		| \<  [D_\ell  , D_k  ]	 D_k  (\vp )	 	\Psi  ,   D_\ell  \Psi \> | 		\\
		\nonumber
		&  \leq 
		\|   \N 	D_k  (\vp )	 	\Psi  	\| 		\, \|	\N \Psi 	\| 		\\
		\nonumber
		&   = 
		\|   	D_k  (\vp )	 \N 	\Psi  	\| 		\, \|	\N \Psi 	\| 		\\
		\nonumber
		& \leq 
		\|	D_k(\vp)	\|  
		\|	 \N \Psi	\|^2 
		\\
		& \leq 
		|\Lambda |
		\|	 \vp 	\|_{\ell^1   }
		\|	 \N \Psi	\|^2 \  .
	\end{align}
	We conclude that there is a constant $C >0$ such that 
	\begin{equation}
\label{TFB F term}
		\nu 
		\big( 
		[
		D_k^*(\vp)
		D_k , V_F 
		]
		\big) \leq 
		C
		|\Lambda |
		\|	 \hat V	\|_{\ell^1 }
		\|	 \vp	\|_{\ell^1 }
		\nu(\N^2 ) \ . 
	\end{equation}
	
	\vspace{2mm}
	
	\noindent \textit{(II) the FB term of \eqref{TFB eq 3}.}
	The relation $   \overline{	\nu ( O )	} = \nu (O^*)$ 
	and a straightforward expansion 
	shows that 
	\begin{align}
\label{TFB eq 4}
		\nu
		\big(
		[
		D_k^*(\vp)
		D_k ,  & V_{FB }	 
		]
		\big) \\
\nonumber 
 & 		 =
		\int_{  \Lambda^{*  } }
		\hat V (\ell)
		\ 
		\nu
		\big(
		[
		D_k^*(\vp)
		D_k ,D^*_\ell 	\, b_\ell 
		]
		\big) 
		\ \d \ell 
		-
		\int_{  \Lambda^{*  } }
		\hat V (\ell)
		\ 
		\overline{
			\nu 
			\big(
			[
			D_k^* 
			D_k(\vp) ,
			D^*_\ell 	\, b_\ell 
			]
			\big) 
		} 
		\ \d \ell 	  . 
	\end{align}
We only estimate the first term in \eqref{TFB eq 4}, since the second one is analogous. 
Indeed, 	we expand the commutator to find that 
	\begin{align}
		\label{lemma remainder eq 1}
		\nu
		\big(
		[
		D_k^*(\vp)
		D_k ,D^*_\ell 	\, b_\ell 
		]
		\big) 
		& = 
		\nu
		\big(
		D_k^*(\vp)
		D^*_\ell 
		[
		D_k , 	\, b_\ell 
		]
		\big) 
		+
		\nu
		\big(
		D_k^*(\vp)
		[
		D_k ,  
		D^*_\ell 
		]
		\, b_\ell 
		\big) 				\\
		\label{lemma remainder eq 2}
		& \quad + 
		\nu
		\big(
		D^*_\ell 
		[
		D_k^*(\vp)
		, 	\, b_\ell 
		]
		D_k 
		\big) 
		+
		\nu
		\big(
		[
		D_k^*(\vp)
		,  
		D^*_\ell 
		]     
		\, b_\ell 
		D_k
		\big) 			 \ . 
	\end{align}
	We bound these four terms in the  following three items below. 

\begin{itemize}[leftmargin=*]
\item
	Since both $[D_k, b_\ell]$ and $b_\ell$ satisfy Type-II estimates, 
	the two 
	terms in \eqref{lemma remainder eq 1} are bounded above in the same way.
	Let us look at the first {\blu one} in detail. 
	Indeed,  for   $\Psi \in \F$ and $\ell \in \supp \hat V$
	we find 
	\begin{align}
		| 
		\<
		\Psi   , 
		D_k^*(\vp)
		D^*_\ell 
		[
		D_k , 	\, b_\ell 
		]
		\Psi 
		\>
		|
		&  \ = \ 
		| 
		\<
		D_\ell
		D_k(\vp)
		\Psi 
		, 
		[
		D_k , 	\, b_\ell 
		]
		\Psi 
		\>
		|			\nonumber \\ 
		&
		\ 	\leq  \ 
		\|	D_k(\vp)	\| 				 
		\		\|	    \N   \Psi 	\|	 \ 
		\nonumber
		\|  [D_k , b_\ell]	 \Psi 	\|  
		\\
		&  \ \lesssim  \ 
		|\Lambda |
		\|	 \vp	\|_{\ell^1 }
		\|	 \N \Psi	\|
		R^{\frac{1}{2}}
		\
		\| \po^{1/2} \Psi\|    
	\end{align}
	where we have used the Type-I estimate  for $D_\ell$, 
	the Type-II estimate for $[D_k , b_\ell]$, 
	the Type-IV estimate  for  $ 	 D_k(\vp)	$,
	and
	the commutation relation $[\N, D_k (\vp)]=0$. 
	
\item
		For	the first term in \eqref{lemma remainder eq 2}  we consider $\Psi \in \F$ and $\ell \in \supp \hat V$. We find
	\begin{align}
\nonumber 
		| \<
		\Psi , 
		D^*_\ell 
		[
		D_k^*(\vp)
		, 	\, b_\ell 
		]
		D_k 
		\Psi
		\> |  
 & \leq 
		\|
		[
		D_k^*(\vp)
		, 
		b_\ell 
		] 
		\|  
		\, 
		\|	 \N\Psi	\|^2 	  \\ 
& \lesssim 
		|\Lambda|
		p_F^{-m}
		\| 	\vp 		\|_{\ell_m^1 }
		\|	 \N\Psi	\|^2  \ . 
	\end{align} 
	where we have  used 
	the Type-I estimate for $ D_k$ and $D_\ell$,  
	and the Type-III estimate $   [D_k^* (\vp) , b_\ell] .$
	
\item
		For	the second term in \eqref{lemma remainder eq 2}  we consider $\Psi \in \F$ and $\ell \in \supp \hat V$. We find
	\begin{align}
\nonumber 
		| \langle 
		\Psi , 
		[
		D_k^*(\vp)
		,  
		D^*_\ell 
		]
		& \, b_\ell 
		D_k
		\Psi
		\rangle | \\ 
		& 
		\leq 
		| \<
		[  		D_\ell 		 ,		 D_k(\vp) 		]
		\Psi , 																\nonumber 
		[  b_\ell  ,  D_k] 
		\Psi
		\>|						
		+ 
		| \<
		D_k^* [  		D_\ell 		 ,		 D_k(\vp) 		]
		\Psi , 
		b_\ell  \Psi
		\>|	\\
		& 
		\lesssim  
		\|	 [D_\ell , D_k(\vp)] \| 
		\| 	\Psi 	\|		 
		\|	 [b_\ell, D_k ] \Psi	\|			\nonumber 
		+ 
		\| [D_\ell , D_k( \vp )]		\|
		\|	 \N \Psi 	\|
		\| b_\ell \Psi		\|		\\
		& \lesssim 
		|\Lambda |
		\|	 \vp 	\|_{\ell^1 }
		R^{\frac{1}{2}}
		\|	(  \N+1)  \Psi	\|
		\
		\| \po^{1/2} \Psi\|    \   , 
	\end{align}
	where, we have  used
	the Type-I estimate for $D_k^*$, 
	Type-II estimates for $[b_\ell, D_k]$  and $b_\ell$, 
	Type-IV estimates for 	$[D_\ell, D_k(\vp)]$. 
\end{itemize}

	We put back the three estimates found in the three items above
	to find that 
	there exists a constant 
	$C>0$ such that   
	\begin{align}
\label{TFB FB term}
		\nu
		\big(		
		[
		D_k^* (\vp ) D_k  ,   V_{FB}
		]
		\big)   \leq C 
		\|	 \hat V	\|_{\ell^1  }
		\|		 	\vp	\|_{\ell^1_m } 
		\, |\Lambda |
		\Big[
		R^{ \frac{1}{2}}
		\nu(\po)^{\frac{1}{2}}
		+ 
		p_F^{- m }
		\nu(\N^2 )^{ \frac{1}{2}	} 
		\Big]
		\nu(\N^2 )^{\frac{1}{2}}
		\   		 \  . 
	\end{align}

	\vspace{2mm}
	
	\noindent \textit{(III) the B term of \eqref{TFB eq 3}.}
	Similarly as we dealt with the second term,  we expand 
	\begin{align}
\label{TFB eq 5}
		\nu
		\big(
		[
		D_k^*(\vp)
		D_k , V_{B}
		]
		\big) 
		& =
		\int_{  \Lambda^{*  } }
		\hat V (\ell)
		\nu 
		\big(
		[ D_k^*(\vp) D_k 	,	b_\ell^* b_\ell 	]	 \d \ell 
		\big) 	\\ 
\label{TFB eq 6}
		& + 
		\frac{1}{2}
		\int_{  \Lambda^{*  } }
		\hat V (\ell)
		\nu 
		\big(
		[ D_k^*(\vp) D_k 	, b_{-\ell} b_\ell 	]	 
		\big) 
		\d \ell 
		\\
		& -
		\frac{1}{2}
		\int_{  \Lambda^{*  } }
		\hat V (\ell)
		\overline{ 
			\nu 
			\big(
			[ D_k^* D_k(\vp ) 	, b_{-\ell} b_\ell 	]	  
			\big) 
		} 
		\d \ell  \ . 
	\end{align}
	We only present a proof of the estimates for 
	the terms in \eqref{TFB eq 5} and \eqref{TFB eq 6}. 
	We do this in (III.1) and (III.2) below.
Since the third one is analogous to the second one,   we omit it. 
In order to ease the notation we shall   omit the indices $k,\ell \in \supp \hat V$. 
	
	\vspace{3mm}
	
	\begin{itemize}[leftmargin=*]
\item
		 \textit{Analysis of \eqref{TFB eq 5}}.  
		We expand the commutator to find that 
		\begin{align}
			\label{lemma FB third term eq 1}
			[D^*(\vp) D ,b^*b]
			=
			D^*(\vp)  b^* [D,  b]
			+
			D^*(\vp)  [D,  b^* ]b 
			+ 
			[D^*(\vp)  , b^*b]  D 
		\end{align}
		and estimate each term separately. 
	Let us fix  a   $\Psi \in \F $. 
	
	\vspace{1mm} 
	
\noindent 	$\blacklozenge$ The first term in \eqref{lemma FB third term eq 1} may be estimated as 
		\begin{align}
			| 	\langle	 \Psi , 		D^*(\vp)  &  b^* [D,  b] 		\Psi		 \rangle | 	 \\ 
			\nonumber 
			&
			 \  \leq  \ 
			\|    [ b ,D( \vp  ) ]		\| 
			\, 
			\|	 \Psi	\|	
			\, 
			\|   [D,b]	 \Psi	\|
			+ 
			\|	 D(\vp )	\| 	
			\, 	
			\|	  b\Psi	\|
			\, 
			\|	  [D,b] \Psi 	\|  \\ 
			&  
		 \ 	\lesssim  \ 
			|\Lambda| p_F^{-m }
			\|	  \vp	\|_{\ell^1_m }
			\|	 \Psi	\|	
		R^{\frac{1}{2}	}
			\|	 \po^{1/2} \Psi 	\|
			+ 
			|\Lambda |
			\|	 \vp	\|_{\ell^1 }
		R
			\|	 \po^{1/2}	 \Psi\|^{2 }		\nonumber \\ 
			& 
			 \  \leq  \ 
			\|   \vp	\|_{\ell^1_m }
			|\Lambda |
			\Big(
			p_F^{-m}
			\|	 \Psi	\| 	
			+  
		R^{\frac{1}{2}	}
			\|	 \po^{1/2} \Psi	\|
			\Big)  
		R^{\frac{1}{2}	}
			\|	 \po^{1/2} \Psi	\|
			\ . 	\nonumber 
		\end{align}
	Here, we have used Type-II estimates for $[D,b]$ and $b$, 
	Type-III estimates for  $[b,D(\vp)]$, 
	and Type-IV estimates for    $D(\vp)$. 
	
	\vspace{1mm }
	
\noindent 	$\blacklozenge$ The second term in \eqref{lemma FB third term eq 1}
		may be estimated  as 
		\begin{align}
			| 	\langle 	 \Psi , 	D^*( & \vp)    [D,  b^* ]b 	\Psi		 \rangle  | 	
			\nonumber  \\
			& 
			\  \leq 				 \ 
			\|	 D(\vp)	\|	
			\, 
			\|   [ b ,D^*  ] \Psi		\|	
			\,
			\|		 b\Psi\|
			+ 
			\|	 [	[ 	 b , D^* 	]	 , D(\vp)]	\|	
			\, 
			\|	 \Psi	\|
			\, 
			\|	 b\Psi	\|		\\ 
			\nonumber
			&  
			\ \lesssim   \ 
			|\Lambda|	\| \vp	\|_{\ell^1 }
		R
			\|	 \po \Psi	\|^2
			+
			|\Lambda| p_F^{-m }
			\|	   \vp	\|_{\ell^1_m}
			\|	 \Psi	\|
		 R^{\frac{1}{2}}
			\|	 \po^{\frac{1}{2}} \Psi 	\|		\\
			\nonumber
			&  \  \lesssim  \ 
			\|	 \vp	\|_{\ell^1_m }
			|\Lambda |
			\Big(
			p_F^{-m}
			\|	 \Psi	\| 	
			+  
		R^{\frac{1}{2}	}
			\|	 \po^{1/2} \Psi	\|
			\Big)  
		R^{\frac{1}{2}	}
			\|	 \po^{1/2} \Psi	\| \ . 
			\nonumber
		\end{align}
Here, we have used Type-II estimates for
$[b,D^*]$ and $b$, 
the Type-III estimate for 
$  [	[b, D^*] ,  D(\vp)	]$, 
and Type-IV
estimates for 
$D(\vp)$.

\noindent 	$\blacklozenge$ 	The third term in \eqref{lemma FB third term eq 1}  may be estimated as  
		\begin{align}
			\label{lemma FB third term eq 3}						\nonumber 
			| \<			\Psi ,		[D^* (\vp) , b^*b ]	 D 	\Psi 		\> |		 
			&
			 \  \leq  \ 
			\|	 [	D^*(\vp)	 , b^*b ]	\|	
				\, 	
					\|	 \Psi	\|	
						\, 	
							\|	D \Psi	\|		\\
			& 
			\  \lesssim  \ 
R   {|	\Lambda	|} p_F^{-m }					 
			\|	  \vp\|_{\ell^1_m}
			\,  \|			 \Psi	\| 	\, \|	 \N \Psi	\|  	 \ . 
		\end{align}
Here, we have used the Type-I
estimate for $D$, and
  Type-III estimates
and the operator norm bound for $b$ (see \eqref{bounds on b})
for 
$ 		\|  [D^*(\vp) , b^*b ] ] \|		 \leq   \|	b^* \|		 \,  \|  [D^*(\vp) , b]	 \|  +  \|		[D^*(\vp) , b^*]  \|		\, \|		b  \|	 . 	$

\vspace{4mm }

\item
 \textit{Analysis of \eqref{TFB eq 6}}.  
		Similarly as before, we expand  the commutator 
		\begin{align}
			\label{lemma FB third term eq 2}
			[	D^*(\vp) D  ,  b b 	]
			=
			D^*(\vp) 		b 	[ D , b]
			+ 
			D^*(\vp)
			[D,b] b 
			+ 
			[D^*(\vp ) , b  b ]		D   
		\end{align} 
		and estimate each term separately. 
		We let $\Psi \in \F$ . 
		
		\vspace{1mm}
\noindent 	$\blacklozenge$ The first term in \eqref{lemma FB third term eq 2} 		may be estimated as 
		\begin{align}
\nonumber 
			|			\<	\Psi , D^*(\vp)	  b [D, b ]		\Psi 		 \>		|	
			& 	 \ \leq  \ 
			\|	 D(\vp)	\|		  \| b \|		\|	 \Psi	\|			 
			\|	 [D,b] \Psi	\|	  \\
 & 		  \	\lesssim  \ 
			| \Lambda | 
			\|	 \vp	\|_{\ell^1 }
		R^{	\frac{3}{2}	}
			\|	 \Psi	\|		\|	 \po^{1/2 } \Psi	\|		\ . 
		\end{align}
Here, we have used the Type-II estimate for $[D,b]$, the Type-IV
estimate for $D(\vp)$, 
and the operator norm bound $\|	 b\| \lesssim R$ . 

\vspace{1mm}

\noindent 	$\blacklozenge$ The second term in 
		\eqref{lemma FB third term eq 2}
		may be estimated as 
		\begin{align}
\nonumber 
			|			\<	\Psi , D^*(\vp)	   [D, b ]	 b 	\Psi 		 \>		|	
			&  \ 	\leq  \ 
			\|	 D(\vp)	\|		 	
			\|	 [D,b]	\|	
			\|	 \Psi 	\|
			\|	b \Psi 	\| \\ 
		&  \ 	\lesssim  \ 
			|\Lambda |
			\|	 \vp	\|_{\ell^1 }
		R^{	\frac{3}{2}	}
			\|	 \Psi	\|		\|	 \po^{1/2 } \Psi	\|		\ . 
		\end{align}
Here, we have used the Type-II estimate for $b$, the Type-IV
estimate for $D(\vp)$, 
and the operator norm bound $\|	[D, b]\| \lesssim R$ . 

\vspace{1mm}
\noindent $\blacklozenge$ The third term in 
		\eqref{lemma FB third term eq 2}
		may be estimated as 
		\begin{align}
\nonumber 
			|			\<	\Psi ,    		D [D^*(\vp) , bb] 	\Psi 		 \>		|	
		& 
			 \  \leq  \ 
		\|	 [	D^*(\vp)	 , bb ]	\|	
		\, 	
		\|	 \Psi	\|	
		\, 	
		\|	D^* \Psi	\|		\\
		& 
		\  \lesssim  \ 
		R   {|	\Lambda	|} p_F^{-m }					 
		\|	  \vp\|_{\ell^1_m}
		\,  \|			 \Psi	\| 	\, \|	 \N \Psi	\|  	\  . 
	\end{align}
Here, we have used the Type-I
estimate for $D^*$, and
Type-III estimates
and the operator norm bound for $b$ (see \eqref{bounds on b})
for 
$ 		\|  [D^*(\vp) , bb ] ] \|		 \leq   \|	b \|		 \,  \|  [D^*(\vp) , b]	 \|  +  \|		[D^*(\vp) , b]  \|		\, \|		b  \|	 . 	$ 
	\end{itemize}
	Putting together the estimates found in the six   points above, 
	we find that 
	there exists a constant $C>0$
	such that  
	for all $ k \in \supp \hat V$
	\begin{equation}
\label{TF B term}
		\nu 
		\big( 
		[
		D_k^*(\vp)
		D_k , V_B
		]
		\big) 
		\leq 
		C
		|\Lambda |
				\|	 \vp	\|_{\ell^1_m }
		\|	 \hat V	\|_{\ell^1 }
\Big[
R^{\frac{3}{2}}
\nu(\N_S)^{\frac{1}{2}}
+
R
\nu(\N_\S)
+ 
\frac{R^{\frac{1}{2}}}{p_F^{m}}
\nu(\N_\S)^{\frac{1}{2}}
+
\frac{R}{p_F^{m}}  
\nu(\N^2)^{\frac{1}{2}}
 \Big] \ . 
	\end{equation}

	\vspace{3mm}

Finally, we can go back to the original decomposition found in \eqref{TFB eq 3}, plug it back in the starting point  \eqref{starting point M delta 1},
and use the estimates found in Eqs. \eqref{TFB F term}, \eqref{TFB FB term} and \eqref{TF B term} to find that
there exists a constant $C>0$
such that 
	\begin{align}
		| M_1^\delta  (t,\vp ) | 
		&  
		\leq C  		 |\Lambda | 
		\lambda t^3 
	 R 														\nonumber 
		\|	 \vp	\|_{\ell^1_m }
\|	 \hat V	\|_{\ell^2}^2 	\| 
	\|   \hat V\|_{\ell^1}	
		 \\
		&  \quad \times 
		\sup_{0 \leq \tau \leq t 	}
		\Big(	
		R^{\frac{3}{2}	}
		\nu_\tau (\po)^{  \frac{1}{2}}
		+ 
		R \nu_\tau (\po )
		+ 
		\frac{R^{ \frac{1}{2}}}{p_F^m}
		\nu_\tau (\po)^{  \frac{1}{2} }
		+ 
		\frac{R}{p_F^m}
		\nu_\tau (\N )^{ \frac{1}{2} }
		\Big)  \ .
	\end{align}
To conclude, we note that $\nu(\N_\S) \leq R^{1/2} \nu(\N )$ so that the third term {\blu on the right hand side} above can be absorbed into the fourth one. 
This finishes the proof. 
\end{proof}

Let us estimate the second term of the right hand side in \eqref{T FB 1 }.

\begin{lemma}
	\label{claim 3}  
For all $m>0$ there exists a constant $C>0$ such that 	
for all $ t  \geq 0 $ and $\vp \in \ell^1 $ 
 	the following estimate holds true 
 	\begin{align}
\label{MR estimate}
 		| 	 M^\calR  (t, \vp ) | 
 		 \,  	\leq   \, 
 		 C 
 		t^2 
 		\|		 \hat V	\|_{\ell^1 }^2 
 		|\Lambda|
 		\|	 \vp 	\|_{\ell^1_m} 
 		\sup_{ \tau \leq t }
 		\Big(
 		R^{\frac{1}{2}}
 		\nu_\tau (\po)^{ 	\frac{1}{2}	 }
 		+ 
 		p_F^{ -m  }
 		\Big) 					
 		\nu_\tau (\N^2)^{ \frac{1}{2} }  
 		\ . 
 	\end{align} 
 \end{lemma}

  \begin{proof} 
Let us fix $m>0$,  $t   \geq 0  $ and $\vp\in \ell^1$.
Going back to the definition of the main term in $\eqref{TFB main term}$,
we plug in   the remainder operator 
 $\calR_{ k , \ell  }$ 
 defined in \eqref{cal R k ell}, 
from which the elementary inequality follows 
\begin{align}
 | 		M^\calR   (t,\vp) | 
	& 
 \ 	\lesssim  \ 
t^2 
\|	 \hat V	\|_{\ell^1}^2
	\sup_{  k,\ell \in \supp \hat V, t_ i \in [0,t]}
\Big|
\nu_{t_2 }
\Big(
D_{t_1}^*(k,\vp)
\calR_{k ,\ell }(t_1, t_2)
D_{t_2}(\ell )
\Big)  
\Big|	 \ . 
\end{align}
 	Let us estimate the supremum quantity  in the above equation. 
 	Since our estimates are uniform    in $t_1,t_2$ we shall omit them. 
 	{\blu Letting} $\Psi \in \F$, we find that 
 	\begin{align}
 		| \<	
 		\Psi			,
 		D_{k}^*( \vp)
 		\calR_{k ,\ell } 
 		D_\ell 
 		\Psi 		\> | 
 		& 
 		\   \leq  \ 
 		| \<	
 		D_{k}^*( \vp)
 		\calR^*_{k ,\ell } 
 		\Psi			,
 		D_\ell 
 		\nonumber 
 		\Psi 		\> |  
 		+ 
 		| \<	
 		\Psi			,
 		[ 
 		D_{k}^*( \vp)
 		, 
 		\calR_{k ,\ell } 
 		] 
 		D_\ell 
 		\Psi 		\> | 		\\
 		& 
 		\  \leq  \ 
 		\|	D_k^*(\vp ) 	\|
 		\|	 \calR_{ k , \ell  } \Psi	\| 
 		\|	 D_\ell \Psi	\|
 		+ 
 		\|	  \Psi 	\|
 		\|	  [	  D_k^*(\vp) , \calR_{k , \ell }	]		\|
 		\|	 D_\ell  \Psi	\|  \ . 
 	\end{align}
Letting $k,\ell \in \supp \hat V, $	we find the following estimates for the   quantities containing $\calR_{k,\ell }$
 	\begin{align}
 		\|		
 		\calR_{ k , \ell  }
 		\Psi
 		\| 
 		\lesssim 
 	R^{\frac{1}{2}}
 		\|	
 		\po^{1/2 }
 		\Psi 
 		\|	 
 		\quad
 		\t{and}
 		\quad 
 		\|		
 		[  D_k^* (\vp ) ,    	\calR_{ k , \ell  }  ] 
 		\| 
 	\lesssim 
 	|\Lambda|   p_F^{-m}
 		\|		  \vp	\|_{\ell^1_m }	    \  . 
 	\end{align}
The proof of these estimates follows the same lines of the proof  of Lemma \ref{lemma type II}
and \ref{lemma type III}, so we shall omit it. 
We combine the last three displayed equations together with Remark \ref{remark states} to conclude 
the proof of the estimate contained in Eq. \eqref{MR estimate}. 
 \end{proof}

Finally, we turn to the proof of the last lemma.
 
 \begin{proof}[Proof of Lemma \ref{lemma TFB main term}]
The triangle {\blu inequality}  and the decomposition $M = M^\delta_0 + M^\delta_1 + M^\calR$
gives $|M - M_0^\delta | \leq |M^\delta_1| + |M^\calR|$. 
It suffices  then to use  the results  contained in Lemma \ref{claim 1}, \ref{claim 2} and \ref{claim 3}. 
 \end{proof}

 \subsection{Analysis of   the remainder terms}
 In this subsection, we estimate the remainder terms $R^{(i)}$
  (see \eqref{TFB reminder terms})
 and give a proof of Lemma \ref{lemma TFB reminder terms}.

 \begin{proof}[Proof of Lemma \ref{lemma TFB reminder terms}]
Throughout the proof, we fix $m>0$, $t  \geq 0 $ and $\vp \in \ell_m^1$.
We make extensive use of the Type-I, Type-II, Type-III and Type-IV estimates contained in 
Lemmas \ref{lemma type I}, 
\ref{lemma type II},
\ref{lemma type III},
and
\ref{lemma type IV}, respectively,
together with the operator bound {\blu $\|	b_k(t) \| \leq R $, see \eqref{bounds on b}.} 
Due to the similarities, we only show all the details for the proof of (1), and 
only give the key estimates for the proofs of (2), (3), and (4). 

\vspace{1.5mm}

 	\noindent \textit{Proof of (1)}
Our starting point is the elementary estimate
\begin{align}
  | 	R^{(1)} (t,\vp)  | 
& 
 \,    \lesssim  		\, 
t^2
\|		 \hat V	\|_{\ell^1}^2 
\sup_{  k,\ell \in \supp \hat V, t_ i \in [0,t]}
\Big|
\nu_{t_2 }
\Big(
D_{ k }^*(t_1,\vp)
B^*_\ell (t_2)
[b_{ k}(t_1),  D_{\ell}(t_2)]
\Big)   
\Big|  
	\ . 
\end{align}
In view of Remark \ref{remark states},
it is sufficient to give estimates for pure states $\Psi \in \F$.  
In order to ease the notation, we shall drop the time variables $t_1,t_2 \in [0,t]$, 
together with the momentum labels $ k ,\ell \in \supp \hat V$. 

 	 	Letting $\Psi \in \F$, we find  that 
 	\begin{align}
 		|
 		\<
 		\Psi ,
 		D^* (\vp ) B^* [b,D ]
 		\Psi 
 		\>			\nonumber 
 		|
 		& 	\ \leq  \ 
 		\|	 D(\vp)	\|	\|	 \Psi	\|
 		\|	 B^*	\|
 		\| [ b,D] \Psi 		\|			\\
 		& \  \lesssim  \ 
 		|	\Lambda | 		\|	 \vp	\|_{\ell^1 }
 		\|	 \Psi 	\| 
 	R^{\frac{3}{2}}
 		\|	 \po^{1/2}	 \Psi \| , 
 	\end{align}
where we used the 
Type-II estimate for $[b,D]$, 
the Type-IV estimate for $D^*(\vp)$,
and the norm bound $ \|	B	\| \leq 2 \|	b	\| \lesssim R$. 
The estimate in Eq. \eqref{lemma remainder 1} now follows from the last two displayed equations,  and $\nu(\1) = 1$. 

\vspace{1.5mm}
 	
 	 	\noindent \textit{Proof of (2)}
 	 	 	Letting $\Psi \in \F$, we find that 
 	 	\begin{align}
 	 		|
 	 		\<
 	 		\Psi ,
 	 		[ D^*(\vp)  , B^* ] D b 
 	 		\Psi 
 	 		\>			\nonumber 
 	 		|
 	 		&  \ 	\leq  \ 
 	 		\|	  [ D^*(\vp)  , b  ] 	\|	\|	 \Psi	\|
 	 		\|  	 D b 		 \Psi 		\|			\\
 	 		& \  \lesssim  \ 
 	 		\|	  [ D^*(\vp)  , b  ] 	\|	\|	 \Psi	\|				\nonumber 
 	 		\|  	 \N  b 		 \Psi 		\|			\\
 	 		& \  \lesssim  \ 
 	 		|\Lambda| p_F^{ -m  }\|   \vp		\|_{\ell^1_m }
R
 	 		\|	 \Psi 	\|\|	\N \Psi 	\|  \ , 
 	 	\end{align}
where we have used the 
Type-II estimate, commutation relations and the norm bound for $b$
to obtain 
 	 	$\|  	 D   b 		 \Psi 		\|   \leq \| \N b \Psi 	\|		\leq \|	 b \N \Psi 	\|	  \lesssim R \|	 \N \Psi \|$ ; 
 	 	and the Type-III estimate for $[D^*(\vp), b ]$. 
 	 	The proof is   finished after one follows  the same argument we used for (1). 
\vspace{1.5mm} 	 	

 	 	 	\noindent \textit{Proof of (3)}
 	 	 	 	Letting $\Psi \in \F$, we find  that 
 	 	 	\begin{align}
 	 	 		|	\< 
 	 	 		\Psi  , 
 	 	 		B^* [D(\vp) , D ] b 
 	 	 		\Psi 
 	 	 		\> | 
 	 	 		& 	\leq 
 	 	 		\|	 B^*	\|
 	 	 		\|	 \Psi 	\|
 	 	 		\|	 [D(\vp ), D ]	\|
 	 	 		\|		 b \Psi	\|			\\
 	 	 		& \lesssim 
R^{\frac{3}{2}}
 	 	 		\|	 \Psi 	\|
 	 	 		|\Lambda|		\|		 \vp\|_{\ell^1 } 
 	 	 		\|		\po^{\frac{1}{2}} \Psi 	\| \ , 
 	 	 	\end{align}
where we used the Type-II estimate for $b$, 
the Type-IV estimate for $[D(\vp) , D ]$,
and the norm bound $\|		B\|\lesssim R$. 
 	 	The proof is   finished after one follows  the same argument we used for (1). 
 	 	 	
 	 	 	\vspace{1.5mm}
 	 	 	 	\noindent \textit{Proof of (4)} 	Letting $\Psi \in \F$, we find that 
 	 	 	 	\begin{align}
 	 	 	 		|
 	 	 	 		\<
 	 	 	 		\Psi ,
 	 	 	 		[D b(\vp) , B^* D]
 	 	 	 		\Psi 
 	 	 	 		\>			\nonumber 
 	 	 	 		|
 	 	 	 		& 	\leq 
 	 	 	 		2 
 	 	 	 		|		\<  \Psi ,  D b(\vp)	 B^* D \Psi 			 \>	|		\\
 	 	 	 		& \lesssim 
 	 	 	 		\|	 D^* \Psi	\|
 	 	 	 		\|	 b(\vp)	\|
 	 	 	 		\|	 B^*	\|							\nonumber 
 	 	 	 		\| D \Psi		\|		\\
 	 	 	 		& \lesssim 
 	 	 	 		p_F^{-m } 
 	 	 	 		|\Lambda|
 	 	 	 		\|	   \vp	\|_{\ell^1_m }
 	 	 	 	R
 	 	 	 		\|	 \N \Psi	\|^2  \   .  
 	 	 	 	\end{align}
where we used the Type-I estimate for $D$ and $D^*$, 
the Type-III estimate for $b(\vp)$,
and the norm bound $\|		B\|\lesssim R$. 
The proof is   finished after one follows  the same argument we used for (1). 
 \end{proof}

 \section{Subleading Order Terms}
 \label{section subleading}
 In this section we analyze the $T_{\alpha,\beta}(t,p)$
 terms of the double commutator expansion \eqref{expansion f}
 that we regard as subleading order terms. 
 So far, out of the nine terms 
 we have analyzed two leading order terms: 
 $T_{F,F}$ in Section \ref{section TFF}
 and $T_{FB,FB}$ in Section \ref{section TFB}. 
 Thus, we shall analyze the {\blu remaining} seven.
 We do this in the following five subsections.

 \subsection{Analysis of $T_{F,FB}$}

The main result of this subsection is the following proposition, 
which gives an estimate on the size of $T_{F,FB}$. 

\begin{proposition}
	[Analysis of $T_{F,FB}$]
\label{prop T F FB}
Let $T_{F,FB}(t,p)$
be the quantity defined in \eqref{T alpha beta}
with $\alpha = F$ and $\beta = FB$, and
let $m>0$.
Then, there exists a constant $C>0$
such that  for all $\vp \in\ell_m^1$ and $ t \geq 0$  
the following estimate holds true 
	\begin{equation}
|
 	T_{F,FB}
	(t, \vp ) 
	| 
 \ 	\leq   \ 
 C 
 t^2  
 \|	 \hat V	\|_{\ell^1}^2 
 | \Lambda|
 \|	 \vp	\|_{\ell^1_m}
\sup_{ 0 \leq \tau \leq t }
 \Big(
 R^{\frac{1}{2 }	} \, 
 \nu_\tau 	( \N^2)^{\frac{1}{2}}
 \nu_\tau  (\po)^{\frac{1}{2}}
 + 
p_F^{-m}
\nu_\tau (\N^2 )
 \Big) 
	\end{equation}
where we recall 
$ 	T_{F,FB} (t, \vp )  =  \< \vp , T_{F,FB} (t)\>$
and
$R = |\Lambda| p_F^{d-1}$. 
	\end{proposition}

\begin{proof} 
For simplicity, we assume $\vp$ is real-valued--in the general case, 
one may expand into real and imaginary parts and use linearity of the commutators. 
Starting from \eqref{T alpha beta phi} we use the self-adjointness of $V_F (t) $ and $N(\vp) = \int_{  \Lambda^{*} } \vp(p) a_p^*a_p \d p $ to get 
the elementary inequality 
\begin{align}
\nonumber 
| 	 & T_{  F , FB } 
 (t   , \vp )
| 
 \\
\nonumber 
 &   \quad  = 
\Big|
\int_0^t 
\int_0^{t_1}
\int_{  \Lambda^{*2 } }
\hat{V}(k) 
\hat V (\ell)
2 \mathrm{Re}
\nu_{t_2 }
\Big( 
 [ [   N(\vp) ,  D_k^* (t_1)  D_k  (t_1) ] ,  D_\ell^*  (t_2) b_\ell (t_2) ]
\Big)  
\d t_1 \d t_2  \d k \d \ell 
\Big|   \\
&  \quad \lesssim 
t^2 \|	 \hat V	\|_{\ell_1}^2
\sup_{  k,\ell \in \supp \hat V, t_ i \in [0,t]}
\Big|
\nu_{t_2}
\Big( 
[ [   N(\vp) ,  D_k^* (t_1)  D_k  (t_1) ] ,  D_\ell^*  (t_2) b_\ell (t_2) ]
\Big)  
\Big| \ . 
\label{starting point T F FB}
%
%
\end{align}
It suffices now to estimate the supremum quantity in the above equation. 
In order to ease the notation, we shall drop the time labels $t_1,t_2\in[0,t]$,
together with the momentum variables $k,\ell \in \supp \hat V$. 
Using the notation $D^* (\vp) \equiv  [N (\vp), D^* ]$ we compute the commutator 
\begin{equation}
\label{F,FB term eq 1}
[ N (\vp)  , D^* D]
=
D^* (\vp) D + D^* D(\vp) \ . 
\end{equation}
We shall only show how to estimate the contribution that arises from 
the  first term  on the right hand side of \eqref{F,FB term eq 1}; the second one is analogous. 
To this end, we expand 
\begin{align}
\label{T F FB eq 1}
[	D^* (\vp)  & D , D^* b	]     \\
 & =
D^* (\vp) [D, D^*] b 
+ 
D^*(\vp ) D^* [D, b] 
+ 
[D^*(\vp),  D^*] b D
+  		
D^* [D^*(\vp) , b] D 
\nonumber
 \ . 
\end{align}
Next, we estimate the expectation of  each term in \eqref{T F FB eq 1} separately.  
In view of Remark \ref{remark states}, it suffices to  provide estimates for pure states $\Psi \in \F$. 
We shall make extensive use of Type-I to Type-IV estimates contained in Lemma \ref{lemma type I}--\ref{lemma type IV}, 
the commutation relations from Lemmas \ref{lemma N commutators} and \ref{lemma NS commutators}, 
and operator bounds of the form $\|	 b	\|_{B(\F)}\lesssim R$.

	\begin{itemize}[leftmargin=*]
	\item  
	\textit{The first term in \eqref{T F FB eq 1}.  }
Using the Type-I estimate for $[D^*,D]$, 
the Type-II estimate for $b$, 
the Type-IV estimate for $D(\vp)$
and the commutation relation 
$[\N , D(\vp)]=0$ 
we find 
\begin{align}
|	 \<   \Psi , 	
D^* (\vp) [D, D^*] b\nonumber
\Psi   \>	|
&  \ 			 	\leq  \ 
\|   [D^* ,D]		 D(\vp) \Psi \|	 
\, 
\| b\Psi		\|		\\
& 	\ \lesssim  \ 	
\nonumber
\|	 \N  D(\vp) \Psi	\|	
R^{\frac{1}{2}}
\|	 \po^{\frac{1}{2}}	 \Psi \|		\\ 
& \ \lesssim \ 
| \Lambda | \|	 \vp\|_{\ell^1 }  \| \N\Psi 		\|
R^{\frac{1}{2}} \|	 \po^{\frac{1}{2}}	 \Psi \|	 \ . 
\end{align}

	\item 
\textit{The second term in \eqref{T F FB eq 1}. }
Using the Type-I estimate for $D$, 
the Type-II estimate for $[D,b]$, 
the Type-IV estimate for $D(\vp)$
and the commutation relation 
$[\N , D(\vp)]=0$ 
we find 
\begin{align}
|	 \<   \Psi , 	
D^*(\vp ) D^* [D, b ] 
\nonumber
\Psi   \>	|
&  \ 			 	\leq  \ 
\|    D 	 D(\vp) \Psi \|	 
\, 
\| [ D,  b  ] \Psi		\|		\\
& 	\ \lesssim  \ 	
\nonumber
\|	 \N  D(\vp) \Psi	\|	
R^{\frac{1}{2}}
\|	 \po^{\frac{1}{2}}	 \Psi \|		\\ 
& \ \lesssim \ 
| \Lambda | \|	 \vp\|_{\ell^1 }  \| \N\Psi 		\|
R^{\frac{1}{2}} \|	 \po^{\frac{1}{2}}	 \Psi \|	 \ . 
\end{align}

	\item  
\textit{The third term in \eqref{T F FB eq 1}.   	}
	Using the Type-I estimate for $D^*$, 
	the Type-II estimate for both $b$ and  $[D,b]$, 
	the Type-IV estimate for $[ D, D(\vp)] $
	and the commutation relation 
	$[\N , [D,D(\vp)] ]=0$ 
	we find 
	\begin{align}
		|	 \<   \Psi , 	
		[D^*(\vp)  ,  D^*] b D
		\nonumber
		\Psi   \>	|		
		& \ \leq  \ 
		|	 \<  [D, D(\vp)] \Psi , 	
		[ b  , D ] 
		\nonumber
		\Psi   \>	|
		+ 
		|	 \<   [ D, D(\vp )]  \Psi , 					\nonumber 
		D  b
		\Psi   \>	|	\\ 
		& \ \leq \ 
		\|	 [ D , D(\vp)]  \Psi  	\| \, \|  [b,D] \Psi	\|
		+ 
		\nonumber
		\|	 D^* [ D ,D(\vp)] \Psi 	\| 
		\, \|	 b \Psi	\|			\\
		&  \ \lesssim  \ 
		| \Lambda | \|	 \vp\|_{\ell^1 }  \| ( \N + 1 ) \Psi 		\|
		R^{\frac{1}{2}} \|	 \po^{\frac{1}{2}}	 \Psi \|	 \ . 
		%
		%
		%
		%
		%
		%
		%
	\end{align}
	
	\item 
  \textit{The fourth term in \eqref{T F FB eq 1}.}   
	Using the Type-I estimate for $D$ and 
	the Type-III estimate for $[D^*(\vp),b]$
	we find 
	\begin{align}
		|
		\< 	\Psi ,
		D^*  [D^*(\vp) , b] D 
		\Psi 
		\>
		|
		\,    \leq  \, 
		\| 	 [D^*(\vp) , b] 	\|
		\|	 D	\Psi\|^2 	  
		\,   \lesssim  \, 
		|	\Lambda 	|  p_F^{-m}
		\|	  \vp	\|_{\ell^1_m }
		\|	 \N \Psi	\|^2  \ . 
	\end{align} 
\end{itemize}

The proof now follows by collecting the previous four estimates in the expansion \eqref{T F FB eq 1},
and plugging them back in \eqref{starting point T F FB}. 
 \end{proof}

  \subsection{Analysis of $T_{F,B}$}
 The main result of this subsection is the following proposition, 
 that gives an estimate on the size of $T_{F,B}$. 
 
 \begin{proposition}
 	[Analysis of $T_{F,B}$]
 	\label{prop T F B}
 	Let $T_{F,B}(t,p)$
 	be the quantity defined in \eqref{T alpha beta}
 	with $\alpha = F$ and $\beta = B$, and
 	let $m>0$.
 	Then, there exists a constant $C>0$
 	such that  for all $\vp \in\ell_m^1$ and $ t \geq 0$  
 	the following estimate holds true 
 	\begin{equation}
 		|
 		T_{F,B}
 		(t, \vp ) 
 		| 
 		\ 	\leq   \ 
 		C 
 		t^2 \|	 \hat V	\|_{\ell^1 }^2 
 	|\Lambda|
 	\|	 \vp	\|_{\ell^1_m}
 	\sup_{ 0 \leq \tau \leq t }
 	\Big(
 	R^\frac{3}{2} \nu_\tau (\po)^{\frac{1}{2}}
 	+ 
 	R p_F^{-m}
 	\nu (\N^2)^{	\frac{1}{2}	}
 	+ 
 	R 
 	\nu_\tau (\po)
 	\Big) 
 	\end{equation}
 	where we recall 
 	$ 	T_{F,B} (t, \vp )  =  \< \vp , T_{F,B} (t)\>$
 	and
 	$R = |\Lambda| p_F^{d-1}$. 
 \end{proposition}

 \begin{proof} 
For simplicity, we assume $\vp$ is real-valued--in the general case, 
one may expand into real and imaginary parts and use linearity of the commutators. 
Starting from \eqref{T alpha beta phi} we use the self-adjointness of $V_F (t) $, $V_B(t)$ and $N(\vp) = \int_{  \Lambda^{*} } \vp(p) a_p^*a_p \d p $ to get 
the elementary inequality 
{\blu
thanks to
  \eqref{eq:N}}
\begin{align}
\label{T F B eq 1}
	| 	  T_{  F , B } 
 	(t   , \vp )
	| 
	&       = 
	\Big|
	\int_0^t 
	\int_0^{t_1}
	 \mathrm{Re} 
	 \, 
	\nu_{t_2 }
	\Big( 
	[ [   N(\vp) , V_F(t_1) ] ,  V_B(t_2)]
	\Big)  
	\d t_1 \d t_2  
	\Big|   \\
\nonumber 
	& 
	  \lesssim 
	t^2 \|	 \hat V	\|_{\ell_1}^2
	\sup_{  k,\ell \in \supp \hat V, t_ i \in [0,t]}
	\Big|
	\nu_{t_2}
	\Big( 
	[ [   N(\vp) ,  D_k^* (t_1)  D_k  (t_1) ] ,   b_\ell^*(t_2) b_\ell (t_2) ]
	\Big)  
	\Big| \\
\nonumber 
&  \,  	+  
	t^2 \|	 \hat V	\|_{\ell_1}^2
	\sup_{  k,\ell \in \supp \hat V, t_ i \in [0,t]}
	\Big|
	\nu_{t_2}
	\Big( 
	[ [   N(\vp) ,  D_k^* (t_1)  D_k  (t_1) ] ,   b_\ell(t_2) b_{-\ell} (t_2) ]
	\Big)  
	\Big| \  , 
	%
	%
\end{align}
where in the last line we used the representation 
of $V_F(t)$ and $V_B(t)$
in terms of $b$- and $D$-operators found in Eqs. \eqref{VF(t)} and \eqref{VB(t)} --the $b^*b^*$ term
is re-written in terms of $bb$ upon taking the  real part of $\nu$. 

\vspace{1mm}

We now estimate  the two supremum quantities in  \eqref{T F B eq 1}, which we shall refer to 
as     an \textit{off-diagonal contribution}, 
 and a \textit{diagonal contribution}, 
 with respect to the operators $b$ and $b^*$. 
In view of Remark \ref{remark states}, it suffices to  provide estimates for pure states $\Psi \in \F$. 
Further, in  order to ease the notation, we omit the time labels $t_1,t_2 \in [0,t]$
and the momentum variables $ k , \ell \in \supp \hat V$.
We   make extensive use of Type-I to Type-IV estimates contained in Lemma \ref{lemma type I}--\ref{lemma type IV}, 
the commutation relations from Lemmas \ref{lemma N commutators} and \ref{lemma NS commutators}, 
and operator bounds of the form $\|	 b	\|_{B(\F)}\lesssim R$.

	\begin{itemize}[leftmargin=*]
  	\item  
\textit{The off-diagonal contribution of \eqref{T F B eq 1}.}
  	We expand the first commutator as follows 
  	\begin{align}
  		[	[ N  (\vp) , D^*D] , b b]
  		= 
  		[D^*(\vp) D , bb] + [D^* D(\vp) , bb]
  		\label{T F B eq 2} \  , 
  	\end{align}
  	where we recall  we use the notation $D^*(\vp)  =  [N(\vp) , D]$. 
  	We shall only show in detail how to estimate the first term in \eqref{T F B eq 2}--the second  term
  	can be estimated in the same spirit. 
  	We   expand the second commutator as follows 
  	\begin{align}
  		\label{T F B eq 3}
  		[  D^*(\vp) D   ,  b b   ] 
  		= 
  		D^*(\vp) b [D,b]
  		+
  		D^*(\vp) [D,b] b 
  		+
  		[D^*(\vp) ,bb] D  \ . 
  	\end{align}
  	We now estimate the three terms  {\blu on  the right hand side} of \eqref{T F B eq 3}.
  	
  	\begin{itemize}[leftmargin=*]
  		\item[$\blacklozenge$]\textit{The first term of \eqref{T F B eq 3}. }
  		Letting $\Psi \in \F$, 	we find that
  		\begin{align}
  			\nonumber 
  			|
  			\< \Psi  , 
  			D^*(\vp) b [D,b]
  			\Psi 	\> 
  			&  
  			\ 	  	\leq  \ 
  			\|	 D(\vp)	\|		\|	 \Psi	\| \|	b	\|		\| [D,b]	 \Psi	\|	   \\
  			&  \ 		\lesssim  \ 
  			|\Lambda| 	\|	 \vp	\|_{\ell^1 }
  			\|	 \Psi	\|
  			R^\frac{3}{2}
  			\|  \po^{\frac{1}{2}}	 \Psi	\|	  \  , 
  		\end{align}
  		where we used 
  		the Type-II estimate for $[D, b]$, 
  		the Type-IV estimate for $D(\vp)$,
  		and the norm bound $ \|	b	\|      \lesssim R$. 
  		\item[$\blacklozenge$]\textit{The second term of \eqref{T F B eq 3}. }
  		Letting $\Psi \in \F$, 	we find that
  		\begin{align}
  			\nonumber 
  			|
  			\< \Psi  , 
  			D^*(\vp) [D,b] b 
  			\Psi 	\> 
  			|
  			& \ 	\leq  \ 
  			\|	 D(\vp)	\|		\|	 \Psi	\| \|	[D,b]	\|		\|   b 	 \Psi	\|	   \\ 
  			& \ 		\lesssim  \ 
  			|\Lambda| 	\|	 \vp	\|_{\ell^1 }
  			\|	 \Psi	\|
  			R^\frac{3}{2}
  			\|  \po^{\frac{1}{2}}	 \Psi	\|	  \  , 
  		\end{align}
  		where we used 
  		the Type-II estimate for $b$, 
  		the Type-IV estimate for $D^*(\vp)$,
  		and the norm bound $ \|	[D,b]	\|   \lesssim R$,

  		\item[$\blacklozenge$]\textit{The third term of \eqref{T F B eq 3}. }
  		Letting $\Psi \in \F$, 	we find that
  		\begin{align}
  			\nonumber
  			\label{FB eq 1} 
  			|
  			\< \Psi  , 
  			[D^*(\vp) ,bb] D 
  			\Psi 	\> 
  			|
  			& \ 	\leq   \ 
  			\|	 [D^*(\vp) ,bb]	\|
  			\|	 \Psi	\|
  			\|	 \N \Psi	\| \\ 
  			&  \  		\lesssim  \ 
  			| \Lambda|
  			\|   \vp	\|_{\ell^1_m}
  			p_F^{-m}
  			R 
  			\|	 \Psi	\|
  			\|	 \N \Psi	\| \  , 
  		\end{align}
  		where we used 
  		the Type-I estimate for $D$, 
  		the Type-III estimate for $ [ D^*(\vp),b]$
  		and the norm bound $   \|	b	\| \lesssim R$. 
  		
  	\end{itemize}
  	We collect the four estimates found above and put them back in \eqref{T F B eq 2} 
  	to find that the off-diagonal contribution satisfies the following upper bound 
  	\begin{align}
  		\Big|
  		\nu  
  		\Big( 
  		[ [   N(\vp) ,  D^*   D   ] ,   b    b   ]
  		\Big)  
  		\Big|   
  		\ 	\lesssim  \ 
  		|\Lambda	|
  		\|	 \vp  \|_{\ell_m^1}
  		\Big(
  		R^{\frac{3}{2}}
  		\nu (\N_\S)^{\frac{1}{2}}
  		+ 
  		\frac{R}{p_F^m}
  		\nu (\N^2)^{\frac{1}{2}}
  		\Big) \ . 
  	\end{align}

  	\item 
  \textit{The diagonal contribution of \eqref{T F B eq 1}.}
  	Similarly as before, we shall  expand the commutator as follows. 
  	\begin{align}
  		[  D^*(\vp) D   ,  b^* b   ] 
  		= 
  		D^*(\vp) b^* [D,b]
  		+
  		D^*(\vp) [D,b^*] b 
  		+
  		[D^*(\vp) ,b^*b] D  \ . 
  		\label{T F B eq 6}
  	\end{align}
  	These  three terms are {\blu estimated} as follows.
  	\begin{enumerate}[leftmargin=*]
  		\item[$\blacklozenge$] \textit{The first term of \eqref{T F B eq 6}}. 
  		Letting $\Psi \in \F$, 	we find that
  		\begin{align}
  			\nonumber 
  			|
  			\langle \Psi  , 
  			D^*(\vp) b^* & [D,b]
  			\Psi 	\rangle  |  \\ 
  			\nonumber 
  			&  
  			\leq 
  			|
  			\<     
  			D(\vp)  b \Psi  , 
  			[D,b]
  			\Psi 	\>  | 
  			+
  			|
  			\<  
  			[   D(\vp)  ,  b  ] 
  			\Psi  , 
  			[D,b]
  			\Psi 	\>  | 	\\
  			&
  			\lesssim   
  			\nonumber 
  			\|	 D(\vp)	\|		\|	 b \Psi  	\|
  			\|   [D,  b ]	 \Psi  	\|
  			+ 
  			\|	 [D(\vp) , b ]	\|	\|	 \Psi	\|	
  			\|	 [D,b] 	\|		\| \Psi 	\|	
  			\\
  			&  \lesssim  
  			|\Lambda|		\|	 \vp	\|_{\ell_m^1 }
  			\Big(
  			R 
  			\|	 \po^{\frac{1}{2}} \Psi	\|^2 
  			+ 
  			p_F^{-m}R 
  			\|	 \Psi	\|^2
  			\Big)  \ . 
  		\end{align}
  		where we used 
  		the Type-II estimate for $b$ and  $[D, b]$, 
  		the Type-III estimate for $[D(\vp),b]$,  
  		the Type-IV estimate for $D(\vp)$. 
  		\item[$\blacklozenge$] \textit{The second term of \eqref{T F B eq 6}}. 
  		Letting $\Psi \in \F$, 	we find that
  		\begin{align}
  			\nonumber 
  			|
  			\langle \Psi  , 
  			D^*(\vp) [D,  & b^*] b 
  			\Psi 	\rangle 
  			|   \\
  			\nonumber 
  			&      \leq  
  			|
  			\<     
  			D(\vp) [D^*, b ] \Psi  , 
  			b
  			\Psi 	\>  | 
  			+
  			|
  			\<  
  			[   D(\vp)  ,  [D^*, b ]   ] 
  			\Psi  , 
  			b 
  			\Psi 	\>  | 	\\
  			&
  			\lesssim   
  			\nonumber 
  			\|	 D(\vp)	\|		\|	 [D^*,b]\Psi 	\|
  			\|  b  \Psi  	\|
  			+ 
  			\|	  [   D(\vp)  ,  [D^*, b ]   ]  	\|	\|	 \Psi	\|	
  			\|	b   \Psi 	\|	
  			\\   
  			&  \lesssim 
  			|\Lambda|		\|	 \vp	\|_{\ell_m^1 }
  			\Big(
  			R 
  			\|	 \po^{\frac{1}{2}} \Psi	\|^2 
  			+ 
  			p_F^{-m}R 
  			\|	 \Psi	\|^2
  			\Big)  \ ,
  		\end{align}
  		where we used 
  		the Type-II estimate for $b$ and  $[D^*, b]$, 
  		the Type-III estimate for $[D(\vp),[D^*,b]		]$,  
  		and  the Type-IV estimate for $D(\vp)$. 
  		\item[$\blacklozenge$] \textit{The third term of \eqref{T F B eq 6}}. 
  		Letting $\Psi \in \F$, 	we find that
  		\begin{align}
  			\nonumber 
  			|
  			\< \Psi  , 
  			[D^*(\vp) ,b^*b] D 
  			\Psi 	\> 
  			|
  			& 	\leq 
  			\|	 \Psi	\| 
  			\| [D^*(\vp) , b^*b ]		\|
  			\|	D\Psi	\|		\\ 
  			\nonumber 
  			& \leq 
  			\Big(
  			\|    b^*		\|  		  \| [D^*(\vp), b ]		\| 
  			+
  			\| [D^*(\vp), b^* ]				\| \|    b  		\|  
  			\Big)
  			\|	\Psi 	\|	
  			\|	 \N \Psi 	\| \\
  			& \lesssim 
  			| \Lambda|
  			\|   \vp	\|_{\ell^1_m}
  			p_F^{-m}
  			R 
  			\|	 \Psi	\|
  			\|	 \N \Psi	\| \ ,
  		\end{align}
  		where we used 
  		the Type-I estimate $D$, 
  		the Type-III estimate for $[D^*(\vp),b]$	and $[D^*[\vp], b^*]$	,  
  		and the norm bound $\|	 b\|\lesssim R $. 
  		
  	\end{enumerate}
  	We gather the three above estimates to find that 
  	the diagonal contribution    satisfies the following upper bound
  	\begin{align}
  		\Big| 	\nu\Big(
  		[ [  N(\vp)  ,  D^*D ] ,  b^*b ] 
  		\Big)  
  		\Big|
  		\ 	\lesssim  \ 
  		|\Lambda	|
  		\|	 \vp  \|_{\ell_m^1}
  		\Big(
  		R 
  		\nu (\N_\S) 
  		+ 
  		\frac{R}{p_F^m}
  		\nu (\N^2)^{\frac{1}{2}}
  		\Big) \ . 
  	\end{align}

  \end{itemize}

The proof of the proposition is finished once we gather the diagonal and off-diagonal contributions
and plug them back in \eqref{T F B eq 1}. 
\end{proof}

  \subsection{Analysis of $T_{FB,F}$}
In this subsection, we analyze the term $T_{FB,F}$.
Our main result is the estimate contained in the next proposition. 
  
  \begin{proposition}
[Analysis of $T_{FB,F}$]
  	\label{prop T FB F}
Let $T_{FB,F}(t,p)$
be the quantity defined in \eqref{T alpha beta}
with $\alpha = FB$ and $\beta = F$, 
and let $m>0$.
Then, there exists a constant $C>0$
such that  for all $\vp \in\ell_m^1$ and $ t \geq 0$  
the following estimate holds true 
  	\begin{align}
 |   		T_{FB,F}(t,\vp) | 
  		&  \,	\leq C   \, 
  		t^2 \|	 \hat V	\|_{\ell^1 }^2 
  		|\Lambda|
  		\|	 \vp	\|_{\ell^1_m}
  		\sup_{ 0 \leq \tau \leq t }
  		\Big(
  		R^\frac{1}{2}
  		\nu_\tau (  \po)^{\frac{1}{2}}
  		+ 
  		p_F^{-m}
  		\nu_\tau (\N^2)^{\frac{1}{2}} 
  		\Big) 
  		\nu_\tau (\N^2)^{\frac{1}{2}}    
  	\end{align}
where we recall 
$ 	T_{FB,F} (t, \vp )  =  \< \vp , T_{FB,F} (t)\>$
and
$R = |\Lambda| p_F^{d-1}$. 
  \end{proposition}
  
  \begin{proof} 
For simplicity, we assume $\vp$ is real-valued--in the general case, 
one may expand into real and imaginary parts and use linearity of the commutators. 
Starting from \eqref{T alpha beta phi} we use the self-adjointness of $V_{FB} (t) $, $V_F(t)$ and $N(\vp) = \int_{  \Lambda^{*} } \vp(p) a_p^*a_p \d p $ to get 
the elementary inequality 
{\blu
	thanks to
	\eqref{eq:N}}
\begin{align}
	\label{starting point T FB F}
		| 	  T_{  FB , F } 
	 	(t   , \vp )
	| 
 &  = 
	\Big|
	\int_0^t 
	\int_0^{t_1}
	\mathrm{Re} 
	\, 
	\nu_{t_2 }
	\Big( 
	[ [   N(\vp) , V_{FB }(t_1) ] ,  V_F(t_2)]
	\Big)  
	\d t_1 \d t_2  
	\Big|   \\
\nonumber 
	& 
	\lesssim 
	t^2 \|	 \hat V	\|_{\ell_1}^2
	\sup_{  k,\ell \in \supp \hat V, t_ i \in [0,t]}
	\Big|
	\nu_{t_2}
	\Big( 
	[ [   N(\vp) ,  D_k^* (t_1)  b_k  (t_1) ] ,   D_\ell^*(t_2) D_\ell (t_2) ]
	\Big)  
	\Big| 
\end{align}
where in the last line we used the representation 
of $V_F(t)$ and $V_B(t)$
in terms of $b$- and $D$-operators found in Eqs. \eqref{VF(t)} and \eqref{VB(t)} --the $D_k^*b^*_{-k}$ term
is re-written in terms of $D^*_k b_k$ upon taking the  real part of $\nu$. 
Next, we estimate the   supremum   in  \eqref{starting point T FB F}. 
In view of Remark \ref{remark states}, it suffices to  provide estimates for pure states $\Psi \in \F$. 
In  order to ease the notation, we omit the variables $t_1,t_2 \in [0,t]$ and  $ k , \ell \in \supp \hat V$.
We  shall  make extensive use of Type-I to Type-IV estimates contained in Lemma \ref{lemma type I}--\ref{lemma type IV}, 
and the commutation relations from Lemmas \ref{lemma N commutators} and \ref{lemma NS commutators}.

\vspace{1mm}

 We expand the first commutator in terms of $D^*(\vp) = [N(\vp) , D ]$  and $b(\vp) = [N(\vp) , b]$
 as follows 
  	\begin{align}
[ [N (\vp) , D^* b ] , D^*D 	]
= 
 [ D^*(\vp) b  , D^*D] 
 + 
  [ D^*  b(\vp)   , D^*D] 
  		  		\label{FB, F eq 1} \  . 
  	\end{align}
We dedicate the rest of the proof to estimate the expectation of the
 two terms {\blu on the right hand side} of \eqref{FB, F eq 1}. 
  	
	\begin{itemize}[leftmargin=*]
   	\item  
  \textit{The first term of \eqref{FB, F eq 1}}
   	We break up the    commutator into three pieces 
   	\begin{align}
   		\label{T FB F eq 1}
   		[D^*(\vp) b  , D^* D]
   		= 
   		D^* (\vp)
   		D^* 
   		[ b,D ] 
   		+ 
   		D^*(\vp) [b,D^*]
   		D
   		+ 
   		[D^*(\vp) , D^* D]
   		b 
   	\end{align}
   	which we now estimate separately.  
	\begin{itemize}[leftmargin=*]
   		\item[$\blacklozenge$] \textit{The first term of \eqref{T FB F eq 1}.} 
   		Letting $\Psi \in \F$, we find that  
   		\begin{align}
   			\nonumber 
   			| 	\< 
   			\Psi  , 
   			D^* (\vp)
   			D^* 
   			[ b,D ] 
   			\Psi 
   			\> | 
   			& 	\ 	\leq  \ 
   			\|	 D D(\vp)	 \Psi \|
   			\|	 [b,D] \Psi 	\|			\\ 
   			& 	\ 	 \lesssim  \ 
   			| \Lambda | 	\|	 \vp	\|_{\ell^1 }
   			R^\frac{1}{2}
   			\|	 \N \Psi	\| 
   			\| \po^{1/2}	 \Psi 	\|  \ ,
   		\end{align}
   		where we used 
   		the Type-I estimate for  $D$, 
   		the Type-II estimate for   $[D, b]$,  
   		the Type-IV estimate for $D(\vp)$,
   		and the commutation relation $[\N , D(\vp)]=0$.

   		\item[$\blacklozenge$] \textit{The second term of \eqref{T FB F eq 1}}.
   		Letting $\Psi \in \F$, we find that  
   		\begin{align}
   			| 	\langle \nonumber 
   			\Psi  , 
   			D^*(\vp) &  [b,D^*]
   			D
   			\Psi 
   			\rangle |  \\ 
   			\nonumber 
   			& 				\ 	\leq  \ 
   			| 	\< 
   			\Psi  , 
   			D^*(\vp)  D [b,D^*]
   			\Psi 
   			\> | 
   			+
   			| 	\< 
   			\Psi  , 
   			D^*(\vp)
   			[ [b,D^*] , 
   			D ] 
   			\Psi 
   			\> |   
   			\\ 
   			& 		\ \lesssim \
   			\nonumber 
   			\|	 D^* D(\vp)	 \Psi \|
   			\|	 [b,D^*] \Psi 	\|
   			+
   			\|		D(\vp)	\| 	\|	 \Psi	\|
   			\|		  [ [b,D^*] , 
   			D ]  \Psi 
   			\| 			\\
   			& \ \lesssim \ 
   			| \Lambda | 	\|	 \vp	\|_{\ell^1 }
   			R^\frac{1}{2}
   			\|	( \N  +\1)  \Psi	\| 
   			\| \po^{1/2}	 \Psi 	\|	\  , 
   		\end{align} 
   		where we used 
   		the Type-I estimate for $D^*$, 
   		the Type-II estimate for $	 [b,D^*]	$ and  $	 [	 [	 b	, D^*	] ,  D	]		$, 
   		the Type-IV estimate for $D(\vp)$,
   		and the commutation relation $[\N, D(\vp)] =0.$
   		
   		\item[$\blacklozenge$] \textit{The third term of \eqref{T FB F eq 1}}.
   		Letting $\Psi \in \F$, we find that  
   		\begin{align}
   			\nonumber 
   			| 	\<   
   			\Psi  , 
   			[D^*(\vp) , D^* D]
   			b 
   			\Psi 
   			\> | 
   			& 	\		\leq \ 
   			\|	 [D(\vp) , D^* D]  \Psi	\|
   			\|	 b \Psi	\|		\\
   			&   			\ \lesssim \
   			| \Lambda | 	\|	 \vp	\|_{\ell^1 }
   			R^\frac{1}{2}
   			\|	 \N  \Psi	\| 
   			\| \po^{1/2}	 \Psi 	\| \   ,
   		\end{align}
   		where we used 
   		the Type-I estimates for $D$ and $D^*$, 
   		the Type-II estimate for $b$, 
   		the Type-IV estimate for $ [ D(\vp) , D]$ and $[D(\vp) , D^* ]$,
   		and the commutation relation $[ \N , [   D(\vp) ,D   ]]=0. $
   		
   	\end{itemize}
   	Upon gathering the last three estimates, we find that the first term  of \eqref{FB, F eq 1}
   	satisfies the following upper bound 
   	\begin{align}
   		\label{T FB F first term bound}
   		| 
   		\nu 
   		\big(
   		[ 	 D^*(\vp) b  , D^* D 	] 
   		\big)  
   		|
   		\lesssim  
   		|\Lambda|	
   		\|	 \vp\|_{\ell_m^1 }
   		R^{\frac{1}{2}}
   		\nu(\N^2 )^{\frac{1}{2}}
   		\nu(\N_\S)^{\frac{1}{2}} \ . 
   	\end{align}
   	
   	\item 
  	  \textit{The second term of \eqref{FB, F eq 1}.}
   	Similarly as before, we break up the    {\blu commutator}  into three pieces 
   	\begin{align}
   		\label{T FB F eq 2}
   		[	D^* b(\vp)	 ,     D^*D]
   		=
   		D^* D^* [b(\vp),D]
   		+ 
   		D^* [ b(\vp)  , D^*] D 
   		+ 
   		b(\vp) [D^* , D^* D] . 
   	\end{align}
   	{\blu These terms can be estimated as follows. } 
   	
	\begin{itemize}[leftmargin=*]
   		\item[$\blacklozenge$] \textit{The first term in \eqref{T FB F eq 2}}. 
   		Letting $\Psi \in \F$, we find that  
   		\begin{align}
   			\nonumber 
   			| 	\< \Psi ,
   			D^* D^* 
   			[b(\vp),D]
   			\Psi  \> | 
   			& 		\leq 
   			\|	 D  D   (\N + 2 )^{ -1 } \Psi	\|
   			\|	 (\N+ 2 ) [ b(\vp)   ,D  ]  \Psi	\| 		\\ 
   			&   			\lesssim 
   			|\Lambda|
   			\|	 \vp	\|_{\ell^1_m}
   			p_F^{-m}
   			\|	\N   \Psi	\|^2  \  , 
   		\end{align}
   		where we used 
   		the Type-I estimate for $D$,
   		the Type-III estimate for $[b (\vp) ,D ]$ 
   		and the pull-through formula $ (\N +2) [b(\vp) ,D ] = [b(\vp), D] \N$.  
   		\item[$\blacklozenge$] 
   		Letting $\Psi \in \F$, we find that 
   		\begin{align}
   			\nonumber 
   			| 	\< \Psi ,
   			D^* 
   			[ b(\vp)  , D^*] 
   			D 
   			\Psi  \> | 
   			& 	\ 	\leq \ 
   			\| D \Psi		\|
   			\|	[ b(\vp)  , D^*] 	\|
   			\|	 D \Psi	\|		\\ 
   			&  \ 	\lesssim  \ 
   			| \Lambda |	\|	 \vp 	\|_{\ell^1_m}
   			p_F^{-m} 
   			\|	 \N \Psi	\|^2 \  , 
   		\end{align} 
   		where we used 
   		the Type-I estimate for $D$,
   		and the Type-III estimate for $[b (\vp) ,D^* ].$ 
   		\item[$\blacklozenge$]
   		Letting $\Psi \in \F$, we find that 
   		\begin{align}
   			\nonumber 
   			| 	\< \Psi ,
   			b(\vp) 
   			[D^* , D^* D]
   			\Psi  \> | 
   			& 	\ 	\leq \ 
   			\|	 b^* (\vp)  \N \Psi	\|
   			\|	[D^* , D^* D]  ( \N +2  )^{-1 }	 \Psi \|			\\ 
   			& 	\ \lesssim   \  
   			| \Lambda |	\|	 \vp 	\|_{\ell^1_m}
   			p_F^{-m} 
   			\|	 \N \Psi	\|^2 \  , 
   		\end{align} 
   		where we used 
   		the Type-I estimate for $[D^*,D^*D]$,
   		the Type-III estimate for $b^* (\vp) $, the pull-through formula $ (\N +2) b(\vp)  = b (\vp)  \N$
   		and the commutation relation $[D^*,\N]=0$. 
   	\end{itemize}
   	
   	Upon gathering the last three estimates, we find that the second term  of \eqref{T F FB eq 1}
   	satisfies the following upper bound 
   	\begin{align}
   		\label{T FB F second term bound}
   		\Big|
   		\nu 
   		\Big(
   		[ D^*   b(\vp) , D^* D ] 
   		\Big)  
   		\Big|
   		\lesssim  
   		|\Lambda|	
   		\|	 \vp\|_{\ell_m^1 }
   		p_F^{-m}\nu(\N^2 ). 
   	\end{align}

   	\item 
  \textit{Conclusion}. The proof of the proposition is finished once we put together the estimates found in Eqs. 
   	\eqref{T FB F first term bound} and \eqref{T FB F second term bound} back in
   	\eqref{FB, F eq 1} . 
   \end{itemize}
  \end{proof}

  \subsection{Analysis of $T_{FB,B}$} 
  
  The main result of this subsection is the following proposition. 
  It contains an estimate on the size of $T_{FB,B}$. 
  
  \begin{proposition}
  	[Analysis of $T_{FB,B}$]
  	\label{prop T FB B}
Let $T_{FB,B}(t,p)$ be the quantity defined in \eqref{T alpha beta} 
with $\alpha = FB$, and $\beta=B$.
Further, let $m>0$.
Then, there exists a constant $C>0$
such that for   all   $\vp \in \ell_m^1$
  	and 
$ t \geq 0 $
such that 
  	\begin{align}
  	 |  
  	 	T_{FB,B}
  		(t,\vp)
  		| 
  		&  \, 	\leq  	\, 
  		C   
  		t^2 \|	 \hat V	\|_{\ell^1 }^2 
  		|\Lambda|
  		\|	 \vp	\|_{\ell^1_m}
  		\sup_{ 0 \leq \tau \leq t }
  		\Big(
  		R^{\frac{3}{2}}  \, 
  		\nu_\tau (  \po)^{\frac{1}{2}}
  		+ 
  		R^2 p_F^{-m}
  		\nu_\tau (\N_1^2)^{\frac{1}{2}} 
  		\Big)  
  	\end{align}
where we recall $  
T_{FB,B}
(t,\vp)  =  \< \vp , T_{FB,B} (t) \>
$  	and  $R =    |\Lambda|  p_F^{d-1}  $. 
  \end{proposition}

  \begin{proof}
For simplicity, we assume $\vp$ is real-valued--in the general case, 
one may expand into real and imaginary parts and use linearity of the commutators. 
Starting from \eqref{T alpha beta phi} we use the self-adjointness of $V_{FB} (t) $, $V_B(t)$ and $N(\vp) = \int_{  \Lambda^{*} } \vp(p) a_p^*a_p \d p $ to get 
the elementary inequality 
\begin{align}
\nonumber 
	| 	  T_{  FB , B} 
	(t   , \vp )
	| 
	&  = 
	\Big|
	\int_0^t 
	\int_0^{t_1}
	\mathrm{Re} 
	\, 
	\nu_{t_2 }
	\Big( 
	[ [   N(\vp) , V_{FB }(t_1) ] ,  V_B (t_2)]
	\Big)  
	\d t_1 \d t_2  
	\Big|   \\
	\label{starting point T FB B}
	& 
	\lesssim 
	t^2 \|	 \hat V	\|_{\ell_1} 
	\sup_{  k \in \supp \hat V, t_ i \in [0,t]}
	\Big|
	\nu_{t_2}
	\Big( 
	[ [   N(\vp) ,  D_k^* (t_1)  b_k  (t_1) ] ,   V_B(t_2) ]
	\Big)  
	\Big| 
\end{align}
where in the last line we used the representation 
of $V_{FB}(t)$ in terms of $b$- and $D$-operators found in \eqref{VFB(t)}--the $D_k^*b^*_{-k}$ term
is re-written in terms of $D^*_k b_k$ upon taking the  real part of $\nu$. 
Next, we estimate the two supremum quantity in  \eqref{starting point T FB B}. 
In view of Remark \ref{remark states}, it suffices to  provide estimates for pure states $\Psi \in \F$. 
In  order to ease the notation, we omit the variables $t_1,t_2 \in [0,t]$. 
We  shall  make extensive use of Type-I to Type-IV estimates contained in Lemma \ref{lemma type I}--\ref{lemma type IV}, 
and the commutation relations from Lemmas \ref{lemma N commutators} and \ref{lemma NS commutators}.

\vspace{1mm}
In terms of $D_k^*(\vp)  =  [ N (\vp) , D_k]$ and $b_k  (\vp)  =  [N(\vp) , b_k]$
we calculate the first commutator to be 
\begin{equation}
  		\label{TFB,B eq 1}
 [	 [   N(\vp) , 	 D^*_k b_k	] , V_B	]	   
	 = 
	[	 D^*_k(\vp) b_k ,    V_B		]	 + 	[	 D^*_k b_k(\vp)  ,    V_B		]	 , \qquad \forall k \in \supp \hat V\ . 
\end{equation}
  	We shall estimate the expectation of the  two terms in \eqref{TFB,B eq 1} separately.

 \begin{itemize}[leftmargin=*]
 	\item   		\textit{The first term of \eqref{TFB,B eq 1}.}
 	We expand $V_B$ into   three additional terms. Namely
 	\begin{align}
 		\nonumber 
 		[  D_k^*(\vp) b_k   , V_B ] 
 		&   = 
 		\int_{  \Lambda^{* } } \hat V (\ell)
 		\Big(
 		[  D_k^*(\vp) b_k   ,  b_\ell^* b_\ell] 
 		+
 		\frac{1}{2}	
 		[  D_k^*(\vp) b_k   ,  b_\ell b_{ - \ell }] 
 		+
 		\frac{1}{2}
 		[  D_k^*(\vp) b_k   ,  b_{-\ell}^* b_\ell^* ] 
 		\Big)
 		\d \ell  \\ 
 		&  \equiv  
 		\int_{  \Lambda^{* } } \hat V (\ell) 
 		\Big( 
 		C_{1} (k,\ell) + C_{2} (k,\ell) + C_{3}(k,\ell) 
 		\Big)  \d \ell  \ . 
 		\label{C1 C2 C3}
 	\end{align}
 	Next, we proceed  	to  analyze  the commutators $C_{j }$ for $ j  =1 ,2,3$  separately.

	\begin{itemize}[leftmargin=*]
 		\item[$\blacklozenge$] \textit{Analysis of $C_1$}. We expand the commutator 
 		\begin{equation}
 			\label{C1,1}
 			C_{1}(k,\ell)
 			= 
 			D_k^* (\vp) [b_k , b_\ell^*] b_\ell
 			+
 			[D_k^*(\vp) , b_\ell^* b_\ell] b_k \ . 
 		\end{equation}
 		Let us recall that the   $[b_k (t) , b_\ell^*(s)]$ 
 		can be calculated explicitly -- see \eqref{boson commutator}. 
 		In particular, it  can be easily verified that for $k,\ell \in \supp \hat V$ it satisfies the   estimate
 		\begin{equation}
 			\label{commutator boson estimate}
 			\|    [  b_k (t),b_\ell^* (s) ]	 \|_{B(\F)}   
 			\lesssim  
 			R
 		\end{equation}
 		Consequently,   $C_{1}$
 		can be estimated as follows.
 		Omitting momentarily the variables $k,\ell \in \supp \hat V$ we find 
 		\begin{align}
 			\nonumber 
 			\big|
 			\<
 			\Psi , 
 			C_{1}  
 			\Psi 
 			\>
 			\big|
 			&   \leq 
 			\big|
 			\<
 			[ b , b^* ] D (\vp) 
 			\Psi ,
 			b
 			\Psi 
 			\>
 			\big|
 			+
 			\big|
 			\<
 			\Psi ,
 			[D^*(\vp) , b^* b] b 
 			\Psi 
 			\>	
 			\big|			\\ 
 			\nonumber 
 			& 
 			\leq 
 			\| [b  , b^*]  \| \, \|   D^* (\vp)	 \Psi	\|
 			\, \| b \Psi		\|
 			+ 
 			\|		[D^*(\vp) , b^* b] 	\|
 			\|	 \Psi	\|		\|	 b \Psi	\|  
 			\\
 			& 
 			\lesssim  
 			R 
 			|\Lambda |		\|	 \vp	\|_{\ell^1 }
 			\|	  \Psi	\|
 			R^{\frac{1}{2} } 
 			\|	 \po^{ \frac{1}{2}  } \Psi	\|
 			+ 
 			| 	 \Lambda| 
 			p_F^{-m}
 			\|	\vp	\|_{\ell^1_m}
 			R^2 
 			\|	 \Psi	\|^2
 			\  ,
 			\label{C1}
 		\end{align}
 		where we used the
 		Type-II estimate for $b$, 
 		the Type-III estimate for $[D^*(\vp), b]$ and $[D^*(\vp) , b^*]$,
 		the Type-IV estimate for $D^*(\vp)$, 
 		the norm bound $\|		b\|		\lesssim R$
 		and the commutator bound \eqref{commutator boson estimate}. 
 		\item[$\blacklozenge$] \textit{Analysis of $C_2$}. 
 		{\blu This term is easier to estimate, as there are no non-zero commutator between the $b$ operators.}  		
 		Namely, 
 		there holds $C_{2} (k, \ell) =   [  D_k^*(\vp)  ,  b_\ell b_{ - \ell }]  b_k .$
 		Thus, we find (omitting the $k,\ell \in \supp \hat V $ variables) 
 		\begin{align}
 			|	 \<  \Psi , C_2 \Psi \>	|
 			\lesssim 
 			| \Lambda|	\|	 \vp	\|_{\ell^1 }
 			\| \hat V		\|_{\ell^1}^2 
 			R^2
 			p_F^{ -m}
 			\|	 \Psi	\|^2 \ . 
 			\label{C2}
 		\end{align}
 		\item[$\blacklozenge$] \textit{Analysis of $C_3$}.
 		This is the most intricate term among the three terms we analyze, because it involves higher-order commutators. 
 		First we decompose 
 		\begin{align}
 			\nonumber 
 			C_{3}(k ,\ell)
 			& 	=	
 			D_k^*(\vp)   b_{-\ell}^*   [   		  b_k ,  b_\ell^* ] 
 			+
 			D_k^*(\vp)     [  		  b_k   ,  b_{-\ell}^*] b_\ell^* 
 			+ 
 			[  D_k^*(\vp)  ,  b_{-\ell}^* b_\ell^* ] 
 			b_k 	\\
 			\label{eq:C3}
 			& \equiv 
 			C_{3,1}(k,\ell)
 			+
 			C_{3,2}(k,\ell)
 			+ 
 			C_{3,3}(k,\ell) \   
 		\end{align}
 		and analyze each term separately.

 		\vspace{1mm}
 		
 		Let us look at the first one.  Omitting the $k,\ell \in  \supp \hat V $ variables we find 
 		\begin{align}
 			\nonumber 
 			\big|
 			\< 
 			\Psi , 
 			C_{3,1}
 			\Psi 
 			\> 
 			\big|
 			& = 
 			\big|
 			\< 
 			b  D (\vp )	 \Psi ,   [   		  b ,  b ^* ] 
 			\Psi 
 			\> 
 			\big|				\\
 			\nonumber 
 			& \leq 
 			\|  b  D (\vp)	 \Psi	\| 
 			\|	 [b  ,b^*] \Psi 	\|		\\
 			\nonumber 
 			& \leq 
 			\|	 [  b , D(\vp)]	\| \|	 \Psi	\|  \|	 [ b, b^*]   \Psi	\|
 			+ 
 			\| D(\vp)		\|	 
 			\|	 b \Psi 	\|
 			\|	 [ b, b^*]   \Psi	\|  \\
 			\nonumber 
 			& \lesssim  
 			( | \Lambda | p_F^{-m} \| \vp		\|_{\ell_m^1}  ) R \|	 \Psi	\|^2 
 			+ 
 			| \Lambda  |    \|	 \vp	\|_{\ell^1 }		R^{\frac{1}{2}} \|	 \po^{\frac{1}{2}} \Psi	\| 
 			R \| \Psi		\|			\\
 			& \leq   
 			|\Lambda | \| \vp	\|_{\ell^1_m }
 			\Big(
 			R^{\frac{3}{2}} 
 			\|	 \Psi	\|	\|	 \po^{ \frac{1}{2} } \Psi	\|
 			+ 
 			R p_F^{-m} \|	 \Psi	\|^2 
 			\Big)  
 		\end{align}
 		where we used 
 		the Type-II estimates for $b$, 
 		the Type-III estimate for $[b, D (\vp)]$, 
 		and the commutator bound  $\|	 [b , b^*]  	\| \leq R $, see Eq. \eqref{commutator boson estimate}. 
 		
 		\vspace{1mm}
 		
 		Let us now look at the second one. 
 		Let us recall that the 
  bosonic commutator  can be written 
 		as $[b_k , b_\ell^*] = \delta(k - \ell) G_k \1 + \calR_{ k , \ell  }$ 
 		where $G_k$ is a scalar, and $\calR_{k,\ell}$ is a remainder operator (see \eqref{commutator of b} for details). 
 		Thus, we find 
 		\begin{align}
 			\nonumber 
 			\big|
 			\langle 
 			\Psi , 
 			C_{3,2}(k ,\ell )
 			\Psi 
 			\rangle 
 			\big|  
 			\nonumber 
 			& \leq 
 			\big|
 			\< 
 			\Psi , 
 			C_{3,1} (k , - \ell)
 			\Psi 
 			\> 
 			\big| 	
 			+ 
 			\big|
 			\< 
 			\Psi , 
 			D_k^*(\vp)     
 			[  \calR_{ k , -  \ell  }   ,   b_{  \ell} ]
 			\Psi 
 			\> 
 			\big| 				\\
 			\nonumber 
 			& \leq 
 			|\Lambda | \| \vp	\|_{\ell^1_m }
 			\Big(
 			R^{\frac{3}{2}} 
 			\|	 \Psi	\|	\|	 \po^{\frac{1}{2}}\Psi	\|			+ 
 			R p_F^{-m} \|	 \Psi	\|^2 
 			\Big)   
 			\\
 			& \quad   	
 			+ 
 			|\Lambda|  \| \vp		\|_{\ell^1 }
 			R^{\frac{1}{2}} \|	 \Psi 	\|		
 			\|   \po^{\frac{1}{2}}	 \Psi	\| 
 		\end{align}
 		where in the last line we used the upper bound for $C_{3,1} (k,\ell )$, 
 		the Type-IV estimate for $D_k^*(\vp)$,
 		and   the  following commutator estimate 
 		\begin{equation}
 			\|	  [  \calR_{ k , \ell  } , b_{- \ell } ]   \Psi	\|
 			\lesssim
 			R^{\frac{1}{2}}
 			\|	 \po^{\frac{1}{2}} \Psi	\| 
 		\end{equation}
 		valid for $k ,\ell \in \supp \hat V$. 
 		
 		\vspace{1mm}
 		
 		Let us now look at the third one. Omitting the $k, \ell \in  \supp \hat V $ variables we find 
 		\begin{align}
 			\big|
 			\< 
 			\Psi , 
 			C_{3,3}
 			\Psi 
 			\> 
 			\big| 	   
 			\leq
 			2
 			\|		b^*	\|
 			\|	  [ D^*(\vp) , b^*  ] 	\|
 			\|	 \Psi 	\|
 			\|	 b \Psi	\|	  
 			\lesssim  
 			|\Lambda|
 			\|	 \vp	\|_{\ell^1_m}
 			R^2
 			p_F^{-m}
 			\|	 \Psi	\|^2 
 		\end{align}
 		where we used the Type-III estimate for $[D^*(\vp) , b^*]$,
 		and the norm bounds $\|	 b \| ,  \|	b^*	\|\lesssim R . $

 		\vspace{1mm}
 		
 		{\blu We can combine   the estimates for $C_{3,1}$, $C_{3,2}$ and $C_{3,3}$ with \eqref{eq:C3}. Namely, we   find that  for all 
 			$k , \ell \in \supp \hat V $ there holds } 
 		\begin{align}
 			\big|
 			\<	\Psi , C_3 (k,\ell) \Psi	\>
 			\big|
 			\lesssim 
 			| \Lambda|	\|	 \vp	\|_{\ell^1  }  
 			\Big(
 			R^{ \frac{3}{2}} 
 			\|	 \Psi	\| 
 			\|		 \po^{\frac{1}{2}} \Psi 	\|
 			+ 
 			R^2
 			p_F^{ -m}
 			\|	 \Psi	\|^2 
 			\Big) \  . 
 			\label{C3}
 		\end{align}
 	\end{itemize}	
 	
 	{\blu  Finally,  	we combine the   estimates that we found for
 		$C_{1}$, $C_{2}$ and $C_{3}$ in 
 		\eqref{C1}, \eqref{C2} and \eqref{C3}, respectively. 
 		More precisely, we find that 
 		the expectation of the first term   in \eqref{TFB,B eq 1}
 		is bounded above by} 
 	\begin{align}
 		\label{T FB B eq 3}
 		\Big|	  
 		\nu 
 		\big(  	   
 		[  D_k^*(\vp) b_k   , V_B ] 
 		\big) 
 		\Big| 
 		\leq 
 		| \Lambda|	\|	 \vp	\|_{\ell^1_m }
 		\| \hat V		\|_{\ell^1} 
 		\Big(
 		R^{ \frac{3}{2}} 
 		\nu(\1)^{\frac{1}{2}}
 		\nu(\po)^{ \frac{1}{2}}
 		+ 
 		R^2
 		p_F^{ -m}
 		\nu(\1)
 		\Big) \ . 
 	\end{align}

 	\item 
 	\textit{The second term of \eqref{TFB,B eq 1}.}
 	This one is easy, we use the {\blu rough} estimate 
 	\begin{align}
 		\label{T FB B eq 2}
 		| \nu 
 		\Big(  	   
 		[  D_k^* b_k (\vp)   , V_B ] 
 		\Big)  | 
 		\leq 
 		| \nu 
 		\big(  	   
 		D_k^* b_k (\vp)   V_B 
 		\big)  | 
 		+ 
 		| \nu 
 		\big(  	   
 		V_B   		 D_k^* b_k (\vp)   
 		\big)  | \ . 
 	\end{align}
 	We estimate these two terms as follows. 
	\begin{itemize}[leftmargin=*]
 		\item[$\blacklozenge$]
 		In view of 
 		$\|	 V_B	\|_{B(\F)}
 		\lesssim 
 		\|	 \hat V	\|_{\ell^1 } 
 		R^2 $
 		we find 
 		for the first term in \eqref{T FB B eq 2} that 
 		\begin{align}
 			\nonumber 
 			\big|
 			\< 
 			\Psi , 
 			D_k^* b_k (\vp)   V_B 
 			\Psi 
 			\> 
 			\big| 	
 			& \leq 
 			\|   b_k^*(\vp) D_k \Psi		\|   
 			\| V_B \Psi 		\| 
 			\\ 
 			& \leq   		\nonumber 
 			\|	 \hat V	\|_{\ell^1 } 
 			\|	 b^*_k (\vp)	\|
 			\|	 \N \Psi	\|
 			R^2 \|	\Psi 	\|  \\
 			&   			\leq 
 			|\Lambda|
 			\|	 \hat V	\|_{\ell^1 } 
 			p_F^{ - m}
 			\|	 \vp	\|_{\ell^1_m}
 			R^2 
 			\|	 \N\Psi	\|\|	 \Psi	\| \ , 
 		\end{align}
 		where we used the Type-I estimate for $D_k$, 
 		and the Type-III estimate for $b_k^*(\vp)$. 
 		
 		\item[$\blacklozenge$]
 		For the second term in \eqref{T FB B eq 2}, 
 		we use the same bound for $V_B$, 
 		together with the 
 		pull through formula $\N b(\vp) = b(\vp) (\N -2 )$
 		to find that 
 		\begin{align}
 			\nonumber 
 			\big| 	
 			\< 
 			\Psi , 
 			V_B   		 D_k^* b_k (\vp)   
 			\Psi 
 			\> 
 			\big| 	
 			&   		\leq 
 			\|	 V_B \Psi	\|
 			\|	 D_k^* (\N+2)^{-1}	\|
 			\|	  ( \N+2)  b_k(\vp) \Psi 	\|			 \\ 
 			\nonumber 
 			& 	\leq 
 			\|	 \hat V	\|_{\ell^1 }   		R^2 \|	 \Psi	\|
 			\|   b_k(\vp)	\N     \Psi	\|	\\
 			&   			\leq 
 			\|	 \hat V	\|_{\ell^1 } 
 			R^2 | \Lambda |	  \|  \vp \|_{\ell_m^1 }		 p_F^{-m}
 			\|	 \Psi	\|		\|  \N    \Psi	\|  \ . 
 		\end{align}
 		Here,   		  we used the Type-I estimate for $D_k^*$, 
 		and the Type-III estimate for
 		$b_k(\vp)$. 
 	\end{itemize}

 	These last two estimates  combined together then imply that 
 	\begin{align}
 		\label{T FB B eq 4}
 		| \nu 
 		\big(  	   
 		[  D_k^* b_k (\vp)   , V_B ] 
 		\big)  | 
 		\leq 
 		\|	 \hat{V}	\|_{\ell^1} 
 		R^2 | \Lambda |	  \|  \vp \|_{\ell_m^1 }		 p_F^{-m}
 		\nu(\1)^{\frac{1}{2}}
 		\nu( \N^2)^{\frac{1}{2}		} \ . 
 	\end{align}

 	\item 
  \textit{Conclusion.} 
 	The proof of the proposition is finished once we gather the estimates  contained in \eqref{T FB B eq 3} and \eqref{T FB B eq 4}, 
 	and plug them back in \eqref{starting point T FB B}.
 	
 \end{itemize}

  \end{proof}

  \subsection{Analysis of $T_{B,\alpha }$}
  Out of the nine terms $T_{\alpha , \beta} (t,\vp)$, 
  those with  $\alpha = B $ are the easiest ones to deal with.
The main result of this subsection is contained in the following proposition.
It contains an estimate for the three terms $T_{B,F}$, $T_{B,FB}$ and $T_{B,B}$. 
  \begin{proposition} 
[Analysis of $T_{B,F}$, $T_{B,FB}$ and $T_{B,B}$]
  	\label{prop T B}
Let $T_{B,F}(t,p)$,
$T_{B,FB}(t,p)$ 
and
$T_{B,FB}(t,p)$
be the quantities defined in \eqref{T alpha beta},
for $\alpha = B$
and $\beta = F$, 
$\beta = FB$
and
$\beta = B$, respectively.
Further, let $m>0$.
Then, there exists a constant $C>0$
such that for all $\vp \in \ell_m^1$
and $t \geq 0$
there holds 
  	\begin{align}
  		\nonumber  
|  		T_{B,F  } (t,\vp)|
+
|  		T_{B,FB  } &  (t,\vp)|
 +
|  		T_{B,B  } (t,\vp)|	\\
 	& 	\leq C  
  		t^2 
  		\|		\hat V 	\|_{\ell^1}^2 
  		| \Lambda |  			\|	 \vp	\|_{ \ell^1_m}
  		R^3 p_F^{-m} 
  		\sup_{0 \leq \tau \leq t}
  		\Big(
  		1 + R^{-2}
  		\nu_\tau (\N^4 )^{\frac{1}{2}	} 
  		\Big) \ , 
  	\end{align}
  	where  we recall 
  	$ T_{\alpha,\beta} (t,\vp) = \<  \vp , T_{\alpha,\beta } (t)  \>   $
  	and 
  	$ R = | \Lambda| p_F^{d-1}$  . 
  \end{proposition}

  \begin{proof} 
In what follows, we let $\alpha$ be either $F$, $FB$ or $B$, 
and we fix $m>0$, $t \geq 0$ and $\vp\in \ell_m^1$.
Starting from \eqref{T alpha beta} one   finds the following elementary bound
\begin{align}
	|	T_{B,\alpha} (t, \vp)	| 
	\lesssim 
	t^2   
	\sup_{t_i\in[ 0 , t] }
	\big|
	\nu_{t_2} 
	\big(
 [ 	[N(\vp)	  , V_B(t_1)	]   , V_\alpha(t_2) ] 
	\big)
	\big|
\end{align}
and so it suffices to estimate the supremum quantity in the above inequality. 
In view of Remark \ref{remark states}, 
it suffices to consider estimates on pure states $\Psi \in \F$.
In order to ease the notation, we drop the time variables $t_1,t_2\in[0,t]$. 
Thus, we find that 
\begin{align}
	|\< \Psi  ,   [   	[N(\vp)	  , V_B  	]   , V_\alpha  ]  \Psi  \>|
	\leq 2 
	|\< \Psi  ,    [ 	[N(\vp)	  , V_B  	]   V_\alpha   \Psi  \>|
	\leq 
	2
	\|		 [N(\vp) , V_B]\|  		\| \Psi		\|		\|	V_\alpha \Psi 	\|
\end{align}
Using the expansion of $V_B$ in terms of $b$-operators (see \eqref{VB(t)}), 
it is straightforward to find that, in  terms of $b_k(\vp)  =   [	N(\vp) ,  b_k 	]$, 
\begin{equation}
	\|  [N(\vp) , V_B]	 	\| 
	\leq 
	 2
	 \|	 \hat V	\|_{\ell^1 }	
	   \|  	 b	 \|  \|	 b_k(\vp) 	\|
	 \lesssim 
\|	 \hat V	\|_{\ell^1 }	
R |\Lambda| p_F^{- m } \|	 \vp	\|_{\ell^1_m}
\end{equation}
where we used     the Type-III estimate on  $b_k(\vp)$ (see Lemma \ref{lemma type III}),
together with the norm bound $\|	b_k	\|\lesssim R$. 
  	On the other hand, we have previously established the estimate 
  	\begin{equation}
\|	 V_\alpha \Psi \|   
  		\lesssim 
  		\|	 \hat V	\|_{\ell^1 }
  		\Big( 
  		\|	 \N^2 \Psi	\|
  		+
  		R^2 	\|	 \Psi	\| 
  		\Big)  \ . 
  	\end{equation}
  	The proof is finished once we gather the last four estimates. 
  \end{proof}

 \section{Proof of Theorem \ref{theorem 1}}
  \label{section proof of theorem}
 We are now ready to give a proof of our main result, Theorem \ref{theorem 1}. 
 We shall make extensive use of the excitation estimates 
 established in Section \ref{section number estimates}.
 Namely, letting $(\nu_t)_{t \in \R}$
 be the interaction dynamics \eqref{interaction dynamics} with initial data satisfying Condition \ref{condition initial data},
 we know that for all $\ell \in \mathbb N $ exists a constant $C>0$ such that 
 for all $ t \geq 0$ there holds 
 \begin{align}
\label{N estimate}
\nu_t( \N^\ell)  & 	\leq C n^\ell \exp(C \lambda R t ) \ , \\
\label{NS estimate}
\nu_t(\N_\S) & \leq (\lambda R \< t\>)^2 \exp ( C  \lambda Rt  ) \ . 
 \end{align}
Here, $ n  = \nu_0(\N) \lesssim R^{1/2}$ is the initial number of particles/holes in the system, and 
  $R = |\Lambda| p_F^{d-1 }$ is our recurrent parameter.

 \begin{proof}  
Throughout the proof, we shall fix the parameter $m>0$.
Let $f_t(p)$ be the momentum distribution of the system, as defined in Def. \ref{definition momentum distribution}. 
In Section \ref{section preliminaries}, 
we performed a   double commutator expansion of $f_t(p)$,  given in \eqref{expansion f}, 
in terms of the quantities $T_{\alpha,\beta} (t,p)$, defined in   Eq. \eqref{T alpha beta}.
It then follows from the triangle inequality that 
for all $t \geq 0 $
\begin{align}
\nonumber 
	\big\|	 f_t - f_0    -   \lambda^2 t  \,  \big(  Q_t[f_0 ]   & +    B_t[f_0] \big) 	\big\|_{\ell_m^{1*}} \\
\nonumber 
& 	\leq  
 \frac{\lambda^2}{ | \Lambda| }
	\Big( 
	\big\|	  T_{F,F} (t) +    t |\Lambda |                  Q_t[f_0 ]   	\big\|_{\ell_m^{1*}}
 +    \big\|	  T_{FB,FB} (t) +   t |\Lambda |B_t[f_0 ]   	\big\|_{\ell_m^{1*}}
 \Big)   \\ 
\nonumber 
& 
 \quad +  \frac{	\lambda^2  }{|\Lambda |} 
\Big(    \|	  T_{F,FB} (t)   	\|_{\ell_m^{1*}}   +       \|	  T_{F,B} (t)  	\|_{\ell_m^{1*}}  \Big) \\ 
\nonumber 
 & 
 \quad +  \frac{	\lambda^2  }{|\Lambda |} 
\Big(    \|	  T_{FB,F} (t)   	\|_{\ell_m^{1*}}   +       \|	  T_{FB,B} (t)  	\|_{\ell_m^{1*}}  \Big) \\ 
 & 
\quad +  \frac{	\lambda^2  }{|\Lambda |} 
\Big(    \|	  T_{B,F} (t)   	\|_{\ell_m^{1*}}   
+ 
      \|	  T_{B,FB} (t)  	\|_{\ell_m^{1*}} 
      + 
        \|	  T_{B,B} (t)  	\|_{\ell_m^{1*}}
       \Big)     \  
       \label{expansion f 2.0}
\end{align}	
where $Q_t$ and $B_t$ 
are the operators defined in Def. \ref{definition Q1} and \ref{definition B}, respectively. 
We shall now estimate the right hand side of \eqref{expansion f 2.0}. 
First, we estimate the leading order terms, previously analyzed in Section \ref{section TFF} and \ref{section TFB}. 
Secondly, we describe the subleading order terms, previously analyzed in Section \ref{section subleading}. \\

 	\noindent  \textsc{Leading order terms.}
 	First, we collect the Boltzmann-like dynamics. 
 	This term emerges from $T_{F,F}$. 
 	Indeed,  it follows from Proposition \ref{prop TFF} and 
Eq. \eqref{N estimate}
 	that
 	there exists  a constant $C   >0$ 
 	such that for all $ t \geq 0 $ 
 	\begin{align}
 		\nonumber 
 		\|	  T_{F,F} (t)  +   t  |\Lambda |   Q_t[f_0]		\|_{\ell_m^{1*}}
 		& 	\leq C 
 		|\Lambda|
 		t^3 
 		\lambda 
 		\sup_{\tau \leq t }
 		\Big(
 		R^2 
 		\nu_\tau (\N^4)^{\frac{1}{2} } 
 		+
 		\nu_\tau (\N^4)
 		\Big) \\ 
 		& 
 		\leq C 
 		|\Lambda|
 		t^3  
 		\lambda 
 		(R^2   + n^2  )  n^2 
 		\exp(C \lambda R t  )
 		\nonumber \\ 
 		& 
 		\leq C
 		|\Lambda|
 		t^3  
 		\lambda 
 		R^2 	 n^2   
 		\exp(C \lambda R t  ) \ , 
 	\end{align}
 	where we have used the assumption  $ n \lesssim  R  $.

 	Now, we collect the interactions between holes/particles and bosonized particle-hole pairs around the Fermi surface. 
 	In view of Proposition \ref{prop TFB} and 
 	Eqs. \eqref{N estimate} and \eqref{NS estimate}
 	we find that there exists a constant  $C >0  $
 	such that for all  $t \geq 0 $  there holds 
 	\begin{align}
 		\nonumber 
 		\| 
 		T_{FB,FB}
 		(t)
 		+ 
 t  		|\Lambda |
 	 & 	    B_t [f_0] 
 		\|_{\ell_m^{1*}}			\\ 
\nonumber 
 		&  \leq C 
 		| \Lambda|     
 		t^2 
 		\sup_{ \tau \leq t }
 		\Big[
 		R^{\frac{1}{2}	}
 		\nu_\tau (\po)^{\frac{1}{2}}
 		\nu_\tau  (\N)^{\frac{1}{2}}
+ 
 		R^{\frac{3}{2}}
 		\nu_\tau (\po)^{\frac{1}{2}} 
 	 	+
 		\frac{R}{ p_F^m}
 		\nu_\tau (\N^2)
 		\Big]	 \\
 		\nonumber 
 		&  \qquad  + 
 		C  
 		| \Lambda|    
 		t^3     
 		\lambda  R  
 		\sup_{ \tau \leq t }
 		\Big[
 		R^{  \frac{3}{2} } \nu_\tau (\po)^{ \frac{1}{2 }}  
 		+ 
 		R \nu_\tau  ( \po )
 		+ 
 		\frac{R}{ p_F^m}
 		\nu_\tau (\N )^{\frac{1}{2 } }
 		\Big]	 \ , 	\\ 
 		\nonumber 
 		&
 		\leq 
 		C 
 		|\Lambda |
 		t^2
 		\Big[ 
 		R^{\frac{1}{2}}
 		\lambda R  \< t \>  n^{\frac{1}{2}}
 		+ 
 		R^{\frac{3}{2}}
 		\lambda R \< t \>  
 		+ 
 		\frac{Rn^2}{p_F^m }	
 		\Big]
 		e^{C \lambda R t } 
 		\\
 		\nonumber 
 		& \qquad  + 
 		C
 		|\Lambda | 
 		t^3 
 		\lambda R  
 		\Big[
 		R^\frac{3}{2}
 		\lambda R \< t   \> 
 		+ 
 		R 
 		( \lambda R  \< t \>  )^2 
 		+ 
 		\frac{R n^{\frac{1}{2}}}{p_F^m}
 		\Big]
 		e^{C \lambda R t } 
 		\ , \\ 
 		\nonumber 
 		& \leq 
 		C
 		| \Lambda|  
 		\Big[
 		t^2 \<t\> \lambda R^{\frac{3}{2}} n^{ \frac{1}{2} }
 		+ 
 		 	 		t^2 \<t\>  \lambda R^{  \frac{5}{2} }
 		+ 
 		t^2 \frac{Rn^2}{p_F^m }	
 		\Big]	
 		e^{C \lambda R t } 
 		\\ 
 		&  \qquad  + 
 		C | \Lambda| 
 		\Big[
 		 		t^3 \<t\> 		 \lambda^2 R^{\frac{7}{2}}  
 		+ 
 		 	 t^3 	\<t\>^2 \lambda^3 R^4 
 		+ 
 		t^3  \frac{\lambda R^2 n^{\frac{1}{2}}}{p_F^m }
 		\Big]
 		e^{C \lambda R t } 
 		\ . 
 	\end{align}
 	Under the assumptions $ 1 \lesssim  n \lesssim  R$  we find the following upper bound,
 	for some constant $C  >0 $.
 	Note that we absorb {\blu polynomials} on the variable $\lambda R \<t\> $ into the exponential factor $\exp(C \lambda R \<t\>)$, after updating the constant $C $. 
 	\begin{align}
 		\nonumber 
 		\| 
 		T_{FB, FB}
 		(t)
 	 & 	+ 
 t   		|\Lambda | 
 		   B_t [f_0] 
 		\|_{\ell_m^{1*}}  \\ 
 		\nonumber 
 		& 	\leq 
 		C |\Lambda|
 		\Big[
 		\lambda t^2  \<t\> R^{ \frac{5}{2} }
 		\Big(
 		1 + \lambda R  		\<t\> + R^{-\frac{1}{2}} (\lambda R  		\<t\> )^2
 		\Big)
 		+ 
 		\frac{t^2   Rn^2 }{p_F^m }
 		\big(
 		1  + \lambda R  t  
 		\big)
 		\Big]
 		e^{C \lambda R \< t\>  }    \\
 		& \leq 
 		C |\Lambda|
 		\Big(
 		\lambda  	 t^2 	\<t\>  R^{\frac{5}{2}}
 		+ 
 		\frac{t^2   Rn^2 }{p_F^m }
 		\Big)
 		e^{C \lambda R \< t\>  } 
 		\ . 
 	\end{align}

 	\noindent  \textsc{Subleading order terms.}
 	In the expansion given by \eqref{expansion f} we have already analyzed the leading order terms given by $T_{F,F}(t)$ and $T_{FB,FB} (t)$. 
 		The remaining seven terms are regarded as subleading order terms. 
 		These can be estimated as follows.

 	Using Proposition \ref{prop T F FB} 
 	and   	Eqs. \eqref{N estimate} and \eqref{NS estimate},  
 	we find
 	that there is  a constant $C>0$ such that 
 	\begin{align}
 		\nonumber 
 		\|  
 		T_{F,FB}
 		(t)
 		\|_{\ell_m^{1* }	}
 		& 	\leq C 
 		t^2  
 		| \Lambda|
 		\sup_{ 0 \leq \tau \leq t }
 		\Big(
 		R^{\frac{1}{2 }	} \, 
 		\nu_\tau 	( \N^2)^{1/2}
 		\nu_\tau  (\po)^{1/2}
 		+ 
 		p_F^{-m}
 		\nu_\tau (\N^2 )
 		\Big) \\ 
 		\nonumber 
 		& 	\leq C 
 		t^2  
 		| \Lambda|
 		\Big(
 		R^{\frac{1}{2 }	} \, 
 		n^2 
 		(\lambda R  		\<t\> )
 		+ 
 		p_F^{-m}
 		n^2 
 		\Big) 
 		e^{C \lambda R t } 
 		\\ 
 		& \leq 
 		C 
 		|\Lambda|
 		\Big(
 		\lambda  	t^2	\<t\> R^{\frac{3}{2}}n^2 
 		+ 
 		\frac{n^2 t^2}{p_F^m }
 		\Big)
 		e^{C \lambda R t } 
 		\ . 
 	\end{align}
 	
 	Using Proposition \ref{prop T F B}
 	 	and   	Eqs. \eqref{N estimate} and \eqref{NS estimate},  
 	we find
 	that there is  a constant $C>0$ such that 
 	\begin{align}
 		\nonumber 
 		\| 	T_{F,B} (t) 
 		\|_{\ell_m^{1* }	}
 		& 	\leq C 
 		t^2 
 		|\Lambda|
 		\sup_{ 0 \leq \tau \leq t }
 		\Big(
 		R^\frac{3}{2} \nu_\tau (\po)^{\frac{1}{2}}
 		+  
 		R 
 		\nu_\tau (\po)
 		+
 		R p_F^{-m}
 		\nu (\N^2)^{	\frac{1}{2}	}
 		\Big)  
 		\\ 
 		\nonumber 
 		& 	\leq C 
 		t^2  
 		| \Lambda|
 		\Big(
 		R^\frac{3}{2}
 		\lambda R  		\<t\>
 		+
 		R (\lambda R  		\<t\> )^2
 		+ 
 		\frac{R n }{p_F^m}
 		\Big)
 		e^{C \lambda R t } 
 		\\ 
 		& 
 		\nonumber 
 		\leq 
 		C 
 		|\Lambda|
 		\Big(
 		\lambda  		t^2	\<t\> R^{\frac{5}{2}}
 		\big(
 		1 + \lambda R^\frac{1}{2}  		\<t\>
 		\big)
 		+ 
 		\frac{R n t^2  }{p_F^m}
 		\Big) 
 		e^{C \lambda R t } 
 		\\ 
 		& \leq 
 		C 
 		|\Lambda|
 		\Big(
 		\lambda  		t^2	\<t\> R^{\frac{5}{2}}
 		+ 
 		\frac{R n t^2  }{p_F^m}
 		\Big) 
 		e^{C \lambda R  		\<t\> }  \ . 
 	\end{align}

 	Using Proposition \ref{prop T FB F} 
 	 	and   	Eqs. \eqref{N estimate} and \eqref{NS estimate},  
 	we find
 	that there is  a constant $C>0$ such that 
 	\begin{align}
 		\nonumber 
 		\|  	T_{FB,F}
 		(t)
 		\|_{\ell_m^{1* }	} 
 		& 	\leq C 
 		t^2 
 		|\Lambda|
 		\sup_{ 0 \leq \tau \leq t }
 		\Big(
 		R^\frac{1}{2}
 		\nu_\tau (  \po)^{\frac{1}{2}}
 		+ 
 		p_F^{-m}
 		\nu_\tau (\N^2)^{\frac{1}{2}} 
 		\Big) 
 		\nu_\tau (\N^2)^{\frac{1}{2}}   
 		\\ 
 		\nonumber 
 		& 	\leq C 
 		t^2  
 		| \Lambda|
 		\Big(
 		R^{1/2}(\lambda R  		\<t\> ) n 
 		+ 
 		\frac{ n^2 }{p_F^m}
 		\Big)
 		e^{C \lambda R t } 
 		\\ 
 		& 
 		\leq 
 		C 
 		|\Lambda|
 		\Big(
 		\lambda  		t^2	\<t\> R^{ \frac{3}{2}} n 
 		+ 
 		\frac{ n^2 t^2  }{p_F^m}
 		\Big) 
 		e^{C \lambda R t } 
 		\ . 
 	\end{align}
 	
 	Using Proposition \ref{prop T FB B} 
 	 	and   	Eqs. \eqref{N estimate} and \eqref{NS estimate},  
 	we find
 	that there is  a constant $C>0$ such that 
 	\begin{align}
 		\nonumber 
 		\|  	T_{FB,B} 
 		(t) 
 		\|_{\ell_m^{1* }	}
 		& 	\leq C 
 		t^2 
 		|\Lambda|
 		\sup_{ 0 \leq \tau \leq t }
 		\Big(
 		R^{\frac{3}{2}}  \, 
 		\nu_\tau (  \po)^{\frac{1}{2}}
 		+ 
 		R^2 p_F^{-m}
 		\nu_\tau (\N^2)^{\frac{1}{2}} 
 		\Big) 
 		\\ 
 		\nonumber 
 		& 	\leq C 
 		t^2  
 		| \Lambda|
 		\Big(
 		R^{\frac{3}{2}}
 		(\lambda R  		\<t\> )
 		+ 
 		\frac{ R^2 n  }{p_F^m}
 		\Big) 
 		e^{C \lambda R t } 
 		\\ 
 		& 
 		\leq 
 		C 
 		|\Lambda|
 		\Big(
 		\lambda  		t^2	\<t\> R^{ \frac{5}{2}} 
 		+ 
 		\frac{ R^2 n  t^2  }{p_F^m}
 		\Big)
 		e^{C \lambda R t } 
 		\ . 
 	\end{align}

 	Using Proposition 
 	 	and   	Eqs. \eqref{N estimate} and \eqref{NS estimate},  
 	we find
 	that there is  a constant $C>0$ such that  
 	\begin{align} 
 		\|  	T_{B  } (t)
 		\|_{\ell_m^{1* }	}
 		& \leq 
 		C 
 		| \Lambda |  		  
 		t^2 
 		R^3 p_F^{-m} 
 		\sup_{0 \leq \tau \leq t}
 		\Big(
 		1 + R^{-2}
 		\nu_\tau (\N^4 )^{\frac{1}{2}	} 
 		\Big) 
 		\leq 
 		C 
 		| \Lambda |  		  
 		t^2 
 		\frac{R^3   }{p_F^m} 
 		e^{C \lambda R t } 
 		\  ,  
 	\end{align}
where we have additionally used the fact that $1 \lesssim n \lesssim R.$  \\

 	\noindent \textsc{Conclusion.}
 	It suffices now to gather all the estimates 
 	for the leading and subleading order terms, and plug them back in the expansion given in Eq. \eqref{expansion f 2.0}
 for the momentum distribution of the system. 
 This finishes the proof of our main theorem. 
 \end{proof}

 \section{Collision operator estimates}
 \label{section fixed volume}
 In this section, we prove the inequalities that were stated in Section \ref{section main results} 
 concerning the 
  three-dimensional torus  $\Lambda $ of fixed length $L>0$. 
 We establish     three lemmas in arbitrary dimension $d \geq1$, 
 and specialize to $d =3 $ when constructing the initial data. 
  
 \subsection{The delta function}
 First, we recall that $\delta_t (x)$
 is the mollified Delta function, defined in \eqref{delta function}.
 Here, we prove the following approximation lemma. 
 \begin{lemma}
 	\label{lemma delta}
 	There is $C>0$ such that for all $  |x| \geq (\frac{2\pi}{L})^2$, 
 	$  |y|  \leq \frac{|x|}{2 \lambda }$
 	and $t>0$
 \begin{equation}
 	 | 	 \delta_t(x + \lambda y )  -  (2/ \pi) t   \delta_{x,0} |  
 	\leq
 	C \, 	 \frac{ 	 (1 - \delta_{x,0})			}{ x^2 }  \frac{1}{t }
 	+    C \delta_{x,0}   \lambda^2  t^3  |y|^2  \ . 
 \end{equation}
 \end{lemma}

\begin{proof}
We consider the decomposition 
	\begin{equation}
\label{delta}
		\delta_t(x + \lambda y ) 
		 =
		  \delta_{x,0}
		 \delta_t( \lambda y ) 
		 + 
		 (1 - \delta_{x,0})
		 \delta_t(x + \lambda y )  \ . 
	\end{equation}
The first term  in \eqref{delta} 
is estimated as follows.
Using $\delta_t(0) = 2 t / \pi $, we find that 
\begin{align}
 |  			
  \delta_t( \lambda y )   
  			  -    2 t / \pi 
  			  | 
  			  =
  			  t |
  			   \delta_1 ( t  \lambda y )  - \delta_1 ( 0 )  
  			   |
  			 \leq  C  t ( t \lambda |y|)^2 \ . 
  			  \end{align}
In the last line,  $C>0 $ is a constant that {\blu satisfies} $ | \delta_1 (z ) - \delta_1(0) | \leq C |z |^2$
for all $z   \in \R $--the constant exists because $\delta_1' (0)=0$, and $\delta_1(z)$ is globally bounded. 
The second term in \eqref{delta} is estimated as follows.
For $|x | \geq 1$  and $\lambda |y| \leq  1/2  $ we have 
\begin{align}
	\delta_t(x + \lambda y )
\leq 
\frac{ 2/\pi }{  t  (x + \lambda y)^2}
\leq 
\frac{ 2/\pi }{  t x^2  (1   -   |x|^{-1 } \lambda |y |)^2}
\leq 
\frac{C}{t x^2  } \ . 
\end{align}
The proof is finished once we put all the inequalities together.
\end{proof}

\subsection{Operator estimates}
Let us now analyze  the time dependence of the operators $Q_t$ and $B_t$.

 \vspace{1mm}
 
Let us recall that $Q_t$ was defined in Def. \ref{definition Q1}, and the time independent operator 
$\mathscr{Q}$ is defined  in the same way, 
but with the discrete Delta function 
$(2/\pi) \delta_\Z (\Delta  e)$ 
replacing the energy mollifier $\delta_t(\Delta E)$.
Here, $\Delta E$ corresponds to the dispersion relation \eqref{dispersion relation},
whereas $\Delta e$ corresponds to  (signed) free dispersion 
$$e(p) = (\chi^\perp(p) - \chi(p) )  \, p^2/2 . $$ 
We shall prove that, under our assumptions for $\hat V$, the following result is true. 
\begin{lemma}
	[Analysis of $Q_t$]
	\label{lemma Q difference}
Assuming that $0< \lambda \|	 \hat V	\|_{\ell_1}\leq 
\frac{1}{2} (	 \frac{2\pi}{L})^2 $, there is $C = C(V)    >0$ such that for all $f \in \ell^1 (	 \Lambda^* 	)$ and $t>0 $ there holds 
\begin{align}
		\|	Q_t[f] - t \mathscr{Q} [f]	\|_{\ell^\infty	 }
	\leq  & C 
	t  \big(1/t^2 + (\lambda t)^2 \big)
	\|	 \widetilde f 	\|_{\ell^\infty }^2
	\|	 f	\|_{\ell^1 } 
	\|	f	\|_{\ell^\infty }    \ . 
\end{align}
\end{lemma}

\begin{proof}
Starting from the definition of  $Q_t[f]$, one finds 
after evaluating the delta functions $\delta(p - p_1) + \delta(p  - p_2 ) - \delta(p - p_3) - \delta(p - p_4)$
that 
\begin{align}
	Q_t[f] - t \mathscr{Q}[f]
	= 
	R_t^+[f] - R_t^- [f]
\end{align}
where on the right hand side we have two remainder terms, corresponding to a gain, and a loss term.
Namely, for $ p \in \Lambda^*$ we have 
\begin{align}
	R_t^+[f] (p) 
	& =  \pi 
		\int 
		\sigma(\vec p ) 
	\Big(
	\delta_t (\Delta E) - 2t/\pi \delta_{\Delta e, 0}
	\Big)
	f(p_3) f(p_4) \widetilde{ f } (p_2) \widetilde f (p) 
  \,	\d p_2 \d p_3 \d p_4 \ , \\ 
		R_t^- [f] (p) 
	& = \pi 
	\int  
	\sigma(\vec p ) 
	\Big(
	\delta_t (\Delta E) - 2t/\pi \delta_{\Delta e, 0}
	\Big)
	 f(p ) f(p_2) \widetilde{ f } (p_3) \widetilde f (p_4) 
 \, 	\d p_2 \d p_3 \d p_4 \ . 
\end{align}
Here, we have denoted $\vec p  = (p , p_2, p_3, p_4)$,  
$\Delta E = E(p) + E(p_2) - E(p_3) - E(p_4)$ 
and
$\Delta e \equiv  \frac{1}{2} (p^2 +  p_2^2  - p_3^2 -p_4^2)$. 
Lemma \ref{lemma delta} with $x=\Delta e$ 
and $y = \mathcal{O}(\|	 \hat V	\|_{\ell^1 })$ now implies that there is $C>0$ such that 
\begin{align}
\label{R+}
	|    R_t^+ [f] (p)  |
 &  \leq  
C (1/t  +  \lambda^2 t^3 \|	 \hat V	\|_{\ell^1}^2  )
\| \widetilde f		\|_{\ell^\infty}^2
	\int  
\sigma(\vec p)
 \, |f(p_3)| \, |  f(p_4)  | \ \d p_2 \d p_3 \d p_4 \  , \\ 
\label{R-}
	|    R_t^- [f] (p)  |
&  \leq  
C (1/t  +  \lambda^2 t^3 \|	 \hat V	\|_{\ell^1}^2  )
\| \widetilde f		\|_{\ell^\infty}^2
	\int 
\sigma(\vec p)
\,  |f(p)| \, |  f(p_2)|   \ \d p_2 \d p_3 \d p_4 \ . 
\end{align}
Next, we consider the following upper bound  
for the coefficients
\begin{align}
\nonumber 
	\sigma(\vec p) 
& 	\leq 
\delta( p + p_2 - p_3 - p_4)
| \hat V (p - p_3) - \hat V (p - p_4)|^2 
+ 
2
\delta( p -  p_2 - p_3  +  p_4)
| \hat V (p - p_3)  |^2   \\ 
\nonumber 
 &  =  \delta( p + p_2 - p_3 - p_4) 
 \Big(
  \hat V (p - p_3)^2 
 + 
  \hat V (p - p_4)^2 
   - 2 \hat V(p - p_3) \hat V(p - p_4)
 \Big) \\
 & 
 \quad + 
 2
 \delta( p -  p_2 - p_3  +  p_4)
 | \hat V (p - p_3)  |^2   \ . 
\end{align}
We insert the above inequality {\blu on the right hand side} of 
\eqref{R+},
and use some elementary manipulations  
to obtain the crude upper bound 
\begin{equation}
	\int 
	\sigma(\vec p )
	| f(p_3)| \, | f( p_4) |\, 
	\d p_2 \d p_3 \d p_4 
	\leq 
	C
	 \|	   \hat V	\|_{\ell^1}  \|	 \hat V	\|_{\ell^\infty } 	
	 	 \|	 f	\|_{\ell^\infty}	\|	f	\|_{\ell^1}      , 
\end{equation}
and the same bound holds for the right hand side of Eq.  \eqref{R-}. 
This finishes the proof after we collect all the estimates, 
 and collect the $\hat V$-dependent  factors into a constant $C>0.$
\end{proof}

Next, we analyze the operator $B_t$, defined in Def. \ref{definition B}, 
and its relation 
to the time independent operator $\mathscr{B}$, defined in the same way but
with $\delta_t(E_1 - E_2 - E_3 - E_4)$
being replaced by 
$  
(2/\pi) \delta_\Z (   e_1 -e_2 - e_3 - e_4).$
While for the operator $Q_t$ an upper bound can be given in terms of
the number of holes 
$n =  |\Lambda| \int  f(p) \d p $, 
the operator $B_t$ depends on the total number of fermions $N$.

\begin{lemma}
	[Analysis of $B_t$]
	\label{lemma B difference}
Assuming that $0< \lambda \|	 \hat V	\|_{\ell_1}\leq 
\frac{1}{2} (	 \frac{2\pi}{L})^2 $, there is $C = C(V)    >0$ such that for all 
$ f \in \ell^1 (	 \Lambda^* 	)$ 
and $t>0 $ there holds  
	\begin{align}
		\|	B_t[f] - t \mathscr{B} [f]	\|_{\ell^\infty}
\ 		\leq   \, C \, 
		t  \,
		\big(   1/ t^2 +  (\lambda t)^2 \big)
\bigg(	 \frac{N}{|\Lambda|}		\bigg)^{ \frac{d-1}{d }}
		\|	 \widetilde f 	\|_{\ell^\infty}  
		\|	f	\|_{\ell^\infty}    \ . 
	\end{align} 
\end{lemma}

\begin{proof}
	Recall that $B = B^{(H)} + B^{(P)}$ is defined  in Def. \ref{definition B}
	in terms of the respective hole and particle interaction terms. 
	Let us look only at the $B^{(H)}$ term, the second one being analogous. 
	We find in terms of $\mathscr B  = \mathscr{B}^{(H)} + \mathscr {B}^{(P)}$
	that for 
$ f \in \ell^1 (	 \Lambda^* 	)$  
	 \begin{align}
 		B_t^{(H)}  [f]    -  t \mathscr B^{(H)} [f]   
 =   L_t [f]   
\end{align}
where we define the following remainder term 
\begin{align}
\nonumber 
	  L_t[f]   (h) 
 & 	  =
	  	2\pi 
	  \int 
	  |\hat V(k)|^2 
\Big( 
 	\rho^{H}_t (h - k ,k)    f(h - k) \widetilde f(h) 
 	-
 			\rho^{H}_t (h,k)    f(h) \widetilde f(h+k ) 
 	\Big) 
 	 \d k    .
\end{align}
Here,  the new  remainder coefficient $\rho_t (h,k)$ are  given by 
	\begin{align}
		\rho_t  
		( h, k)
		&   \equiv 
		\chi(h)\chi(h+k)
	  \int 
		\chi( r )
		\chi^\perp ( r + k )
	\Big( 
		\delta_t 
( 
\widetilde{\Delta E }	  
	) 
		 -  \frac{2t}{\pi} 
		 \delta_\Z   ( \widetilde {\Delta e }	) 
		 \Big)  
		\d r    
	\end{align}
where we denote $\widetilde{\Delta E} = 	E_h - E_{ h + k  }  -  E_r     -  E_{r+k }    $
and 
$\widetilde{\Delta e} =     e_h - e_{h+k} - e_r - e_{r+k }   .$
Thus, it follows from Lemma \ref{lemma delta}
with $x = \widetilde{\Delta e}$
and $|y| \leq \| \hat V		\|_{\ell^1}$
that there is $C>0$ such that 
\begin{align}
\nonumber 
 \|	  B_t  [f]  -  t \mathscr{B}   [f]	  \|_{\ell^\infty }  
\nonumber 
&  \leq 
 C 
 (1 /t +  t^3  \lambda^2 \|	 \hat V	\|_{\ell^1}^2 )
 \|	  \widetilde f	\|_{\ell^\infty}
 \|	f	\|_{\ell^\infty } 
	  \int 
  |\hat V (k)|^2 
 \chi (r) \chi^\perp (r+k ) \d r \d k  \ \\
 & \leq 
  C 
 (1 /t +  t^3  \lambda^2 \|	 \hat V	\|_{\ell^1}^2 )
 \|	  \widetilde f	\|_{\ell^\infty}
 \|	f	\|_{\ell^\infty } 
 \|	 \hat V	\|_{\ell^1}^2
 N^{ \frac{d-1}{d} } \ . 
\end{align}
In the last line, we have used the geometric estimate 
$\int  \chi(r) \chi^\perp(r+k)\d r \lesssim  (N / |\Lambda|)^{\frac{d-1}{d}}$, 
valid for $ k \in \supp  \hat V$.
This finishes the proof after we absorb $\hat V$ into the constant $C>0.$
\end{proof}

 \subsection{Example of Initial Data}
 \label{appendix 1}
 In the remainder of this section, we work in three spatial dimensions $d=3$. 
The inequality contained in     Theorem  \ref{corollary}
 becomes a meaningful approximation for $f_t$
 provided $f_0$ is such that 
 \begin{equation}
 	\label{lower bound}
 	\|	 \mathscr{Q} [f_0]	\|_{\ell_m^{1*}}  
 	+ 
 	\|	 \mathscr B [f_0 ]\|_{\ell_m^{1*}} 
 	\gg  
 	\|	  \t{Rem}_1 (T)	\|_{\ell_m^{1*}}   \  . 
 \end{equation} 

Clearly, we need a lower bound on $\hat V$.
For simplicity, we assume the following. 
Recall $r>0$ from Condition \ref{condition 3}. 
\begin{condition}
	\label{condition 4}
 $\hat V (k)$ is rotationally symmetric and $\hat V (0,0,|k|)>0$
 for all $|k| \leq r$. 
\end{condition}

\vspace{2mm}

 In the rest of this section, we construct examples of initial data $f $
 for which the  lower bound \eqref{lower bound} holds true.
We recall here that  we denote by $\S$ the Fermi surface defined in \eqref{S definition}, 
in terms of the parameter $r>0$. 

\vspace{2mm}

We consider initial data with delta support in the union of the sets, 
with properties that we describe in Condition \ref{condition 3} below.

\begin{definition}
	\label{definition initial data}
Let $ n\geq 1$, and consider
sets   $P  =  \{ p_k\}_{k=1}^n  \subset\B^c /  \S$, 
and   $H =   \{  h_\ell \}_{\ell=1}^n \subset \B / \S $. 
We  define 
\begin{equation}
	f_{H,P}(p) 
	\equiv
	\sum_{ q \in H \cup P  }      \delta (p-q ) \ . 
\end{equation}	
\end{definition}

For simplicity, we shall simply write $f \equiv f_{H,P}$.
Note that one may  easily construct an initial state $\nu : B(\F) \rightarrow \C $ 
with momentum distribution 
 $f  $ by considering the pure state associated to the Slater determinant 
 \begin{equation}
\label{nu}
 	\nu (\mathcal{ O } ) 
 	\equiv 
 	\< \Psi   , \mathcal O \Psi  \>_\F 	 
 	\quad
 	\t{with}
 	\quad 
 	\Psi  \equiv
 	\frac{1}{|\Lambda|}
 	  \prod_{p \in H \cup P  }   a^*_p 	\  \Omega  \ . 
 \end{equation}
As we have already argued in Section \ref{section main results}, 
the state  $\nu$ {\blu satisfies} Condition \ref{condition initial data}. 
We will additionally assume the following support conditions 

\begin{condition}\label{condition 3}
	We assume that the sets $H$ and $P$ satisfy the following two conditions. 
\begin{enumerate}
	\item $|x - y| > r  $ for all pairwise different $x,y \in H \cup P $ 
	\item There exists a  constant    $\ve>0$ such that the following holds: 
	for all $ q \in  H \cup P  $  there exists $i \in \{ 1,2,3\}$  
	such that  
	\begin{equation} 
		 \ve \ p_F^2  \  \leq \  | q_i |^2 \  \leq \  (1 - \ve) \  p_F^2 \  . 
	\end{equation}
\end{enumerate}
\end{condition}

\noindent  \textit{Estimates for $\mathscr{Q}[f]$}. 
The upper bound 
$\|	 \mathscr{Q} [f ]	\|_{\ell^\infty } \leq C    \|	 \widetilde f  	\|_{\ell^\infty }^2 \|	 f 	\|_{\ell^1 }  \|	f 	\|_{\ell^\infty } $
can be established in an analogous way as we did for Lemma \ref{lemma Q difference}. 
Consequently, one easily finds that for all $f $ as in Definition \ref{definition initial data} there holds 
\begin{equation}
 \|	 \mathscr{Q} [f ]	\|_{\ell^\infty }	 \leq C n 
\end{equation}
for a constant $C>0$, independent of $n$.

\vspace{2mm}

\noindent \textit{Estimates for $\mathscr{B}[f ]$}.
Let us recall that  in the present case, the 
operator $\mathscr{B}[f]$ has the  following
decomposition  into holes and particles 
\begin{align}
	 \mathscr{B}  &   =   \mathscr B^{(H)}  +  \mathscr{B}^{(P)} 	\label{B2} 
\end{align}
where 
each of 
   the operators acts  on $\ell^1   (  \Lambda^* )$  as follows 
\begin{align}
\mathscr B^{(H)}[f](h) 
\nonumber 
	&   	 = 
 \frac{2 \pi }{\pi/2 }
\int 
	|\hat V(k)|^2 
	\Big(
	\alpha^{H}_t (h - k ,k)    f(h - k) \widetilde f(h) 
	- 
	\alpha^{H}_t (h,k)    f(h) \widetilde f(h+k ) 
	\Big) \d k 
  \\
	\nonumber 
 \mathscr{B}^{(P)} [f](p) 
	& 	= 
\frac{ 2 \pi }{ \pi/2  }
\int  
	|\hat V(k)|^2 
	\Big(
	\alpha^{P}_t ( p + k ,k)    f( p + k) \widetilde f(p) 
	- 
	\alpha^{P}_t (p,k)    f(p) \widetilde f( p  - k ) 
	\Big) \d k    
\end{align} 
for $f \in \ell^1$ and $p,h\in \Lambda^*$. 
Here,  the coefficients  $\alpha^H$   and  $\alpha^P$ are given as follows  
\begin{align}
	\alpha^H 
	( h, k)
	&   \equiv 
	\chi(h)\chi(h+k)
\int 
	\chi( r )
	\chi^\perp ( r + k )
\delta_\Z  \big(  				 r  \cdot k   -  h \cdot k 		\big) 
 \d r      \\ 
	\alpha^P ( p, k)
	& 	\equiv 
	\chi^\perp (p)\chi^\perp ( p -k)
\int 
	\chi( r )
	\chi^\perp ( r + k )
  \delta_\Z 
  \big( 
r \cdot k   -  (p - k) \cdot k 
\big) \d r   
\end{align}
for all $p,h, k  \in \Lambda^*.$ 
In particular, we have evaluated the free dispersion relation $\Delta e $ in terms
of $p$, $h$ and $k$. 
 
\vspace{2mm}

Certainly, it is sufficient to analyze the counting function defined as 
\begin{align}
	\nonumber 
	N(q,k) &  \equiv 
\int  \chi(r) \chi^\perp(r+k) \delta_\Z    
\big(
	r \cdot k   -  q  \cdot k
	\big)
	\d r 	\\
& 	 = 
	 \frac{1}{|\Lambda|}
	 \Big|
	 \Big\{
	 r \in	 \big(	 \frac{2\pi}{L}\Z	\big)^3			 :  |r | \leq p_F ,  \ |r + k | > p_F , \ r \cdot k = q \cdot k
	 \Big\} 
	 \Big|
\end{align}
for $q \in    \big(	 \frac{2\pi}{L}\Z	\big)^3	  $ and    $ 1  \leq |k| \leq r$ -- where $r \sim |\supp \hat V|$.

\begin{remark}
Geometrically, $N(q,k)$ counts the number of lattice points that lie in the intersection of
the 
\textit{lune set}
$L(k) = \{ r \in \big(	 \frac{2\pi}{L}\Z	\big)^3	:  |r | \leq p_F ,  \ |r + k | > p_F    \}$
and the plane $H(q,k)\equiv \{	 r \in \big(	 \frac{2\pi}{L}\Z	\big)^3	  : r \cdot k = q \cdot k \}$.  
Notice that  $  L(k) \cap H(q,k)$
is nonempty only if
$|q \cdot k | \leq p_F |k|$. 
\end{remark}

 \begin{lemma}[Upper bound for $\mathscr{B} $]
There is a constant $C = C (\hat V)$ such that for all 
$f\in \ell^{\infty} ( \Lambda^* )$ 
the following bound holds true 
\begin{equation}
 \|		\mathscr{B}	[f]	\|_{\ell^\infty    }
 \leq C 	
 (N / |\Lambda|)^{1/3} 
  \|	f	\|_{\ell^\infty    } 	\|	1 - f 	\|_{\ell^\infty  } 
 		  \ . 
\end{equation}
\end{lemma}

\begin{proof}
Let us first given an upper bound for the counting function $N(q,k)$, 
for   $(q,k)$ with $ q \in\big(	 \frac{2\pi}{L}\Z	\big)^3	 $ and $ 1 \leq |k| \leq r $. 
Indeed, let us assume that $ | q \cdot k|  \leq p_F |k|$ for otherwise $N(q,k)=0$.
Then, a standard integral estimate shows that for a constant $C>0$
	\begin{equation} 
		N (q,k) 
		\leq 
		C
		\int_{\R^3 }
		\1 (	 |x| \leq p_F +1  	)
		\1 (|x  + k   | \geq p_F -1 ) 
		\1 ( |x \cdot k - q \cdot k | \leq 1 	 	)  \d x \ . 
	\end{equation}
	We now evaluate the last integral by changing variables so that $x \cdot k = x_3 |k|$.
	Indeed, denoting $\hat k = k/ |k|$  we find  using cylindrical coordinates  
	\begin{align}
		\nonumber 
		N(q,k) 
		& 	\leq 
		C
		\int_{	q \cdot \hat k -1 / |k | }^{q \cdot \hat k + 1/ |k|}
		d x_3
		\int_0^\infty 
		r d  r  
		\1 \Big(
		( p_F -1)^2 - (x_3+|k|)^2  \leq  r^2   \leq (p_F +1)^2 - x_3^2
		\Big)  \\ 
		& \leq C(k) p_F  \ , 
		\label{bound N}
	\end{align}
	where we have evaluated the last integral and used the upper bound $|q \cdot k | \leq p_F |k|$. 
	
	\vspace{2mm}
	
	Going back to the operator $\mathscr{B}[f]$, one may readily find that 
	for a constant $ C >0$ there holds 
	\begin{equation}
		\|	 \mathscr{B} [f]	\|_{\ell^\infty}
		\leq 
		C
		 \|	f 	\|_{\ell^\infty}
		\|1 -f 		\|_{\ell^\infty} 
\int 
			|\hat V (k)|^2 
		\sup_{q \in \Lambda^*  }		N(q,k) \d k 
		\leq 
			C
			\|	\hat V\|_{\ell^2 }^2 
		\|	f 	\|_{\ell^\infty}
		\|1 -f 		\|_{\ell^\infty} 
	 \ 	p_F 
	\end{equation}
	where we have used the bound \eqref{bound N} for the counting function
	in terms of $p_F$. 
	This finishes the proof after we use $p_F \sim (N / |\Lambda|)^{1/3}$
	\end{proof}

In order to give a lower bound for $\mathscr{B} [f]$, 
we take $f$ as in Definition \ref{definition initial data}
satisfying   Condition \ref{condition 3}. 
It turns out that one can  easily calculate the leading order term of the asymptotics of $N(q,k)$
provided  $k $ is chosen parallel to one of the
basis vectors, and $q$ is large enough in this direction.
We do this in the following lemma.

\begin{lemma}[Counting function]
	\label{lemma number}
Let $k = (0,0,  \pm | k|)  \in \supp \hat V $ and  let $q  \in \Lambda^* $ 
 satisfy the lower bound  $  \pm q_3 \geq C p_F^{2/3}$. 
Then, the following asymptotics holds true 
\begin{equation}
	N(q,k ) = 
	 \frac{  q \cdot k}{2 \pi L }
	  \Big(
	1 +  \mathcal O  (p_F^{-1/3})
	\Big) \ , \qquad p_F \rightarrow \infty \ . 
\end{equation} 
The same result holds for $ k  = ( \pm 1,0,0) |k|$ and $k =(0,\pm 1,0) |k|$
provided $  \pm q_1  \geq C p_F^{2/3}$ and $ \pm  q_2 \geq C p_F^{2/3}$. 
\end{lemma}

\begin{remark}
	Before we turn to the proof, let us note that for $ k = (0,0,|k|)$
one can explicitly calculate that 
	\begin{equation}
		\label{number}
		N(q ,k) = 
 |\Lambda|^{-1 }
		\big|
		\big\{		 x \in 	 (\frac{2\pi}{L}\Z)^2 	 : 		 { p_F^2 - ( q_3 +  |k| )^2	}	 < |x|^2 \leq  { p_F^2 -  q_3^2  	}	 	 		 \big\}
		\big|   \ .
	\end{equation}
Note the area of the above annulus is   $\pi (2  q_3 |k|+ |k|^2) >0  .  $ 
Determining the leading order term 
of the  asymptotics of the   counting function  \eqref{number} with $ L = 2\pi$
	is a problem  that has received attention in  other fields; 
	see for instance \cite{BleherLebowitz1994,BleherLebowitz1995,Colzani2021,HughesRudnik2004,Major1992} and the references therein.  
Let us try to explain (informally) why it is, in general,  more challenging than the usual Gauss circle problem. 
Namely, we note that if  $q_3$  is sufficiently small relative to $p_F$
	 (this is the so-called  \textit{thin annulus} situation), 
	 the area of the corresponding annulus may be comparable 
to the remainder term
that comes from lattice point counting. 
 This would be the case for holes $h \in H$ of sufficiently small norm, relative to $p_F$. 
	In our case,  the lower bound $ q_3  \geq C\,   p_F^{2/3}$ 
	(introduced in Condition \ref{condition 3})
	is sufficiently large and the problem is avoided altogether. In other words, we avoid the  thin annulus regime in our analysis.  
\end{remark}

\begin{remark}
In  the proof of Lemma \ref{lemma number} 
we compute the asymptotics of $N(q,k)$ as $p_F \rightarrow \infty$
for particular values of $ k \in  (	 \frac{2\pi}{L}\Z	)^3 $, 
with remainder $o (p_F^{2/3})$. 
These particular values are enough to establish the desired lower bounds on
$\mathscr{B}[f]$, 
for $f$ verifying Condition \ref{condition 3}. 
The asymptotics for \textit{arbitrary} values of 
$ k \in (	 \frac{2\pi}{L}\Z	)^3 $
has also been computed in the literature for $L = 2 \pi$
with remainder $O	\big( 	 \log(p_F)^{2/3} p_F^{2/3}	\big) $. 
See e.g.  \cite[Eq. (B.85)]{Christiansen2023-t}. 

\end{remark}

\begin{proof}[Proof of Lemma \ref{lemma number}]
First, we recall some estimates from the Gauss circle problem.
Namely, let us denote by 
  $n(r) \equiv  | \{ x\in\Z^2 : |x|^2 \leq r^2 \}|$
  the area of the circle $\pi r^2$. 
It is  known that 
the remainder $E(r) \equiv  n(r) -  \pi r^2$ 
satisfies the following bound: 
for all $\ve>0$ there exists $C_\ve$ and $r_\ve>0$ such that 
\begin{equation}
\label{reminder E}
|E(r) | \leq C_\ve  r^{\theta +\ve} \ , \qquad    \forall r \geq r_{\ve}. 
\end{equation}
Here,   $\theta	 =   262/416 < 2/3$ 
is (to the authors best knowledge)
the current best power for the bound \eqref{reminder E}, and is due to Huxley 
\cite{bourgain}. 

\vspace{2mm}

We now assume $ k = (0,0,|k|)$. 
Let us now use \eqref{reminder E} 
with $\theta + \ve < 2/3$. 
Indeed,  as  $p_F \rightarrow\infty$ the area of the annulus is the difference between the area of two concentric circles and one finds 
\begin{align}
\nonumber 
&  N(q,k) \\ 
\nonumber 
 &    =
 \frac{1}{L^3}
  \Big(
   n \Big(  \big(\frac{L}{2\pi}	\big) \sqrt{     p_F^2 - q_3^2  }  	\Big)
   - 
      n \Big( \big(\frac{L}{2\pi}	\big)     \sqrt{  p_F^2 - ( q_3+|k|)^2  }  \Big)
 \Big) \ ,  \\ 
 & = 
  \frac{1}{L^3}
 \nonumber 
 \Big(
\pi  \big(\frac{L}{2\pi}	\big)^2  \, (2  q_3 |k| + |k|^2)
+
    E \Big(  \big(\frac{L}{2\pi}	\big)   \sqrt{p_F^2 - q_3^2}		\Big)
 - 
 E\Big(   \big(\frac{L}{2\pi}	\big)  \sqrt{p_F^2 - (q_3+|k|)^2}		\Big)
 \Big) \ ,  \\ 
   & = 
\frac{1}{2\pi L } q_3 |k|
   \Big(
   1 +
o  \bigg(		\frac{   1 +|k| }{	   p_F^{1/3}	}		\bigg) 
   \Big)   \ . 
\end{align}
Here, we have used the lower bound  
$ \ve p_F^2 \leq p_F^2 - q_3^2 \leq p_F^2$
and similarly for $p_F^2 + (q_3  + |k|)^2$; 
see Condition \ref{condition 3}. 
This finishes the proof in view of 
$q \cdot  k = q_3 |k|$. 
\end{proof}

We are now ready to give a lower bound for the $\mathscr{B}$ operator.
We do this by evaluating the function $\mathscr{B}[f]$ over the points $q\in  ( \frac{2\pi}{L}\Z)^3 $ where $f$ has
non-trivial support. 

\begin{lemma}[Lower bound for $\mathscr{B}$]
	Let $f$ be as in Definition \ref{definition initial data}
	satisfying Condition \ref{condition 3}. 
	Then, there exsists $C_\Lambda>0$ such that for all $ q \in H \cup P$
	there holds 
	\begin{equation}
		|	\mathscr{B} [f] (q)	| \geq C_\Lambda  N^{1/3} \ . 
	\end{equation}
	
\end{lemma}

 \begin{remark}
 	Since $H\cup P$ does not intersect the Fermi surface, the above lemma combined with the  previous 
 	upper bound implies that 
 	\begin{equation}
c_\Lambda   N^{1/3}
 \leq 		\|	   \mathscr{B}[f] 	\|_{\ell^\infty(	\Lambda^* / \S 	)} 	  
 \leq 
 C_\Lambda  N^{1/3} 
 	\end{equation}
for constants $c_\Lambda, C_\Lambda$ depending on the volume $|\Lambda | = L^d$,  and $N \geq 1$ large enough. 
 \end{remark}
 
 \begin{proof}
 	We prove the lemma only for $q = p \in P$ since the proof for $q = h \in H$ is analogous. 
 	To this end, we notice that thanks to Condition \ref{condition 3}, 
 	it holds that  $f (p)=1$  and $f(p+k)=0$  for all $ k \in \supp\hat V$.
 	Consequently, the ``gain term" vanishes and one is left with 
 	a simplified loss term. Namely, there holds 
 	\begin{equation}
 		\mathscr{B}^{(P)}[f](p) 
 		=   - 
4  \int 
 	   |\hat V (k)|^2 
 		N(p - k , k )  \d k 
 	\end{equation}
Let $ i \in\{1,2,3\}$ be the index for which the lower bound holds true 
$|p _i | \geq \ve p_F $. 
Assume without loss of generality that $i=3 $. 
Thanks to our assumption on $\hat V$ given by Condition \ref{condition 4}, 
there exists $ k_* = (0,0,1)|k|$ with $\hat V (k_*)>0 $. 
Then, $|p_3 -  |k| | \geq C p_F $
and we may use Lemma \ref{lemma number} 
with $q = p-k_*$.
Hence, we find that for some $c_\Lambda>0 $
\begin{equation}
	|		 		\mathscr{B}^{(P)}[f](p) 		|
	\ \geq  \ 
	C | \hat V(k_*)	|^2  
	N(p-k_*)
 \ 	\geq \ 
	C | \hat V(k_*)	|^2   
	\big( p_3 -  |k|  \big)
\ \geq  \ 
c_\Lambda N^{1/3}
\end{equation}
This finishes the proof. 
 \end{proof}

\section{Proof of Theorem \ref{thm3}}

The main purpose of this section
is the comparison of the operators $\mathscr{C}$
and $\mathscr{B}$, acting on $\ell^1(\Lambda^*)$.
In what follows, in order to shorten some lengthy expressions, we will use the following notation
for the three-fold lattice sum
\begin{equation}
	\int \equiv \int \d p ' \d p_* \d p_*' 
\end{equation}
which are summing over the pre-collisional variable $p_* \in \Lambda^*$, 
and the post-collisional variables $p', p'_* \in \Lambda^*$. 
We also introduce the measure  
\begin{equation}
\d 	\sigma \equiv 
	b(p p_*  p' p'_* ) \, \d p_* \d p' \d  p_*' 
\end{equation}
where the function $b$ was defined in \eqref{kernel}.
We also recall  the following notations
\begin{equation}
	F = F(p)
	\qquad 
	F_* = F(p_*)
	\qquad 
	F' = F(p')
	\qquad 
	F'_* = F (p'_*)  
\end{equation}
which will be extensively used throughout this section. 

\vspace{1mm}

Let us recall that, with these notations,  the operator $\mathscr{C}$ acts on $\ell^1(\Lambda^*)$  as follows 
	\begin{align}
\label{operatorC}
	\mathscr C  [F] 
	\equiv  
	\int 
	\Big(
	F' F'_* (1 - F) (1 - F_*) - 
	F F_* (1 - F' ) (1 - F_*')
	\Big)   \d \sigma  \ . 
\end{align}  
Similarly,    the operator $\mathscr{B}= \mathscr{B}^{(P)} + \mathscr{B}^{(H)}$
acts on $\ell^1(\Lambda^*)$
and was defined as the limits of the operator $ \frac{1}{t}  B_t$ given in Definition \ref{definition B}.
The effect of  taking the limit is that we replace the approximate  delta function  $\delta_t(\Delta E)$ with $\frac{1}{\pi /2} \delta_\Z ( \Delta E  )$. 
Here, we write them as 
after a change of variables: 
\begin{align}
\nonumber 
		\mathscr B ^{ (P)}  [ f ] 
& 	 =    (1 - \chi )  
	\int 
	(1 - \chi ' )
	\Big(
 \chi_* (1- 	\chi_*') 
  f'	 	 ( 1 - f ) 
 -
\chi_*'		(1- \chi_*	 )
   f 	 	 		 ( 1 - f'  ) 
	\Big)   \d \sigma_P  \\ 
	\label{operatorB}
		\mathscr B ^{ (H )}  [ f ] 
	&  = \chi 
	\int 
\chi' 
	\Big(
 		 \chi_* (1- 	\chi_*') 
 		f' (1 -f )
	-
 \chi_*'		(1- \chi_*	 ) 
    f 	 	 		 ( 1 - f'  ) 
	\Big)   
	\d \sigma _H  \ . 
\end{align}
The measures
are written  as 
\begin{align}
	\label{kernel2}
	\d \sigma_P& = 
	 \frac{2 \pi}{( \pi/2 )  }   \,  \delta (p + p_* - p' - p'_*)
	\delta_{\Z } 
	( p^2 + p_*^2  -  (p')^2  - (p_*) '|^2		  )  
	| \widehat V ( p - p ' )		 			|^2    \ ,  \\
	 \d \sigma_H   & = 
		 \frac{2 \pi}{( \pi/2 ) } 	   \,  \delta (p + p_* - p' - p'_*)
	\delta_{\Z } 
	( p^2 + p_*^2  -  (p')^2  - (p_*) '|^2		  )  
	| \widehat V ( p - p ' )		 			|^2      \ . 
\end{align}

\subsection{Statements}
The main result of this section is the comparison between the operators $\mathscr{C}[F]$
and $\mathscr{ B}[f]$
after one changes variables from the original picture 
to the particle-hole picture. 
We assume that  the perturbation $f$  does not have support around the surface $\S$, in the sense of Condition \ref{condition initial data}.
More precisely:

\begin{proposition}
	\label{prop F}
	Let $\mathscr{C}$ and $\mathscr{B}$ be the operators described by \eqref{operatorC} and  \eqref{operatorB}.
	Let $F,f \in \ell^1(\Lambda^*)$ satisfy $ 0 \leq F, f \leq  1$ and  assume that they are related through 
	\begin{equation}
		\label{Ff}
		F = (1 - \chi) f + \chi (1 -f ) \ . 
	\end{equation}
Furthermore, let $\hat V$ verify Condition \ref{cond potentials}
and assume that $f |_{\S} \equiv 0 $, where $\S $ is the Fermi surface  given by \eqref{S definition}.
Then,   it holds that for all $ p \in \S $
\begin{equation}
	\mathscr{C}[F] = (1 -\chi) \mathscr{B}[f]  - \chi  \mathscr{B}[f] + R [f]
\end{equation}
where $R[f]$ is a remainder that {\blu satisfies} the estimate 
\begin{equation}
\|		R[f]			\|_{\ell^\infty(	 \Lambda^* \backslash \S 	)}		 \leq  C 	\|	 \widehat V	\|_{\ell^2}^2 	 \|	f	\|_{\ell^1 (\Lambda^*)} \ . 
\end{equation}
\end{proposition}

\begin{proof}
We    decompose the operator $\mathscr{C}$ into its particle and hole contributions. 
We then      decompose it further 
into the gain and collision terms. More precisely, we write 
 \begin{equation}
 	\label{C decomposition}
 	\mathscr{C}[F] = \chi \mathscr C^+ [F] 		-\chi \mathscr{C}^- [F]			+  (1 - \chi)	\mathscr{C}^+ [F] - (1 - \chi)	\mathscr{C}^- 	 [F]
 \end{equation} 
and analyze each term separately. 

\vspace{1mm}

\textit{Observation 1}. 
Let us recall that the measure $\d \sigma = b( p p_* p' p_*')\d p_* \d p ' \d p_* ' $
is defined in  \eqref{kernel} in terms of the symmetrized matrix elements
\begin{equation}
	\label{V}
	|	 \widehat V( p - p') -  \widehat V ( p - p_*')	|^2 =   
	 |   \widehat V( p - p')|^2 
	 + 
	 |   \widehat V( p - p_* ')|^2  
	  -2 
	   \widehat V( p - p') 
  \widehat V ( p - p_*') 
  \ . 
\end{equation}
Observe that inside the integral we can  change  variables  and obtain 
$		 |   \widehat V( p - p_* ')|^2   \mapsto 	 |   \widehat V( p - p')|^2   $ for the second term. 
On the other hand, observe that for fixed $ p \in \S $ the product 
$	   \widehat V( p - p') 
\widehat V ( p - p_*') $ forces \textit{all} momenta  to lie close to each other, within a radius $O(1)$
within the Fermi surface. Since $f {|_\S}=0 $ such term drops out. 
Thus, without loss of generality  we assume throughout the proof 
that the measure is determined  by $2 | \hat V( p  - p')|^2$
rather than  \eqref{V}. 
This factor $2$ enters the $\mathscr{ B}$ operators in the kernels \eqref{kernel2}.

\vspace{2mm}
\textit{Observation 2}. Let us make an  observations that will facilitate the analysis of each term.
Recall that we regard $p$ and $p'$ as pre and post-collisional momenta (resp. $p_*$ and $p_*'$)
and that the measure  $\d \sigma $ is proportional to  $|\widehat V( p - p' )|^2$
 (resp.  $| \widehat V( p_* - p'_*)|^2 $  thanks to momentum conservation). 
Recall also that  the map $k \mapsto \hat V(k)$ is supported in a ball of radius $r>0$, 
and  the Fermi surface $\S$ is defined as a neighborhood of the Fermi surface of order $3r$. 
Therefore, under the integral sign $ \int  \d \sigma $ we have that 
the product of $\chi$ functions induce the following restrictions. 
\begin{align}
	\label{obs}
	\chi (1  - \chi ')  & = \chi ( 1- \chi ')  ( \1_\S \times \1_\S )( p ,p') 
\end{align}
and similarly for  $	\chi_*  (1  - \chi '_*) $. 
Since $f  {|_\S } = 0 $ we then find that several combinations simplify. For instance: 
\begin{equation}
	\label{obs2}
	\chi (1 - \chi ' ) f   = 0  
	\qquad
	\t{and}
	\qquad 
		\chi (1 - \chi ' )  (1  - f )   = 		\chi (1 - \chi ' )    \ . 
\end{equation}

\vspace{1mm }

\textit{The first term of \eqref{C decomposition}}. 
Thanks to \eqref{Ff} we have  $\chi  (1-F)  = \chi f  $. Thus, 
\begin{equation}
\chi	\mathscr{C}^+ [F]
  = 
  \int 
 \chi  f  F'  \, F'_*    (1 - F_*)   \d \sigma  \ . 
\end{equation}
Next,  
we can compute 
thanks to  \eqref{Ff}  for $F'$ and the observation \eqref{obs}
\begin{equation}
    \chi f  F ' 
	   \  =     \ 
	   \chi '  \chi  
	   	 	 f   (1 - f' ) 
	  \  	 + \  (1 - \chi ' ) \chi 
	   	 	   f    f ' 
 \ 	   	 	   =		\    \chi '  \chi  
	   	 	   f   (1 - f' )  \ . 
\end{equation}
 Thereby 
\begin{equation} 
\label{C+1}
		\chi	\mathscr{C}^+ [F]
		= 
		\int 
		\chi  \chi '   f  (1 -f ') F'_*   (1 - F_*)  \d \sigma \ . 
\end{equation}
Next, we compute using \eqref{Ff}  the four contributions
of $F_* (1  - F_*)$
which then subsequently simplify thanks to a variation of \eqref{obs2}: 
\begin{align}
\label{C+2}
	F_* ' (1 - F_*)    = 
 \chi_* ' 
 \chi_* 
 ( 1 - f_* ' ) f_* 
 +
 \chi_* ' 
 (1-  \chi_* )   
 +
 ( 1 -	 \chi_* ' ) (1 -  \chi_*)
 f_* '  (1-f_*  )   \ . 
\end{align}
We conclude that by putting \eqref{C+1} and \eqref{C+2} together that 
\begin{align}
	\label{Q1}
	\nonumber 
			\chi	\mathscr{C}^+ [F] 
			&  = 
				\int  
			\chi  \chi ' 	\chi_* ' 
			(1-  \chi_* )   f  (1 -f ') 
		  \d \sigma \\ 
		  	\nonumber 
& 	+ 		 	\int 
			 \chi  \chi '  	 \chi_* ' 
			 \chi_*   f  (1 -f ') 
( 1 - f_* ' ) f_* \d \sigma   \\
& + 	\int 
\chi  \chi ' ( 1 -	 \chi_* ' ) (1 -  \chi_*)    f  (1 -f ') 
f_* '  (1-f_* ) \d \sigma  \ . 
\end{align}
Let us now show that the second  and third terms can be bounded in the  $\ell^\infty$ norm by $\|	f	\|_{\ell^1}$. 
To see this, 
we write $ \d \sigma   \leq C | V( p '- p 	)|^2 \delta(	p + p_*  - p' - p'_*	) \d p_* \d p' \d p_* ' $
and
use the uniform bounds $ 0  \leq \chi_*' , \chi_* , f '  \leq 1 $   to find 
\begin{align}
	\nonumber
\Big| 
		\int 
	\chi  \chi '   f  (1 -f ') 
 	\chi_* ' 
	\chi_* 
	( 1 - f_* ' ) f_*	 \d \sigma  	\Big|   	
& \ 	\leq  \ 
	C 
	\int 
	f_*   |V ( p ' - p)	|^2 \delta(	  p'_* + p ' - p - p_*	)  d p_* d p ' d p_* '    \\
&  \ 	\leq  \  C \|	f\|_{\ell^1} \|	 \widehat V	\|_{\ell^2}^2 \ . 
\label{eq1}
\end{align} 
For the third term, the analogous bound yields 
\begin{align}
\nonumber\Big|  
	\int 
	\chi  \chi '   f  (1 -f ') 
	( 1 -	 \chi_* ' ) (1 -  \chi_*)
	f_* '  (1-f_* )	 \d \sigma 
		 \Big|    
		&  	\leq 			\nonumber
 C \|	f\|_{\ell^1} \|	 \widehat V	\|_{\ell^2}^2 \ . 
 \label{eq2}
\end{align}
We conclude that
\begin{equation}
		\chi	\mathscr{C}^+ [F]  
		 = 
			\int 
		\chi  \chi '  	\chi_* ' 
		(1-  \chi_* )  f  (1 -f ') 
	 	\d \sigma   + R_1 [f]
\end{equation}
where $R_1 [f]$ satisfies the claimed $\ell^1$ estimate. 

\vspace{3mm }

\textit{The second term of   \eqref{C decomposition}}.
 We use \eqref{Ff} for $1 - F' $  and then
 a variation of \eqref{obs2} to find 
\begin{align}
  	\chi	\mathscr{C}^- [F]  
  & = 	\int  
  \chi  \chi '  (1 -f ) f '  F_*
  (1 - F_* ' )\d \sigma   + 
  \int   
  \chi  (1   - \chi ' )      F_*   (1 - F_* ' ) \d \sigma   \  . 
\end{align}  
Next, we expand  $ F_*   (1 - F_* ' )$ with \eqref{Ff} and  then use a variation of \eqref{obs2} to find 
\begin{align}
	F_*  (1 - F_* ' )  
	& = 
	\chi_* 
	\chi_*  ' 
	( 1 - f_*  ) f_*   ' 
	+
	\chi_* 
	(1-  \chi_* '  )  	 
	+
	( 1 -	 \chi_*  ) (1 -  \chi_* ' )
	f_*   (1-f_* '  )   \ . 
\end{align}
The combination of the last two identities
gives a total combination of six terms. 
 Three of them contain  one (and only one) 
 of  the combinations 
 $\chi ( 1 - \chi ' )$, 
 $ \chi ' (1 - \chi )$, 
 $\chi_*  (1 - \chi_* ' )$
 and
 $\chi_*'  (1 - \chi_*  )$. 
 These will be the leading terms. We thus expand the operator and write
 \begin{align}
\nonumber 
 	  & 	\chi	\mathscr{C}^- [F]   \\ 
 	  \nonumber 
 	&   	 = 
 	  	 \int  
 	  	 \chi  \chi '  
 	  	 \chi_*( 1 - \chi_*')
 	  	 (1 -f ) f ' 
 	  	 + 
 	  	 \chi (1 - \chi ')
 	  	 \bigg[
 	  	 	\chi_* 
 	  	 \chi_*  ' 
 	  	 ( 1 - f_*  ) f_*   '  
 	  	 + 
 	  	 	( 1 -	 \chi_*  ) (1 -  \chi_* ' )
 	  	 f_*   (1-f_* '  )  
 	  	 \bigg] \\ 
 	  	 \nonumber 
 	 &  	 + 
 	  	 \int \chi \chi ' \chi_* \chi_* '   (1  -f )f '    (1 - f_*)f_* ' 
 	  	+ 
 	  	 \int \chi \chi ' (1 - \chi_* )(1- \chi_* ')   (1  -f )f '    f_* (1 - f_*' ) \\
&  	  	 + 
 	  	 \int \chi (1 - \chi ') \chi_*(1 - \chi_*' ) \ . 
\label{Q2}
 \end{align}
The two terms in the second line can be combined into a single operator $R_2[f]$, which is  easily estimated    in  $\ell^\infty$ norm 
as in
\eqref{eq1}. 

\vspace{3mm}

\textit{The third term of   \eqref{C decomposition}}.
This term can be computed   analogously  as we did with the  $\chi \mathscr C^{ - }[F]$ term. 
Namely, we use the relations \eqref{Ff} 
to expand various $F$ factors and then use the observation \eqref{obs2} 
to simplify them. 
A straightforward although long  computation shows
\begin{align}
	\label{Q3}
	\nonumber 
	(1 - \chi)\mathscr{C}^+ [F]   
	& 	=
	\int  (1 - \chi)(1 - \chi') 		 \chi_* ' (1 - \chi_*)  f' (1 -f ) \\
		\nonumber 
& 	 \quad + 
	\int   \chi ' (1 - \chi)   \chi_* '     \chi_*   f_* (1 -f_* ' ) \\
		\nonumber 
&  \quad 	+ 
	\int  \chi '   (1 - \chi)     (1 - \chi_* ' ) (1 - \chi_* )f_* ' (1 - f_*) \\
		\nonumber 
& \quad  	+ 
 	\int(1 - \chi) (1 - \chi ' ) \chi_* ' \chi_* f' (1 - f ) f_* (1 - f_* ' ) \\
 		\nonumber 
	&  \quad + 
 	\int(1 - \chi) (1 - \chi ' ) ( 1  - \chi_* ' ) ( 1 -  \chi_*)  f '  (1 -f ) f_* ' (1  - f_* )  \\
& \quad  	 + 
 	 \int  \chi '  (1 - \chi )  \chi_* ' (1 - \chi_* )  \ . 
\end{align}
The first three terms are the corresponding leading order terms, as they contain
one (and only one)
combination of 
$\chi ( 1 - \chi ' )$, 
$ \chi ' (1 - \chi )$, 
$\chi_*  (1 - \chi_* ' )$
and
$\chi_*'  (1 - \chi_*  )$. 
The fourth and fifth term 
can be estimated similarly as in \eqref{eq1}
and will be denoted by $R_3[f]$.

\textit{The fourth term of \eqref{C decomposition}}. We proceed analogously as
we did for the term $\chi \mathscr C^+[F]$. 
We use \eqref{Ff} for $1 - F' $  and then
a variation of \eqref{obs2} to find 
\begin{align}
	\label{Q4}
	\nonumber
(1  - \chi )	\mathscr{C}^-[F] 
&  = 
 \int (1 -\chi) (1 -   \chi ') \chi_* (1 - \chi_* ' ) f (1  - f ' ) \\
 	\nonumber
&   \qquad + 
 \int (1 - \chi ) (1 - \chi') (1 - \chi_* ) (1 - \chi_* ') f (1 - f ') f_* (1 - f_*') \\
  &  \qquad + 
  \int (1 - \chi )(1 - \chi') \chi_* \chi_* '
  f (1 - f')
  f_* ' (1 - f_* ) \  . 
\end{align}
The first term is leading, and the second and third terms are remainder terms, denoted by $R_4[f]$, 
which satisfy the estimate \eqref{eq1}. 

\textit{Conclusion}. 
Next, we put  \eqref{C decomposition}, \eqref{Q1}, \eqref{Q2}, \eqref{Q3} and \eqref{Q4} together.
To this end, we denote $R[f ] \equiv \sum_i R_i[f]$ the sum
of the remainder terms that satisfy \eqref{eq1}. Further, we use  the fact 
that the Fermi ball is a stationary solution of the quantum Boltzmann equation, i.e.
\begin{equation}
	\mathscr C [\chi ] = 0  \ . 
\end{equation}
We thus find  that for all $ p \in \S $
\begin{align}
\nonumber
	\mathscr C [F]
	  =&   
	 \int \chi \chi ' \chi_* ' (1 - \chi_*) f ( 1 - f ' ) \d \sigma  \\
	 \nonumber
& -  \int \chi \chi ' \chi_* (1  - \chi_*') (1 -f )f '  \\
\nonumber
 &  - 
 \int  \chi (1 - \chi')
 \chi_* \chi_* '  (1 - f_*) f '_* \\
 \nonumber
& - 
    \int \chi (1 - \chi')
 (1 - \chi_*) (1 - \chi_* ' )     f_* (1 -  f '_*)  \\
 \nonumber
	& 	+
\int  (1 - \chi)(1 - \chi') 		 \chi_* ' (1 - \chi_*)  f' (1 -f ) \\
\nonumber
\nonumber 
& 	   + 
\int     \chi ' (1 - \chi)   \chi_* '     \chi_*   f_* (1 -f_* ' ) \\
\nonumber
\nonumber 
&    	+ 
\int  \chi   '   (1 - \chi)     (1 - \chi_* ' ) (1 - \chi_* )f_* ' (1 - f_*) \\
& -
 \int (1 -\chi) (1 -   \chi ') \chi_* (1 - \chi_* ' ) f (1  - f ' )   + R[f] \ . 
\end{align}
Finally, we find that when restricting to $ p \in \S$ the factor $\chi (1 - \chi')$ vanishes, thanks to the measure $\d \sigma $ and \eqref{obs2}.
Thus, for all such momenta, 
one finds after re-arranging terms 
\begin{align}
	\nonumber
	\mathscr C [F]
	& 	=
 (1 - \chi ) 	 	\int	(1 - \chi')    \big[		 	 \chi_* ' (1 - \chi_*)  f' (1 -f )  
	- 
   \chi_* (1 - \chi_* ' ) f (1  - f ' ) 
	\big]  \d \sigma  \\
&   \qquad \ \ - 
	 	 \chi  	\int 	\chi '   \big[		
   \chi_* (1  - \chi_*') f '  (1 -f ) 
	-
  \chi_* ' (1 - \chi_*) f ( 1 - f ' )  
	\big] 
	\d \sigma   
 + R[f] \ .
\end{align}
The right hand side can be identified with the $\mathscr{B}$ operators given in \eqref{operatorB}.
 This finishes the proof. 
\end{proof}

\noindent  

\vspace{3mm}

\noindent \textbf{Acknowledgements.}
E.C. is deeply  grateful to Michael Hott for several stimulating discussions regarding quantum Boltzmann dynamics, as well as for    many helpful comments
that helped improve the present manuscript.  
In addition, the authors would like to express their sincere gratitude towards the anonymous referee for careful and thorough reading of an earlier version of this manuscript; their feedback, comments, and suggestions helped us greatly enhance the quality of our article. 

The work of E.C.  was  supported  by  the Provost’s Graduate Excellence Fellowship at The University of Texas at Austin. 
T.C. gratefully acknowledges support by the NSF through grants DMS-1151414 (CAREER), DMS-1716198, DMS-2009800, and the RTG Grant DMS-1840314 {\em Analysis of PDE}.

  \vspace{3mm}
  
\noindent \textbf{Data availability.} 
This manuscript has no associated data.

   \vspace{3mm}
  
\noindent   \textbf{Conflict of interest}. The authors declare that there is no conflict of interest.

\end{document}